\documentclass[10pt]{article}
\usepackage{amsthm,amsfonts,amssymb,euscript,amsmath,comment,slashed,color,enumerate,upgreek,mathtools,float}

\usepackage[affil-it]{authblk}
\usepackage{bbm}
\usepackage{graphics}
\usepackage{enumerate}
\usepackage[margin=1in,dvips]{geometry}
 \DeclareMathOperator*{\esssup}{ess\,sup}
 \usepackage[utf8]{inputenc}

\newcommand{\f}{\frac}
\newcommand{\rd}{\partial}
\newcommand{\nab}{\nabla}
\newcommand{\alp}{\alpha}

\newcommand{\bt}{\beta}

\newcommand{\gi}{(g^{-1})}
\newcommand{\mfg}{\mathfrak g}

\newcommand{\ls}{\lesssim}
\newcommand{\de}{\delta}
 \usepackage{bm}
 \usepackage[hidelinks]{hyperref} 
 	\usepackage[dvipsnames]{xcolor}
 \def\f {\frac}
 \def\rd {\partial}
 \def\ls {\lesssim}
 \def\de {\delta}
 \def\ep {\epsilon}
 \def\i {\infty}
 \def\alp {\alpha}
 \def\bt {\beta}
 \def\Db {\langle D_x \rangle}
  
 \def\la {\langle}
 \def\ra {\rangle}
 \def\th {\theta_k}
 \def\nab {\nabla}
 \def\wo2 {\la x\ra^{-\f r2}}

    \newcommand{\Bes}{  B^{u_k,u_{k'}}_{\infty,1}(\Sigma_t)}
      \newcommand{\fracDu}{ \la  D_{u_k,u_{k'} } \ra  }

\newcommand{\srd}{\slashed{\rd}}
  \newcommand{\partialuk}{  \srd_{{u}_k}}
   \newcommand{\partialukp}{  \srd_{{u}_{k'}}}
      \newcommand{\partialukukp}{  
      \srd_{{u}_{k} {u}_{k'}}^2}
  
    \newcommand{\partialukpukp}{  \srd_{{u}_{k'} {u}_{k'}}^2}
  \def\ls {\lesssim}
  \def\om {\omega}
  \def\th {\theta}
  \def\rd {\partial}
  
  \def\ep {\epsilon}

  \def\nab {\nabla}
  \def\f {\frac}

  \def\i {\infty}
  \def\de {\delta}
  \newcommand{\ud}{\mathrm{d}}
  \def\alp {\alpha}
  \def\bt {\beta}

  \def\mfg {\mathfrak{g}}
  \newcommand{\tphi}{ \widetilde{\phi}_{k}}
  \def\Omg {\Omega} 
  \newcommand{\Sd}{S_{\delta}}
    \newcommand{\sdelta}{\delta^{\frac{1}{2}}}
  
  \def\Db {\langle D_x \rangle}
   
  \def\la {\langle}
  \def\ra {\rangle}
  \def\th {\theta}
  \def\nab {\nabla}
  \def\wo2 {\la x\ra^{-\f r2}}
   \def\gi {(g^{-1})}
   \newcommand{\barg}{(\bar{g}^{-1})^{\nu\color{black}\bt}}

 \theoremstyle{plain}
 \newtheorem{theorem}{Theorem}[section]
\newtheorem{proposition}[theorem]{Proposition}
\newtheorem{lemma}[theorem]{Lemma}

\newtheorem{rmk}[theorem]{Remark}
\newtheorem{definition}[theorem]{Definition}
\newtheorem{corollary}[theorem]{Corollary}

 \newtheorem{thm}[theorem]{Theorem}

 \newtheorem{lem}[theorem]{Lemma}
 \newtheorem{prop}[theorem]{Proposition}
 \newtheorem{cor}[theorem]{Corollary}
 
 \newtheorem{defn}[theorem]{Definition}

 \newcommand{\RR}{\mathbb{R}}

 \newcommand{\Lgeo}{L_k^{geo}}
 \newcommand{\bnabla}{\bar{\nabla}}

 \newcommand{\barL}{X_k}

 \newcommand{\n}{ \vec{n}}
 \newcommand{\rphi}{ \phi_{reg}}

 \newcommand{\T}{ \mathbb{T}_{\mu \nu}}
 \newcommand{\TE}{ \mathbb{T}}

 \newcommand{\tboxone}{\blue{\Box^{1}}}
 \newcommand{\tboxtwo}{\blue{\Box^{2}}}
 \newcommand{\bg}{\breve{g}}

 \newcommand{\pfstep}[1]{\vspace{.5em} {\it \noindent #1.}}

 \newcommand{\blue}[1]{{\color{black} #1}}
 \newcommand{\magenta}[1]{{\color{black} #1}}
 
 \newcommand{\green}[1]{{\color{black} #1}}
 
 \def\f {\frac}
 \def\rd {\partial}
 \def\ls {\lesssim}
 \def\de {\delta}
 \def\ep {\epsilon}
 \def\i {\infty}
 \def\alp {\alpha}
 \def\bt {\beta}
 \def\Db {\langle D_x \rangle}
  
 \def\la {\langle}
 \def\ra {\rangle}
 \def\th {\theta}
 \def\nab {\nabla}
 \def\wo2 {\la x\ra^{-\f r2}}

 \usepackage{color}

 \numberwithin{equation}{section}

\title{Nonlinear interaction of three impulsive gravitational waves II:\\
the wave estimates}

\author{Jonathan Luk\thanks{jluk@stanford.edu}}
\affil{\small  Department of Mathematics, Stanford University, 450~Serra~Mall~Building~380,~Stanford~CA~94305-2125,~United~States~of~America \ }

\author{Maxime Van de Moortel\thanks{mmoortel@princeton.edu}}
\affil{\small  Department of Mathematics, Princeton University, Fine~Hall,~Washington~Road,~Princeton~NJ~08544,~United~States~of~America \ }

\begin{document}

\maketitle

\begin{abstract}
This is the second and last paper of a series aimed at solving the local Cauchy problem for polarized $\mathbb U(1)$ symmetric solutions to the Einstein vacuum equations featuring the nonlinear interaction of three small amplitude impulsive gravitational waves. Such solutions are characterized by their three singular ``wave-fronts'' across which the curvature tensor is allowed to admit a delta singularity.
		
		Under polarized $\mathbb U(1)$ symmetry, the Einstein vacuum equations reduce to the Einstein--scalar field system in $(2+1)$ dimensions. In this paper, we focus on the wave estimates for the scalar field in the reduced system. The scalar field terms are the most singular ones in the problem, with the scalar field only being Lipschitz initially. We use geometric commutators to prove energy estimates which reflect that the singularities are localized, and that the scalar field obeys additional fractional-derivative regularity, as well as regularity along appropriately defined ``good directions''.  The main challenge is to carry out all these estimates using only the low-regularity properties of the metric. Finally, we prove an anisotropic Sobolev embedding lemma, which when combined with our energy estimates shows that the scalar field is everywhere Lipschitz, and that it obeys additional $C^{1,\theta}$ estimates away from the most singular region.

\end{abstract}

 	\tableofcontents

\section{Introduction}

	\paragraph{\blue{The impulsive gravitational waves}.}
	\blue{In this paper and \cite{LVdM1}, our} main goal is to construct and give a precise description of a large class of local solutions to the Einstein vacuum equations
	\begin{equation}\label{EE}
	Ric(^{(4)}g) = 0
	\end{equation} 
	which feature the nonlinear, transversal interaction of \textbf{three impulsive gravitational waves}. An impulsive gravitational wave is a (weak) solution to the Einstein vacuum equations for which the Riemann curvature tensor has a delta singularity supported on a null hypersurface. Interaction of impulsive gravitational waves is then represented by solutions to \eqref{EE} featuring the transversal intersection of such singular hypersurfaces.

	In our work, we impose a polarized $\mathbb U(1)$ symmetry assumption. In other words, we consider a $(3+1)$-dimensional Lorentzian manifold $(I\times \mathbb R^2\times \mathbb S^1\color{black}, ^{(4)}g)$, where $I\subseteq \mathbb R$ is an interval, and assume that the metric takes the following form
	\begin{equation}\label{eq:intro.form.of.metric}
	^{(4)}g= e^{-2\phi}g+ e^{2\phi}(dx^3)^2,
	\end{equation}
	where $\phi:I\times \mathbb R^{2}\to \mathbb R$ is a scalar function and $g$ is a Lorentzian metric on $I\times \mathbb R^{2}$, i.e.~they are independent of the \blue{$\mathbb S^1 = \mathbb R/ (2\pi \mathbb Z)$}-direction\blue{, which we parameterize by the coordinate $x^3$}. The Einstein vacuum equations then reduce to the $(2+1)$-dimensional Einstein--scalar field system
	\begin{equation}\label{eq:Einstein.scalar.field}
	\begin{cases}
	Ric(g) = 2\ud \phi \otimes \ud \phi, \\
	\Box_g \phi = 0.
	\end{cases}
	\end{equation}
	
	The following is an informal version of our main theorem (see \cite[Theorem~5.2]{LVdM1} for a precise statement):
	\begin{theorem}[\blue{Informal main theorem for impulsive gravitational waves}]\label{thm:intro}
		Given a polarized $\mathbb U(1)$-symmetric initial data set corresponding to three (non-degenerate) small-amplitude impulsive gravitational waves propagating towards each other, there exists a weak solution to the Einstein vacuum equations corresponding to the given data up to and beyond the transversal interaction of these waves. In particular, in the solution, the metric is everywhere Lipschitz and is $H^2_{loc}\cap C_{loc}^{1,\th}$ for some $\th \in (0,\f 14)$ away from the three null hypersurfaces corresponding to impulsive gravitational waves. 
	\end{theorem}
	
	\paragraph{\blue{The $\de$-impulsive gravitational} waves.}
	We began the proof of Theorem~\ref{thm:intro} in part I of our series \cite{LVdM1}. We introduced the notion of $\de$-impulsive gravitational waves, which are smooth approximations of the impulsive gravitational waves at a length scale $\de>0$. In our setup, these waves are of small, $O(\ep)$, amplitude, but being $\de$-impulsive means that their second derivatives could be of pointwise size $O(\ep \de^{-1})$ in $\de$-neighborhoods around the null hypersurfaces on which the singularity propagates. They can be viewed as more realistic solutions to \eqref{EE} which are ``quantitatively impulsive'' but without an actual singularity. For this reason, the study of $\de$-impulsive waves is a problem of independent interest that we will also address: we give below an informal version of our result on $\de$-impulsive waves (see \cite[Theorem~5.6]{LVdM1} for a precise statement).
		\begin{theorem}[\blue{Informal main theorem for $\de$-impulsive gravitational waves}]\label{deltathm:intro}
			Given a polarized $\mathbb U(1)$-symmetric initial data set corresponding to three small-amplitude $\de$-impulsive gravitational waves propagating towards each other, there exists a smooth solution to the Einstein vacuum equations corresponding to the given data up to and beyond the transversal interaction of these waves. 
			
			Moreover, for all sufficiently small $\de>0$, \magenta{the following holds:}
			\begin{itemize}
			\item {[Local existence]}. \magenta{The solution exists up to time $1$, independently of $\de$.} 
			\item \green{[Uniform estimates]. T}he solution satisfies $\de$-dependent estimates consistent with $\delta$-approximations of actual impulsive waves\magenta{.}
			\item \magenta{In particular,} the metric is uniformly Lipschitz in $\de$ \blue{everywhere, and obeys uniform-in-$\de$ $H^2\cap C^{1,\theta}$  (for $\theta\in(0,\frac{1}{4})$)  estimates away from the $\de$-impulsive gravitational waves}.
			\end{itemize}
		\end{theorem}
		
		As it turns out, the proof of our main Theorem~\ref{thm:intro} regarding actual impulsive waves reduces to the proof of Theorem \ref{deltathm:intro} on $\de$-impulsive waves. We indeed proved on the one hand in \cite{LVdM1} that given any non-degenerate\footnote{We recall that the non-degeneracy assumption in \cite{LVdM1}  is only used to solve the constraint equations. Moreover, given $O(\ep)$ data, the non-degeneracy assumption can be guaranteed by adding an $O(\ep^{\f 65})$ smooth perturbation; see \cite[Remark~4.7]{LVdM1}.} initial data representing three small amplitude impulsive gravitational waves propagating towards each other, the initial data can be approximated by those for $\de$-impulsive gravitational waves for all small enough $\de>0$. On the other hand, we proved in \cite{LVdM1} via a limiting argument that to any such one-parameter (indexed by $\de$) family of $\de$-impulsive gravitational waves solutions \emph{corresponds an actual impulsive gravitational waves solution}, provided that the $\de$-impulsive waves satisfy specific quantitative estimates for all small $\de>0$.

	 \blue{Because of the above reduction, the remaining goal} is to prove \blue{the quantitative wave estimates for the} $\de$-impulsive waves as stated in Theorem \ref{deltathm:intro}. By the above, this step completes our resolution of the \blue{local} Cauchy problem for three actual impulsive gravitational waves, i.e.~it completes the proof of Theorem \ref{thm:intro}.

	\paragraph{\blue{W}ave estimates \blue{for the $\de$-impulsive waves}.}
	
	\blue{In view of the form of the metric \eqref{eq:intro.form.of.metric} in polarized $\mathbb U(1)$ symmetry, the estimates for the original $(3+1)$-dimensional metric $^{(4)}g$ naturally separate into those for the reduced metric $g$ and for the scalar field $\phi$.
	From now on, we will work in the reduced picture: we will refer to $g$ as the ``metric'' part, and $\phi$ as the ``wave'' part. }
	
	\magenta{In the context of Theorems~\ref{thm:intro} and \ref{deltathm:intro}, the wave part is more singular. Indeed, for an impulsive gravitational wave, $\rd\phi$ has a jump discontinuity across a null hypersurface, while $g$ is more regular. Correspondingly, for a $\de$-impulsive gravitational wave, $|\rd \phi|$ is of size $O(1)$, and  $|\rd^2 \phi|$ is of size $O(\de^{-1})$  in a $\de$-neighborhood of a null hypersurface. Thus, in Theorem~\ref{deltathm:intro}, when we prove that the $(3+1)$-dimensional metric $^{(4)}g$ is uniformly Lipschitz in $\de$ everywhere and obeys uniform-in-$\de$ $H^2\cap C^{1,\th}$ estimates (\green{for} $\th \in (0, \f 14)$) away from the $\de$-impulsive gravitational waves, the main challenge is to prove these bounds for $\phi$.}
	
%The estimates for the $\mathbb U(1)$-reduced equation \eqref{eq:Einstein.scalar.field} naturally separate into those for the (reduced) metric $g$ and for the scalar field $\phi$. 

In part I of our series \cite{LVdM1}, we proved estimates for the metric $g$, as well as for some associated null hypersurfaces, assuming estimates for $\phi$ which are consistent  with  the spacetime having three interacting $\de$-impulsive waves. 

In this paper, we carry out the remaining task, which is to obtain the estimates for $\phi$ assumed in \cite{LVdM1}\blue{, thus closing a bootstrap argument}.  

 In fact, given the estimates in \cite{LVdM1}, and recalling from \eqref{eq:Einstein.scalar.field} that $\phi$ satisfies a linear wave equation,  we can think of this as a statement concerning the linear wave equation with $\de$-impulsive wave data on a background with rough metric. (See Section~\ref{sec:related.rough.coeff} for further discussions.) The following is an informal version of the main result in this paper:
\begin{theorem}[Informal version of the main result in this paper]\label{thm:intro.wave}
    Suppose that
    \begin{itemize}
        \item the initial data for $\phi$ correspond to three small-amplitude $\de$-impulsive gravitational wave propagating towards each other, and
        \item there is a smooth Lorentzian metric $g$   in $[0,T_B)\times \mathbb R^2$ such that the geometric estimates for the reduced $(2+1)$-dimensional metric and null hypersurfaces in \cite{LVdM1} hold.
    \end{itemize}
    
    Suppose $\phi$ is the solution to the linear wave equation
    $\Box_g \phi =0$ with the prescribed data.
    Then, for all sufficiently small $\de>0$, the following holds in $[0,T_B)\times \mathbb R^2$:
    \begin{itemize}
			\item The solution $\phi$ satisfies $\de$-dependent estimates consistent with $\delta$-approximations of actual impulsive waves.
			\item $\phi$ is Lipschitz  uniformly-in-$\de$  everywhere, and obeys uniform-in-$\de$ $H^2\cap C^{1,\th}$ estimates (\green{for} $\th \in (0, \f 14)$) away from the $\de$-impulsive gravitational waves.
			\end{itemize}
\end{theorem}

The precise version of Theorem~\ref{thm:intro.wave} can be found in %\footnote{\blue{Notice that Theorems~\ref{thm:bootstrap.Li} and \ref{thm:energyest} do not explicitly refer to the geometric estimates in \cite{LVdM1}. Nonetheless, in the proof we will indeed first use the bootstrap assumptions and results in \cite{LVdM1} to obtain the geometric estimates; see Section~\ref{sec:partI}.}}
 Theorem~\ref{thm:bootstrap.Li} and Theorem~\ref{thm:energyest}. In particular, we refer the reader  
\begin{itemize}
\item to Section~\ref{section:data_recalling} for the precise assumptions on the initial data of the $\de$-impulsive gravitational waves, 
\item to Section~\ref{sec:partI} for the geometric estimates that we need, and 
\item to Section~\ref{sec:main.estimates} for the precise wave estimates that we prove.
\end{itemize}

According to the results in \cite{LVdM1}, the estimates in Theorem~\ref{thm:intro.wave} complete  the proof of Theorem~\ref{deltathm:intro}.

\color{black}

\paragraph{Comments on the wave estimates.}
 The main issue at stake is that we want to propagate a bound for \blue{$\| \partial \phi\|_{L^{\i}(\RR^2)}$ everywhere and a bound for $\|\rd\phi\|_{C^\th(\RR^2)}$ (for $\th \in (0,\f 14)$) away from the most singular region,} while the initial data of $\phi$ are very rough from the point of view of isotropic $L^2$-based Sobolev spaces. Indeed, \magenta{recall that for an impulsive gravitational wave, $\rd\phi$ initially has a jump discontinuity across a curve. Thus, for the $\de$-impulsive wave,} in terms of isotropic $L^2$-based Sobolev spaces \magenta{$H^s$}, the data for $\phi$ \magenta{only obey the following\footnote{Indeed, it is easy to check that a function with a jump discontinuity along a smooth curve in $\RR^2$ is locally in $H^{\f 1 2 -}(\RR^2) \cap L^{\infty}(\RR^2)$.} $\de$-independent bound:}
 \begin{equation}\label{eq:intro.no.better.bound}
 \| \partial \phi \|_{H^{\f 1 2 -}(\RR^2)} \magenta{\leq \ep.}
 \end{equation}
This is far too weak to control the Lipschitz and H\"older norms  (and is even below the threshold to close the estimates for local existence of the quasilinear problem).

It turns out that in order to close a bootstrap argument, to propagate uniform-in-$\de$ Lipschitz bounds for $\phi$, and to obtain improved H\"older regularity away from the wave fronts, we need to design energies that exploit the specific nature of the $\delta$-impulsive waves. More precisely, we will use the following more subtle ``improved regularity'' in the problem: \begin{enumerate}
	\item \label{a}  \label{c} \emph{[Anisotropy]}. \color{black} We prove that each of the three impulsive waves propagates along specific directions: this property can be proven by differentiating $\phi$ by vector fields tangential to the wave front. 	
	\item \label{b}  \emph{[Hierarchy of $\de$-dependent estimates: ``short pulse bounds'']}. Related to the localization, the solution satisfies a hierarchy of $\de$-dependent bounds involving large and small quantities, in a manner that is similar to Christodoulou's short pulse estimates in \cite{dC2009}. \color{black}
	\item \emph{[Localization]}. \color{black}  The singular parts are initially localized, and we prove that they remain localized in $\de$-neighborhoods of $3$ null hypersurfaces throughout the evolution.
\end{enumerate}

\color{black}
In the energy estimates, it is important that we employ a combination of geometric and fractional derivatives so as to capture the above features. The main challenge for closing these energy estimates is that due to the quasilinear coupling, the metric is of very limited regularity, and we need to propagated the energy bounds for such rough metrics. 

%Having designed our energies to encode geometric and fractional derivatives, we then prove anisotropic Sobolev embedding theorems so that the boundedness of these energies imply 
%\begin{enumerate}[i)]
%	\item the boundedness of the Lipschitz norm $\|\partial \phi\|_{L^{\i}(\RR^2)}$ (and in fact, importantly for \cite{LVdM1}, a slightly stronger Besov norm) \emph{everywhere}, and
%	\item improved $C^{1,\theta}$ estimates, with $\theta\in (0,\frac{1}{4})$, for the scalar field $\phi$ \emph{away from the most singular region.}
%\end{enumerate}  

We will further discuss these estimates and sketch the main ideas of the proof in \textbf{Section~\ref{sec:method}}. After the discussion of the proof, \blue{we will discuss some related works in \textbf{Section~\ref{sec:related.works}}. Finally,} we will outline the remainder of the paper in \textbf{Section~\ref{sec:outline}}.

	\subsection{Ideas of the proof}\label{sec:method}
    
\blue{This section will be organized as follows. 

We begin with the geometric setup in  \textbf{Section~\ref{sec:intro.basic.geometry}}. Then in \textbf{Section~\ref{sec:intro.part.I}}, we briefly recall the estimates for the geometric quantities derived in \cite{LVdM1}.

Turning to the scalar wave, we first introduce in \textbf{Section~\ref{scaling.section}} the regular-singular composition of the scalar wave, which plays an important role in the analysis. Roughly speaking, this decomposes the scalar wave into a regular part and singular parts, where the latter are localized and propagating in specific directions.

We then address the proof of the wave estimates, which is the focus of this paper. Our main wave estimates are $L^2$-based. However, importantly, our $L^2$-based energies are designed so as to obtain the global Lipschitz estimates as well as the improved H\"older bounds away from the singular region (cf.~Theorem~\ref{thm:intro}).
\begin{itemize}
    \item In \textbf{Section~\ref{sec:loc.est}}, we discuss the $L^2$-based estimates up to the second derivative. These estimates already capture particular features of the $\de$-impulsive waves, including  its anisotropy and localization.
    \item In \textbf{Section~\ref{sec:intro.embedding}}, we motivate the various higher order $L^2$ norms that we use by two anisotropic Sobolev-type embedding results. This is related to the Lipschitz and improved H\"older estimates.
    %\item \textbf{Section~\ref{sec:intro.part.I}} is an interlude in which we discuss the estimates for the geometric quantities derived in \cite{LVdM1}. 
    \item Finally, in \textbf{Section~\ref{sec:intro.higher.regularity}}, we explain the ideas in the proof of the higher order $L^2$-based estimates. In particular, we will discuss how the proof of these estimates are intertwined with the control for the geometry that we discussed in Section~\ref{sec:intro.part.I}.
\end{itemize}

}
    
	\subsubsection{The basic geometric setup}\label{sec:intro.basic.geometry}

	\paragraph{Elliptic gauge.} We recall the basic geometric setup in \cite{LVdM1}. First, we construct a solution in an elliptic gauge, i.e.~the ($(2+1)$-dimensional reduced) Lorentzian manifold $(I\times \RR^2),g)$ takes the form $I\times \RR^2 = \underset{t \in  I}{\cup} \Sigma_t$ and
	\begin{equation}\label{eq:elliptic.gauge.intro}
	g = -N^2 dt^2 +e^{2\gamma} \de_{i j} (dx^i + \beta^i dt) (dx^j + \beta^j dt),
	\end{equation}
	where $\de_{ij}$ is the Kronecker symbol\color{black}, the constant-$t$ hypersurfaces $\Sigma_t$ are maximal, and (as a consequence) the metric components $\mfg \in \{N,\,\gamma,\,\bt^i\}$ satisfy semilinear elliptic equations which are schematically of the form
	\begin{equation}\label{eq:mfg.semilinear.elliptic}
	\Delta \mfg = (\rd\phi)^2 + (\rd_x \mfg)^2.
	\end{equation}
		\begin{figure}[H]
		
		\begin{center}
			
			\includegraphics[width=63 mm, height=50 mm]{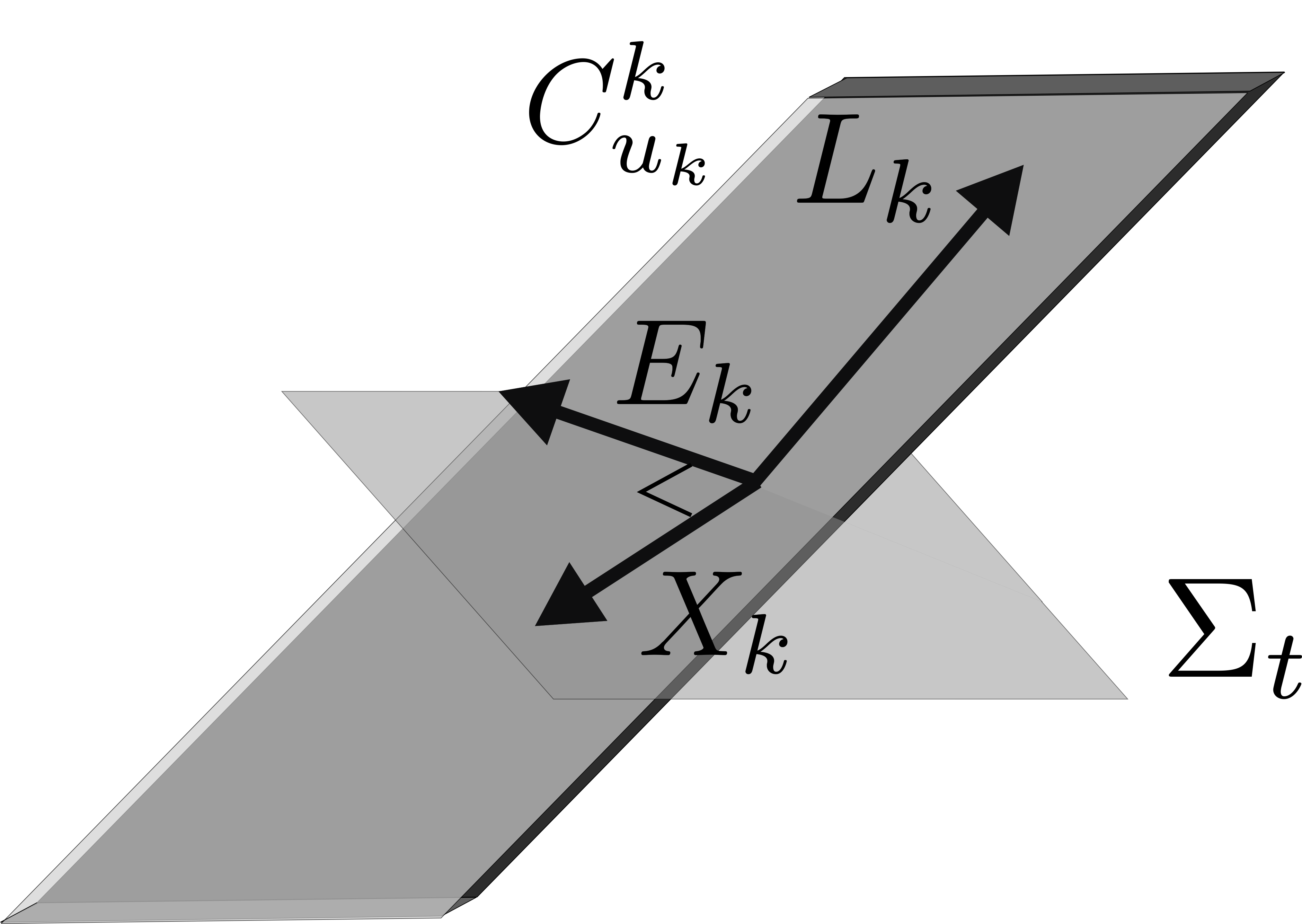}
			
		\end{center}
		\caption{The vector fields $\{L_k, E_k, X_k\}$:  $L_k$ is null, $E_k$, $X_k$ are space-like and tangent to $\Sigma_{t}$. }\label{Fig.null.frame}
	\end{figure}
	
	\paragraph{Eikonal functions and geometric vector fields\blue{.}} In addition to the metric itself, we constructed --- dynamically defined --- eikonal functions $\{u_k\}_{k=1,2,3}$, satisfying $g^{-1}(du_k, du_k) = 0$, which capture the direction of propagation of the $\de$-impulsive gravitational waves. Associated with each eikonal function $u_k$, we constructed a frame of vector fields $\{L_k, E_k, X_k\}$, where $L_k$ and $E_k$ are tangential to the constant-$u_k$ (null) hypersurfaces $C_{u_k}^k$\color{black} and $X_k$ is tangent to $\Sigma_t$ and orthogonal to $E_k$ as depicted in Figure \ref{Fig.null.frame}. These eikonal functions and geometric vector fields are important for capturing the propagation and interaction of the $\de$-impulsive waves, as we will further explain in Sections \ref{scaling.section} and \ref{sec:loc.est} below.

	\subsubsection{Summary of the \blue{geo}metric estimates from part I}\label{sec:intro.part.I} \blue{Continuing our discussion on geometry, we} recall some of the estimates for the geometric quantities that we obtained in \cite{LVdM1}. As we will see, one of the challenges in proving the wave estimates is to contend with the low regularity of the metric. 
	
	Different components of the metric components in the elliptic gauge \eqref{eq:elliptic.gauge.intro} and different derivatives of the Ricci coefficients with respect to the $\{L_k, E_k, X_k\}$ frame obey different bounds. \blue{Especially for the highest order wave estimates, w}e will use the precise bounds for these geometric objects\blue{.} %In what follows, we recall the notation $\mathfrak g$ for the metric components $\{N,\gamma,\beta_i\}$  in gauge \eqref{eq:elliptic.gauge.intro}.
	\begin{enumerate}
		\item For the metric components in the elliptic gauge\blue{, denoted with the schematic notation $\mathfrak g \in \{N,\gamma,\beta_i\}$}, we have the \blue{following} regularity estimate\blue{s for all $R>0$:}
		\begin{equation}\label{eq:intro.metric.bounds}
		\|\rd_i\mathfrak g\|_{W^{1,\infty}\cap W^{1+s',2}(\Sigma_t)}\ls \ep^2,\quad \|\rd_t\mathfrak g\|_{W^{1,\f{2}{s'-s''}}(\Sigma_t \cap B(0,R))}\ls_R \ep^2,
		\end{equation}
		where $0< s'' < s' < \f 12$ are fixed but arbitrary parameters, to be explained later.
		
		Note that no estimates were obtained for $\rd_t^2 \mfg$.
		\item The Ricci coefficients $\chi_k := g(\nabla_{E_k} L_k, E_k)$ and $\eta_k := g(\nabla_{X_k} L_k, E_k)$ associated to the null frame $\{L_k, E_k, X_k\}$
		are considerably less regular. Denote $\kappa_k\in \{\chi_k, \,\eta_k\}$, and introduce coordinates $(t_k,u_k,\th_k)$ with $u_k$ the eikonal function {from Section \ref{sec:intro.basic.geometry}}, $t_k = t$ and $\th_k$ such that $L_k \th_k = 0$. Then \cite{LVdM1} gives
		\begin{equation}\label{eq:intro.kappa}
		\sum_{\kappa_k \in \{\chi_k,\eta_k\}} \Big(\color{black}\|\kappa_k \|_{L^\i(\Sigma_t)} + \|L_k \kappa_k \|_{L^\i(\Sigma_t)}  \Big)\color{black}\ls \ep^2.
		\end{equation}
		
		Observe that $L_k \kappa_k$ is estimated at the same regularity as $\kappa_k$: this is because $\kappa_k$ satisfies a transport equation in the $L_k$ direction due to the Einstein equations (see \cite[Lemma~2.22]{LVdM1}).
		
		{The other $E_k$ and $X_k$ derivatives are less regular and only obey mixed $L^2/L^{\i}$ or $L^2$ bounds:}
		\begin{equation}\label{eq:intro.Ricci.worst}
		\|E_k \kappa_k\|_{L^\i_{t_k} L^\i_{u_k}  L^2_{\th_k}} +\|X_k \chi_k\|_{L^\i_{t_k} L^\i_{u_k} L^2_{\th_k}} \ls \ep^2,\quad \|X_k \eta_k\|_{L^2(\Sigma_t\cap B(0,R))} \ls_R \ep^2.
		\end{equation}
		
		{Note that} $X_k \chi_k$ obeys a similar bound as $E_k\kappa_k$, but to bound $X_k \eta_k$, we need $L^2$ in both $u_k$ and $\theta_k$.
		
		To obtain higher order estimates, we are only allowed to commute with an extra $L_k$ derivative:
		\begin{equation}\label{eq:Lkappa.even.better}
		\|L_k X_k \chi_k \|_{L^2(\Sigma_t \cap B(0,R))},\quad \smashoperator{\sum_{\kappa_k \in \{\chi_k,\eta_k\}}} \| L_k^2 \kappa_k \|_{L^2(\Sigma_t \cap B(0,R))},\quad \smashoperator{\sum_{\kappa_k \in \{\chi_k,\eta_k\}}} \|L_k E_k \kappa_k \|_{L^2(\Sigma_t \cap B(0,R))} \ls_R \ep^2.
		\end{equation}
		Notice that as in \eqref{eq:intro.Ricci.worst} $\eta_k$ obeys slightly weaker bounds than $\chi_k$ {and moreover there is no estimate to control $L_k X_k \eta_k$.}
	\end{enumerate}
	In general, the derivatives of $\mathfrak g$ obey better bounds than $\chi_k$, $\eta_k$ (due to ellipticity of \eqref{eq:mfg.semilinear.elliptic}). However, spatial ellipticity does not merge well with $\rd_t$ derivatives: $\rd_t \mfg$ only obeys weaker bounds, and $\rd_t^2 \mfg$ is not controlled in our argument at all. On the other hand, while $\chi_k$, $\eta_k$ obeys weaker bounds, they behave better with respect to $L_k$ derivatives (which contains a $\rd_t$ component); see \eqref{eq:intro.kappa}, \eqref{eq:Lkappa.even.better}. (Additionally, one needs to control various non-trivial commutators when going back and forth between (1) the eikonal quantities constructed with $L_k$ and $E_k$ and (2) the metric coefficients in the elliptic gauge \eqref{eq:elliptic.gauge.intro}. We will not get into details here, except for remarking that they can be controlled using the geometric estimates in \cite{LVdM1}.)
	
	\subsubsection{Regular-singular decomposition and the singular zones} \label{scaling.section} 
	We will impose that the $\de$-impulsive waves are of small amplitude $\ep>0$. The length scale $\de$ at which each $\de$-impulsive wave is localized is required to satisfy $0< \de \ll \ep$.
	
	We begin by decomposing $\phi$ into a regular and three singular parts. This is achieved by solving an auxiliary characteristic-Cauchy problem so that
	$$\phi = \rphi + \sum_{k=1}^3 \tphi, \text{ where }\Box_g \rphi=0 \text{ and }\Box_g \tphi=0 \text{ for } k=1,2,3,$$ 
	where each $\tphi$ corresponds to a $\de$-impulsive wave propagating along the constant-$u_k$ null hypersurfaces $C^k_{u_k}$ and $\rphi$ is an error term which is more regular. 
	The part $\phi_{reg}$ is regular everywhere in the sense that \begin{equation*} 
		\| \rphi \|_{H^{2+s'}(\RR^2)} \ls \ep,
		\end{equation*} for some $s' \in (0, \f 12)$; see Section \ref{sec:rphi}. The remainder of Section~\ref{sec:method} will thus be devoted to the discussion of  the singular parts $\tphi$.
		
		Each $\tphi$ is initially regular away from the region $\{-\de \leq u_k \leq 0\}$ and is in fact constructed to \emph{vanish} for $u_k \leq -\de$.  In the region $\{-\de \leq u_k \leq 0\}$,   the \color{black} first and second derivatives of $\tphi$ only obey  initially \color{black} the   following schematic \color{black}  bounds
		\begin{equation} \label{delta.intro}
	|\partial \tphi | \ls \ep,\ |\partial^2 \tphi| \ls \ep \de^{-1}\magenta{.}
		\end{equation} (Notice that these are exactly the size estimates one obtains by smoothing out at a scale $u_k \approx \de$ an initial function $\phi_{rough}$ of amplitude $\ep$ whose generic first derivatives $\rd \phi_{rough}$ have a jump continuity across the curve given by $\{u_k = 0\}$ and whose generic second (distributional) derivatives $\rd^2 \phi_{rough}$ have a delta singularity supported on $\{u_k = 0\}$.) Because of \eqref{delta.intro}, $\tphi$ is initially no better than $\|\tphi \|_{H^2(\Sigma_0)} \ls \ep \de^{-\f 12}$ and, in terms of $L^2$-based Sobolev spaces, it is only the $\|\rd \tphi \|_{H^s(\Sigma_0)}$ norms, for $s< \f 12$, that obey the uniform-in-$\de$ bounds $\|\partial \tphi \|_{H^{s}(\Sigma_0)} \ls_s \ep$.

		An important use of the dynamically constructed eikonal functions that we mentioned earlier is they can track the location of singularities. For each $k=1,2,3$, define the corresponding singular zone by 
		\begin{equation} \label{S.def.intro}
		S^k_\de := \{ -\de \leq u_k \leq \de\}
		\end{equation} 
		(slightly larger than the initial singular zone $\{-\de \leq u_k \leq 0\}$), measured with respect to the eikonal functions. We will show that throughout the evolution, the most singular part of $\tphi$ is localized in $S^k_\de$. As a first guide to the estimates, the reader can keep in mind that we will prove the following bounds   inside and outside  \color{black} $S^k_\de$:
	\begin{itemize}
		
			\item \textit{[Interior of the singular zone $S^k_\de$].} Within this singular region  $S^k_\de$ (see \eqref{S.def.intro}), our bounds can be no better than the initial estimates \eqref{delta.intro}. We will in fact prove estimates consistent with the $\de$-weights in \eqref{delta.intro}. Namely, we prove the $L^2$-based bound
		\begin{equation} \label{H2.sing.intro}
		\| \tphi \|_{H^2(S^k_\de)} \ls \ep \cdot \de^{-\f 12},
		\end{equation} 
		as well as the Lipschitz bound for $\tphi$
		\begin{equation}\label{Lip.intro}
		\| \rd \tphi \|_{L^\i(S^k_\de)} \ls \ep.
		\end{equation}

		\item \color{black} \textit{[Exterior of the singular zone $S^k_\de$].} We prove that the following estimate holds 
		\begin{equation} \label{H2.ext.intro}
			\| \tphi \|_{H^2(\RR^2\setminus S^k_\de)} \ls \ep.
		\end{equation} 
		Note that this is better than the bounds \eqref{delta.intro} in the singular zone for the initial data.
			
			Moreover, in terms of $L^\i$ based norms, we will show an  improved  H\"older estimate (compare with the Lipschitz estimate \eqref{Lip.intro} \magenta{above}) for $\tphi$ outside of the singular zone $S^k_\de$, i.e.\ for some $\theta \in(0,\frac{1}{4})$ \begin{equation} \label{Holder.intro}
				\|  \tphi\|_{C^{1,\theta}(\RR^2 \setminus S^k_{\de})} \ls \ep.
				\end{equation}

		\end{itemize}
		
	We will further explain the proof of the estimates \blue{\eqref{H2.sing.intro}--\eqref{Holder.intro}}. In order to derive these bounds\color{black}, we will need to prove that improved regularity is exhibited for derivatives with respect to $\{L_k, E_k\}$, the vector fields tangential to constant-$u_k$ hypersurfaces, as well as to derive higher order estimates.%; see Section \ref{sec:loc.est} below.

	\subsubsection{The  $H^2$  energy \color{black} estimates: anisotropic  estimates\color{black}, short-pulse bounds and slice-picking}\label{sec:loc.est}
	%\color{red} We first discuss the proof of the $L^2$-based (energy) estimates (see the next Section~\ref{sec:intro.embedding} for the discussion of the $L^{\infty}$-based estimates). \color{black}
	\blue{We first discuss our $L^2$ based energy estimates for $\tphi$ up to the second derivative. (The $L^\infty$-based estimates will be discussed in Section~\ref{sec:intro.embedding} and the higher order $L^2$-based estimates will be explained in Section~\ref{sec:intro.higher.regularity}.)} One of the main challenges of this problem is that the $H^2$ norm of $\phi$ is no better than $\de^{-\f 12}$ (recall \eqref{H2.sing.intro}). Already at the $H^2$ level, we capture the following features of the solutions in our energy estimates (these will again play a role in the Lipschitz (as in \eqref{Lip.intro}) and improved H\"older bounds (as in \eqref{Holder.intro}); see Section \ref{sec:intro.embedding}):
	\begin{enumerate}
		\item \label{step1}[\emph{Anisotropy}]. Derivatives in the geometric directions $L_k$ and $E_k$ are ``good'' derivatives for $\tphi$ that are better behaved than others. This phenomenon will allow us to prove anisotropic $H^2$-estimates where one general derivative is replaced by a ``good derivative''.
	
		\item \label{step2} [\emph{Short pulse bounds}]. As we mentioned above, the singularity leading to a large $H^2$ norm is only localized in a ``small'' region  of length $\sim \de$. At the same time, in the singular region, some (integrated) bounds can be proven to be \emph{$\de$-small} using the small $\de$ length  as a source of smallness.
		\item \label{step3} [\emph{Localization}].  We prove that the singularity for $\tphi$ is localized in a small region $S^k_{\delta}$ around a null hypersurface. Indeed, we show that $\tphi$ obeys uniform-in-$\delta$ $H^2$ bounds \emph{away} from $S^k_{\delta}$ as in \eqref{H2.ext.intro}. To show such bounds, we rely on a novel \emph{slice-picking argument} exploiting the \blue{anisotropic bounds  and the short pulse bounds}. \blue{A $\de$-independent $H^2$ bound} can then be propagated towards the future of this good hypersurface.

	\end{enumerate}

	\blue{In the steps below, we explain in more detail these features of our (up to $H^2$ level) energy estimates.
	}
	
	%how the above allows us to prove, respectively,
	%\begin{itemize}
	%		\item sharp anisotropic $H^2$-estimates on $\tphi$ with no $\delta$-weight, involving the vector fields $L_k$ and $E_k$ (Step \ref{step1});
	%		\item sharp $\de$-small anisotropic $H^2$-estimates involving $L_k$ and $E_k$ inside the singular zone $S^k_{\de}$ (Step \ref{step2});
	%		\item sharp isotropic $H^2$-estimates on $\tphi$ with no $\delta$-weights \emph{outside of  the singular zone $S^k_{\de}$} (Step \ref{step3}).
	%\end{itemize}
	
	\textbf{Step \ref{step1}: \blue{Anisotropic energy estimates captured by} the good geometric derivatives.} \blue{\green{At} the lowest order, our regularity assumption allows us to easily prove a $\de$-independent  bound
	\begin{equation}\label{eq:intro.lowest.energy}
	    \| \rd \tphi \|_{L^2(\Sigma_t)} \ls \ep.
	\end{equation} 
	 In fact, we can put in an extra fractional $s'\in (0,\f 12)$ derivative (cf.~\eqref{eq:intro.no.better.bound}) and prove
	\begin{equation}\label{eq:intro.Hs'}
	    \| \rd \Db^{s'} \tphi \|_{L^2(\Sigma_t)} \ls \ep.
	\end{equation}
	 However, a}s mentioned in \eqref{H2.sing.intro}, \blue{at the second derivative level, we prove an estimate no better than the following:}
	\begin{equation}\label{eq:intro.large.energy}
	\|\rd^2 \tphi\|_{L^2(\Sigma_t \cap S^k_\de)} \sim \|\rd^2 \tphi\|_{L^2(\Sigma_t)} \ls \ep  \cdot \de^{-\f 1 2}.
	\end{equation}
	As we indicated above, despite \eqref{eq:intro.large.energy}, not all derivatives are equally bad. Since $\tphi$ is essentially propagating along constant-$u_k$ hypersurfaces $C_{u_k}$, we have better regularity properties for derivatives in the directions tangential to $ C_{u_k}$ i.e.~directions spanned by $\{L_k, E_k\}$ (see Section~\ref{sec:intro.basic.geometry} and Figure \ref{Fig.null.frame}). Indeed, we prove that $\rd L_k \tphi$ and $\rd E_k \tphi$ are more regular and on constant-$t$ hypersurfaces $\Sigma_t$:
	\begin{equation}\label{eq:intro.good.energy}
	\sum_{Y_k \in \{ L_k, E_k\}} \|\rd Y_k \widetilde{\phi}_k \|_{L^2(\Sigma_t)} \ls \ep.
	\end{equation}
	
	\textbf{{Step \ref{step2}:} The short pulse  bounds \color{black} in the singular region.} The next feature of $\tphi$ to be emphasized is that the large $H^2$ norm (recall~\eqref{eq:intro.large.energy}) is only localized in a \magenta{small} region $S^k_{\delta}$ {(recall \eqref{S.def.intro})} of length scale $\sim \de$. The first observation towards proving the localization is the following: while $S^k_\de$ is a singular region for $\tphi$ in the sense that \eqref{eq:intro.large.energy} cannot be improved, some \magenta{small-in-$\de$} bounds hold \magenta{for the lower derivatives} in $S^k_\de$.
	
	To see this, first observe that since the initial data for $\tphi$ is chosen so that $\tphi = 0$ for $u_k \leq -\de$, finite speed of propagation implies that $\tphi = 0$ on the null hypersurface $\{u_k = -\de\}$ and in fact on the whole half-space $\{u_k \leq  -\de\}$. Using this vanishing and the smallness of the $\de$ length scale, we can propagate a hierarchy of $\de$-dependent estimates for $\tphi$ and its derivatives in the singular region $S^k_\de$. (This is reminiscent of the short pulse estimates of Christodoulou, originally introduced to tackle the problem of the formation of trapped surfaces for the Einstein vacuum equations \cite{dC2009}.) In particular, we prove the \emph{smallness} estimate for the $H^1$ norm of $\tphi$:
	\begin{equation}\label{eq:intro.small.energy.1}
	\|\rd \widetilde{\phi}_k\|_{L^2(\Sigma_t\cap S^k_\de)} \ls \ep \cdot \de^{\f 12}.
	\end{equation} 
	This is consistent with the initial data bound \eqref{delta.intro} (and the Lipschitz estimate \eqref{Lip.intro} that we hope to prove): $\rd \tphi$ is bounded by $\ep$ pointwise, and the smallness arises from the smallness of the $\de$-length scale. 
	Moreover, in this region, $\rd L_k \tphi$ and $\rd E_k \tphi$ also obey similar smallness bounds, which are better than \eqref{eq:intro.good.energy}:
	\begin{equation}\label{eq:intro.small.energy.2}
	\sum_{Y_k \in \{ L_k, E_k\}} \|\rd Y_k \widetilde{\phi}_k\|_{L^2(\Sigma_t\cap S^k_\de)} \ls \ep\cdot  \de^{\f 12}.
	\end{equation}
	
	\textbf{{Step \ref{step3}:}  Localization \blue{using} \color{black} a slice-picking argument.} 
	The short pulse bounds \eqref{eq:intro.small.energy.1}--\eqref{eq:intro.small.energy.2} allow us to use a slice-picking argument to prove that $\tphi$ obeys $H^2$ bounds \emph{with no $\de$-weights} when $u_k \geq \de$, i.e.~beyond the singular region $S^k_\de$. 
	
	\begin{figure}
\centering{
\def\svgwidth{17pc}
%% Creator: Inkscape 1.1 (c4e8f9e, 2021-05-24), www.inkscape.org
%% Accompanies image file 'regions.pdf' (pdf, eps, ps)
\begingroup%
  \makeatletter%
  \providecommand\color[2][]{%
    \errmessage{(Inkscape) Color is used for the text in Inkscape, but the package 'color.sty' is not loaded}%
    \renewcommand\color[2][]{}%
  }%
  \providecommand\transparent[1]{%
    \errmessage{(Inkscape) Transparency is used (non-zero) for the text in Inkscape, but the package 'transparent.sty' is not loaded}%
    \renewcommand\transparent[1]{}%
  }%
  \providecommand\rotatebox[2]{#2}%
  \newcommand*\fsize{\dimexpr\f@size pt\relax}%
  \newcommand*\lineheight[1]{\fontsize{\fsize}{#1\fsize}\selectfont}%
  \ifx\svgwidth\undefined%
    \setlength{\unitlength}{567.11133867bp}%
    \ifx\svgscale\undefined%
      \relax%
    \else%
      \setlength{\unitlength}{\unitlength * \real{\svgscale}}%
    \fi%
  \else%
    \setlength{\unitlength}{\svgwidth}%
  \fi%
  \global\let\svgwidth\undefined%
  \global\let\svgscale\undefined%
  \makeatother%
  \begin{picture}(1,0.79154239)%
    \lineheight{1}%
    \setlength\tabcolsep{0pt}%
    \put(0,0){\includegraphics[width=\unitlength,page=1]{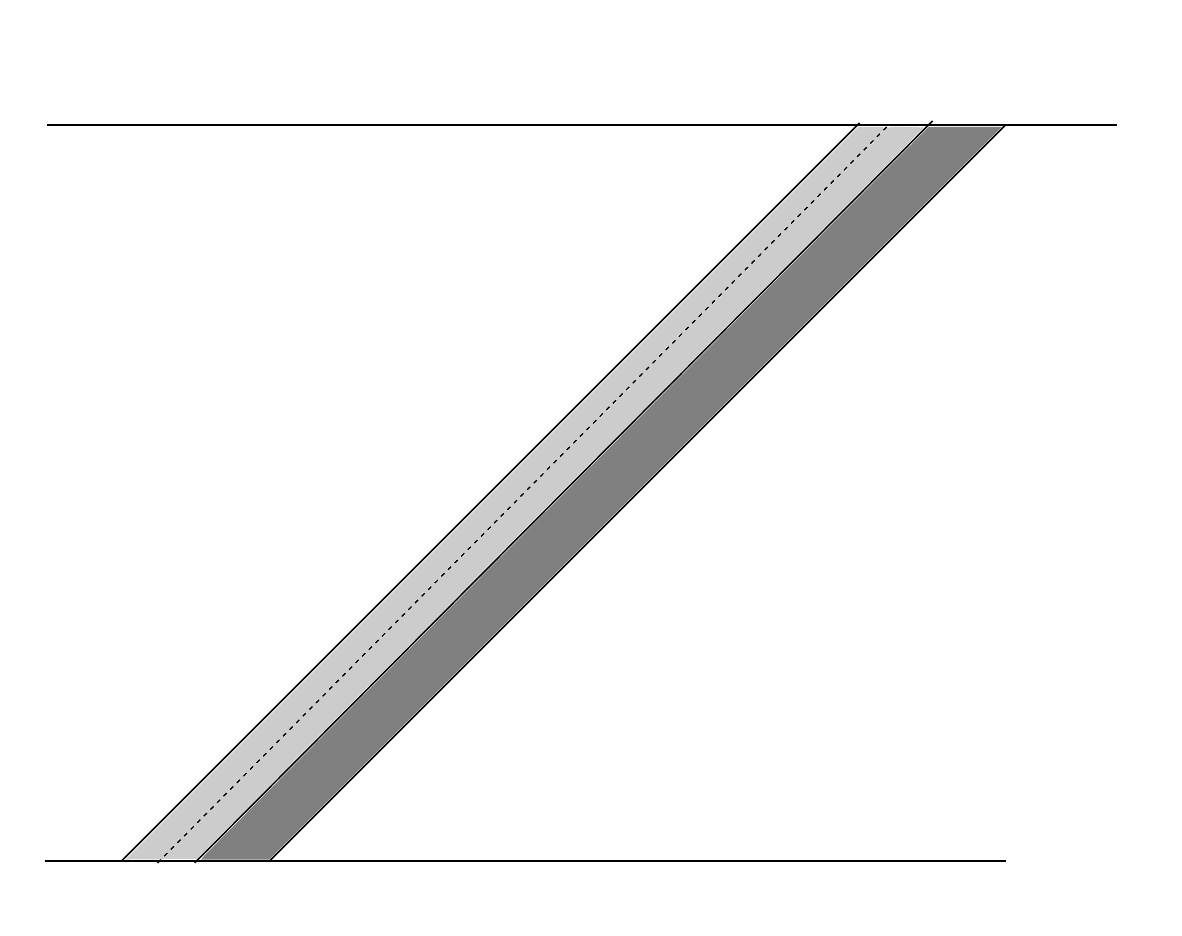}}%
    \put(0.61423318,0.29824734){\makebox(0,0)[lt]{\lineheight{1.25}\smash{\begin{tabular}[t]{l}$\widetilde{\phi}_k\equiv 0$\end{tabular}}}}%
    \put(0,0){\includegraphics[width=\unitlength,page=2]{regions.pdf}}%
    \put(0.63171016,0.77525923){\makebox(0,0)[lt]{\lineheight{1.25}\smash{\begin{tabular}[t]{l}$u_k=u_k^*$\end{tabular}}}}%
    \put(-0.00160145,0.00372984){\makebox(0,0)[lt]{\lineheight{1.25}\smash{\begin{tabular}[t]{l}$u_k=\delta$\end{tabular}}}}%
    \put(0.19550409,0.00574716){\makebox(0,0)[lt]{\lineheight{1.25}\smash{\begin{tabular}[t]{l}$u_k=0$\end{tabular}}}}%
    \put(0.36949902,0.0083584){\makebox(0,0)[lt]{\lineheight{1.25}\smash{\begin{tabular}[t]{l}$u_k=-\delta$\end{tabular}}}}%
    \put(0.78805892,0.07762691){\makebox(0,0)[lt]{\lineheight{1.25}\smash{\begin{tabular}[t]{l}$\Sigma_0$\end{tabular}}}}%
    \put(0.87118116,0.70005348){\makebox(0,0)[lt]{\lineheight{1.25}\smash{\begin{tabular}[t]{l}$\Sigma_1$\end{tabular}}}}%
  \end{picture}%
\endgroup%
}
\caption{The regions for the slice-picking argument}\label{fig:regions}
\end{figure}
	
	Consider Figure~\ref{fig:regions}. For $u_k\leq -\de$, we have $\phi_k \equiv 0$. The initial $\|\rd^2\tphi\|_{L^2(\{-\de\leq u_k \leq 0\})}$ norm is large --- of size $O(\ep \de^{-\f 12})$--- when $-\de\leq u_k \leq 0$ (the darker shaded region), while the initial $\|\rd^2\tphi\|_{L^2(\{ u_k \geq 0\})}$ norm is of size $O(\ep)$ away from the darker shaded region (including in the lightly shaded region, which is also of length scale $\de$). In both the darker shaded region and and lightly shaded region, we can prove the estimates \eqref{eq:intro.large.energy}, \eqref{eq:intro.small.energy.1} and \eqref{eq:intro.small.energy.2}. 
	
	Squaring, integrating \eqref{eq:intro.small.energy.2} over $t$ and using Fubini's theorem to switch the $t$ and $u_k$ integrals, we have 
	$$\sum_{Y_k \in \{ L_k, E_k\}} \int_{0}^\de \|\rd Y_k \widetilde{\phi}_k\|_{L^2(C_{u_k}^k)}^2 \, du_k \ls \sum_{Y_k \in \{ L_k, E_k\}} \int_0^T \|\rd Y_k \widetilde{\phi}_k\|_{L^2(\Sigma_t\cap S^k_\de)}^2 \, dt \ls (\ep \de^{\f 12})^2\ls \ep^2 \delta,$$
	where $C_{u_k}^k$ is a constant $u_k$-null hypersurface. The mean value theorem implies that \textbf{there exists $u_k^*\in [0,\de]$} (the dotted line in the lightly shaded region after the short pulse) such that the integral over the $u_k = u_k^*$ null hypersurface $C_{u_k^*}^k$ satisfies 
	\begin{equation} \label{eq:intro.de.independent.flux}
	\sum_{Y_k \in \{ L_k, E_k\}} \|\rd Y_k \widetilde{\phi}_k\|_{L^2(C_{u_k^*}^k)}^2 \ls \ep^2.
	\end{equation}

	Using standard energy estimates (assuming sufficient bounds for the metric), in order to estimate $\| \rd^2 \tphi\|_{L^2(\{ u_k \geq u_k^*\})}$ after $C_{u_k^*}^k$, it suffices to bound (a) the data on $\Sigma_0$ in the region $\{u_k \geq u_k^*\}$ and (b) the flux  
%(1) and (2) look too much like equations I find 
	\begin{equation*}
	\sum_{Y_k \in \{ L_k, E_k\}} ( \| Y_k \widetilde{\phi}_k\|_{L^2(C_{u_k^*}^k)}^2+ \|\rd Y_k \widetilde{\phi}_k\|_{L^2(C_{u_k^*}^k)}^2),
	\end{equation*}
	i.e.~on $C_{u_k^*}^k$ we only need bounds where at least one derivative is tangential to $C_{u_k^*}^k$. Since (a) we have improved data bound on $\Sigma_0\cap \{u_k \geq 0\}$ and (b) $u_k^*$ is picked so that we have a $\de$-\underline{in}dependent bound \eqref{eq:intro.de.independent.flux} of this flux, we obtain $H^2$ estimates in the region $\{u_k \geq \de \} \subseteq \{u_k \geq u_k^*\}$ \emph{with no $\de$ weights}, i.e.\ 
	\begin{equation}\label{eq:intro.after.energy}
	\| \rd^2 \tphi \|_{L^2(\Sigma_t \setminus S^k_{\de})} \ls \ep.
	\end{equation}
	In other words, the worst bound \eqref{eq:intro.large.energy} is indeed only saturated in $S^k_\de$.

	\subsubsection{\blue{Anisotropic embedding results and the Lipschitz and H\"older estimates}}\label{sec:intro.embedding}
	
    \blue{Recall that we aim at proving the Lipschitz bound \eqref{Lip.intro} and the H\"older bound \eqref{Holder.intro}. This necessitates $L^2$ estimates beyond those discussed in Section~\ref{sec:loc.est}. Below, we will explain the embedding results adapted to our setting, and the precise higher order $L^2$ estimates that we will need.
    
    Our embedding results will be used to control $\rd\tphi$, where $\rd$ denotes a derivative in the (original) coordinates of the elliptic gauge. In order to take advantage of the good derivatives, we will also introduce another coordinate system on each $\Sigma_t$ as follows. Given $k \in \{1,2,3\}$, pick $k' \in \{1,2,3\}$ with $k \neq k'$. Then $(u_k,\, u_{k'})$ forms a coordinate system on $\RR^2$ for any fixed $t$. Denote by $(\srd_{u_k}, \srd_{u_{k'}})$ the corresponding coordinate derivatives.
	%The following are the main anisotropic embedding results that we will use. 
	The reader should already think that $\srd_{u_{k'}}$ is the ``good'' derivative for $\tphi$, i.e.~it is parallel to $E_k$,  while $\srd_{u_k}$ is a ``bad'' derivative, and that $\srd$ denotes a general derivative in the $(u_k, u_{k'})$ coordinates.}
	
	\paragraph{Almost Lipschitz bounds.} \blue{By comparing with the initial data estimates, one sees that t}he bounds \eqref{eq:intro.large.energy}, \eqref{eq:intro.small.energy.2} and \eqref{eq:intro.after.energy} in Section \ref{sec:loc.est} are already the best $H^2$ estimates that can be proven. Heuristically, the bounds of Section~\ref{sec:loc.est} 
 are almost sufficient to obtain the desired Lipschitz estimate \eqref{Lip.intro} for $\phi$\, except that when trying to use Sobolev embedding, one encounters a logarithmic divergence in the summation over frequency scales \green{in a Littlewood--Paley decomposition}. However, for any fixed $p \in [1, \infty)$, the $H^2$ bounds \eqref{eq:intro.large.energy}, \eqref{eq:intro.small.energy.2} and \eqref{eq:intro.after.energy} are still sufficient to give an $L^p$ bound:
	\begin{itemize}
		\item {\emph{[$L^p$ bounds away from the singular zone]}}. Away from the singular zone $S^k_\de$, \blue{the} standard Sobolev embedding $H^1(B(0,R)) \rightarrow L^{p}(\blue{B(0,R)})$  give 
		\begin{equation}\label{eq:intro.very.stupid.embedding}
		\|\rd \tphi \|_{L^p(\Sigma_t \setminus S^k_\de)} \ls  \|\rd \tphi \|_{L^2(\Sigma_t \setminus S^k_\de)}  +  \|\srd \rd \tphi \|_{L^2(\Sigma_t \setminus S^k_\de)}.
		\end{equation}
		\blue{By \eqref{eq:intro.lowest.energy} and \eqref{eq:intro.after.energy} (after justifying that $\srd \rd \tphi$ and $\rd^2\tphi$ are comparable), the right-hand side of \eqref{eq:intro.very.stupid.embedding} is bounded by $\ep$, independently of $\de$.}
		\item  \color{black} {\emph{[$L^p$ bounds inside the singular zone]}}. To treat the singular region, note	that one can prove a refined version of Sobolev embedding that takes into account the directions of the derivatives \blue{and makes use of the localization of the singular region. Introducing a cutoff function $\rho_k$ localizing $\rd\tphi$ near $S^k_\de$, we have} 
		\begin{equation}\label{eq:intro.stupid.embedding}
		\begin{split}
		\|\rho_k \rd \tphi \|_{L^p(\Sigma_t)} \ls &\:  \|\srd (\rho_k \rd\tphi) \|_{L^2(\Sigma_t)}^{\f 12} \|\srd_{u_{k'}} (\rho_k\rd\tphi) \|_{L^2(\Sigma_t)}^{\f 12} + \| \rho_k \rd\tphi \|_{L^2(\Sigma_t)} \\
		\ls &\:  \de^{\f 12} \|\srd (\rho_k\rd\tphi) \|_{L^2(\Sigma_t)} + \de^{-\f 12} \|\srd_{u_{k'}} (\rho_k\rd\tphi) \|_{L^2(\Sigma_t)} + \| \rho_k \rd\tphi \|_{L^2(\Sigma_t)},
		\end{split}
		\end{equation}
	    \magenta{where the second line follows from the first using the Cauchy--Schwarz inequality.} Now even though in our setting $\|\tphi\|_{H^2(\Sigma_t\cap S^k_\de)} \sim \ep \de^{-\f 12}$, we have \emph{smallness} in the good derivatives estimate \eqref{eq:intro.small.energy.2}. Thus, modulo controlling the coordinate change and the vector field $E_k$,  \eqref{eq:intro.large.energy} and \eqref{eq:intro.small.energy.2} respectively imply that $\de^{\f 12} \|\srd (\rho_k\rd\tphi) \|_{L^2(\Sigma_t)} \ls \ep$ and $\de^{-\f 12} \|\srd_{u_{k'}} (\rho_k\rd\tphi) \|_{L^2(\Sigma_t)} \ls \ep$\magenta{. Using also \eqref{eq:intro.lowest.energy} to control $\| \rho_k \rd\tphi \|_{L^2(\Sigma_t)}$, this shows that $\|\rho_k \rd \tphi \|_{L^p(\Sigma_t)} \ls \ep$}.

	\end{itemize}

	\paragraph{{Anisotropic Sobolev embedding adapted to the problem}.} In order to improve \eqref{eq:intro.very.stupid.embedding}, \eqref{eq:intro.stupid.embedding}, we prove two anisotropic embedding results, designed particularly for our setting for which we can exploit \blue{the anisotropy and localization of our $L^2$ estimates}. \green{In the following, we will only give the embedding estimates when applied to $\rd \tphi$ (or a cutoff version of $\magenta{\partial}\tphi$). These} estimates are key ingredients in our proof of \blue{\eqref{Lip.intro} and \eqref{Holder.intro}}, since they provide the summability over all frequencies that we were lacking in the above paragraph.

	\begin{itemize}
	
	\item Our first embedding result \blue{(cf.~Theorem~\ref{embeddingexterior})} is a H\"older estimate on a half space\footnote{We overlook here the ambiguity in whether the $L^2$, $C^\theta$, norms are taken with respect to the $(x^1,x^2)$ coordinates or the $(u_k, u_{k'})$ coordinates, since we showed in \cite{LVdM1} that $(x^1,x^2) \mapsto (u_k, u_{k'})$ is a $C^1$ diffeomorphism. \blue{A similar comment applies to \eqref{eq:intro.anisotropic.2} below.}}. For $s'' \in (0, \f 12)$,
	\begin{equation}\label{eq:intro.anisotropic.1}
 \|\rd\tphi\|_{C^{0,\frac{s}{2}}(\Sigma_t \setminus S^k_\de)} \ls_{s} \| \rd\tphi \|_{L^2(\Sigma_t \setminus S^k_\de)} + \|\srd \rd\tphi\|_{L^2(\Sigma_t \setminus S^k_\de)} + \| \srd_{u_{k'}} \la D_{u_k,u_{k'}} \ra^{s''} \rd\tphi \|_{L^2(\Sigma_t)},
	\end{equation}
	where $\la D_{u_k,u_{k'}} \ra^{s''}$ is the fractional derivative operator in the $(u_k, u_{k'})$ coordinates. \color{black}

	\blue{The estimate \eqref{eq:intro.anisotropic.1} could be compared with \eqref{eq:intro.very.stupid.embedding}, where the extra term $\| \srd_{u_{k'}} \la D_{u_k,u_{k'}} \ra^{s''} \rd\tphi \|_{L^2(\Sigma_t)}$ on the right-hand side not only allows us to sum over all frequencies \green{in a Littlewood--Paley decomposition}, but also lets us obtain extra H\"older regularity (as long as we are away from $S^k_\de$).}
	
	\item \blue{Our} second embedding result \blue{(cf.~Theorem~\ref{embeddingThmInterior})} is an $L^\i$ estimate, involving $\de$ weights on the right-hand side: 
	\begin{equation}\label{eq:intro.anisotropic.2}
	\begin{split}
	\|\rho_k\rd\tphi\|_{L^\i(\Sigma_t)} \ls &\: \de^{-\f 12}\| \rho_k\rd\tphi \|_{L^2(\Sigma_t)} + \de^{\f 12}\| \srd_{u_{k'}} \srd (\rho_k\rd\tphi) \|_{L^2(\Sigma_t)} \\
	&\: + \de^{\f 12}\| \srd (\rho_k\rd\tphi) \|_{L^2(\Sigma_t)} + \de^{-\f 12} \| \srd_{u_{k'}} (\rho_k\rd\tphi) \|_{L^2(\Sigma_t)}\blue{,}
	\end{split}
	\end{equation}
where $\rho_k$ is a cutoff as in \eqref{eq:intro.stupid.embedding}.

Notice that \eqref{eq:intro.anisotropic.1} and \eqref{eq:intro.anisotropic.2} in particular gives the global Lipschitz estimate (cf.~\eqref{Lip.intro}):
\begin{equation}\label{Lip.intro.global}
    \| \rd \tphi \|_{L^\i(\Sigma_t)} \ls \mbox{RHSs of \eqref{eq:intro.anisotropic.1} and \eqref{eq:intro.anisotropic.2}}.
\end{equation}
	The \color{black} reader may want to compare \eqref{eq:intro.anisotropic.2} \color{black} with \eqref{eq:intro.stupid.embedding}. \blue{The two new $\de$-weighted terms $\de^{-\f 12}\| \rho_k\rd\tphi \|_{L^2(\Sigma_t)}$ and $\de^{\f 12}\| \srd_{u_{k'}} \srd (\rho_k\rd\tphi) \|_{L^2(\Sigma_t)}$ allow us to sum over all frequencies. In fact, this even allows us to control a Besov norm $\| \rho_k \rd \tphi\|_{B^{0}_{\infty,1}(\RR^2)}$, which is crucial for closing an endpoint elliptic estimate in part I; see \cite[Section~1.1.4]{LVdM1}. Notice also that $\de^{-\f 12}\| \rho_k\rd\tphi \|_{L^2(\Sigma_t)} \ls \ep$ by \eqref{eq:intro.small.energy.1}.}

\end{itemize}

By \eqref{eq:intro.anisotropic.1} and \eqref{Lip.intro.global}, proving that $\phi$ is Lipschitz uniformly-in-$\de$ with additional H\"older regularity away from the $\de$-impulsive waves reduces to showing  $\mbox{RHSs of \eqref{eq:intro.anisotropic.1} and \eqref{eq:intro.anisotropic.2}} \lesssim \ep$, and  will thus require the following  %\eqref{eq:intro.anisotropic.1} and \eqref{eq:intro.anisotropic.2}, 
 main higher order estimates %that we need are 
\begin{equation}\label{eq:higher.order.motivated}
    \| \srd_{u_{k'}} \la D_{u_k,u_{k'}} \ra^{s''} \rd\tphi \|_{L^2(\Sigma_t)}\ls \ep,\quad \de^{\f 12}\| \srd_{u_{k'}} \srd (\rho_k\rd\tphi) \|_{L^2(\Sigma_t)} \ls \ep.
\end{equation}
Recalling that $\srd_{u_{k'}}$ can be thought of as a good derivative, we see that the first bound is an estimate combining fractional and good geometric derivatives while the second bound is a higher order $\de$-weighted estimates involving a good geometric derivative.

	\textbf{The most difficult part of the paper is then to obtain the \blue{bounds in \eqref{eq:higher.order.motivated}} under the very limited regularity of the metric.} We will explain these $L^2$ estimates in the \blue{next subsection}.

	\subsubsection{\blue{The higher order energy estimates}}\label{sec:intro.higher.regularity}
	
	%We need to propagate, using the wave equation, higher order (i.e.~more than $2$ derivatives) estimates so as to apply the embeddings \eqref{eq:intro.anisotropic.1} and \eqref{eq:intro.anisotropic.2}. %For this we take special combinations of the commutators \eqref{eq:commutator.blocks}.
\paragraph{The main higher order energy estimates.}	
	We now explain \magenta{the higher order energy estimates we prove to obtain} \eqref{eq:higher.order.motivated}. Corresponding to \blue{the first term} \magenta{in \eqref{eq:higher.order.motivated}}, we prove
	\begin{equation}\label{eq:intro.top.order.1}
	\|\rd E_k \Db^{s''} \tphi\|_{L^2(\Sigma_t)} \ls  \ep,\quad \|\rd L_k \Db^{s''} \tphi\|_{L^2(\Sigma_t)} \ls \ep.
	\end{equation}

	Corresponding to \blue{the second term in \eqref{eq:higher.order.motivated}}, we prove
	\begin{equation}\label{eq:intro.top.order.2}
	\|\rd E_k \rd_i \tphi\|_{L^2(\Sigma_t)} \ls  \ep\de^{-\f 12},\quad \|\rd L_k L_k \tphi\|_{L^2(\Sigma_t)} \ls \ep \de^{-\f 12}.
	\end{equation}
	
	One can think that the $E_k$'s \blue{in the first terms in \eqref{eq:intro.anisotropic.1}, \eqref{eq:intro.anisotropic.2} above are the good derivatives $\srd_{u_{k'}}$}, since $\srd_{u_{k'}}$ is parallel to $E_k$. %However, it should also be noted that the form of \eqref{eq:intro.top.order.1} and \eqref{eq:intro.top.order.2} are not exactly those needed for \blue{\eqref{eq:higher.order.motivated}}: indeed, \blue{many of} the derivatives \blue{in \eqref{eq:higher.order.motivated}} (including the fractional ones!) need to be taken with respect to the coordinate vector fields in the $(u_k, u_{k'})$ coordinate system instead of the coordinates vector fields in the elliptic gauge or the geometric vector fields. Nevertheless,
	Once  \eqref{eq:intro.top.order.1} and \eqref{eq:intro.top.order.2} are obtained, the bounds from \cite{LVdM1} allow us to control all necessary commutator terms (even though some of them are top order), convert \eqref{eq:intro.top.order.1}--\eqref{eq:intro.top.order.2} into estimates in the $(u_k, u_{k'})$ coordinate system, and \color{black} to apply them for    \eqref{eq:higher.order.motivated}; see Section~\ref{dphiLinftysection}.\color{black}
	
\paragraph{The initial regularity of the wave.} Note that  \eqref{eq:intro.top.order.1} and \eqref{eq:intro.top.order.2} 
are consistent with the initial regularity of the wave, and we thus assume they are satisfied initially. In particular, \eqref{eq:intro.top.order.1} is a statement that the fractional regularity energy estimate \eqref{eq:intro.Hs'} still holds after a suitable commutation with the good derivatives $E_k$ and $L_k$. The main challenge, however, is to \emph{propagate} such regularity \emph{with only very limited regularity of the metric}.

\paragraph{The  estimates involving $L_k$ in  \eqref{eq:intro.top.order.1} and \eqref{eq:intro.top.order.2}.}		Furthermore, notice that only the respective first bounds in \eqref{eq:intro.top.order.1} and \eqref{eq:intro.top.order.2} are used for the anisotropic Sobolev embedding. However, in order to handle some commutators that arise, it is important to simultaneously \color{black} prove the second bounds in \eqref{eq:intro.top.order.1} and \eqref{eq:intro.top.order.2}\blue{.} %(see also Section~\ref{sec:intro.commutator}). 	
	
 %To explain the interaction between the higher order wave estimates and the geometry, let us summarize the bounds we have for the geometric quantities:
	\medskip
	
	\paragraph{\magenta{Ideas of proof of \eqref{eq:intro.top.order.1}.}} \magenta{We now explain the proof of \eqref{eq:intro.top.order.1}.}
	\begin{itemize}
	
		\item To prove \eqref{eq:intro.top.order.1}, we bound the commutator terms $[\Box_g, E_k \Db^{s''}]\tphi$ and $[\Box_g, L_k \Db^{s''}]\tphi$. It is important to both (1) use fractional derivatives with respect to the elliptic gauge (as opposed to geometric) coordinates and (2) commute with $\Db^{s''}$ first before commuting with the geometric vector fields. This way we exploit the better regularity of the metric components in the elliptic gauge.  (Indeed, we will not be able to control either $[\Box_g, E_k \la D_{u_k,\blue{u_{k'}}}\ra^{s''}]\tphi$ or $[\Box_g, \Db^{s''} E_k]\tphi$.)\color{black}
		\item The commutator term $[\Box_g, E_k \Db^{s''}]\tphi$ schematically gives rise to \blue{error} terms of the form
		\begin{equation}\label{eq:intro.commute.EDs.error}
		(\Db^{s''} \rd \rd_i \mfg)(\rd \tphi),\quad (\rd \rd_{i} \mfg)(\rd \Db^{s''}\tphi),\quad (\rd \mfg) (\rd^2 \Db^{s''}\tphi).
		\end{equation}
		\blue{The terms $(\Db^{s''} \rd \rd_i \mfg)(\rd \tphi)$ can be controlled using the metric bound \eqref{eq:intro.metric.bounds} together with \eqref{Lip.intro}. To control the terms $(\rd \rd_{i} \mfg)(\rd \Db^{s''}\tphi)$, we  use \eqref{eq:intro.Hs'} and combine it with \eqref{eq:intro.metric.bounds} (recall that $0<s''<s'<\f 12$).  There is a slight subtlety here: the reason that we need to introduce two different exponents $0<s''<s'<\f 12$ and  estimate $\rd \Db^{s'}\tphi$, $\rd E_k \Db^{s''}\tphi$, $\rd L_k \Db^{s''}\tphi$ with the slightly different order of derivatives is because for the term $(\rd_t \rd_{i} \mfg)(\rd \Db^{s''\color{black}}\tphi)$, we do not have $L^\i$ bounds for $\rd_t \rd_{i} \mfg$ (see \eqref{eq:intro.metric.bounds}).}
	 
		\item The third type of error terms \green{in \eqref{eq:intro.commute.EDs.error}}, i.e.~the terms $(\rd \mfg) (\rd^2 \Db^{s''}\tphi)$, are more subtle because we do not control general derivatives $\rd^2 \Db^{s''}\tphi$. To close our argument, we need show that the only such term arising in the commutator is schematically of the form $\rd E_k \Db^{s''} \tphi$. To achieve this, we need to give a sharp expression for the commutator with fractional derivatives to isolate the main {$\rd E_k \Db^{s''} \tphi$} term. This in turn requires a %slight
		refinement of the usual Kato--Ponce type commutator estimates; see already Proposition~\ref{prop:commute.3}. 
		\item \blue{When showing that} the top-order derivative  $\rd^2 \Db^{s''}\tphi$ \color{black} \green{from the above bullet point} is morally $\rd E_k \Db^{s''} \tphi$\blue{, t}he term we obtain is $E_k \Db^{{s''}-2} \rd^3_{i \nu\color{black}\bt} \tphi$. Since $\Db^{-2} \rd^2_{ij}$ is a bounded operator on $L^2$-based Sobolev spaces, the term can be thought of as like $\rd E_k \Db^{s''} \tphi$ if at least one of $\nu\color{black}$, $\bt$ is a spatial index. However, the term becomes much more challenging when $(\nu\color{black},\bt) = (t,t)$ so that we need to use the wave equation to convert the times indices into spatial ones, and in the process we are required to handle a large number of commutator terms.
		\item Since we consider the nonlocal operator $\Db$, the terms involved are no longer compacted supported. An additional challenge is that the metric components diverge logarithmically near spatial infinity (a difficulty well-known in the $(2+1)$-dimensional case); and moreover the components $L^i_k$, $E^i_k$ of the commutators $L_k = L_k^i \rd_i$ and $E_k = E^i_k \rd_i$ also grow near spatial infinity. We \blue{therefore} use \emph{weighted} estimates\footnote{\blue{We would like to thank an anonymous referee for suggesting us to handle this instead with a commutator of the form $\varpi \Db^{s''}$, where $\varpi$ is compactly supported. While we have not implemented this, we do believe that this would lead to some simplifications of our arguments.}}, including when understanding terms like $E_k \Db^{{s''}-2} \rd^3_{i \nu\color{black}\bt} \tphi$ described in the above point. 
		\item Finally, the considerations above by themselves cannot control $[\Box_g, L_k \Db^{s''}]\tphi$. This is because $L_k$ (in the $\rd_t$, $\rd_1$, $\rd_2$ basis {from the elliptic gauge \eqref{eq:elliptic.gauge.intro})} has a $\rd_t$ component and thus  the \color{black} result is a term schematically like
		$$(\Db^{s''} \rd_t^2 \mfg )(\rd \tphi),$$
		in addition to terms similar to those we encountered in $[\Box_g, E_k \Db^{s''}]\tphi$.
		
		Recall that (see Section~\ref{sec:intro.part.I}) we do not have \underline{any} bounds for $\rd_t^2 \mfg$. To resolve this issue, we note that such a term can be traced back to a total $\rd_t$-derivative, i.e.~we can write 
		\begin{equation}\label{eq:intro.wave.for.LDsphi}
		\Box_g L_k \Db^{s''}\tphi=F + \rd_t C.
		\end{equation}
		While $\rd_t C$ cannot be controlled,  $F$, $C$ and $\rd_i C$ can be controlled {in $L^2(\Sigma_t)$ using the same methods as for  $[\Box_g, E_k \Db^{s''}]\tphi$}. Now the key observation is that in the energy estimate, we schematically have a bulk integral of the form
		\begin{equation}\label{eq:the.equation.without.a.name}
		\int (\rd_t L_k \Db^{s''}\tphi) (\rd_t C) .
		\end{equation}
		To address \eqref{eq:the.equation.without.a.name}, we \emph{integrate by parts} in $t$. For the $\rd_t^2 L_k \Db\tphi$ term, we can use the wave equation \eqref{eq:intro.wave.for.LDsphi} so that up to lower order terms, we obtain three terms to be controlled
		\begin{equation}\label{eq:intro.IBP}
		\int  (\rd^2_{\nu\color{black} i} L_k \Db^{s''}\tphi) C + \int F C + \int C^2.
		\end{equation}
		For the first term, we integrate by parts again in the spatial $\rd_i$ derivative. We can thus bound these terms using the estimates we have for $F$, $C$ and $\rd_i C$.
	\end{itemize}

	\paragraph{\magenta{Ideas of proof of \eqref{eq:intro.top.order.2}.}} \magenta{Finally, we explain the proof of \eqref{eq:intro.top.order.2}.}
	\begin{itemize}
		\item Similar to the proof of \eqref{eq:intro.top.order.1}, the exact choice of commutators matters. We will use $E_k\rd_i$ and $L_k^2$ as commutators, so that we need to bound $[\Box_g, E_k \rd_i]\tphi$ and $[\Box_g, L_k^2]\tphi$. By contrast, we could for instance neither control $[\Box_g, E_k E_k]\tphi$, $[\Box_g, \rd_i E_k]\tphi$ (since we lack general second derivative control of $\chi_k$ and $\eta_k$) nor $[\Box_g, L_k \rd_i]\tphi$ (since we lack $L^\infty$ estimates for $\rd_t \rd_i \mfg$).
	
		\item Terms that arise in $[\Box_g, E_k \rd_i]\tphi$ are {schematically} $$\rd \mfg E_i \rd^2\tphi,\ \rd \mfg L_k^2 \rd\tphi,\ \rd^2 \mfg \rd^2\tphi,\ \rd^3 \mfg \rd \tphi.$$ As in many of the previous estimates, it is important that these terms have some structure. First, in $\rd^2 \mfg$ and $\rd^3\mfg$ there are at most one $\rd_t$ derivative (recall that we do not control $\rd_t^2 \mfg$\color{black}). Second, because we do not control $\rd_i \rd_t\mfg$ in $L^\i$ (see \eqref{eq:intro.metric.bounds}), we would not be able to bound $\rd_i \rd_t \mfg\ \rd^2\tphi$ in general. Fortunately, the commutator has a useful structure  in \color{black} that only $\rd_i \rd_t \mfg\ \rd L_k\tphi$ or $\rd_i \rd_t \mfg\ \rd E_k \tphi$ arise.
		\item There are some further subtleties in the bounds for $[\Box_g, L_k^2]\tphi$.
		\begin{itemize}
			\item $[\Box_g, L_k^2 ]\tphi$ contains terms with second derivatives of $\chi_k$ and $\eta_k$ (which we do not in general control). Importantly, exactly because we are commuting with $L_k$ twice, one of the two derivatives on $\chi_k$ and $\eta_k$ must be $L_k$ so that we can use \eqref{eq:Lkappa.even.better}.
			\item Another dangerous term {that arises} is $(L_k L_k X_k \log N)(L_k \tphi)$, since schematically {it is of the form} $\rd_t^2 \rd_i \mfg$ (recall we do not have any control over two time derivatives of $\mfg$!). This can be treated with an integration by parts argument similar to \eqref{eq:intro.wave.for.LDsphi}--\eqref{eq:intro.IBP} in the proof of \eqref{eq:intro.top.order.1}.
		\end{itemize}

	\end{itemize}

	\subsection{Comments and related works}\label{sec:related.works}
	
	We refer the reader to the introduction of \cite{LVdM1} for discussions on impulsive gravitational waves and other related works in general relativity. Instead, we restrict ourselves to discussing previous works on wave estimates (for linear and nonlinear wave equations) related to those in this paper and how our work connects to this existing literature.
	
	\subsubsection{Geometric and harmonic analysis techniques for quasilinear wave equations}
	
    As we saw from Section~\ref{sec:method}, our result in this paper is based on a combination of techniques from geometric analysis and harmonic analysis. Related techniques are used in many low-regularity problems for quasilinear wave equations. We refer the readers to \cite{sA1988, mDcLgMjS2019, sK2001, sK2003, sKiR2003, sKiRjS2015, hSdT2005, qW2014} for a sample of results.
    
    In the specific context of low-regularity solutions to quasilinear hyperbolic equations featuring one or more singularities propagating along null hypersurfaces, geometric methods using well-chosen coordinate systems and commuting vector fields are often employed; see \cite{sA1993, HLHF, jL2013, LR1, LR2, LR3}. In the present paper, we extend the methods in these works but further combine them with techniques from harmonic analysis  to handle the interaction of three ($\de$-)impulsive waves. 
	
	\subsubsection{Linear wave equations with rough coefficients}\label{sec:related.rough.coeff}
	
	While our main goal in this paper is to prove wave estimates so as to complete the program in \cite{LVdM1}, when taken on its own, the present paper concerns proving estimates for a linear scalar wave equation with rough coefficients. Indeed, as seen in Theorem~\ref{thm:intro.wave}, the main result in this paper takes the following form: assuming certain bounds on the metric and suitable commuting vector fields, then one can propagate $\de$-impulsive waves type estimates under the flow of the linear wave equation. Such a formulation does not explicitly refer to general relativity. % (although the assumptions on the metric and the commuting vector fields would be somewhat unnatural). 
	In this context, let us also remark that the techniques we introduce can also be easily adapted to deal with linear wave equations of the form
	\begin{equation}
	    g^{\nu\color{black}\bt} \rd^2_{\nu\color{black}\bt} \phi + B^{\nu\color{black}} \rd_{\nu\color{black}} \phi + V\phi = 0
	\end{equation}
	with suitable regularity assumptions on $g^{\nu\color{black}\bt}$, $B^\nu$ and $V$.
	
	Though not directly related to this paper, we mention a small sample of works concerning estimates for linear wave equations with rough coefficients; see \cite{aHjR2020, Tataruloss, dTdaG2005}.

	\subsubsection{Interactions of singularities for semilinear wave equations}
	
	Our main result Theorem~\ref{thm:intro} can be viewed as a result on the interaction of singularities. In the setup of \eqref{eq:Einstein.scalar.field}, the nonlinear interaction is hidden in the coupling between the scalar wave and the metric. In the literature, interaction of singularity results are often studied for the following type of simpler semilinear models:
	\begin{equation}\label{eq:stupid.model}
	    \Box \phi = F(\phi),
	\end{equation}
	where $\Box$ is the standard wave operator on $\RR^{2+1}$ and $F:\mathbb R\to \mathbb R$ is a smooth function. See for instance \cite{aSB2020,aSByW2018,mB1983,mB1988,jmB1984,jmB1986,yKmLgU2018,mLgUyW2018,rMnR1985,jRmcR1980,jRmR1982,gUyW2018,mZ1994}.
	
	We remark that even though our methods are specifically designed to handle the rough metric, they can be easily applied to the model problem \eqref{eq:stupid.model}. Indeed, given initial data which represent three small-amplitude impulsive waves, we can smooth them out to $\de$-impulsive waves and introduce the decomposition $\phi= \rphi + \sum_{k=1}^3 \tphi$, where
	$$\Box \tphi = 0,\quad \Box \rphi = F(\phi).$$
	(Notice that this is slightly different from Section~\ref{scaling.section}.)
	It is then not difficult to see that one can propagate all the $L^2$ estimates that we prove in this paper. (In fact, the proof would be by far easier than that in this paper.) In particular, after taking the $\de\to 0$ limit, this shows that the solution remains Lipschitz everywhere and has additional $H^2$ and H\"older regularity away from propagating singularities.
	
	Let us note that it is also interesting to study \magenta{interactions of singularities for} semilinear wave equations where the nonlinearity depends also on the derivative of the solution \cite{hyC1987, aSB1990} (e.g., nonlinearities satisfying the classical null condition). However, the techniques introduced in this paper do not immediately apply to these models.
	\color{black}

\subsection{Outline of the paper}\label{sec:outline}
The remainder of the paper is structured as follows. \begin{itemize}
	\item In \textbf{Section \ref{geometry.section}}, we introduce the geometric setup, the equations in various coordinate systems and the main notations that will be used throughout the paper.
	\item  In \textbf{Section \ref{sec:norms}}, we introduce the function spaces and norms that we will use in the paper.
	\item In \textbf{Section \ref{sec:main.results}}, we give a precise version of our main results, whose rough versions were already presented as Theorem \ref{thm:intro} and Theorem \ref{deltathm:intro}. 
	\item In \textbf{Section \ref{sec:partI}}, we recall the main results of Part I \cite{LVdM1}, including the estimates for the metric components and for the null hypersurfaces.
	\item Most of the remainder of the paper is devoted to the proof of the energy estimates. We begin with some preliminaries towards the energy estimates.
	\begin{itemize}
	\item  In \textbf{Section \ref{IBP.section}}, we prove a technical integration by parts lemma that will be important in the proof of the energy estimates.
	\item In \textbf{Section \ref{sec:EE}}, we give the proof of basic energy estimates with an arbitrary source term. 
	\item In \textbf{Section \ref{1commuted.section}}, we compute and estimate the commutators between various vector fields and the wave operator in preparation for the proof of higher order energy estimates.
	\end{itemize}
	\item Using the above preliminaries, we first prove energy estimates for $\tphi$ up to second derivatives (\blue{see}~Section~\ref{sec:loc.est}):
	\begin{itemize}
	\item In \textbf{Section \ref{firstcommutedsection}}, we prove our basic energy estimates up to second derivatives.
	\item In \textbf{Section \ref{exterior}}, we obtain improved energy estimates up to second derivatives (\blue{see}~\eqref{eq:intro.small.energy.1}, \eqref{eq:intro.small.energy.2}, \eqref{eq:intro.after.energy}).
	\end{itemize}
	\item {We then prove higher order energy estimates for $\tphi$ (\blue{see}~\eqref{eq:intro.top.order.1} and \blue{\eqref{eq:intro.top.order.2}}):}
	\begin{itemize}
	\item  In \textbf{Section \ref{highest}}, we prove energy estimates involving up to three derivatives of $\tphi$ (and $\rphi$).
	\item In \textbf{Section \ref{unlochighestfrac}}, we prove fractional energy estimates for $\tphi$ and its good derivatives.
	\end{itemize}
	\item In \textbf{Section \ref{sec:rphi}}, we prove energy estimates for $\rphi$, the regular part of the solution.
	\item In \textbf{Section \ref{sec:wave.final}}, we combine the results of all previous sections to conclude the proof of our energy estimates.
	\item In \textbf{Section \ref{dphiLinftysection}}, we prove an anisotropic Sobolev embedding result. Using our energy estimates from Section \ref{sec:wave.final}, we apply the embedding result to obtain Lipschitz and improved H\"older bounds.
\end{itemize}

\subsection*{Acknowledgements}   We would also like to thank two anonymous referees for their useful comments.\color{black}

Part of this work was carried out when M.~Van de Moortel~was a visiting student at Stanford University. J.~Luk~is supported by a Terman fellowship and the NSF grants DMS-1709458 and DMS-2005435. 

\section{Summary of the geometric setup}\label{geometry.section}

 	 In this section, we recall the geometric setup introduced in \cite{LVdM1}, as well as some useful computations.
 	
 	In Sections~\ref{ellipticgaugedef}, we introduce the symmetry assumption and the elliptic gauge in the symmetry-reduced spacetime.
 	
 	In Section \ref{XELexpressionsection}, we introduce the eikonal functions $u_k$ and the geometric vector field $(L_k, E_k, X_k)$ for $k=1,2,3$ (\blue{see}~Section~\ref{sec:intro.basic.geometry}). In Section~\ref{riccinullframesection}, we compute the covariant derivatives and commutators with respect to these geometric vector fields.
 	
 	In connection with eikonal functions, we introduce in Section \ref{relationXELgeocoordinatesection} various different systems of coordinates. Some computations regarding the change of coordinates between these coordinate systems are given in Section~\ref{coordinatetransform}.

 	\subsection{Elliptic gauge and conformally flat spatial coordinates} \label{ellipticgaugedef}

 	\begin{defn}[$\mathbb U(1)$ symmetry]\label{def:U1}
 	We say that a (3+1) Lorentzian manifold $(\mathcal{M} = \RR^2 \times \mathbb{S}^1 \times I, ^{(4)}g)$, where $I\subseteq \RR$ is an open interval, has \textbf{polarized $\mathbb U(1)$ symmetry} if the metric $^{(4)}g$ can be expressed as: 
 		
 		\begin{equation}\label{eq:U1}
 		^{(4)}g = e^{-2\phi} g + e^{2\phi} (dx^3)^2,
 		\end{equation}
 		where $\phi$ is a scalar function on $I \times \RR^2 $ and $g$ is a $(2+1)$ Lorentzian metric\footnote{Note that since $\phi$ and $g$ are defined on $\RR^2 \times \RR$, they do not depend on $x^3$, the coordinate on $\mathbb{S}^1$.} on  $ I \times\RR^2$.
 	\end{defn}
 	
 	\begin{defn}[The foliation $\Sigma_t$]\label{def:Sigmat}
 	Given a space-time as in Definition~\ref{def:U1}, we foliate the $2+1$ space-time $( I \times \RR^2, g)$ with slices $\{ \Sigma_t\}_{t\in I}$ where $\Sigma_t$ are spacelike. We will later make a particular choice of $t$; see Definition \ref{def:gauge}. The metric can then be written as
 		
 		\begin{equation} \label{metric2+1}
 		g = -N^2 dt^2 +\bar{g}_{i j} (dx^i + \beta^i dt) (dx^j + \beta^j dt).
 		\end{equation}
 	\end{defn}
 	
 In the above, and the remainder of the paper, we use the convention \blue{that} \textbf{lower case Latin indices refer the the \blue{spatial} coordinates $(x^1, x^2)$}, \blue{and} \textbf{lower case Greek indices to refer to spacetime coordinates $(x^0,x^1,x^2) := (t,x^1,x^2)$.} \textbf{Repeated indices are \blue{always summed over}}\blue{: repeated lower case Latin indices are summed over $i,j, \cdots = 1,2$ and repeated lower case Greek indices are summed over $\mu,\nu, \cdots = 0,1,2$.}
 	
 	\begin{defn}\label{def:miscellaneous}
 		Given $( I \times \RR^2, g)$ and $\{ \Sigma_t\}_{t\in I}$  as \color{black} in Definition~\ref{def:Sigmat}.
 		\begin{enumerate}
 			\item (Space-time connection) Denote by $\nabla$ the Levi--Civita connection for $g$.
 			\item (Induced metric) Denote by $\bar{g}$ the induced metric on the two-dimensional slice $\Sigma_t$.
 			\item (Spatial connection) Denote by $\bnabla$ the orthogonal projection of $\nab$ onto $T\Sigma_t$ and $T^*\Sigma_t$ (and their tensor products)\footnote{We remark that for $Y,\,Z\in \Gamma(T\Sigma_t)$, $\bnabla_Y Z$ coincides with the derivative with respect to the Levi--Civita connection for $\bar{g}$.}. 
 			\item (Normal to $\Sigma_t$) Denote by $\n$ the future-directed unit normal to $\Sigma_t$; $\n$ admits the following expression
 				\begin{equation} \label{defnormal}
 				\n =  \frac{  \partial_t - \beta^i \partial_{i} }{N}.
 				\end{equation}
 				and satisfies $g(\n,\n)=-1$. Note that we have the following commutation formula \begin{equation} \label{nspatial-spatialn}
\left[\n, \partial_q \right]= \partial_q \log(N) \cdot \n+\frac{\partial_q \beta^i}{N} \cdot \partial_i.
 				\end{equation}
 				
 				 Define also $e_0$ to be the vector field
 			\begin{equation}\label{def:e0}
 			e_0=  \partial_t - \beta^i \partial_{i} = N\cdot \n.
 			\end{equation}
 			\item (Second fundamental form) Define $K$ to be the second fundamental form on $\Sigma_t$:
 			\begin{equation} \label{Kdef}
 			K(Y,Z) = g( \nabla_Y \n , Z),
 			\end{equation}
 			for every $Y,\,Z \in T\Sigma_t $.  
			\end{enumerate} 
 	\end{defn}
 	
 	\begin{defn}[Gauge conditions]\label{def:gauge}
 		We define our gauge conditions (assuming already \eqref{eq:U1}) by the following:
 		\begin{enumerate}
 			\item For every $t\in I$, $\Sigma_t$ is required to be maximal, i.e.
 			\begin{equation} \label{maximality}
 		(\bar{g}^{-1})^{ i j}  K_{i j} = 0 .
 			\end{equation}
 		Note that \eqref{maximality} defines the coordinate $t$.
 			\item We  choose the coordinate system on $\Sigma_t$ so that $\bar{g}_{i j}$ is conformally flat: this gauge condition is written as
 			\begin{equation} \label{gauge}
 			\bar{g}_{ i j} = e^{2\gamma} \delta_{i j},
 			\end{equation}
 			where from now on $\de$ denotes the Kronecker delta.
 		\end{enumerate}
 	\end{defn}

 	We collect some simple computations:
 	
 	\begin{lem}
 	The following holds given $g$ of the form \eqref{metric2+1} satisfying Definition~\ref{def:gauge}:
 		\begin{enumerate}
 			\item The inverse metric $g^{-1}$ is given by
 			\begin{equation} \label{inversegelliptic}
 			g^{-1}=\frac{1}{N^2}\left(\begin{array}{ccc}-1 & \beta^1 & \beta^2\\
 			\beta^1 & N^2e^{-2\gamma}-\beta^1\beta^1 & -\beta^1 \beta^2\\
 			\beta^2 & -\beta^1 \beta^2 & N^2e^{-2\gamma}-\beta^2\beta^2
 			\end{array}
 			\right).
 			\end{equation} 	
 			\item The space-time volume form associated to $g$ is given by 
 			\begin{equation} \label{volelliptic}
			dvol= Ne^{2\gamma} dx^1 dx^2 dt.
			\end{equation}
 			The volume form on the spacelike hypersurface $\Sigma_t$ induced by $g$ is given by
 			\begin{equation} \label{volelliptict=0} 
			dvol_{\Sigma_t}= e^{2\gamma} dx^1 dx^2.
			\end{equation}
 			\item The wave operator (i.e.~the Laplace--Beltrami operator associated to $g$) is given by
 			\begin{equation} \label{Box2+1}
 			\Box_g f = \frac{-e_0^2 f}{N^2} + e^{-2 \gamma} \delta^{i j} \partial^{2}_{i j} f + \frac{e_0 N}{N^3} e_0 f + \frac{ e^{-2 \gamma}}{N}  \delta^{i j} \partial_{i} N \partial_{j} f= - \n^2 f + e^{-2 \gamma} \delta^{i j} \partial^{2}_{i j} f  + \frac{ e^{-2 \gamma}}{N}  \delta^{i j} \partial_{i} N  \partial_j f ,
 			\end{equation}
 			where $e_0$ and $\n$ are as in Definition~\ref{def:miscellaneous}.% and (from now on) $\rd_i$ denotes the coordinate derivatives in $(t,x^1,x^2)$ coordinates.
 			\item The condition \eqref{maximality} can be rephrased as
 			\begin{equation} \label{maximality2}
 			\partial_q \beta^q = 2 e_0 (\gamma),
 			\end{equation} 
 			\item The second fundamental form is given by\footnote{This follows from
 				$$K_{i j } = \frac{2e_0(e^{2\gamma})}{N} \delta_{i j} - \frac{2e^{2\gamma}}{N} (\de_{jl} \partial_i \beta^l + \de_{il} \partial_j \beta^l) $$
 				together with \eqref{maximality2}.}
 			\begin{equation} \label{maximality3}
 			K_{i j } = \frac{  e^{2\gamma}}{2N} \cdot \left(    \partial_q \beta^q  \cdot  \delta_{i j}  - \partial_i \beta^q \cdot  \delta_{q j}   -  \partial_j \beta^q \cdot \delta_{i q}\right):=- \f{e^{2\gamma}}{2 N}(\mathfrak L\bt)_{ij},
 			\end{equation} 
 			where $\mathfrak L$ is the conformal Killing operator $(\mathfrak L \bt)_{ij}:= - \partial_q \beta^q  \cdot  \delta_{i j} + \partial_i \beta^q \cdot  \delta_{q j}  +  \partial_j \beta^q \cdot \delta_{i q}$.
 		
 		\end{enumerate}
 	\end{lem}

 	Finally, we compute the connection coefficients with respect to $\{e_0,\rd_1,\rd_2\}$:
 	
 	\begin{lem} \label{Christoffel}
 	Given $g$ of the form \eqref{metric2+1} satisfying Definition~\ref{def:gauge},	
 		\begin{equation} \label{e0e0e0}
 		g(\nabla_{e_0} e_0, e_0) = -N \cdot e_0 N,
 		\end{equation}
 		\begin{equation} \label{e0e0ei}
 		g(\nabla_{e_0} e_0, \rd_i) = -	g(\nabla_{\rd_i} e_0, e_0)=g(\nabla_{e_0} \rd_i, e_0)=N \cdot \partial_i N,
 		\end{equation}		
 		\begin{equation} \label{eje0ei} \begin{split}
 		g(\nabla_{\rd_j} e_0, \rd_i) = &\:	g(\nabla_{e_0} \rd_j, \rd_i)-e^{2\gamma} \cdot \partial_j \beta^l \delta_{i l}=-	g(\nabla_{\rd_j} \partial_i, e_0) \\ = &\: \frac{e^{2\gamma} }{2}\cdot \left( 2e_0 \gamma \cdot \delta_{i j}-\partial_i \beta^q \cdot \delta_{ j q}  -\partial_j \beta^q  \cdot \delta_{i q}                          \right) = \frac{e^{2\gamma} }{2}\cdot \left( \partial_q \beta^q  \cdot \delta_{i j}-\partial_i \beta^q \cdot \delta_{ j q}  -\partial_j \beta^q \cdot \delta_{i q}                          \right). \end{split} \end{equation}
 		
Moreover,
 		\begin{equation} \label{Diej}
 		\nabla_{\rd_i} \rd_j= \frac{e^{2\gamma} }{2N}\cdot \left( \partial_q \beta^q \cdot \delta_{i j}-\partial_i \beta^q \cdot \delta_{ j q}  -\partial_j \beta^q  \cdot \delta_{i q}                          \right) \n + \left(\delta^q_i \partial_j \gamma + \delta^q_j \partial_i \gamma-\delta^{i j} \delta^{q l} \partial_l \gamma  \right) \rd_q, \end{equation}
 		\begin{equation} \label{D0e0}
 		\nabla_{e_0} e_0= \frac{e_0 N }{N}\cdot e_0+  e^{-2\gamma} \delta^{i j} N \partial_i N \rd_j,\end{equation}
 		\begin{equation} \label{D0ei}
 		\nabla_ {e_0}\rd_i= \nabla_{\rd_i} e_0+\partial_i \beta^j \partial_j = \frac{\partial_i N }{N}\cdot e_0+\frac{1 }{2}\cdot \left( \partial_q \beta^q \cdot \delta_{i}^{ j}+\partial_i \beta^j   -\delta_{ i q} \delta^{j l}\partial_l \beta^q            \right) \rd_j,  \end{equation}
 		
 	\end{lem}

 	\subsection{Eikonal functions and null frames}\label{XELexpressionsection}
 	
 We will define three eikonal functions together with null hypersurfaces and null frames. Each of these will later be chosen to \blue{be} adapted to one propagating wave.
 	
 	\begin{defn}[Eikonal functions]
 	Given a space-time $(  I \times \RR^2 , g)$ of the form \eqref{metric2+1} satisfying Definition~\ref{def:gauge}, define three eikonal functions $u_k$, $k=1,2,3$ corresponding to the three impulsive waves, as  the unique solutions to
 		
 		\begin{equation} \label{eikonal1}
 	\gi^{\nu\color{black}\bt} \partial_{\nu\color{black}} u_k \partial_{\beta}  u_k =0,
 		\end{equation}
 		\begin{equation} \label{eikonalinit}
 		(u_k)_{|\Sigma_0}= a_k +  c_{k j }x^j,
 		\end{equation}
 	which satisfies $e_0 u_k >0$. Here, $a_k$, $c_{k j } \in \RR$ are constants obeying  the following three conditions \begin{equation} \label{cnormalization}
 		\sqrt{c_{k1}^2+c_{k2}^2}= 1 ,
 		\end{equation} \begin{equation} \label{cangle} \begin{split}
 		| c_{k 1} \cdot c_{k' 1} + c_{k 2} \cdot  c_{k' 2}| \geq \upkappa_0 ,\\  | -c_{k 2} \cdot c_{k' 1} + c_{k 1} \cdot c_{k' 2}|= 1-|c_{k 1} \cdot c_{k' 1}+ c_{k 2} \cdot c_{k' 2}| \geq \upkappa_0 ,\end{split}
 		\end{equation}  for some fixed constant $\upkappa_0 \in (0,\frac{\pi}{2})$  and for every  $k \neq k' \in \{1,2,3\}$. 
 	\end{defn} 
 	\begin{defn}[Sets associated with the eikonal functions] \label{def:eikset}
 		Let $u_k$ ($k=1,2,3$) satisfying \eqref{eikonal1} and \eqref{eikonalinit} in $(  I \times \RR^2, g)$ be given.
 		\begin{enumerate}
 			\item For all $w \in \mathbb{R}$, define
 			\begin{equation} \label{defCu}
 			C_{w}^k = \{ (t,x): u_k(t,x)=w\}, \quad C_{\leq w}^k := \bigcup\limits_{u_k \leq w}  	C_{u_k}^k, \quad C_{\geq w}^k := \bigcup\limits_{u_k \geq w}  	C_{u_k}^k.
 			\end{equation}
 			\item For all $w_1,\,w_2 \in \mathbb{R}$,  $w_2 \in \mathbb{R}$, define \begin{equation} \label{defStwosided}
S^k(w_1,w_2):= \bigcup\limits_{ w_1 \leq u_k \leq w_2}  	C_{u_k}^k.
 			\end{equation}
 			\item  Define (what we will later understand as) ``the singular zone'' for $\tphi$: for any $\delta_0>0$ \begin{equation} \label{defS}
 			S_{\delta_0}^k :=  S^k(-\delta_0,\delta_0)= \bigcup\limits_{-\delta_0 \leq u_k \leq \delta_0 } 	C_{u_k}^k.
 			\end{equation}
 		\end{enumerate}
 	\end{defn}

 	\begin{defn}[Definition of the null frame]\label{def:null.frame}
 		\begin{enumerate}
 			\item Define the null vector $\Lgeo$ associated to the eikonal function $u_k$ by
 			\begin{equation} \label{Lgeodefinition}
 			\Lgeo=-\gi^{\nu\color{black}\bt} \partial_{\beta} u_k \cdot   \partial_{\nu\color{black}}.
 			\end{equation}
 			\item Define $L_k$ to be the vector field parallel to $\Lgeo$ which satisfies $L_k t = N^{-1}$, i.e.
 			\begin{equation} \label{Ldefinition}
 			L_k= \mu_k \cdot \Lgeo,\quad \mu_k =  (N \cdot \Lgeo t)^{-1}.
 			\end{equation}
 			\item  Define the vector field $X_k$ to be the unique vector tangential to $\Sigma_{t}$  which is everywhere orthogonal (with respect to $\bar{g}$) to $C_{u_k}^k \cap \Sigma_t$ and such that $g(X_k,L_k)=-1$.
 			\item Define $E_k$ to be the unique vector field which is tangent to $C_{u_k}^k \cap \Sigma_t$, satisfies $ g(E_k,E_k)=1$ and such that $(X_k,E_k)$ has the same orientation as $(\partial_1,\partial_2)$. 
 		\end{enumerate}
 	\end{defn}
 	
 	\begin{lem}\cite[Lemma~2.11]{LVdM1}
 		\begin{enumerate}
 			\item $\Lgeo$ is null and geodesic, i.e.
 			\begin{equation} \label{eikonal2}
 		g(\Lgeo, \Lgeo) = 0,\quad \nabla_{\Lgeo} \Lgeo=0.
 			\end{equation}
 			\item The following holds:
 			\begin{equation}\label{eq:silly.tangential}
 			L_k u_k = E_k u_k = 0,\quad E_k t = X_k t = 0,\quad L_k t = N^{-1},\quad X_k u_k = \mu_k^{-1}.
 			\end{equation}
 			\item The normal $\n$ can be expressed in terms of $X_k$ and $L_k$ as:	
 			\begin{equation} \label{nXEL}
 			\n = L_k +  X_k .
 			\end{equation}
 			\item The triplet $(X_k,E_k,L_k)$ forms a null frame, i.e.~it satisfies %\footnote{This terminology precisely means that the triplet  $(X_k,E_k,L_k)$ satisfies \eqref{XELframecondition}.}, which obeys the following orthonormality relations: 
 			\begin{equation} \label{XELframecondition}
 			g(L_k,X_k) = -1, \hskip 7 mm g(E_k,L_k)=g(E_k,X_k)=g(L_k,L_k)=0, \hskip 7 mm g(E_k,E_k)=g(X_k,X_k)=1 .
 			\end{equation}
 			\item $g^{-1}$ can be given in terms of the $(X_k, E_k, L_k)$ frame by
 			\begin{equation} \label{inversegXEL}
 			g^{-1} = -  L_k \otimes L_k - L_k \otimes X_k  -  X_k \otimes L_k  + E_k \otimes E_k.
 			\end{equation}
 		\end{enumerate}
 	\end{lem}	
\subsection{Ricci coefficients with respect to the $\{X_k, E_k, L_k\}$ frame} \label{riccinullframesection}

We now define some Ricci coefficients in terms of the frame ${\{X_k, E_k, L_k\}}$:
\begin{equation} \label{chi}
\chi_k= g(\nabla_{E_k} L_k, E_k)=-g(\nabla_{E_k} E_k,L_k),
\end{equation}
\begin{equation} \label{eta}
\eta_k=   g(\nabla_{X_k} L_k, E_k)=- g(\nabla_{X_k} E_k, L_k).
\end{equation}
All the other Ricci coefficients can, in fact, be determined from $\chi_k$, $\eta_k$ and contractions of $K$.
\begin{lem}\cite[Lemma~2.19]{LVdM1} \label{riccibarXEL} The following identities hold:
	\begin{equation} \label{EL}
	\nabla_{E_k} L_k = \chi_k \cdot E_k -K(E_k,X_k) L_k ,
	\end{equation}
	\begin{equation} \label{LE}
	\nabla_{L_k} E_k = (E_k \log(N)-K(E_k,X_k) )\cdot L_k ,
	\end{equation}
	\begin{equation} \label{EL-LE}
	\left[E_k,L_k \right] = \chi_k \cdot E_k - E_k \log(N) \cdot L_k.
	\end{equation}	\begin{equation} \label{EX}
	\nabla_{E_k} X_k = K(E_k,X_k) X_k +  (K(E_k,E_k)- \chi_k) \cdot E_k +  K(E_k,X_k) L_k,
	\end{equation}
	\begin{equation} \label{XE}
	\nabla_{X_k} E_k =  \eta_k X_k + K(E_k,X_k) L_k %=\frac{|X_k|}{2} \cdot \eta_k \cdot \n + \frac{1}{2} \cdot \eta_k \cdot X_k			+\frac{1}{2} \cdot \left( \barnu -|X_k|^2 \cdot \bareta \right)L_k 
	, \end{equation}
	\begin{equation} \label{EX-XE}
	\left[E_k,X_k\right]  =   (K(E_k,X_k)-\eta_k) \cdot X_k +  (K(E_k,E_k)- \chi_k) \cdot E_k, 
	\end{equation} 			\begin{equation} \label{LX}
	\nabla_{L_k} X_k =(-K(E_k,X_k) + E_k\log N) \cdot E_k  -(K( X_k,X_k) -X_k\log(N)) \cdot  X_k-( K( X_k,X_k) -X_k\log(N) )\cdot L_k,
	\end{equation}
	\begin{equation} \label{XL}
	\nabla_{X_k} L_k = \eta_k \cdot E_k -K(X_k,X_k) \cdot L_k,
	\end{equation}
	\begin{equation} \label{LX-XL}
	[L_k,X_k] =-(K(E_k,X_k)\ -E_k\log N  +\eta_k) \cdot E_k  - (K( X_k,X_k) -X_k\log(N)) \cdot  X_k+  X_k\log(N) \cdot L_k,
	\end{equation}	
	\begin{equation} \label{EE2}
	\nabla_{E_k} E_k=  \chi_k \cdot X_k+  K(E_k,E_k) \cdot L_k,\end{equation} 
	\begin{equation} \label{XX}
	\nabla_{X_k} X_k = K(X_k,X_k) \cdot X_k+  (K(E_k,X_k)-\eta_k) \cdot E_k + K(X_k,X_k) \cdot  L_k,
	\end{equation}	
	\begin{equation} \label{LL}
	\nabla_{L_k}L_k = (K( X_k,X_k) -X_k\log(N)) \cdot L_k. 
	\end{equation}

\end{lem}

 \subsection{Geometric coordinate systems $(u_k,\theta_k,t_k)$ and $(u_k,u_{k'})$} \label{relationXELgeocoordinatesection}
 
 \subsubsection{Spacetime coordinate system $(u_k,\theta_k,t_k)$} 
 We now introduce the coordinate $\theta_k$ such that $(u_k,\theta_k,t_k)$ is a regular coordinate system on $I \times \mathbb R^2$.
 \begin{defn}
 	\begin{enumerate}
 		\item Given $u_k$ satisfying \eqref{eikonal1}--\eqref{eikonalinit}, and fixing some constants $b_k$, define $\th_k$ by
 		\begin{equation} \label{thetadef}
 		L_k\theta_k =0,
 		\end{equation}
 		\begin{equation} \label{thetainit}
 		(\theta_k)_{|\Sigma_0}=  b_k + c_{k j}^{\perp} x^j,
 		\end{equation} 
 		where $c_{k 1}^{\perp} = -c_{k 2}$ and $c_{k 2}^{\perp} = c_{k 1}$, and $c_{k i}$ are the constants in \eqref{eikonalinit}.
 		\item Let $t_k = t$. 
 		\item Denote by $(\partial_{u_k},\partial_{\theta_k}, \partial_{t_k})$ the coordinate vector fields in the $(u_k,\theta_k,t_k)$ coordinate system. (Note that we continue to use $\rd_t$ to denote the coordinate derivative in the $(x^1,x^2,t)$ coordinate system of Section~\ref{ellipticgaugedef}.)
 	\end{enumerate}
 \end{defn}

 \begin{lem}\label{lem:XEL.in.rd}  \cite[Lemma~2.13]{LVdM1}
 	Defining $\varTheta_k = (E_k\theta_k)^{-1}$ and $\Xi_k = X_k\th_k$, we have
 	\begin{equation} \label{thetaE}
 	L_k = \f 1N \cdot \rd_{t_k},\quad E_k =  \varTheta_k^{-1} \cdot \rd_{\th_k}, \quad X_k = \mu_k^{-1}\cdot \partial_{u_k}+ \Xi_k \cdot  \partial_{\theta_k}.
 	\end{equation}
 \end{lem}

 \begin{lem} \cite[(2.46)]{LVdM1}
The metric $g$ in the $(t_k,u_k,\theta_k)$ coordinate system is given by
 		\begin{equation} \label{gthetau}
 		g= \varTheta_k^2 \,d\theta^2_k - 2 \mu_k N \,dt_k\, du_k - 2 \mu_k \varXi_k  \varTheta_k^2 \,du_k\, d\theta_k + \mu_k^2(1+  \varXi_k^2  \varTheta_k^2 ) \,du_k^2.
 		\end{equation} 
 	
 \end{lem}

 \subsubsection{Spatial coordinate system $(u_k,u_{k'})$ on $\Sigma_t$} \label{ukuk'coordinatesection}
 
 Fix $k,\,k'\in \{1,2,3\}$ with $k \neq k'$. Introduce the spatial coordinate system $(u_k,u_{k'})$. So as to distinguish it from other coordinate derivatives, we define the coordinate vector fields on $\Sigma_t$ in the $(u_k,u_{k'})$ coordinate system by $(\partialuk, \partialukp)$. 
 
 We now express $(\partialuk, \partialukp)$ in terms of $(X_k,E_k)$ in the following lemma: 
 \begin{lem} \label{ukuk'toXE.lemma}\cite[Lemma~2.18]{LVdM1}
 	The vector fields $X_k$ and $E_k$ can be expressed in the $(u_k, u_{k'})$ coordinate system as follows:
 	\begin{equation} \label{Xukukprime}
 	X_k =\mu_{k}^{-1} \cdot  \partialuk+ \mu_{k'}^{-1} \cdot g(X_k,X_{k'}) \cdot 	\partialukp	,
 	\end{equation}
 	\begin{equation} \label{Eukukprime}
 	E_k= \mu_{k'}^{-1} \cdot g(E_k,X_{k'}) \cdot  	\partialukp.
 	\end{equation}
 	The above transformation can be inverted to give
 	\begin{equation} \label{partialukpEX}
 	\partialukp=  \mu_{k'} \cdot g(E_k,X_{k'}) ^{-1} E_k,
 	\end{equation}
 	\begin{equation} \label{partialukEX}
 	\partialuk = \mu_k X_k- \frac{ \mu_k\cdot g(X_k,X_{k'})}{ g(E_k,X_{k'})} \cdot E_k.
 	\end{equation}
 	
 \end{lem}

 	\subsection{Transformations between different vector field bases and different coordinate systems} \label{coordinatetransform}

 	\subsubsection{Relations on $\Sigma_t$ between $(X_k,E_k)$ and the elliptic coordinate vector fields $(\partial_1,\partial_2)$ }
 	
 Recall that we fixed the orientation of $(X_k, E_k)$ to be the same  as $(\rd_1, \rd_2)$.
 	
 	\begin{lem}\cite[Lemma~2.16]{LVdM1}
 		We have the following identity between $E^i_k$ and $X^i_k$ : 
 		\begin{equation} \label{EXinellipticcoord}  \begin{split}
 		E^1_k=-X_k^2, \quad
 		E^2_k=X_k^1. \end{split}
 		\end{equation}
 		
 		Moreover, the coordinate vector fields $(\partial_1,\partial_2)$ can be expressed in terms of $(E_k,X_k)$ as : 
 		
 		\begin{equation} \label{partial1EX}
 		\partial_1 = e^{2\gamma} \cdot \left( -X^2_k \cdot E_k +  E^2_k \cdot X_k \right),
 		\end{equation}
 		\begin{equation} \label{partial2EX}
 		\partial_2 = e^{2\gamma} \cdot \left( X^1_k \cdot E_k -  E^1_k \cdot X_k \right).
 		\end{equation}

 	\end{lem}

 	\subsubsection{Elliptic coordinate derivatives of $(u_k,\theta_k,t_k)$}	Recall that $t_k=t$ so by definition of the elliptic gauge, $\partial_i t_k =0$ and $\partial_t t_k =1$.	Now we are going to compute the non-trivial coefficients of the Jacobian between the elliptic coordinate system $(x^1,x^2,t)$ and the geometric coordinate system $(u_k,\theta_k,t_k)$.

 	 	\begin{lem} \label{duprop}  \cite[Lemma~2.17]{LVdM1} 		We have the following identities :  		\begin{equation} \label{ucartderivative} 		\partial_i u_k = e^{2\gamma} \cdot \mu_k^{-1} \cdot  \delta_{i j}X^j_k, 		\end{equation} 		\begin{equation}  \label{dtui}  	\partial_{t} u_k =  \beta^q \partial_q u_k+ N \cdot \mu_k^{-1} .\end{equation} 		\begin{equation} \label{thetacartderivative}	\partial_i \theta_k = e^{2\gamma}  \delta_{i j} \cdot \left(\varTheta_k^{-1} \cdot E^j_k + \Xi_k  \cdot X^j_k\right), 		\end{equation}		\begin{equation} \label{dttheta} 		\partial_t \theta_k = \beta^i \partial_i \theta_k+ N \cdot   \Xi_k = e^{2\gamma} \cdot \beta_j \cdot \left(\varTheta_k^{-1} \cdot E^j_k + \Xi_k  \cdot X^j_k\right)+ N \cdot  \Xi_k. 		\end{equation} 	 		Moreover, for all vector field $Y$ in the tangent space of $\Sigma_t$ we have \begin{equation} \label{Yu}Yu_k = \mu_k^{-1}  \cdot g(Y,X_k). 		\end{equation}

 		\end{lem}

 	 		\section{Function spaces and norms}\label{sec:norms}
 	 	
 	 	This section is devoted to the definition of all the function spaces and norms that are used throughout the remainder of the paper.
 	 	
 	 	\subsection{Pointwise norms}
 	 	
 	 	\begin{defn}\label{def:pointwise.norm}
 	 		Define the following pointwise norms in the coordinate system $(t,x^1,x^2)$ associated to the elliptic gauge (see Section~\color{black}\ref{ellipticgaugedef}):
 	 		\begin{enumerate}
 	 			\item Given a scalar function $f$, define
 	 			$$|\rd_x f|^2 := \sum_{i=1}^2 (\partial_{i}f)^2, \quad 
 	 			|\partial f|^2 :=  \sum_{\bt =0}^2 (\partial_{\bt} f)^2.$$
 	 			\item Given a higher order tensor field, define its norm and the norms of its derivatives componentwise, e.g.
 	 			$$|\bt|^2:= \sum_{i=1}^2 |\bt^i|^{2}\color{black},\quad |\rd_x \bt|^2:= \sum_{i,j =1}^2|\rd_i \bt^j|^2,\quad |K|^2:= \sum_{i,j=1}^2 |K_{ij}|^2, \quad |\rd K|^2:= \sum_{\bt=0}^2 \sum_{i,j=1}^2 |\rd_\bt K_{ij}|^2\quad etc.$$
 	 			\item Higher derivatives are defined analogously, e.g.
 	 			$$|\rd^2 f|^2 := \sum_{\bt\color{black},\sigma = 0}^2  (\rd_{\bt\color{black}\sigma}^2  f)^2, \quad |\rd \rd_x K|^2:= \sum_{\substack{\bt\color{black} =0,1,2 \\ i,j,l = 1,2}} |\rd_{\bt\color{black}}\rd_i K_{jl}|^2, \quad  etc.$$
 	 		\end{enumerate}
 	 	\end{defn}
 	 	
 	 	\subsection{Lebesgue and Sobolev spaces on $\Sigma_t$}
 	 	
 	 	Unless otherwise stated, all Lebesgue spaces are defined with respect to the measure $dx^1\, dx^2$ (which is in general different from the volume form induced by $\bar{g}$). 
 	 	
 	 	Before we define the norms, we define the following weight function.
 	 	\begin{defn}[Japanese brackets]\label{def:JapBra}
 	 		Define $ \la x \ra := \sqrt{1+|x|^2}$ for $x\in \RR^2$ and $\la s \ra := \sqrt{1 + s^2}$ for $s\in \RR$.
 	 	\end{defn}

 	 	\begin{defn}[$C^k$ and H\"older norms]\label{def:Holder}
 	 		For $k\in \mathbb N \cup \{0\}$ and $s \in (0,1)$, define $C^k(\Sigma_t)$ to be the space of continuously \underline{spatially} $k$-differentiable functions with respect to elliptic gauge coordinate vector fields $\rd_x$ with norm 
			$$\|f \|_{C^k(\Sigma_t)} := \sum_{|\bt\color{black}|\leq k} \sup_{\Sigma_t} |\rd_x^{\bt\color{black}} f|,$$ 
			and define $C^{k,s}(\Sigma_t) \subseteq C^k(\Sigma_t)$ with \underline{spatial} H\"older norm defined with respect to the elliptic gauge coordinates as 
			$$\|f\|_{C^{k,s}(\Sigma_t)} := \|f \|_{C^k(\Sigma_t)} + \underset{\substack{ x,y\in \Sigma_t \\ x\neq y}}{\sup}\ \underset{|\bt\color{black}| = k}{\sum} \f{|\rd_x^{\bt\color{black}} f(x) - \rd_x^{\bt\color{black}} f(y)|}{|x-y|^s}.$$
 	 	\end{defn}
		
		In the later parts of the paper, we will need to consider H\"older spaces in both the $(x^1,x^2)$ coordinates and the $(u_k, u_{k'})$ coordinates. When we need to emphasize the distinction, we will use the notation $C^{0,\sigma}_{x^1,x^2}(\Sigma_t) = C^{0,\sigma}(\Sigma_t)$ and $C^{0,\sigma}_{u_k, u_{k'}}(\Sigma_t)$.
 	 	
 	 	\begin{defn}[Standard Lebesgue and Sobolev norms]\label{def:Sobolev.norm}
 	 		\begin{enumerate}
 	 			\item For $k\in \mathbb N \cup \{0\}$ and $p \in [1,+\infty)$, define the (unweighted) Sobolev norms $$ \| f \|_{W^{k,p}(\Sigma_t)} := \sum_{|\bt| \leq k } \left(\int_{\Sigma_t}  |\partial_x^{\bt} f|^p(t,x^1,x^2)\, dx^1 dx^2 \right)^{\frac{1}{p}}.$$
 	 			For $k\in \mathbb N \cup \{0\}$, define
 	 			$$\| f \|_{W^{k,\infty}(\Sigma_t)} := \sum_{|\bt|\leq k} \esssup_{(x^1,x^2)\in \Sigma_t} |\partial_x^{\bt} f|\color{black}(t,x^1,x^2).$$
 	 			\item Define $L^p(\Sigma_t) :=W^{0,p}(\Sigma_t)$ and $H^k(\Sigma_t):=W^{k,2}(\Sigma_t)$.
 	 		\end{enumerate}
 	 	\end{defn}
 	 	
 	 	\begin{defn}[Fractional Sobolev norms]\label{def:fractional.Sobolev.norm}
 	 		For $s\in \mathbb R\setminus (\mathbb N \cup \{0\})$, define $H^{s}(\Sigma_t)$ by
 	 		$$\|f\|_{H^{s}(\Sigma_t)} := \|\Db^s f\|_{L^2(\Sigma_t)}.$$
 	 		where $\Db^s$ is defined via the (spatial) Fourier transform $\mathcal F$ (in the $x$ coordinates) by $\mathcal F(\Db^s f):= \la \xi\ra^s \mathcal F$.
 	 	\end{defn}
 	 	
 	 	\begin{defn}[Weighted norms]\label{def:weighted.Sobolev.norm}
 	 		\begin{enumerate}
 	 			\item For $k\in \mathbb N \cup \{0\}$, $p \in [1,+\infty)$ and $r\in \RR$, define the weighted Sobolev norms by
 	 			$$ \| f \|_{W^{k,p}_r(\Sigma_t)}= \sum_{|\bt| \leq k } \left(\int_{\Sigma_t} \la x \ra^{p \cdot (r+|\bt|)} |\partial_x^{\bt} f|^p(t,x^1,x^2) \,dx^1 \,dx^2 \right)^{\f 1p},$$
 	 			with obvious modifications for $p = \infty$. 
 	 			\item Define also $L^p_r(\Sigma_t) := W^{0,p}_r(\Sigma_t)$ and $H^k_r(\Sigma_t) := W^{k,2}_r(\Sigma_t)$. Moreover, define $C^k_r(\Sigma_t)$ as the closure of Schwartz functions under the $L^\i_r(\Sigma_t)$ norm.
 	 		\end{enumerate}
 	 	\end{defn}
 	 	
 	 	\begin{defn}[Mixed norms]\label{def:mixed.Sobolev.norm}
 	 		We will use mixed Sobolev norms, mostly in the $(u_k,\th_k,t_k)$ coordinates in spacetime or the $(u_k,u_{k'})$ coordinates on $\Sigma_t$. Our convention is that the norm on the right is taken first. For instance,
 	 		$$ \| f \|_{L^2_{u_{k'}} L^{\infty}_{u_k}(\Sigma_t)}= \left(\color{black}\int_{u_{k'} \in \RR} (\sup_{u_k \in \RR} f(t,u_k,u_{k'}))^2 d u_{k'}\right)\color{black}^{\frac{1}{2}},$$
 	 		and analogously for other combinations.
 	 	\end{defn}
 	 	
 	 	\begin{defn}[Norms for derivatives]
 	 		We combine the notations in Definition~\ref{def:pointwise.norm} with those in Definitions~\ref{def:Sobolev.norm}--\ref{def:mixed.Sobolev.norm}. For instance, given a scalar function $f$,
 	 		$$\|\rd f\|_{L^2(\Sigma_t)} :=  \left( \color{black} \int_{\Sigma_t}  \sum_{\bt = 0}^2\color{black} |\rd_{\bt\color{black}} f|^2 \, dx^1\, dx^2 \right) \color{black} ^{\f 12},$$
 	 		and similarly for $\|\rd_x f \|_{L^2(\Sigma_t)}$, $\|\rd\rd_x f\|_{L^2(\Sigma_t)}$, etc.
 	 	\end{defn}

 	 	\subsection{The Littlewood--Paley projection and Besov spaces in $(u_k,u_{k'})$ coordinates}\label{sec:LPukukp}
 	 	
 	 	Assume for this subsection that $k \neq k'$, so that $(u_k, u_{k'})$ forms a coordinate system on $\Sigma_t$. 
 	 	
 	 	\begin{defn}[Littlewood--Paley projection]\label{def:Littlewood.Paley}
 	 		Define the Fourier transform in the $(u_k, u_{k'})$ coordinates by
 	 		$$(\mathcal F^{u_k, u_{k'}} f)(\xi_k, \xi_{k'}) =  \iint_{\mathbb R^2} f(u_k, u_{k'}) e^{-2\pi i(u_k \xi_k + u_{k'} \xi_{k'})}\, du_k \, du_{k'}.$$
 	 		Let $\varphi:\mathbb R^2 \to [0,1]$ be radial, smooth such that $\varphi(\xi) = \begin{cases} 1\quad \mbox{for $|\xi|\leq 1$} \\
 	 		0\quad \mbox{for $|\xi|\geq 2$}
 	 		\end{cases}$, where $|\xi| = \sqrt{|\xi_k|^2 + |\xi_{k'}|^2}$. 
 	 		
 	 		Define $P_0^{u_k,u_{k'}}$ by 
 	 		$$P_0^{u_k,u_{k'}} f := (\mathcal F^{u_k, u_{k'}})^{-1} (\varphi(\xi) \mathcal F^{u_k, u_{k'}} f),$$
 	 		and for $q\geq 1$, define $P_q^{u_k,u_{k'}} f$ by
 	 		$$P_q^{u_k,u_{k'}} f := ( \mathcal F^{u_k, u_{k'}})^{-1} ((\varphi(2^{-q} \xi) - (\varphi(2^{-q+1} \xi)) \mathcal F^{u_k, u_{k'}} f(\xi)).$$
 	 	\end{defn}
 	 	
 	 	\begin{defn}[The Besov space $B^{u_k,u_{k'}}_{\infty,1}$]\label{def:Besov}
 	 		Define the Besov norm $\Bes$ by
 	 		$$ \|f \|_{\Bes}:= \sum_{q \geq 0} \|P^{u_k, u_{k'}}_q f \|_{L^{\infty}(\Sigma_t)}.$$
 	 	\end{defn}
 	 	
 	 	\subsection{Lebesgue norms on $C^k_{u_k}$ and $\Sigma_t\cap C^k_{u_k}$}
 	 	Recall the definition of $C^k_{u_k}$ from Definition~\ref{def:eikset}. The $L^2$ norm on $C^k_{u_k}$ is defined with respect to the measure $d\th_k\, dt_k$.
 	 	\begin{defn}[$L^2$ norm on $C^k_{u_k}$]
 	 		For every fixed $u_k$, define the $L^2(C_{u_k}^k([0,T)))$ norm by
 	 		$$\| f\|_{L^2(C_{u_k}^k([0,T)))} :=   \left( \color{black} \int_0^T \int_{\mathbb R} |f|^2(u_k,\th_k,t_k) \, d\th_k\, dt_k  \right) \color{black} ^{\f 12}.$$
 	 	\end{defn}
 	 	
 	 	The $L^2$ norm $\Sigma_t\cap C^k_{u_k}$ is defined with respect to the measure $d\th_k$.
 	 	\begin{defn}[$L^2$ norm on $\Sigma_t \cap C_{u_k}^k$]\label{def:L2.in.th_k}
 	 		For every fixed $t$ and $u_k$ (and recall $t=t_k$), define the $L^2_{\th_k}(\Sigma_t \cap C_{u_k}^k)$ norm by
 	 		$$\| f\|_{L^2_{\th_k}(\Sigma_t \cap C_{u_k}^k)} :=   \left( \color{black} \int_{\mathbb R} |f|^2(u_k,\th_k,t_k) \, d\th_k \right) \color{black} ^{\f 12}.$$
 	 	\end{defn}

\section{Main results}\label{sec:main.results}

\subsection{Data assumptions for $\de$-impulsive waves} \label{section:data_recalling}

 Recall that in our companion paper \blue{\cite{LVdM1}}, we defined what it means for $(\phi,\phi',\gamma,K)$ to be an \textbf{admissible initial data set featuring three $\delta$-impulsive waves with parameters $(\ep,s',s'',R,\upkappa_0,\delta)$}  (\cite{LVdM1}, Definition 4.8) for parameters in the ranges $\delta>0$, $\ep>0$, $0 < s'' < s'<\frac{1}{2}$, $0< s'-s''<\frac{1}{3}$, $R>10$ and $\upkappa_0>0$. \blue{Here, $\phi$, $\gamma$ and $K$ are the initial data for these quantities, while $\phi'$ will be the initial data of $\n \phi$ (where $\n$ is as in \eqref{defnormal}).}

In particular, we recall that this definition requires that there exists $\rphi$, $\rphi'$, $\tphi$, $\tphi'$ ($k=1,2,3$) such that  $\mathrm{supp}(\rphi),\,\mathrm{supp}(\rphi'),\,\mathrm{supp}(\tphi),\, \mathrm{supp}(\tphi') \subseteq B(0,\f R2):= \{ (x^1,x^2) \in \Sigma_0 ,\ \sqrt{(x^1)^2+ (x^2)^2} < \f R 2\}$ and $\mathrm{supp}(\tphi) \cup \mathrm{supp}(\tphi') \subseteq \{u_k \geq -\delta\}$ (where $u_k$ solves the equation \eqref{eikonal1}, \eqref{eikonalinit} with parameters obeying the conditions \eqref{cnormalization}, \eqref{cangle} for $k=1,2,3$) and moreover that the following (in)equations be satisfied on $\Sigma_0$: \begin{equation} \label{Data1}
\phi= \rphi + \sum_{k=1}^3 \widetilde{\phi}_k, \hskip 5 mm\phi'= \rphi' + \sum_{k=1}^3 \tphi',
  \end{equation} 	\begin{equation} \label{data1}
  \|\rphi\|_{H^{2+s'}(\Sigma_0)} + \|\rphi' \|_{H^{1+s'}(\Sigma_0)} \leq \ep,
  \end{equation}
  \begin{subequations}
  	\begin{align}
  	\label{eq:assumption.rough.energy}
  	\|\tphi\|_{W^{1,\infty}(\Sigma_0)} + \|\tphi\|_{H^{1+s'}(\Sigma_0)} + \|\tphi' \|_{L^\infty(\Sigma_0)} + \|\tphi'\|_{H^{s'}(\Sigma_0)} \leq &\: \ep, \\
  	\label{eq:assumption.rough.energy.commuted}
  	\|E_k\tphi\|_{H^{1+s''}(\Sigma_0)} + \|E_k \tphi' \|_{H^{s''}(\Sigma_0)} + \| \tphi' - X_k \tphi \|_{H^{1+s''}(\Sigma_0)} \leq &\: \ep,
  	\end{align}
  \end{subequations}
  
  \begin{align} 
  %		 \| \tphi\|_{H^1(\Sigma_0 \cap S^k_{2\de})}+ \|  \tphi'\|_{L^2(\Sigma_0 \cap S^k_{2\de})} + \sum_{ Z_k \in \{E_k,L_k\} } (\|  Z_k  \tphi\|_{H^1( \Sigma_0 \cap S^k_{2\de})}+ \|  Z_k   \tphi'\|_{L^2( \Sigma_0 \cap S^k_{2\de})})\leq &\: \epsilon \cdot \delta^{\frac{1}{2}}, \\
  \label{eq:delta.waves.1}
  \|  \tphi\|_{H^2(\Sigma_0)}+ \| \tphi'\|_{H^1(\Sigma_0)} +  \| E_k \tphi\|_{H^2(\Sigma_0)}+ \| E_k  \tphi'\|_{H^1(\Sigma_0)} + \|\tphi' - X_k\tphi \|_{H^2(\Sigma_0)} \  \leq &\:  \epsilon \cdot \delta^{-\frac{1}{2}}, \\
  \label{eq:delta.waves.2}
  \| \tphi \|_{H^2(\Sigma_0 \setminus S^k(-\de,0))} + \| \tphi' \|_{H^1(\Sigma_0 \setminus S^k(-\de,0))} \leq &\: \ep. %\\
  %	 \label{eq:delta.waves.4}
  %	 \|  \tphi\|_{H^3(\Sigma_0)}+ \| \tphi'\|_{H^2(\Sigma_0)}  \lesssim &\:  \epsilon \cdot \delta^{-\frac{3}{2}}.
  \end{align}

% \eqref{cangle} satisfied by the data:

In the sequel, we shall always consider solutions of the system of equations constituted of \eqref{Data1},  and  \begin{equation}
 Ric_{\mu \nu}(g) = 2 \partial_{\mu} \phi  \partial_{\nu} \phi ,
 \end{equation}\begin{equation}
 \Box_g \tilde{\phi}_1 = \Box_g \tilde{\phi}_2 = \Box_g \tilde{\phi}_3= \Box_g \rphi =0,
 \end{equation}
 with data on $\Sigma_0$ given by an \textbf{admissible initial data set featuring three $\delta$-impulsive waves $(\phi,\phi',\gamma,K)$  with parameters $(\ep,s',s'',R,\upkappa_0,\delta)$} in the above ranges and assuming $0<\ep<\ep_0$, $0<\delta<\delta_0$ with $0<\delta_0<\ep_0$ sufficiently small.

\subsection{Bootstrap assumptions} \label{bootstrapsection}

We will only state the bootstrap assumptions for wave part of the solution. In Part I of our series, we also had bootstrap assumptions for geometric quantities (metric components, Ricci coefficients, etc.), but we then improved all of those assumptions in \cite{LVdM1}. In other words, the results in \cite{LVdM1} can be rephrased as saying that the bounds for the geometric quantities can be proven under the bootstrap assumptions \eqref{BA:rphi}--\eqref{BA:Li} for the wave part. (The more precise statements from \cite{LVdM1} will be recalled below in Propositions~\ref{prop:main.metric.est}, \ref{prop:main.frame.est} and \ref{prop:main.Ricci.est} and Lemmas~\ref{lem:rd.in.terms.of.XEL}, \ref{lem:jacobian}, \ref{lem:angle}.) 

In our main estimates (see theorems in Section~\ref{sec:main.estimates} below, we will work under the following bootstrap assumptions, where the solution is assumed to remain regular in $[0,T_B)$ for some $T_B \in (0,1)$. \color{black}
\medskip

\noindent\underline{\textbf{Energy estimates for $\rphi$.}}
\begin{align}  
\sup_{0 \leq t < T_B} (\| \phi_{reg} \|_{H^{s'}(\Sigma_t)}+\|\partial \phi_{reg} \|_{H^{s'}(\Sigma_t)} + \| \partial^2 \phi_{reg} \|_{H^{s'}(\Sigma_t)}) \leq &\: \epsilon^{\frac{3}{4}}. \label{BA:rphi}
\end{align}

\noindent\underline{\textbf{Energy estimates for $\tphi$.}}
\begin{subequations}
\begin{align} 
\label{bootstrapsmallnessenergy}
\sup_{0 \leq t < T_B} (\| \partial \tphi \|_{L^2(\Sigma_t)} + \sum_{Z_k \in \{L_k,\,E_k\}}\| Z_k \partial \tphi\|_{L^2(\Sigma_t)} ) \leq &\: \epsilon^{\frac{3}{4}}, \\
 \label{tphiH2bootstrap}
\sup_{0 \leq t < T_B} \|  \partial^2 \tphi \|_{L^2(\Sigma_t)} \leq &\: \epsilon^{\frac{3}{4}} \cdot \delta^{-\frac{1}{2}}, \\
 \label{tphiH3/2bootstrap}
\sup_{0 \leq t < T_B} \|  \rd \Db^{s'} \tphi \|_{L^2(\Sigma_t)} \leq &\: \epsilon^{\frac{3}{4}}, \\
\label{EtphiH2bootstrap}
\sup_{0 \leq t < T_B} \| \partial E_k \partial_x \tphi \|_{L^2(\Sigma_t)} \leq &\: \epsilon^{\frac{3}{4}} \cdot \delta^{-\frac{1}{2}}.
\end{align}
\end{subequations}

\noindent\underline{\textbf{Improved energy estimates for $\tphi$.}}

\begin{subequations}\begin{align}
 \label{bootstrapbadunlocenergyhyp}
\sup_{0 \leq t < T_B} ( \|\partial \tphi\|_{L^2(\Sigma_t \cap S_{2\de}^k)} +\sum_{Z_k \in \{L_k,\,E_k\}}\| Z_k \partial \tphi\|_{L^2(\Sigma_t \cap S_{2\de}^k)} ) \leq &\: \epsilon^{\frac{3}{4}} \cdot \sdelta, \\
\label{BA:away.from.singular} \sup_{0 \leq t < T_B}  \| \partial^2 \tphi\|_{L^2(\Sigma_t \setminus \Sd^k)}  \leq  &\: \epsilon^{\frac{3}{4}}.
\end{align}
\end{subequations}

%\begin{equation} \label{bootstrapbadunlocenergyhyp}
%\sup_{0 \leq t < T_B}\|\partial \tphi\|_{L^2(\Sigma_t \cap \Sd)} +  \sup_{-\delta \leq u_k \leq \delta} \sum_{Z_k \in \{L_k,\,E_k\}} \|Z_k \tphi\|_{L^2(C^k_{u_k}\cap[0,T_B))} \lesssim \epsilon^{\frac{3}{4}} \cdot \sdelta,
%\end{equation}	
%\begin{equation} \label{bootstrapELbadunlocenergyhyp}
%\sup_{0 \leq t < T_B}  \sum_{Z_k \in \{L_k,\,E_k\}}\| Z_k \partial \tphi\|_{L^2(\Sigma_t \cap \Sd)} +  \sup_{-\delta \leq u_k \leq \delta} \sum_{Y_k,Z_k \in \{L_k,\,E_k\}} \|Y_k Z_k \tphi\|_{L^2(C^k_{u_k}\cap [0,T_B))} \lesssim  \epsilon^{\frac{3}{4}} \cdot  \sdelta,
%\end{equation}

\noindent\underline{\textbf{Flux estimates for the wave variables.}}
\begin{subequations}
\begin{align}
 \max_k \sup_{u_k \in \mathbb R} \sum_{Z_k \in \{L_k, E_k\}} (\| Z_k \rd_x \rphi\|_{L^2(C^k_{u_k}([0,T_B)))} + \| Z_k \rphi\|_{L^2(C^k_{u_k}([0,T_B)))})  \leq &\: \ep^{\f 34}, \label{BA:flux.for.rphi}\\ 
\max_{k,k'} \sup_{u_{k'} \in\RR} \sum_{Z_{k'} \in \{L_{k'},\,E_{k'}\}} (  \| Z_{k'} \rd_x \tphi\|_{L^2(C^{k'}_{u_{k'}}([0,T_B))\setminus \Sd^k)} +   \| Z_{k'} \tphi\|_{L^2(C^{k'}_{u_{k'}}([0,T_B)))}) \leq &\: \epsilon^{\frac{3}{4}}, \label{BA:flux.for.tphi.improved} \\
\max_{k} \sup_{u_{k} \in\RR} (\| L_k \rd_x \tphi\|_{L^2(C^{k}_{u_k}([0,T_B)))} + \| E_k^2 \tphi\|_{L^2(C^{k}_{u_k}([0,T_B)))} )  \leq &\: \epsilon^{\frac{3}{4}}, \label{BA:flux.for.tphi.improved.2} \\
\max_{k,k'} \sup_{u_{k'} \in\RR} \sum_{Z_{k'} \in \{L_{k'},\,E_{k'}\}} \| Z_{k'} \rd_x \tphi\|_{L^2(C^{k'}_{u_{k'}}([0,T_B)))} \leq  &\:\epsilon^{\frac{3}{4}} \cdot  \delta^{-\frac{1}{2}}. \label{BA:flux.for.tphi}
\end{align}
\end{subequations}

\noindent\underline{\textbf{Besov and $L^\i$ estimates for the wave variables.}}
\begin{subequations}
\begin{align}  
\sup_{0 \leq t < T_B} \max_{(k,k'): k\neq k'} \| \partial \rphi \|_{ \Bes} \leq &\:\epsilon^{\frac{3}{4}}, \label{rphiBbootstrap} \\  
\sup_{0 \leq t < T_B}\max_{k': k' \neq k} \| \partial \tphi \|_{ \Bes} \leq &\: \epsilon^{\frac{3}{4}}, \label{tphiBbootstrap} \\
\sup_{0 \leq t < T_B}(\| \partial \rphi \|_{ L^\infty(\Sigma_t)} + \max_k \| \partial \tphi \|_{ L^\infty(\Sigma_t)}) \leq &\: \ep^{\f 34}. \label{BA:Li}
\end{align}
\end{subequations}

\subsection{Main wave estimates}\label{sec:main.estimates}

The following are the main results for this paper. They are stated and assumed in Part I in order to prove the main existence result for $\de$-impulsive waves and impulsive waves.

\subsubsection{The main Lipschitz and improved H\"older bounds}\label{sec:aprioriestimates2}

\begin{defn}\label{def:Lipschitz.control.norm}
	Define $\mathcal E(t)$ to be the following norm, 
	\begin{equation*}
	\begin{split}
	\mathcal E(t) := &\: \| \partial \Db^{s'} \tphi\|_{L^2(\Sigma_t)} + \| E_k \rd \tphi\|_{L^2(\Sigma_t)} + \| \rd E_k \Db^{s''} \tphi\|_{L^2(\Sigma_t)} + \de^{\f 12} (\| \partial^2 \tphi\|_{L^2(\Sigma_t)} + \|\rd E_k \rd \tphi\|_{L^2(\Sigma_t)}) \\
	&\: + \de^{-\f 12} \|\partial \tphi\|_{L^2(\Sigma_t \cap S_{2\de}^k)} + \de^{-\f 12} \| E_k \partial \tphi\|_{L^2(\Sigma_t \cap S_{2\de}^k)}  +\| \partial^2 \tphi\|_{L^2(\Sigma_t \setminus \Sd^k)}  + \|\rd^2 \Db^{s'} \rphi \|_{L^2(\Sigma_t)}.
	\end{split}
	\end{equation*}
\end{defn}

The following is the main result for obtaining Lipschitz and improved H\"older bounds. It is stated in our previous paper \cite{LVdM1} as \cite[Theorem~7.3]{LVdM1}, and will be proven in this paper.
\begin{theorem}\label{thm:bootstrap.Li}
	Let $(\phi,\phi',\gamma,K)$ be an admissible initial data set featuring three $\de$-impulsive waves with parameters $(\ep,s',s'',R,\upkappa_0)$ (as defined in \cite[Definition 4.3]{LVdM1}) for some $0<\delta<\delta_0$, $0<\ep<\ep_0$, $0 < s'' < s'<\frac{1}{2}$, $0< s'-s''<\frac{1}{3}$, $R>10$ and $\upkappa_0>0$, where $0<\delta_0<\ep_0$ are additionally assumed to be sufficiently small.
	
	Assume the bootstrap assumptions of Section \ref{bootstrapsection} i.e.\ \eqref{BA:rphi}--% \eqref{bootstrapsmallnessenergy}, \eqref{tphiH2bootstrap}, \eqref{tphiH3/2bootstrap}, \eqref{EtphiH2bootstrap}, \eqref{bootstrapbadunlocenergyhyp}, \eqref{BA:away.from.singular},  \eqref{BA:flux.for.rphi},  \eqref{BA:flux.for.tphi.improved},  \eqref{BA:flux.for.tphi.improved.2}, 
	\eqref{BA:Li} hold for some $T_B \in (0,1)$. Then

	$$\mbox{LHSs of \eqref{rphiBbootstrap}--\eqref{BA:Li}} + \sup_{0 \leq t < T_B }(\|\rd\rphi\|_{C^{0,\f{s''}{2}}(\Sigma_t)} + \| \partial \tphi \|_{ C^{0,\frac{s''}{2}}(\Sigma_t \cap C^k_{\geq \delta}) }) \ls \mathcal E,$$
	where the implicit constant in $\ls$ depend only on $s',s'',R,\upkappa_0$.
\end{theorem}

The proof of Theorem~\ref{thm:bootstrap.Li} will be carried out in Section~\ref{dphiLinftysection}.

\subsubsection{Energy estimates}\label{sec:aprioriestimates3}

The following is the main wave energy estimates stated as \cite[Theorem~7.4]{LVdM1}, which we will prove in this paper.

\begin{thm}\label{thm:energyest}
	
		Let $(\phi,\phi',\gamma,K)$ be an admissible initial data set featuring three $\de$-impulsive waves with parameters $(\ep,s',s'',R,\upkappa_0)$ (as defined in \cite[Definition 4.3]{LVdM1}) for some $0<\delta<\delta_0$, $0<\ep<\ep_0$, $0 < s'' < s'<\frac{1}{2}$, $0< s'-s''<\frac{1}{3}$, $R>10$ and $\upkappa_0>0$, where $0<\delta_0<\ep_0$ are additionally assumed to be sufficiently small.
	
	Assume the bootstrap assumptions of Section \ref{bootstrapsection} i.e.\ \eqref{BA:rphi}--% \eqref{bootstrapsmallnessenergy}, \eqref{tphiH2bootstrap}, \eqref{tphiH3/2bootstrap}, \eqref{EtphiH2bootstrap}, \eqref{bootstrapbadunlocenergyhyp}, \eqref{BA:away.from.singular},  \eqref{BA:flux.for.rphi},  \eqref{BA:flux.for.tphi.improved},  \eqref{BA:flux.for.tphi.improved.2}, 
	\eqref{BA:Li} hold for some $T_B \in (0,1)$. Then \begin{enumerate}
		\item \label{wavethm.part1} there  exists $C = C(s',s'',R,\upkappa_0) >0$ such that \eqref{BA:rphi}--% \eqref{bootstrapsmallnessenergy}, \eqref{tphiH2bootstrap}, \eqref{tphiH3/2bootstrap}, \eqref{EtphiH2bootstrap}, \eqref{bootstrapbadunlocenergyhyp}, \eqref{BA:away.from.singular},  \eqref{BA:flux.for.rphi}, \eqref{BA:flux.for.tphi.improved}, \eqref{BA:flux.for.tphi.improved.2}, 
		\eqref{BA:flux.for.tphi} hold with  $C \ep$ in place of $\ep^{\frac{3}{4}}$,
		\item \label{wavethm.part2} the following estimate involving the norm $\mathcal E$ is satisfied:
		$$  \mathcal E \ls \ep,$$
		where the implicit constant in $\ls$ depend only on $s',s'',R,\upkappa_0$.
		\item  \label{wavethm.part3}The following wave energy estimates are satisfied: 
		\begin{subequations}
			\begin{align} 
			\label{eq:main.thm.rphi} 
			\|\rphi \|_{H^{2+s'}(\Sigma_t)} + \|\rd_t \rphi \|_{H^{1+s'}(\Sigma_t)} \ls &\: \ep,\\
			\label{eq:main.thm.tphi.1}
			\| \tphi \|_{H^{1+s'}(\Sigma_t)} +  \| \rd_t \tphi \|_{H^{s'}(\Sigma_t)} \lesssim &\: \epsilon, \\
			\label{eq:main.thm.tphi.2}
			\|    L_k  \tphi \|_{H^{1+s'}(\Sigma_t)}+  \|  E_k  \tphi \|_{H^{1+s'}(\Sigma_t)} +   \| \rd_t L_k \tphi \|_{H^{s'}(\Sigma_t)}+ \| \rd_t E_k \tphi \|_{H^{s'}(\Sigma_t)}  \lesssim &\: \epsilon.
			\end{align}
		\end{subequations}  \begin{equation}\label{eq:smooththeorem.1}
	\|\rd^2 \tphi\|_{L^2(\Sigma_t)} + \sum_{ \substack{ Y_k^{(1)}, Y_k^{(2)}, Y_k^{(3)} \in \{ X_k, E_k, L_k\} \\ \exists i, Y_k^{(i)} \neq X_k} } \|Y_k^{(1)} Y_k^{(2)} Y_k^{(3)} \tphi \|_{L^2(\Sigma_t)} \ls \ep\cdot \de^{-\f 12},
\end{equation}\begin{equation}\label{eq:smooththeorem.top}
\|\phi \|_{H^3(\Sigma_t)} + \|\n \phi\|_{H^2(\Sigma_t)} \ls \ep \cdot \de^{-\f 12} + \|\phi \|_{H^3(\Sigma_0)} + \|\n \phi\|_{H^2(\Sigma_0)},
\end{equation}
\begin{equation}\label{eq:smooththeorem.2}
\| \partial^2  \tphi \|_{ L^2(\Sigma_t \setminus S^k_\de)}  \lesssim \epsilon,
\end{equation}
where, as before, the implicit constant in $\ls$ depend only on $s',s'',R,\upkappa_0$.
	\end{enumerate}

\end{thm}

The proof of Theorem~\ref{thm:energyest} will occupy most of this paper. The conclusion of the proof can be found in Section~\ref{sec:wave.final}.

In view of the parameters that the implicit constants are allowed to depend on in Theorem~\ref{thm:bootstrap.Li} and Theorem~\ref{thm:energyest}, \textbf{from now on, constants $C>0$ or implicit constants in $\ls$ are allowed to depend only on $s',s'',R,\upkappa_0$. We will also often take $\ep_0$ and $\de_0$ to be sufficiently small without further comments.}

\section{Estimates from part I}\label{sec:partI}

In this section, we will assume that $(\phi,\phi',\gamma,K)$ constitute an \textbf{admissible initial data set featuring three impulsive waves with parameters $(\ep,s',s'',R,\upkappa_0,\delta)$}  (recall Section \ref{section:data_recalling}) for parameters in the ranges $0<\delta< \delta_0$, $0<\ep<\ep_0$, $0 < s'' < s'<\frac{1}{2}$, $0< s'-s''<\frac{1}{3}$, $R>10$ and $\upkappa_0>0$, with $0<\delta_0<\ep_0$ sufficiently small. Moreover, we will assume that the bootstrap assumptions of Section \ref{bootstrapsection} i.e.\ \eqref{BA:rphi}--\eqref{BA:Li} hold for some $T_B \in (0,1)$.

We recall the following results follow from \cite[Theorem~7.1]{LVdM1} and its proof.

\begin{lemma}\label{lem:support}
\cite[Lemma~8.1]{LVdM1} The following holds on $\Sigma_t$ for all $t\in [0,T_B)$ and for $k=1,2,3$:
\begin{enumerate}
\item $\mathrm{supp}(\rphi),\,\mathrm{supp}(\tphi) \subseteq B(0,R)$,
\item $\mathrm{supp}(\tphi) \subseteq \{(t,x): u_k(t,x) \geq - \de \}$.
\end{enumerate}
\end{lemma}

Next, we collect some estimates for the metric components in the elliptic gauge proven in \cite{LVdM1}. The first three statements are directly\footnote{Strictly speaking, to obtain the inequality  $\| \Db^{s'} (\varpi (\rd_x \rd_t \mfg)) \|_{L^2(\Sigma_t)} \ls \ep^{\f 32}$ requires using Theorem~\ref{KatoPonce} in addition to \cite[Propositions~9.20]{LVdM1}; we omit the straightforward details.\color{black}} from \cite{LVdM1}, while the fourth statement can be easily derived from the first three.
\begin{proposition}\label{prop:main.metric.est}
\begin{enumerate}
\item Defining $\alp=0.01$, the metric component quantities 
$$\mathfrak{g}\in \{e^{2\gamma}-1,\ e^{-2\gamma}-1,\ \bt^j,\ N-1,\ N^{-1}-1,\ g_{\nu\color{black}\bt} - m_{\nu\color{black}\bt},\ \gi^{\nu\color{black}\bt} - m^{\nu\color{black}\bt} \}$$ 
(where $m$ is the Minkowski metric) satisfy the following estimates:
\begin{equation}\label{eq:g.main}
\sup_{0 \leq t < T_B} (\|\mathfrak{g} \|_{W^{2,\infty}_{-\alp}(\Sigma_t)} + \|\rd_t \mathfrak{g}\|_{L^\i_{-\alp}(\Sigma_t)} + \|\rd_t \mathfrak{g}\|_{W^{1,\f{2}{s'-s''}}_{-s'+s''-\alp}(\Sigma_t)}) \ls \ep^{\f 32}.
\end{equation}
\item \cite[Proposition~9.21]{LVdM1} Taking $\mfg$ as in the previous part,
\begin{equation}\label{eq:g.top}
\sup_{0 \leq t < T_B} \| \rd \rd_x^2 \mfg \|_{L^2(\Sigma_t)} \ls \ep^{\f 32} \de^{-\f 12}.
\end{equation}
\item \cite[Propositions~9.8, 9.20]{LVdM1} Let $\varpi$ be a cutoff such that $\varpi \equiv 1$ on $B(0,2R)$ and $\varpi \equiv 0$ on $\mathbb R^2 \setminus B(0,3R)$. Then for $\mfg$ as above,
\begin{equation}\label{eq:g.top.fractional}
\sup_{0 \leq t < T_B} \|\Db^{s'} \rd^2_{x} \mfg \|_{L^2(\Sigma_t)} + \| \Db^{s'} (\varpi (\rd_x \rd_t \mfg)) \|_{L^2(\Sigma_t)} \ls \ep^{\f 32}.
\end{equation}
\item Define $\Gamma^\lambda = \gi^{\nu\color{black}\bt} \Gamma^{\lambda}_{\nu\color{black}\bt}$, where $\Gamma^{\lambda}_{\nu\color{black}\bt}$ are the Christoffel symbols. Then the following estimates hold:
\begin{equation}\label{eq:Gamma}
\begin{split}
\sup_{0 \leq t < T_B}(\| \Gamma^\lambda \|_{L^{\infty}_{-{3}\alp}(\Sigma_t)}+ &\: \| \Gamma^\lambda \|_{W^{1,\f{2}{s'-s''}}_{-s'+s''-{3}\alp}(\Sigma_t)}) \ls \ep^{\f 32},\quad \sup_{0\leq t <T_B} \|\rd^2_x \Gamma^\lambda \|_{L^2(\Sigma_t\cap B(0,3R))}\ls \ep^{\f 32} \delta^{-\f 1 2}, \\
&\:  \sup_{0 \leq t < T_B} \|\Db^{s'} (\varpi \rd_x \Gamma^\lambda) \|_{L^2(\Sigma_t)} \ls \ep^{\f 32}.
\end{split}
\end{equation}
\end{enumerate}
\end{proposition}

\begin{proposition}\label{prop:main.frame.est} \cite[Lemma~8.2, Lemma~8.4, Proposition 10.5]{LVdM1}
The following estimates hold for $L_k^\bt$, $E_k^i$ and $X_k^i$:
\begin{equation}\label{eq:frame.1}
\sup_{0 \leq t < T_B} (\| L_k^\bt \|_{L^{\infty}_{-\ep}(\Sigma_t)} + \| E_k^i \|_{L^{\infty}_{-\ep}(\Sigma_t)} + \| X_k^i \|_{L^{\infty}_{-\ep}(\Sigma_t)})\ls1,
\end{equation}
\begin{equation}\label{eq:frame.2}
\sup_{0 \leq t < T_B} (\|\rd_t L_k^t\|_{L^{\infty}_{-4\alp}(\Sigma_t)} + \|\rd_t L_k^i\|_{L^{\infty}_{1-4\alp}(\Sigma_t)} + \|\rd_x L_k^\bt\|_{L^{\infty}_{1-4\alp}(\Sigma_t)} + \|\partial X^i_k\|_{L^{\infty}_{1-4\alp}(\Sigma_t)}+ \|\partial E^i_k\|_{L^{\infty}_{1-4\alp}(\Sigma_t)})  \lesssim  \epsilon^{\frac{5}{4}} ,
\end{equation}
\begin{equation}\label{eq:frame.3}
\sup_{0 \leq t < T_B}(\|\rd^2 E^i_k\|_{L^2(\Sigma_t\cap B(0,3R))} + \|\rd^2 X^i_k\|_{L^2(\Sigma_t\cap B(0,3R))} + \|\rd\rd_x L^{\nu\color{black}}_k\|_{L^2(\Sigma_t\cap B(0,3R))} ) \ls \ep^{\f 54}.
\end{equation}

\end{proposition}

\begin{lem}\label{lem:rd.in.terms.of.XEL} \cite[Lemma~8.3, Lemma~10.4]{LVdM1}
	For any sufficiently regular function $f$, and for all $(x,t) \in \RR^2 \times [0,T_B)$:
	\begin{align} 
	\label{spatialintermsofEX}
	|\partial_i f|(x,t) \ls &\: \la x \ra^{\ep} \left( |E_k f |(x,t)+ |X_k f |(x,t)\right), \\
	\label{timeintermsofLandspace}
	|\partial_t f|(x,t) \ls &\: \la x \ra^{\ep} (|L_k f|(x,t) + |\rd_x f|(x,t) ), \\
	\label{timeintermsofLEX}
	|\partial_t f|(x,t) \ls &\: \la x\ra^{\ep} \left( |L_k f |(x,t)+ |X_k f |(x,t)\right) + \la x \ra^{-1+\ep} |E_k f |(x,t),
	\end{align} and for second derivatives, the following estimates hold\begin{align}
	\label{eq:rdrdx.in.terms.of.geometric}
	\|\rd \rd_x f \|_{L^2(\Sigma_t\cap B(0,3R))} \ls &\: \sum_{Y_k \in \{L_k,X_k,E_k \}} \sum_{Z_k \in \{X_k,E_k\} }\| Y_k Z_k f \|_{L^2(\Sigma_t\cap B(0,3R))} + \| \rd_x f \|_{L^2(\Sigma_t\cap B(0,3R))}, \\
	\label{eq:rdrd.in.terms.of.geometric}
	\|\rd^2 f \|_{L^2(\Sigma_t\cap B(0,3R))} \ls &\: \sum_{Y_k,Z_k \in \{L_k,X_k,E_k \}} \| Y_k Z_k f \|_{L^2(\Sigma_t\cap B(0,3R))} + \| \rd f \|_{L^2(\Sigma_t\cap B(0,3R))}.
	\end{align}
\end{lem}

\begin{proposition}\label{prop:main.Ricci.est} \cite[Propositions~9.22, 10.1, 10.2, 10.3]{LVdM1} 
	The following estimates hold:
\begin{align}\label{eq:K}
\sup_{0 \leq t < T_B}(\| K \|_{L^{\infty}_{2-\alp}(\Sigma_t)} + \| \rd_x K \|_{L^{\infty}_{2-\alp}(\Sigma_t)} + \|\rd_t K \|_{L^{\f 2{s'-s''}}_{2-s'+s''+\alp}(\Sigma_t)}) \ls &\: \ep^{\f 32}, \\
\label{eq:Lchi.Leta}
\sup_{0 \leq t < T_B}(\|\chi_k\|_{L^\i_{1-\alp}(\Sigma_t)} + \|\eta_k\|_{L^\i_{1-\alp}(\Sigma_t)} + \|L_k \chi_k \|_{L^\i_{1-\alp}(\Sigma_t)} + \| L_k \eta_k \|_{L^\i_{1-\alp}(\Sigma_t)}) \ls &\:\ep^{\f 32}, \\
\label{eq:dxchi} \sup_{0 \leq t < T_B, u_k \in \RR}(\|\rd_x\chi_k \|_{L^2_{\theta_k}(\Sigma_t \cap C^k_{u_k})} + \|E_k \eta_k \|_{L^2_{\theta_k}(\Sigma_t \cap C^k_{u_k})}) \ls &\: \ep^{\f 32}, \\
\label{eq:Lrdchi}
\sup_{0 \leq t < T_B} ( \|L_k \rd_x \chi_k \|_{L^2(\Sigma_t \cap B(0,R))} + \|L_k E_k \eta_k \|_{L^2(\Sigma_t \cap B(0,R))} + \sum_{\kappa_k \in \{ \chi_k, \eta_k\}} \| L_k^2 \kappa_k \|_{L^2(\Sigma_t \cap B(0,R))} )\ls &\: \ep^{\f 32},\\
	\label{mu.main.estimate}
\sup_{0 \leq t < T_B}(	\|  \log \mu_k - \gamma_{asymp} \om(|x|) \log |x| \|_{L^{\infty}_{1-\alp}(\Sigma_t)} +\| \partial_x \mu_k  \|_{L^{\infty}_{1-\alp}(\Sigma_t)} ) \lesssim &\: \ep^{\f 32} , \\
\label{vartheta.main.estimate}
\sup_{0 \leq t < T_B}(\| \log(\varTheta_k) - \gamma_{asymp} \om(|x|) \log|x|\|_{L^{\infty}_{1-2\alp}(\Sigma_t)} + \| \la x \ra^{-\alp} \rd_x \log \varTheta_k \|_{L^2_{\th_k}(\Sigma_t\cap C^k_{u_k})}) \lesssim &\: \ep^{\f 32}\blue{,}
\end{align}
where $\gamma_{asymp} \in [0, C\ep)$ is a constant defined by $\lim_{|x|\to \infty} \f{\gamma_{|\Sigma_0}}{\log|x|}$; see \cite[Definition 4.2]{LVdM1}.
\end{proposition}

%More specifically, the following lemma follows from Corollary 8.6 and Proposition 8.7 in \cite{LVdM1}: 
The following lemma gives estimates on various change\green{s} of variables:
\begin{lem} \label{lem:jacobian} \begin{enumerate}

\item \cite[Corollary~8.6]{LVdM1} 	For any $k\neq k'$,  the map $(x^1,x^2)\mapsto (u_k, u_{k'})$ is a $C^1$-diffeomorphism on $\Sigma_t$ with entry-wise pointwise estimates independent of $\de$: 
%and its  Jacobian determinant $(\breve{J}_{k,k'})_{\Sigma_t}(x)$ defined by $du_k \wedge d  u_{k'}  =(\breve{J}_{k,k'})_{\Sigma_t} dx^1  \wedge dx^2  $ obeys the following estimate 
%\begin{equation}\label{jacobianx.ukpuk}
%\sup_{0 \leq t < T_B}\|(\breve{J}_{k,k'})_{\Sigma_t}\|_{L^{\infty}(\Sigma_t)}+ \|(\breve{J}_{k,k'})_{\Sigma_t}^{-1}\|_{L^{\infty}(\Sigma_t)} \ls 1. \end{equation}	
	\begin{align}
 &|\partial_i u_k|,\, |\partial_i u_{k'}| \ls 1, \label{eq:COV.1}
 \\   1 \ls\ & |\det \begin{bmatrix}  \f{\rd u_k}{\rd x^1} & \f{\rd u_j}{\rd x^1}  
 \\  \f{\rd u_k}{\rd x^2} & \f{\rd u_j}{\rd x^2} \end{bmatrix}|\ \ls 1. \label{eq:COV.2}
	\end{align}
%$$\left|\begin{bmatrix}  \f{\rd u_k}{\rd x^1} & \f{\rd u_j}{\rd x^1}  \\  \f{\rd u_k}{\rd x^2} & \f{\rd u_j}{\rd x^2} \end{bmatrix} \right|\ls 1,\quad \left| \begin{bmatrix}  \f{\rd u_k}{\rd x^1} & \f{\rd u_j}{\rd x^1}  \\  \f{\rd u_k}{\rd x^2} & \f{\rd u_j}{\rd x^2} \end{bmatrix}^{-1} \right| \ls 1.$$

\item \cite[Proposition~8.7]{LVdM1} For any $k=1,2,3$, \begin{equation} \label{d^2u}
|\rd^2_{ij} u_k| \ls \ep^{\f 54}.
\end{equation}

\item \cite[(2.11), (2.47), (7.2a), (7.2b), (7.3d), (7.3e)]{LVdM1} The Jacobian determinant $J_k$ corresponding to the transformation $(x^1,x^2) \rightarrow (u_k,\theta_k)$, defined by $du_k \wedge d \theta _k =J_k \, dx^1  \wedge dx^2$, obeys the following estimate 
\begin{equation} \label{jacobian}
\sup_{0 \leq t < T_B} (\|J_k\|_{L^{\infty}(\Sigma_t)}+ \|J_k^{-1}\|_{L^{\infty}(\Sigma_t)}) \ls 1.
	\end{equation}

	\end{enumerate}
\end{lem}

Lastly, as a consequence of \eqref{cangle} we have the following estimate: 
\begin{lem}\label{lem:angle}
\cite[(8.13)]{LVdM1} For all $k \neq k'$ we have 
\begin{equation}\label{anglecontrol} 
\begin{split}
\frac{ \upkappa_0}{2} \leq|g(E_k,X_{k'}) | \leq 2.
\end{split}
\end{equation}
%\begin{equation} \label{anglecontrol2}
%\frac{\upkappa_0^2}{4} \leq |\partial_{t_k} u_{k'}|_{|B(0,3R)} \leq 2
%\end{equation}
\end{lem}

\section{An integration by parts lemma}\label{IBP.section}

In this section, we prove an integration by parts lemma (Proposition~\ref{IBPmainestimateprop}) which will later be useful to control the energy (see Corollary~\ref{cor:main.weighted.energy}).

The main purpose of the estimate in Proposition~\ref{IBPmainestimateprop} will be to handle inhomogeneous terms in wave equations which can be written as an $e_0$ derivative. In other words suppose $\Box_g v = f_1+ h \cdot e_0 f_2$, we want to get rid of the time derivative $e_0 f_2$ and to replace it by $\partial_x f_2$ using the wave equation \eqref{Box2+1} (see the first term in the right-hand side of \eqref{IBPmainestimate}). 

As a first step towards Proposition~\ref{IBPmainestimateprop}, we first prove a simple lemma:

\begin{lem}\label{lem:basic.IVP}
For any two smooth functions $h_1$, $h_2$ which are Schwartz class for every $t \in [0,T_B)$, the following holds for all $T\in (0, T_B)$:
\begin{equation}
\left| \int_0^T \int_{\Sigma_t}  h_2 \cdot e_0 h_1\color{black} \, dx dt+ \int_0^T \int_{\Sigma_t}  h_1 \cdot e_0 h_2 \, dx dt \right| \ls \sup_{t\in [0,T)} \| h_1\|_{L^2(\Sigma_t)} \|h_2 \|_{L^2(\Sigma_t)}. 
\end{equation}
\end{lem}
\begin{proof}
Since $e_0 = \partial_t - \beta^i \partial_i$, an explicit computation gives
$$ \int_0^T \int_{\Sigma_t}  e_0 h_1 \cdot h_2 \,dx dt = \int_0^T \int_{\Sigma_t}  (\partial_i  \beta^i \cdot  h_1 \cdot h_2-h_1 \cdot e_0 h_2) \, dx dt+ \int_{\Sigma_T}   h_1 \cdot h_2 \,dx dt- \int_{\Sigma_0}   h_1 \cdot h_2 \,dx dt. $$ 
Using the $L^\i$ bound for $\rd_i \bt^i$ in Proposition~\ref{prop:main.metric.est}, and applying H\"older's inequality, we obtain the desired estimate. \qedhere
\end{proof}

The following is the main result of the \blue{section}.
\begin{prop} \label{IBPmainestimateprop}
Let $v$ be a smooth function which is Schwartz on $\Sigma_t$. Suppose $\Box_g v = f_1+ h \cdot e_0 f_2$.

Assume that $h$ satisfies the bounds
\begin{equation}\label{eq:h.assumption.for.IBP}
\|\la x \ra^{-\alp} h\|_{L^\i(\Sigma_t)} + \|\la x \ra^{-\alp} \rd h \|_{L^\i(\Sigma_t)} \ls 1.
\end{equation}

 Then, the following estimate holds for all $r\geq 1$ and all $T \in [0,T_B)$: 
\begin{equation}  \label{IBPmainestimate} 
\begin{split}
&\: |\int_0^T \int_{\Sigma_t}  \la x\ra^{-2r} (\Box_g v) (e_0v )e^{2\gamma}    dx dt| \\
\ls  &\:  \sup_{t \in [0,T)} \| \la x \ra^{-(r+2\alp)} \rd v \|_{L^2(\Sigma_t)}   \| \la x \ra^{-\f r 2} f_2 \|_{L^2(\Sigma_t)}   +\int_0^T ( \|\la x \ra^{-\f r2} f_1 \|_{L^2(\Sigma_t)}^2   + \|\la x \ra^{-\f r2} f_2 \|_{L^2(\Sigma_t)}^2 )\, dt \\
&\:   + \int_0^T \|\la x\ra^{-(r+2\alp)} \rd v \|_{L^2(\Sigma_t)} \cdot \Big( \|\la x \ra^{-\f r2} f_1 \|_{L^2(\Sigma_t)} + \| \la x \ra^{-\f r2} f_2 \|_{L^2(\Sigma_t)} + \| \la x \ra^{-\f r2} \partial_x f_2 \|_{L^2(\Sigma_t)}\Big) \, dt.
\end{split}
\end{equation}
\end{prop}
\begin{proof}
We first write 
\begin{equation}
\begin{split}
& \int_0^T \int_{\Sigma_t} \la x\ra^{-2r} (\Box_g v) (e_0v )e^{2\gamma}    dx dt  \\
= &\: \int_0^T \int_{\Sigma_t} \la x\ra^{-2r} f_1 (e_0v )e^{2\gamma}    dx dt + \int_0^T \int_{\Sigma_t}  \la x\ra^{-2r} h (e_0 f_2) (e_0v )e^{2\gamma}    dx dt \\ 
= &\:  \int_0^T \int_{\Sigma_t} \la x\ra^{-2r} f_1 (e_0v )e^{2\gamma}    dx dt   - \int_0^T \int_{\Sigma_t}     \la x\ra^{-\frac{3r}{2}} h (e_0   \la x\ra^{-\frac{r}{2}}) f_2 (e_0v )e^{2\gamma}    dx dt \\
&\: + \int_0^T \int_{\Sigma_t}     \la x\ra^{-\frac{3r}{2}}h [e_0 (  \la x\ra^{-\frac{r}{2}}f_2)] (e_0v )e^{2\gamma}    dx dt. \label{eq:b4.integrating.by.parts}
\end{split}
\end{equation}

The first term in \eqref{eq:b4.integrating.by.parts} can be easily controlled as follows, using the Cauchy--Schwarz inequality and the fact that $T \leq T_B \leq 1$: 
\begin{equation}\label{eq:IBP.prelim1}
\Big| \int_0^T \int_{\Sigma_t} \la x\ra^{-2r} f_1 \cdot (e_0v )e^{2\gamma}      dx dt \Big| \ls \sup_{0 \leq t \leq T } \|\la x\ra^{-\f r2} f_1 \|_{L^2(\Sigma_t)} \cdot \| \la x\ra^{-(r+2\alp)} \rd v \|_{L^2(\Sigma_t)},
\end{equation}
where we have bounded $\| \la x\ra^{-\f{r}2+2\alp} e^{2\gamma}\|_{L^\i(\Sigma_t)} \ls 1$ and $\| \la x\ra^{-\f{r}2+2\alp} \bt^j e^{2\gamma}\|_{L^\i(\Sigma_t)} \ls 1$ using \eqref{eq:g.main}.

For the second term in \eqref{eq:b4.integrating.by.parts}, notice that $e_0  \la x\ra^{-\f{r}2} = \f r2 \beta_i x^i  \la x\ra^{-\f{r}2-2}$. Hence, by H\"older's inequality, Proposition~\ref{prop:main.metric.est} and \eqref{eq:h.assumption.for.IBP}, we get 
\begin{equation}\label{eq:IBP.prelim2}
\begin{split}
&\: \Big| \int_0^T \int_{\Sigma_t}     \la x\ra^{-\frac{3r}{2}} h (e_0   \la x\ra^{-\frac{r}{2}}) f_2 (e_0v )e^{2\gamma}   \, dx \, dt  \Big| \\
\ls  &\: \int_0^T \|  \la x\ra^{-\alp} h \|_{L^{\infty}(\Sigma_t)} \cdot  \|\la x\ra^{-\f r2}  f_2 \|_{L^2(\Sigma_t)} \cdot \|\la x\ra^{-(r+2\alp)}\rd v \|_{L^2(\Sigma_t)} \, dt\\
\ls &\: \int_0^T \|\la x\ra^{-\f r2}  f_2 \|_{L^2(\Sigma_t)} \cdot \|\la x\ra^{-(r+2\alp)}\rd v \|_{L^2(\Sigma_t)}\, dt. 
\end{split}
\end{equation}  

For the third term in \eqref{eq:b4.integrating.by.parts}, we integrate by parts. Using Lemma~\ref{lem:basic.IVP} with $h_1 = \la x\ra^{-\frac{3r}{2}} h  (e_0v )e^{2\gamma}$, $h_2 =   \la x\ra^{-\frac{r}{2}}f_2$, and then using H\"older's inequality, Proposition \ref{prop:main.metric.est} and \eqref{eq:h.assumption.for.IBP}, we obtain
\begin{equation}\label{eq:IBP.prelim}
\begin{split}
&\:  \Big|  \int_0^T \int_{\Sigma_t} \la x\ra^{-\frac{3r}{2}} \cdot h \cdot [e_0   (\la x\ra^{-\frac{r}{2}}f_2) ] (e_0v )e^{2\gamma}     \,dx\, dt  \Big|  \\
\ls &\: \int_0^T  [\|\la x \ra^{-\alp}  h \|_{L^{\infty}(\Sigma_t)} + \| \la x \ra^{-2\alp} (e_0 h) \|_{L^{\infty}(\Sigma_t)}] \cdot \|\la x\ra^{-\f r2}  f_2 \|_{L^2(\Sigma_t)} \cdot \|\la x \ra^{-(r+2\alp)} e_0 v \|_{L^2(\Sigma_t)}\, dt \\
&\: +  \Big| \int_0^T \int_{\Sigma_t}  \la x\ra^{-2r} h \cdot  f_2 (e_0^2 v )e^{2\gamma}    \,dx\, dt \Big| +  \sup_{ t \in [0,T)}\| \la x\ra^{-\frac{3r}{2}} h  (e_0v )e^{2\gamma} \|_{L^2(\Sigma_t)}   \| \la x\ra^{-\frac{r}{2}}f_2 \|_{L^2(\Sigma_t)}  \\
\ls &\: \int_0^T \|\la x\ra^{-\f r2}  f_2 \|_{L^2(\Sigma_t)} \cdot \|\la x \ra^{-(r+2\alp)} \rd v \|_{L^2(\Sigma_t)} \, dt+  \Big| \int_0^T \int_{\Sigma_t}  \la x\ra^{-2r} \cdot h \cdot  f_2 (e_0^2 v )e^{2\gamma}  \,dx\, dt \Big|   \\
&\:  +\sup_{t \in [0,T)}\| \la x \ra^{-(r+2\alp)} \rd v \|_{L^2(\Sigma_t)}   \| \la x \ra^{-\f r 2} f_2 \|_{L^2(\Sigma_t)}.
\end{split}
\end{equation}  

Now, we use \eqref{Box2+1} to write $ e_0^2 v= -N^2 \Box_g v +  e^{-2 \gamma} N^2 \delta^{i j} \partial^{2}_{i j} v + e_0 \log(N) \cdot e_0 v +e^{-2 \gamma} N  \delta^{i j} \partial_{i} N \partial_{j} v$ so that
\begin{equation}\label{eq:IBP.I-IV}
\begin{split}
&\:  \int_0^T \int_{\Sigma_t}  \la x\ra^{-2r} \cdot e^{2\gamma}  \cdot h \cdot  f_2 \cdot (e_0^2 v )    dx dt  \\
= &\: -\overbrace{ \int_0^T \int_{\Sigma_t} \la x\ra^{-2r} \cdot e^{2\gamma}   \cdot h \cdot  f_2 \cdot  N^2 \cdot (\Box_g v )dx dt}^{I}  + \overbrace{ \int_0^T \int_{\Sigma_t}  \la x\ra^{-2r} \cdot  h \cdot  f_2  \cdot N^2  \cdot (\delta^{i j} \partial^{2}_{i j} v ) dx dt}^{II} \\ 
&\: +  \underbrace{\int_0^T \int_{\Sigma_t} \la x\ra^{-2r} \cdot e^{2\gamma} \cdot   h \cdot  f_2 \cdot e_0 \log(N) \cdot (e_0 v )dx dt}_{III}  +  \underbrace{\int_0^T \int_{\Sigma_t}  \la x\ra^{-2r} \cdot  h \cdot  f_2 \cdot  N  \cdot(\delta^{i j} \partial_{i} N \partial_{j} v) dx dt}_{IV}.
\end{split} 
\end{equation} 

We start with the easiest terms $III$ and $IV$: an immediate application of the Cauchy--Schwarz inequality, Proposition \ref{prop:main.metric.est} and \eqref{eq:h.assumption.for.IBP} yields: 
\begin{equation}\label{eq:IBP.III.IV}
 |III|+|IV| \ls \epsilon^{\f 32}  \int_0^T \|\la x \ra^{-\f r2} f_2 \|_{L^2(\Sigma_t)} \cdot \|\la x \ra^{-(r+2\alp)}\partial v \|_{L^2(\Sigma_t)}\, dt.
\end{equation}

For $II$ in \eqref{eq:IBP.I-IV}, we integrate by parts in $\rd_i$, and then use H\"older's inequality, \eqref{eq:h.assumption.for.IBP} and \eqref{eq:g.main} to obtain
\begin{equation}\label{eq:IBP.II}
\begin{split}
|II| \ls &\: \int_0^T  \|\la x\ra^{r+2\alp} \rd_x (\la x\ra^{-2r} h f_2 N^2) \|_{L^2(\Sigma_t)} \|\la x\ra^{-(r+2\alp)} \rd_x v\|_{L^2(\Sigma_t)} \, dt \\
\ls &\: \int_0^T  (\| \la x\ra^{-\f r2} f_2 \|_{L^2(\Sigma_t)} + \| \la x\ra^{-\f r2} \rd_x f_2 \|_{L^2(\Sigma_t)}) \|\la x\ra^{-(r+2\alp)}\rd_x v\|_{L^2(\Sigma_t)}\, dt.
\end{split}
\end{equation}

We now turn to the main term $I$ in \eqref{eq:IBP.I-IV}. We write again $\Box_g v = f_1+ h \cdot  e_0 f_2$. The term involving $f_1$ can be estimated directly using H\"older's inequality, \eqref{eq:h.assumption.for.IBP} and Proposition \ref{prop:main.metric.est}. For the term involving $e_0 (f_2^2)$, we integrate by parts again with Lemma~\ref{lem:basic.IVP}, and then bound the resulting terms using H\"older's inequality, \eqref{eq:h.assumption.for.IBP} and \eqref{eq:g.main}. (The weights functions involved in the integration by parts argument can be treated as in \eqref{eq:IBP.prelim2}). We thus obtain
\begin{equation}\label{eq:IBP.I} 
\begin{split}
|I| \lesssim &\:  \int_0^T \|\la x\ra^{-\f r2} f_1 \|_{L^2(\Sigma_t)} \cdot \|\la x\ra^{-\f r2} f_2 \|_{L^2(\Sigma_t)} \, dt + \Big| \int_0^T \int_{\Sigma_t}  \la x\ra^{-2r} e^{2\gamma}   \cdot h^2 \cdot  e_0(f_2^2) \cdot  N^2  dx dt \Big| \\ 
\ls &\: \int_0^T (\| \la x\ra^{-\f r2} f_1 \|_{L^2(\Sigma_t)}^2 +    \|\la x\ra^{-\f r2} f_2 \|_{L^2(\Sigma_t)}^2)\, dt.
\end{split}
\end{equation} 
Plugging \eqref{eq:IBP.III.IV}--\eqref{eq:IBP.I} into \eqref{eq:IBP.prelim} and \eqref{eq:IBP.I-IV} gives the desired bounds for the third term in \eqref{eq:b4.integrating.by.parts}. Combining this with \eqref{eq:IBP.prelim1} and \eqref{eq:IBP.prelim2} yields the conclusion of the proposition. \qedhere
\end{proof}

\section{Basic energy estimates and commutator estimates}\label{sec:EE}

In this section, we prove some basic energy estimates which will be repeatedly used in the later part of the paper.

 	\subsection{Stress-energy-momentum tensor and deformation tensor} \label{deformsection}

 	\begin{lem}\label{lem:local.energy.cons}
 		\begin{enumerate}
 			\item Defining $\T[v] = \partial_{\mu}v \partial_{\nu}v -\frac{1}{2} g_{\mu \nu} \gi^{\blue{\sigma}\bt} \partial_{\blue{\sigma}} v \partial_{\bt} v$, and suppose $\n$ and $(X_k, E_k, L_k)$ are as in \eqref{defnormal} and Definition~\ref{def:null.frame}. Then for $k=1,2,3$,
 			\begin{equation} \label{Tnn}
 			\TE[v](\n,\n) = \frac{1}{2} \cdot \left( (\n v)^2 + (X_k v)^2+ (E_k v)^2 \right) = \f 12 \left( (\n v)^2 + e^{-2\gamma} (\rd_x v)^2 \right) ,
 			\end{equation}
 			\begin{equation} \label{Tni.Tij}
 			\TE[v](\n,\rd_i) = (\n v)(\rd_i v),\quad \TE(v)(\rd_i,\rd_j) = (\rd_i v) (\rd_j v) - \f 12 \de_{ij} (- e^{2\gamma} (\n v)^2 + (\rd_x v)^2),
 			\end{equation}
 			\begin{equation} \label{TLn}
 			\TE[v](L_k,\n) = \frac{1}{2} \cdot \left( (L_k v)^2  + (E_k v)^2 \right).
 			\end{equation}
 			\item $\T[v]$ satisfies
 			$$\gi^{\blue{\sigma}\nu} \nab_{\blue{\sigma}} \T[v] = \Box_g v \cdot \rd_\mu v.$$
 			\item Defining in addition $^{(\n)}\pi(Z_1,Z_2)= \frac{1}{2} \left( g(\nabla_{Z_1}\n, Z_2)+g(\nabla_{Z_2}\n, Z_1) \right)$, we have
 			\begin{equation}\label{eq:T.pi}
 			\begin{split}
 			&\: \TE[v]_{\mu\nu} {}^{(\n)}\pi^{\mu\nu} \\
 			%= &\: \TE(v)(\n, \n) {}^{(\n)}\pi(\n,\n) - 2 e^{-2\gamma} \de^{ii'} \TE(v)(\n, \rd_i) {}^{(\n)}\pi(\n, \rd_{i'}) + e^{-4\gamma}  \de^{ii'} \de^{jj'} \TE(v)(\rd_i, \rd_j) {}^{(\n)}\pi(\rd_{i'}, \rd_{j'}) \\
 			= &\: - e^{-2\gamma} \de^{i l \color{black}} (\n v) (\rd_i v)(\rd_{l \color{black}} \log N) + e^{-4\gamma}\de^{i l \color{black}}\de^{j q \color{black}} K_{ l q \color{black} }[(\rd_i v) (\rd_j v) - \f 12 \de_{ij} (- e^{2\gamma} (\n v)^2 + (\rd_x v)^2)].
 			\end{split}
 			\end{equation}

 		%	\begin{equation}\label{eq:T.pi.L} 			\begin{split} 			&\: \TE(v)_{\mu\nu} {}^{(L_k)}\pi^{\mu\nu} \\
 			%= &\: \TE(v)(\n, \n) {}^{(\n)}\pi(\n,\n) - 2 e^{-2\gamma} \de^{ii'} \TE(v)(\n, \rd_i) {}^{(\n)}\pi(\n, \rd_{i'}) + e^{-4\gamma}  \de^{ii'} \de^{jj'} \TE(v)(\rd_i, \rd_j) {}^{(\n)}\pi(\rd_{i'}, \rd_{j'}) \\ 			= &\: - e^{-2\gamma} \de^{ii'} (\n v) (\rd_i v)(\rd_{i'} \log N) + e^{-4\gamma}\de^{ii'}\de^{jj'} K_{i'j'}[(\rd_i v) (\rd_j v) - \f 12 \de_{ij} (- e^{2\gamma} (\n v)^2 + (\rd_x v)^2)]. 			\end{split} 			\end{equation}
 	
 		\end{enumerate}
 	\end{lem}
 	\begin{proof}

 		Parts~1 and 2 are explicit computations.
 		
 		We turn to 3. By Lemma~\ref{Christoffel}, $\nabla_{\n} \n = e^{-2\gamma} \de^{ij} (\rd_i \log N) \rd_j$. Hence, using also \eqref{Kdef}, we have
 		\begin{equation} \label{Pin}
 		^{(\n)}\pi(\n, \n ) = 0,\quad ^{(\n)}\pi(\n,\rd_i )= \frac{1}{2} \cdot \rd_i \log N,\quad ^{(\n)}\pi(\rd_i, \rd_j)= K_{ij}.
 		\end{equation} 
 		
 		We compute using \eqref{inversegelliptic} that
 		\begin{equation} 
 		\begin{split}
 		&\: \TE[v]_{\mu\nu} {}^{(\n)}\pi^{\mu\nu} \\
 		= &\: \TE[v](\n, \n) {}^{(\n)}\pi(\n,\n) - 2 e^{-2\gamma} \de^{i l \color{black}} \TE[v](\n, \rd_i) {}^{(\n)}\pi(\n, \rd_{l \color{black}}) + e^{-4\gamma}  \de^{i l \color{black}} \de^{j q \color{black}} \TE[v](\rd_i, \rd_j) {}^{(\n)}\pi(\rd_{l \color{black}}, \rd_{q \color{black}}),
 		\end{split}
 		\end{equation}
 		which implies the desired conclusion after plugging in \eqref{Tnn}, \eqref{Tni.Tij} and \eqref{Pin}. \qedhere
 		
 	\end{proof}

\subsection{Volume forms}
\begin{lemma}\label{lem:volume.forms} \cite[Lemma~2.14]{LVdM1}
The spacetime volume form induced by $g$ is given by $$dvol = N \cdot e^{2\gamma} \, dx^1\,  \wedge dx^2\blue{\wedge}\, dt = \mu_k \cdot N \cdot \varTheta_k \, dt_k\,  \wedge du_k\, \wedge d\th_k. $$
The volume form on $\Sigma_t$ induced by $\bar{g}$ is given by \begin{equation}\label{eq:vol.Sigma}dvol_{\Sigma_t} = e^{2\gamma}\, dx^1\,   \wedge  dx^2 = \mu_k^2 \varTheta_k^2 \, du_k  \wedge \, d\th_k.\end{equation}
Let $dvol_{C_{u_k}^k}$ be the volume form on $C_{u_k}^k$ such that $du_k \wedge dvol_{C_{u_k}} =dvol$. Then\begin{equation} \label{volCuk} 			dvol_{C_{u_k}^k}=  -  \mu_k \cdot N \cdot \varTheta_k \,   dt_k\, \wedge d\theta_k=  \mu_k \cdot N \cdot \varTheta_k d\theta_k\,    \wedge  dt_k. \end{equation}
\end{lemma}

\subsection{The main energy estimate}

In this subsection, we prove two basic energy estimates.

The first estimate (Proposition~\ref{prop:EE}) applies only to compactly supported functions (so that weights can be ignored\color{black}), and allows for localization in the $u_k$ variable. The second estimate (Proposition~\ref{prop:EE.weighted}) is a weighted estimate for general (not necessarily compactly supported) functions which does not allows for $u_k$-localization\footnote{One could combine the \color{black} two energy estimates to obtain a more general proposition, incorporating both weights and $u_k$-localization. We will not need such a general statement, and therefore only prove the easier estimates.}. 

The following is the first general energy estimate.
\begin{proposition}\label{prop:EE}
	Given $k\in \{1,2,3\}$, any $T\in [0, T_B)$ and any $-\infty\leq U_0< U_1\leq +\infty$, define 
	\begin{equation}\label{def:D.domain}
	\mathcal D^{(k),T}_{U_0,U_1} := \{ (t,x) \in \mathbb R \times \mathbb R^2 : t \in [0,T],\, u_k(t,x)\in [U_0, U_1] \}.
	\end{equation}
	
	For any\footnote{In particular, $k'$ could be the same as $k$, and could also be different from $k$. The same comment applies to Propositions~\ref{prop:commute.with.spatial.direct}, \ref{prop:commute.with.E.direct} and \ref{prop:commute.with.L.direct}.} $k'\in \{1,2,3\}$, the following holds for all solutions $v$ to $\Box_g v = f$, with $\mathrm{supp}(v),\,\mathrm{supp}(f) \subseteq \{(t,x): |x|\leq R\}$, with a constant depending only on $R$:
	\begin{equation}\label{eq:basic.EE.desired}
	\begin{split}
	&\: \sup_{t\in [0, T)} \|\rd v \|_{L^2(\Sigma_t\cap \mathcal D^{(k),T}_{U_0,U_1})} + \sup_{u_{k'} \in \RR} \sum_{Z_{k'} \in \{L_{k'},\,E_{k'}\}} \|Z_{k'} v\|_{L^2(C^{k'}_{u_{k'}}\cap \mathcal D^{(k),T}_{U_0,U_1})} \\
	\ls &\: \|\rd v \|_{L^2(\Sigma_0\cap \mathcal D^{(k),T}_{U_0,U_1})} + \sum_{Z_k \in \{L_k,\,E_k\}} \|Z_k v\|_{L^2(C^k_{U_0} \cap \mathcal D^{(k),T}_{U_0,U_1})} + \int_0^{T} \|f \|_{L^2(\Sigma_t \cap \mathcal D^{(k),T}_{U_0,U_1})} \, dt.
	\end{split}
	\end{equation}
\end{proposition}
\begin{proof} 
\pfstep{Step~1: The case $k' = k$} By Lemma~\ref{lem:local.energy.cons} (Point 2) and $\Box_g v =f$, we have $\nabla^{\nu}\T(v) = f \partial_{\mu} v $. \blue{Hence,}
\begin{equation}\label{eq:div.form.for.EE}
\begin{split}
%&\: 
\nabla^{\nu} \left( \T[v] \n^{\mu} \right) = \TE[v]_{\mu\nu} {}^{(\n)}\pi^{\mu\nu} + f \cdot \n v.  
\end{split}
\end{equation} 

Fix $T\in [0, T_B)$ and $U_0$, $U_1$ as in the statement of the proposition. For every $\tau \in [0,T]$ and $U\in [U_0,U_1]$, define $\mathcal D_{U_0,U}^{(k),\tau} := \{ (t,x) \in \mathbb R \times \mathbb R^2 : t \in [0,\tau),\, u_k(t,x)\in [U_0, U) \}$. Note that clearly $\mathcal D_{U_0,U}^{(k),\tau} \subseteq \mathcal D^{(k),T}_{U_0,U_1}$.

Integrating \blue{$\nabla^{\nu} \left( \T[v] \n^{\mu} \right)$} on the spacetime region $\mathcal D_{U_0,U}^\tau$ and using Stokes' theorem,  we obtain (for $\T = \T[v]$)
\begin{equation}\label{eq:EE.main.id.prelim}
\begin{split}
&\: \int_{\blue{\Sigma_\tau \cap \mathcal D_{U_0,U}^{(k),\tau}}} \mathbb T(\n,\n)\, dvol_{\Sigma_t} - \int_{\blue{\Sigma_0 \cap \mathcal D_{U_0,U}^{(k),\tau}}} \mathbb T(\n,\n)\, dvol_{\Sigma_0} + \int_{\blue{C^k_{U}\cap \mathcal D_{U_0,U}^{(k),\tau}}} \mathbb T(\n,(-du_k)^\sharp)\, dvol_{C^k_{U}} \\
= &\: \int_{\mathcal D_{U_0,U}^{(k),\tau}} \nabla^{\nu} \left( \T \n^{\mu} \right) \, dvol.
\end{split}
\end{equation}

Using $(-du_k)^\sharp = \mu_k^{-1} L_k$ (by \eqref{Lgeodefinition}  and \eqref{Ldefinition}), the computations for $\mathbb T$ in \eqref{Tnn}, \eqref{TLn}, the computations for the volume forms in Lemma~\ref{lem:volume.forms}, and \blue{the computations for $\nabla^{\nu} \left( \T[v] \n^{\mu} \right)$ in \eqref{eq:div.form.for.EE} and \eqref{eq:T.pi}}, we obtain using \eqref{eq:EE.main.id.prelim} that 
\begin{equation}\label{eq:EE.main.id} 
\begin{split}
&\:  \underbrace{\f 12\int_{\Sigma_\tau \cap \mathcal D_{U_0,U}^{(k),\tau} }[ e^{2\gamma}(\n v)^2 +(\partial_x v)^2] \,dx^1 \,dx^2}_{=:I} + \underbrace{\int_{ C_{U}^k \cap \mathcal D_{U_0,U}^{(k),\tau} }  \frac{\varTheta_k N}{2 } \cdot [  (L_k v)^2+ (E_k v)^2] \,dt_k \, d\theta_k}_{=:II} \\ 
&\: -  \underbrace{ \int_{\mathcal D_{U_0,U}^{(k),\tau}}\left[- e^{-2\gamma} \de^{il} (\n v) (\rd_i v)(\rd_{l} \log N) \right] \cdot Ne^{2\gamma} \,dx^1\, dx^2 \,dt }_{=:III} \\ 
&\: - \underbrace{ \int_{\mathcal D_{U_0,U}^{(k),\tau} }\left[ e^{-4\gamma}\de^{il}\de^{jq} K_{lq}[(\rd_i v) (\rd_j v) - \f 12 \de_{ij} (- e^{2\gamma} (\n v)^2 + (\rd_x v)^2)] + f \cdot \n v\right] \cdot Ne^{2\gamma} \,dx^1\, dx^2 \,dt }_{=:IV} \color{black} \\
= &\: \underbrace{ \frac{1}{2}  \int_{\Sigma_0 \cap \mathcal D_{U_0,U}^{(k),\tau} }[ e^{2\gamma}(\n v)^2 +(\partial_x v)^2] \,dx^1 \,dx^2}_{=:V} +\underbrace{\int_{ C_{U_0}^k \cap \mathcal D_{U_0,U}^{(k),\tau} }  \frac{\varTheta_k N}{2} \cdot[ (L_k v)^2+ (E_k v)^2] \,dt_k\, d\theta_k}_{=:VI}.
	\end{split}
	\end{equation} 
	
	By Proposition~\ref{prop:main.metric.est}, Proposition~\ref{prop:main.Ricci.est} and support properties of $v$, 
	\begin{equation}\label{eq:EE.main.id.1} 
	I + II \gtrsim  \|\rd v\|_{L^2(\Sigma_\tau \cap \mathcal D_{U_0,U}^{(k),\tau})}^2 + \sum_{Z_{k} \in \{L_{k},\,E_{k}\}} \|Z_{k'} v\|_{L^2(C^{k}_{U}\cap \mathcal D_{U_0,U}^{(k),\tau})}^2.
	\end{equation}
	Using Proposition~\ref{prop:main.metric.est}, Proposition~\ref{prop:main.Ricci.est}, the Cauchy--Schwarz inequality and Young's inequality, we get that
	\begin{equation}\label{eq:EE.main.id.2}
	\begin{split}
 |III| + |IV| \color{black}  \leq &\: C \ep^{\f 32} \int_0^\tau \|\rd v \|_{L^2(\Sigma_t \cap \mathcal D_{U_0,U}^{(k),\tau})}^2 \, dt + C \int_0^\tau \|\rd v \|_{L^2(\Sigma_t \cap \mathcal D_{U_0,U}^{(k),\tau})} \|f \|_{L^2(\Sigma_t \cap \mathcal D_{U_0,U}^{(k),\tau})} \, dt \\
	\leq &\: ( \frac{1}{2} + C \ep^{\f 32}) \sup_{t \in [0,T]}  \|\rd v \|_{L^2(\Sigma_t \cap \mathcal D_{U_0,U}^{(k),\tau})}^2 + C \cdot (\int_0^T  \|f \|_{L^2(\Sigma_t \cap \mathcal D_{U_0,U}^{(k),\tau})} \, dt )^2.
	\end{split}
	\end{equation}
	Finally, the data terms can be controlled using Proposition \ref{prop:main.metric.est} and Proposition \ref{prop:main.Ricci.est} applied at $t=0$: 
	\begin{equation}\label{eq:EE.main.id.3}
 |V| + |VI| \color{black} \ls \|\rd v\|_{L^2(\Sigma_0\cap \mathcal D_{U_0,U}^{(k),T})}^2 + \sum_{Z_{k} \in \{L_{k},\,E_{k}\}} \|Z_{k} v\|_{L^2(C^{k}_{U_0}\cap \mathcal D_{U_0,U}^{(k),T})}^2.
	\end{equation}
	
Plugging the estimates \eqref{eq:EE.main.id.1}--\eqref{eq:EE.main.id.3} into \eqref{eq:EE.main.id}, and taking supremum over all $\tau \in [0,T]$ and $U\in [U_0, U_1]$, we obtain 
	\begin{equation}\label{eq:EE.main.step.1}
	\begin{split}
	&\: \sup_{t\in [0,T]} \|\rd v\|_{L^2(\Sigma_t\cap \mathcal D_{U_0,U}^{(k),T})}^2 + \sup_{u_k \in \RR} \sum_{Z_{k} \in \{L_{k},\,E_{k}\}} \|Z_{k} v\|_{L^2(C^{k}_{u_k}\cap \mathcal D_{U_0,U}^{(k),T})}^2 \\
	\leq &\: ( \frac{1}{2} +  C \ep^{\f 32}) \sup_{t \in [0,T]}  \|\rd v \|_{L^2(\Sigma_t \cap  \mathcal D_{U_0,U}^{(k),T})}^2 + C \|\rd v\|_{L^2(\Sigma_0\cap \mathcal D_{U_0,U}^{(k),T})}^2 + C  \sum_{Z_{k} \in \{L_{k},\,E_{k}\}} \|Z_{k} v\|_{L^2(C^{k}_{U_0}\cap \mathcal D_{U_0,U}^{(k),T})}^2 \\
	&\: + C  (\int_0^T  \|f \|_{L^2(\Sigma_t \cap \mathcal D_{U_0,U}^{(k),T})} \, dt )^2.
	\end{split}
	\end{equation}
 Note that while the supremum at first only gives $\sup_{u_k \in [U_0, U_1]}$ for the second term on the left-hand side of \eqref{eq:EE.main.step.1}, we can change this to $\sup_{u_k \in \RR}$ after noticing that $C^k_{u_k}\cap \mathcal D = \emptyset$ if $u_k \in \RR \setminus [U_0, U_1]$. \color{black}
	
	The first terms on the right-hand side of \eqref{eq:EE.main.step.1} can be absorbed to the left-hand side for $\ep_0$ sufficiently small, giving
	\begin{equation}\label{eq:EE.main.step.1.1}
	\begin{split}
	&\: \sup_{t\in [0,T]} \|\rd v\|_{L^2(\Sigma_t\cap \mathcal D_{U_0,U_1}^{(k),T})}^2 + \sup_{u_k \in \RR} \sum_{Z_{k} \in \{L_{k},\,E_{k}\}} \|Z_{k} v\|_{L^2(C^{k}_{u_k}\cap \mathcal D_{U_0,U_1}^{(k),T})}^2 \\
	\ls &\:  \|\rd v\|_{L^2(\Sigma_0\cap \mathcal D_{U_0,U}^{(k),T})}^2 + \sum_{Z_{k} \in \{L_{k},\,E_{k}\}} \|Z_{k} v\|_{L^2(C^{k}_{U_0}\cap \mathcal D_{U_0,U}^{(k),T})}^2  + (\int_0^T  \|f \|_{L^2(\Sigma_t \cap \mathcal D_{U_0,U}^{(k),T})} \, dt )^2.
	\end{split}
	\end{equation}
	The bound \eqref{eq:EE.main.step.1.1} gives the control of the first term in \eqref{eq:basic.EE.desired}, and of the second term in \eqref{eq:basic.EE.desired} when $k'=k$.
	
	\pfstep{Step~2: The general case} To complete the proof of \eqref{eq:basic.EE.desired}, we need to bound the second term on the left-hand side of \eqref{eq:basic.EE.desired}, corresponding to the flux on $C^{k'}_{u_{k'}}$ in the case $k' \neq k$.  Fix $k' \neq k$ and $U' \in \RR$, integrate the same \blue{quantity $\nabla^{\nu} \left( \T[v] \n^{\mu} \right)$} but now on $\mathcal{D}':=\mathcal D_{U_0,U}^{(k),T}\cap \mathcal D_{-\infty,U'}^{(k'),T} = \mathcal D_{U_0,U}^{(k),T} \cap \{(t,x): u_{k'}(t,x) \leq U' \}$, and use Stokes' theorem. We then obtain an analogue of \eqref{eq:EE.main.id}, except with $\mathcal D_{U_0,U}^{(k),T}$ replaced by $\mathcal D'$, and with an additional flux term $\int_{ C_{U'}^{k'} \cap\mathcal D' }  \frac{\varTheta_{k'}N}{2 } \cdot [  (L_{k'} v)^2+ (E_{k'} v)^2] \,dt_{k'} \, d\theta_{k'}$  on the left-hand side.
	
	We now control the bulk terms (i.e.~terms corresponding to $III$ and $IV$ in \eqref{eq:EE.main.id}) in the same manner as in Step~1. Since $\mathcal D'\subseteq \mathcal D_{U_0,U}^{(k),T}$, we obtain an analogue of \eqref{eq:EE.main.step.1.1}, but with the control of an addition flux term on the left-hand side:
	\begin{equation} \label{eq:EE.main.step.2}
	\begin{split}
	&\: \sup_{t\in [0,T]} \|\rd v\|^2_{L^2(\Sigma_t \cap\mathcal{D}') } +  \sup_{u_k \in \RR} \sum_{Z_{k} \in \{L_{k},\,E_{k}\}} \|Z_{k} v\|_{L^2(C^{k}_{u_k}\cap \mathcal D')}^2 +  \sum_{Z_{k'} \in \{L_{k'},\,E_{k'}\}} \|Z_{k'} v\|_{L^2(C^{k'}_{U'}\cap \mathcal D')}^2 \\ 
	\lesssim  &\: \|\rd v\|_{L^2(\Sigma_0\cap \mathcal D_{U_0,U}^{(k),T})}^2 + \sum_{Z_{k} \in \{L_{k},\,E_{k}\}} \|Z_{k} v\|_{L^2(C^{k}_{U_0}\cap \mathcal D_{U_0,U}^{(k),T})}^2  + (\int_0^T  \|f \|_{L^2(\Sigma_t \cap \mathcal D_{U_0,U}^{(k),T})} \, dt )^2.
	\end{split}
	\end{equation}
	We now take supremum over all $U' \in \RR$. Noting that $C^{k'}_{U'}\cap \mathcal D' = C^{k'}_{U'}\cap \mathcal D_{U_0,U}^{(k),T}$, we deduce from \eqref{eq:EE.main.step.2} that
	\begin{equation} \label{eq:EE.main.step.2.1}
	\begin{split}
	&\:   \sup_{u_{k'} \in \RR}\sum_{Z_{k'} \in \{L_{k'},\,E_{k'}\}} \|Z_{k'} v\|_{L^2(C^{k'}_{u_{k'}}\cap \mathcal D_{U_0,U}^{(k),T})}^2 \\ 
	\lesssim  &\: \|\rd v\|_{L^2(\Sigma_0\cap \mathcal D_{U_0,U}^{(k),T})}^2 + \sum_{Z_{k} \in \{L_{k},\,E_{k}\}} \|Z_{k} v\|_{L^2(C^{k}_{U_0}\cap \mathcal D_{U_0,U}^{(k),T})}^2  + (\int_0^T  \|f \|_{L^2(\Sigma_t \cap  \mathcal D_{U_0,U}^{(k),T})} \, dt )^2.
	\end{split}
	\end{equation}
	\eqref{eq:EE.main.step.2.1} thus bounds the second term in \eqref{eq:basic.EE.desired} when $k'\neq k$. Combining this with \eqref{eq:EE.main.step.1.1} concludes the proof.  \qedhere

\end{proof}
%Next, we prove a version of Proposition \ref{prop:EE} using $L_k$ as multiplier instead of $\n$: this gives a so-called ``good energy estimate'':\begin{proposition}\label{prop:EE.good}		For any $k'\in \{1,2,3\}$, the following holds for all solutions $v$ to $\Box_g v = f$, with $\mathrm{supp}(v),\,\mathrm{supp}(f) \subset \{(t,x): |x|\leq R\}$, with a constant depending only on $r$ and $R$:	\begin{equation*}	\begin{split}	&\: \sup_{t\in [0, T)} \sum_{Y_k \in \{L_k,E_k\}}\|Y_k v \|_{L^2(\Sigma_t\cap \mathcal D)} + \sup_{u_{k'} \in [U_0,U_1)} \|L_{k'} v\|_{L^2(C^{k'}_{u_{k'}}\cap \mathcal D)} \\	\ls &\: \|\rd v \|_{L^2(\Sigma_0\cap \mathcal D)} + \sum_{Z_k \in \{L_k,\,E_k\}} \|L_k v\|_{L^2(C^k_{U_0})} + \int_0^{T} \|f \|_{L^2(\Sigma_t\cap \mathcal D)} \, d\tau.	\end{split}	\end{equation*}\end{proposition}

Finally, we prove a version of Proposition \ref{prop:EE} with weights. (Note that the weights can clearly be improved, but will not be relevant for later applications.)
\begin{proposition}\label{prop:EE.weighted}
	Let $v$ be a smooth function which is Schwartz for on $\Sigma_t$ for all $0 \leq t < T_B$.
	
	Then for all $r\geq 1$, the following holds for any $T\in [0, T_B)$, with a constant depending only on $r$:
	\begin{equation}\label{eq:EE.weighted}
	\begin{split}
	&\: \sup_{t\in [0, T)} \|\la x\ra^{-(r+2\alp)} \rd v \|_{L^2(\Sigma_t)}^2  
	\ls  \|\la x\ra^{-\f r2} \rd v \|_{L^2(\Sigma_0)}^2 + \sup_{t\in [0, T)} \Big| \int_0^t \int_{\Sigma_\tau} \la x\ra^{-2r} (e_0 v)(\Box_g v)\,e^{2\gamma}\, dx\, d\tau \Big| .
	\end{split}
	\end{equation}
\end{proposition}
\begin{proof}  
Using the multiplier $ \la x\ra^{-2r} \n$, we have, by \eqref{eq:T.pi},  
\begin{equation} \label{identitymultipliernweights}
\begin{split}
\nabla^{\nu} \left( \T[v] \la x\ra^{-2r}\n^{\mu} \right)
= &\: - \la x\ra^{-2r} e^{-2\gamma} \de^{il} (\n v) (\rd_i v)(\rd_{l} \log N) \\
&\:  + \la x\ra^{-2r} e^{-4\gamma}\de^{il}\de^{jq} K_{lq}  [(\rd_i v) (\rd_j v) - \f 12 \de_{ij} (- e^{2\gamma} (\n v)^2 + (\rd_x v)^2)] \\ 
&\: + \la x\ra^{-2r} (\Box_g v) \cdot \n v + \gi^{\nu\blue{\sigma}} \T[v] (\rd_{\blue{\sigma}}  \la x\ra^{-2r}) \n^{\mu}.
\end{split} 
\end{equation}

	Using \eqref{inversegelliptic}, it can be computed that 
	\begin{equation*} 
	\begin{split} 
	\gi^{\nu \sigma} (\rd_{\sigma}  \la x\ra^{-2r})\rd_\nu =\frac{-2r \la x\ra^{-2r-2} }{N^2}  \left( \beta ^i x_i \partial_t + (N^2 e^{-2\gamma} \de^{ij} - \bt^i \bt^j )x_j \rd_i \right),
	\end{split}
	\end{equation*} 
which in particular implies, by the metric estimates of Proposition \ref{prop:main.metric.est} and the bound in Lemma~\ref{lem:rd.in.terms.of.XEL}, that
	$$ | \gi^{\nu \sigma} \T[v] (\rd_{\sigma} \la x\ra^{-2r}) \n^{\mu} | \lesssim  \la x\ra^{-2r-1+ 10\alp} [(\n v)^2 + (X_k v)^2+ (E_k v)^2],$$ 
	where the implicit constant is allowed to depend on $r$.
	
	Then we use Proposition \ref{prop:main.metric.est}  (specifically $e^{-2\gamma} \lesssim \la x\ra^{\alp} $, $|\partial_x \log N| \lesssim \epsilon \la x\ra^{-1+\alp} $,  $|K|\lesssim \epsilon \la x\ra^{-1} $ and $|T_{\mu \nu}[v]| \lesssim \la x \ra^{2\alp} \cdot |\partial v |^2$, where the indices $\mu$ and $\nu$ are in the coordinate system $(t,x^1,x^2)$) and Lemma~\ref{lem:rd.in.terms.of.XEL}, and repeating the argument of Proposition \ref{prop:EE} we get, integrating \eqref{identitymultipliernweights} on $\{ 0 \leq t' \leq t\}$: 
	\begin{equation}\label{eq:EE.weighted.main} 
	\begin{split}
	&\: \int_{\Sigma_t  }e^{2\gamma} \la x\ra^{-2r}[(\n v)^2 + (X_k v)^2+ (E_k v)^2] \, dx^1 dx^2  \\
	\lesssim &\:  \underbrace{ \int_{\Sigma_0 }e^{2\gamma} \la x\ra^{-2r}[(\n v)^2 + (X_k v)^2+ (E_k v)^2] \, dx^1 dx^2 }_{=:I} \\
	&\: +  \underbrace{ \int_0^t\int_{\Sigma_{t'}} \la x\ra^{-2r-1+ 10\alp} [(\n v)^2 + (X_k v)^2+ (E_k v)^2]  \, dx^1 dx^2 dt'  }_{=:II} \\ 
	&\: + \underbrace{ \left|\int_0^t\int_{\Sigma_{t'}} \la x\ra^{-2r} (\Box_g v) \cdot\n v  \cdot N e^{2\gamma} \, dx^1 dx^2 dt' \right|}_{=:III}.
	\end{split}
	\end{equation}

	We now bound each term on the right-hand side of \eqref{eq:EE.weighted.main}.	
	\begin{enumerate}
		\item For term $I$, we note that by Proposition \ref{prop:main.metric.est} and Proposition~\ref{prop:main.frame.est}, $e^{2\gamma} \la x\ra^{-2r}[(\n v)^2 + (X_k v)^2+ (E_k v)^2] \ls  \left| \la x \ra^{-\f r2} \rd v\right|^2$, and thus $I \ls \|\la x\ra^{-\f r2} \rd v \|_{L^2(\Sigma_0)}^2$. 
		\item For term $II$, since $10\alp = 0.1 <1  - \alp \color{black}$, it can be absorbed to the left-hand side using Gr\"onwall's inequality.
		\item Finally, we just keep the term $III$ as it is (which is allowed on the right-hand side of \eqref{eq:EE.weighted}), since $N \n = e_0$.
	\end{enumerate}

	Combining the above bounds, it follows that 
	$$\int_{\Sigma_t  }e^{2\gamma} \la x\ra^{-2r}[(\n v)^2 + (X_k v)^2+ (E_k v)^2] \, dx^1 dx^2 \ls \mbox{RHS of \eqref{eq:EE.weighted}}.$$
	
	Finally, notice that by \eqref{spatialintermsofEX}, \eqref{timeintermsofLEX} and Proposition \ref{prop:main.metric.est},
	$$\|\la x\ra^{-(r+2\alp)} \rd v \|_{L^2(\Sigma_t)}^2 \ls \int_{\Sigma_t  }e^{2\gamma} \la x\ra^{-2r}[(\n v)^2 + (X_k v)^2+ (E_k v)^2] \, dx^1 dx^2$$
	which therefore gives the desired result. \qedhere
	
\end{proof}

\begin{cor}\label{cor:main.weighted.energy}
Let $v$ be a smooth function which is Schwartz on $\Sigma_t$ for all $0\leq t <T_B$. Suppose $\Box_g v = f_1+ h \cdot e_0 f_2$, where
$f_1$, $f_2$ and $h$ are all smooth and Schwartz on  $\Sigma_t$ for all $0\leq t <T_B$, and $h$ satisfies \eqref{eq:h.assumption.for.IBP}.

Then for all $r \geq 1$, the following holds for any $T\in [0, T_B)$, with a constant depending only on $r$:
	\begin{equation*}
	\begin{split}
	&\: \sup_{t\in [0, T)} \|\la x\ra^{-(r+2\alp)} \rd v \|_{L^2(\Sigma_t)}^2  \\
	\ls &\: \|\la x\ra^{- \f r2} \rd v \|_{L^2(\Sigma_0)}^2   +   \sup_{t\in [0, T)} \| \la x\ra^{- \f r2} f_2\|_{L^2(\Sigma_t)}^{2}  \\
	&\: + \int_0^T \Big( \|\la x\ra^{-\f{r}{2}} f_1 \|_{L^2(\Sigma_t)}^2+  \|\la x\ra^{-\f{r}{2}} f_2 \|_{L^2(\Sigma_t)}^2 +  \|\la x\ra^{-\f{r}{2}} \rd_x f_2 \|_{L^2(\Sigma_t)}^2\Big)\, dt. 	
	\end{split}
	\end{equation*}

\end{cor}
\begin{proof}
We first apply Proposition~\ref{prop:EE.weighted} so that
\begin{equation*}
	\begin{split}
	&\: \sup_{t\in [0, T)} \|\la x\ra^{-(r+2\alp)} \rd v \|_{L^2(\Sigma_t)}^2 
	\ls  \|\la x\ra^{- \f r2} \rd v \|_{L^2(\Sigma_0)}^2 +  \sup_{t\in [0, T)} \Big| \int_0^t \int_{\Sigma_\tau} \la x\ra^{-2r} (e_0 v) \cdot (\Box_g v) \, e^{2\gamma}\, dx\, d\tau \Big|. 
	\end{split}
	\end{equation*}
Controlling the last term by Proposition~\ref{IBPmainestimateprop}, we obtain
\begin{equation*}
\begin{split}
&\: \sup_{t\in [0, T)} \|\la x\ra^{-(r+2\alp)} \rd v \|_{L^2(\Sigma_t)}^2 \\
	\ls &\: \|\la x\ra^{- \f r2} \rd v \|_{L^2(\Sigma_0)}^2 
+ \int_0^T ( \|\la x \ra^{-\f r2} f_1 \|_{L^2(\Sigma_t)}^2   +   \|\la x \ra^{-\f r2} f_2 \|_{L^2(\Sigma_t)}^2) \, dt  \\
&\:  +\sup_{t \in [0,T)}\| \la x \ra^{-(r+2\alp)} \rd v \|_{L^2(\Sigma_t)}   \| \la x \ra^{-\f r 2} f_2 \|_{L^2(\Sigma_t)}   \\
&\:+ \int_0^T \|\la x\ra^{-(r+2\alp)} \rd v \|_{L^2(\Sigma_t)} \cdot \Big( \|\la x \ra^{-\f r2} f_1 \|_{L^2(\Sigma_t)} + \| \la x \ra^{-\f r2} f_2 \|_{L^2(\Sigma_t)} + \| \la x \ra^{-\f r2} \partial_x f_2 \|_{L^2(\Sigma_t)}\Big)  \, dt.
\end{split}
\end{equation*}
For the terms on the last two lines, we use Young's inequality and absorb $\sup_{t\in [0,T)} \|\la x\ra^{-(r+2\alp)} \rd v \|_{L^2(\Sigma_t)}$ to the left-hand side. For the terms on the last line, we additionally use the Cauchy--Schwarz inequality in $t$, giving the desired inequality. \qedhere

\end{proof}

\section{Basic estimates for the commutations with the wave operator} \label{1commuted.section}

\subsection{Two auxiliary estimates}

To streamline the later exposition \blue{in this section, before we even consider the commutations with the wave operator, we first} prove in this subsection two auxiliary estimates in Propositions~\ref{prop:dxchi.f} and \ref{holdertypeprop}. They concern second derivatives of the metric (in the geometric coordinates or in the elliptic gauge coordinates) which are not bounded in $L^\i$.

The estimates in this \blue{sub}section apply either to $\Sigma_t$ or to a half space in the $u_k$ variable. In the remainder of this subsection, for a fixed $k$, we will use $\mathcal D$ to denote one of the following sets:
\begin{equation}\label{eq:D.half.space}
\mathcal D= \{(t,x)\in [0,T_B)\times \mathbb R^2: u_k(t,x) \geq U_0\},\quad \mathcal D= \{(t,x) \in [0,T_B)\times \mathbb R^2:u_k(t,x) \leq U_1\},
\end{equation}
where $U_0 \in [-\infty,\infty)$, $U_1 \in (-\infty,\infty]$. Notice that either $\Sigma_t \cap \mathcal D = \Sigma_t$ (when $U_0 = -\infty$ or $U_1 = \infty$) or $\Sigma_t \cap \mathcal D$ is a half-space in $u_k$ (when $U_0$ or $U_1$ is finite).

Before we proceed to the first auxiliary estimate in Proposition~\ref{prop:dxchi.f}, we first need the following simple lemma.

\begin{lemma}
 Let $k\neq k'$, and let $\mathcal D$ be one of the sets in \eqref{eq:D.half.space}. For all $f$ which is sufficiently regular,
 		\begin{equation} \label{Sobolev1Dtheta} 
 		\| f \|_{ L^2_{u_k} L^{\infty}_{u_{k'}} (\Sigma_t\cap \mathcal D) } \lesssim  \| f \|_{ L^2(\Sigma_t\cap \mathcal D) } + \| E_k f \|_{ L^2_{\f 1 {16}}(\Sigma_t\cap \mathcal D) } .
 		\end{equation}
 \end{lemma} %_{\f 1{16}}_{\f 1{16}}
	
	\begin{proof}
 	First, by the standard 1-dimensional Sobolev embedding,
	\begin{equation}\label{eq:silly.1D.Sobolev}
	 \| f \|_{ L^2_{u_k} L^{\infty}_{u_{k'}} (\Sigma_t\cap \mathcal D) } \lesssim  \| f \|_{ L^2_{ u_k,u_{k'}}(\Sigma_t\cap \mathcal D) } + \| \srd_{u_{k'}} f \|_{ L^2_{ u_k,u_{k'}}(\Sigma_t \cap \mathcal D) }.
	\end{equation}
Finally, by \eqref{partialukpEX}, \eqref{mu.main.estimate}, \eqref{anglecontrol} and Lemma \ref{lem:jacobian}, we obtain \eqref{Sobolev1Dtheta}. (Note that we need a small positive weight, e.g., $\langle x \rangle^{\f 1 {16}}$, here because there is a $\mu_{k'}$ factor in \eqref{partialukpEX}, which grows slowly at infinity according to \eqref{mu.main.estimate}.) \qedhere
 \end{proof}

\begin{prop} \label{prop:dxchi.f}
Fix $k$ and let $\mathcal D$ be one of the sets in \eqref{eq:D.half.space}. The following holds for any sufficiently regular $f$:
\begin{equation} \label{L2Linftyestimate}
\| \partial_x \chi_k  \cdot f \|_{L^2_{-\f 14}(\Sigma_t \cap \mathcal D)} + \| E_k \eta_k  \cdot f \|_{L^2_{-\f 14}(\Sigma_t \cap \mathcal D)} \lesssim  \epsilon^{\frac{3}{2}} \cdot (\| f \|_{L^2_{-\f 18}(\Sigma_t \cap \mathcal D)} + \| E_k  f \|_{L^2_{-\f 18}(\Sigma_t \cap \mathcal D)}).
	\end{equation}
\end{prop}
\begin{proof} 
By \eqref{jacobian}, we can bound the $L^2(\Sigma_t)$ norm in either the $(x^1,x^2)$ or the $(u_k,\th_k)$ coordinates. Denoting $h \in \{ \partial_1 \chi_k, \partial_2 \chi_k, E_k \eta_k  \}$, we first use H\"older's inequality in the $(u_k,\theta_k)$ coordinate system to obtain 
\begin{equation}\label{eq:dxchi.f.1}
\begin{split}
 \|\la x\ra^{-\f 14} h \cdot f \|_{L^2(\Sigma_t\cap \mathcal D)} \lesssim &\: (\sup_{u_k \in \RR}\| h \|_{L^{2}_{\theta_k}((\Sigma_t \cap \mathcal D) \cap C_{u_k})}) \cdot \|  \la x\ra^{-\f 1{4}} f \|_{L^{2}_{u_k}L^\i_{\th_k} (\Sigma_t \cap \mathcal D)} \\
 \ls &\: \ep^{\f 32} \cdot \|  \la x\ra^{-\f 1{4}} f \|_{L^{2}_{u_k}L^\i_{\th_k} (\Sigma_t\cap \mathcal D)}.
\end{split}
\end{equation}
Notice now that for a fixed $k'\neq k$, the $L^2_{u_k}L^\i_{\th_k}$ norm is equal to the $L^2_{u_k} L^\i_{u_{k'}}$ norm. Hence, by \eqref{Sobolev1Dtheta},
\begin{equation}\label{eq:dxchi.f.2}
\begin{split}
\|  \la x\ra^{-\f 1{4}} f \|_{L^{2}_{u_k}L^\i_{\th_k} (\Sigma_t \cap \mathcal D)} \ls &\: \| \la x \ra^{-\f 14} f \|_{L^{2}(\Sigma_t \cap \mathcal D)} + \| \la x \ra^{-\f 18} E_k f \|_{L^{2}(\Sigma_t \cap \mathcal D)} + \| \la x \ra^{\f 1{16}} (E_k \la x \ra^{-\f 1{4}}) f \|_{L^{2}(\Sigma_t \cap \mathcal D)} \\
\ls &\: \| \la x \ra^{-\f 1 8} f \|_{L^{2}(\Sigma_t \cap \mathcal D)} + \| \la x \ra^{-\f 1 8} E_k f \|_{L^{2}(\Sigma_t \cap \mathcal D)},
\end{split}
\end{equation}
where we have used $|E_k \la x \ra^{-\f 1{4}}| \ls \la x \ra^{-\f {5}{4} + \ep}$ (by Proposition \ref{prop:main.frame.est}). 

Combining \eqref{eq:dxchi.f.1} and \eqref{eq:dxchi.f.2} then yields \eqref{L2Linftyestimate}. \qedhere

\end{proof}

We now turn to our second auxiliary estimate.
\begin{prop} \label{holdertypeprop}
%For every $1 \leq p <+\infty$ we have the following estimate, under the bootstrap assumptions of section \ref{bootstrapsection}:
Let $\mathcal D$ be one of the sets in \eqref{eq:D.half.space} for some $k$. For all  smooth function $f$ which is Schwartz class on $\Sigma_t$ for all $t \in [0,T_B)$, the following estimate holds for all $t \in [0,T_B)$, where $\mfg \in \{e^{2\gamma}-1,\ e^{-2\gamma}-1,\ \bt^j,\ N-1,\ N^{-1}-1,\ g_{\nu\color{black}\bt},\ \gi^{\nu\color{black}\bt} \}$:
\begin{equation} \label{holderLpeq} 
 \| \partial \partial_x \mfg \cdot f \|_{L^2(\Sigma_t\cap\mathcal D)}  \lesssim \ep^{\f 32} \| f\|_{H^1(\Sigma_t\cap\mathcal D)}.
\end{equation}
\end{prop}
\begin{proof}
 We use H\"{o}lder's inequality to obtain
$$ \| \partial \partial_x \mfg \cdot f \|_{L^2(\Sigma_t\cap\mathcal D)} \lesssim  \| \partial \partial_x \mfg \|_{L^{4}(\Sigma_t\cap\mathcal D)} \cdot \|  f \|_{L^{4}(\Sigma_t\cap\mathcal D)}. $$
Since $s'-s''<\f 12$, Proposition~\ref{prop:main.metric.est} implies an $L^4$ estimate for $\rd\rd_x \mfg$. On the other hand, $f$ can be controlled using the Sobolev embedding $H^{1}(\RR^2) \hookrightarrow L^4(\RR^2)$  or $H^1(\RR^2_+) \hookrightarrow L^4(\RR^2_+$). (In the case where $\Sigma_t\cap \mathcal D$ is a half-space in $u_k$, we perform Sobolev embedding for the half-space in the $(u_k, u_{k'})$ coordinates (for $k' \neq k$), and  note that the $H^1_{x^1,x^2}(\Sigma_t)$, $H^1_{u_k,u_{k'}}(\Sigma_t)$ norms, or the $L^4_{x^1,x^2}(\Sigma_t)$, $L^4_{u_k,u_{k'}}(\Sigma_t)$ norms, are equivalent by \blue{\eqref{eq:COV.1}--\eqref{eq:COV.2}}).
 Hence,
$$ \| \partial \partial_x \mfg \cdot f \|_{L^2(\Sigma_t\cap\mathcal D)} \lesssim   \epsilon^{\frac{3}{2}} \cdot  \|  f \|_{L^{4}(\Sigma_t\cap\mathcal D)} \ls \ep^{\f 32} \| f\|_{H^1(\Sigma_t\cap\mathcal D)}. $$ 
\end{proof} 

\subsection{Computations of the commutators}
 	
 	\subsubsection{The wave operator}\label{waveopsection}%Wave operator expressed in the XEL frame} 

 	\begin{lem} For any $C^2$ function $v$:
	\begin{equation} \label{waveop} \begin{split}
 \Box_g v 
= &\:  -L_k^2 v - 2X_k L_k v+ E_k^2 v + 2 \eta_k  E_k v -( X_k \log(N))   L_k v  -  \chi_k  \barL v \\ 
= &\: -L_k^2 v - 2L_k X_k v+ E_k^2 v + 2( E_k \log(N)-K(E_k,E_k))  E_k v\\ &\:  +( X_k \log(N))   L_k v + (- \chi_k- 2K(X_k,X_k)+ 2 X_k \log(N))  \barL v.
\end{split}
\end{equation}
 	\end{lem}
 \begin{proof}
By \eqref{inversegXEL}, we have
$$ \Box_g v  % =-\nabla^2_{ L_k L_k } v -2\nabla^2_{ L_k X_k } v+ \nabla^2_{ E_k E_k } v 
 = \color{black} -L_k^2 v + (\nabla_{L_k} L_k )v -2 \barL ( L_k v) + 2(\nabla_{\barL} L_k )v + E_k^2 v -  (\nabla_{E_k} E_k)v.$$

Hence, by \eqref{LL}, \eqref{XL} and \eqref{EE2}, we get 
$$\Box_g v =   -L_k^2 v - 2X_k L_k v+ E_k^2 v + 2 \eta_k  E_k v -(K(X_k,X_k)+ K(E_k,E_k)+ X_k \log(N))   L_k v  -  \chi_k  \barL v.$$
Now, by \eqref{maximality} and the fact that $(X_k, E_k)$ is a $g$-orthonormal frame, we have $K(X_k,X_k)+ K(E_k,E_k)=0$, proving the first equality of \eqref{waveop}. The second equality follows from the first combined with \eqref{LX-XL}. \qedhere
 \end{proof}

 	%	From this formula with \eqref{barLu}, \eqref{ELu}, we see that  
 	
 	%	$$ \Box_g u_k = - \frac{1}{2} \cdot \chi_k.$$
 	
 	%	In addition, we can also apply $\Box_g$ operator to one of the vector fields $L_k$, $\barL$ and $E_k$. Using the results of section \ref{riccinullframesection} with \eqref{inversegnullframe}, we get \begin{equation} \label{BoxLvector} 	\Box_g L_k = \left( E_k \chi_k  + (\eta_k-\zeta_k) \cdot \chi_k \right) E_k+\left( -E_k \zeta_k  -2\eta_k \cdot \zeta_k +\frac{1}{2} \chi_k \cdot (\barchi+\baromega)\right) L_k, 	\end{equation} 	\begin{equation} \label{BoxEvector} 	\Box_g E_k = \left(  \chi_k \cdot \barchi -2\eta_k \cdot \bareta\right) E_k+\left( \frac{1}{2} E_k \barchi \cdot  -\barL (\bareta)+ \eta_k \cdot \barchi-\frac{1}{2} \cdot \zeta_k \cdot \barchi\right) L_k+\frac{1}{2} \cdot ( E_k \chi_k + \zeta_k \cdot \chi_k) \barL, 	\end{equation} 	\begin{equation} \label{BoxbarLvector} 	\begin{split} 	\Box_g \barL = \left(  -2L_k \barnu+E_k \barchi -\chi_k \cdot (2\barnu+  \frac{1}{2} \cdot \chi_k)\right) E_k+\left( -L_k \baromega + \zeta_k \cdot (2\barnu+\chi_k)+  \frac{1}{2} \cdot \barchi^2\right) L_k \\ +( E_k \zeta_k + \zeta_k^2 - \frac{1}{2} \cdot \zeta_k \cdot \barchi+ \frac{1}{2} \cdot \chi_k \cdot \barchi) \barL, \end{split} 	\end{equation}
 	
 In the three subsubsections below, we compute the commutator of $\Box_g$ with $\rd_l$, $E_k$ and $L_k$ respectively. We introduce the following conventions: for each commutator we divide into three types of terms; see the statements of Lemmas~\ref{waveopcommspatiallemma}, \ref{waveopcommElemma} and \ref{waveopcommLlemma}. $I$ has second derivatives of metric and first derivatives of $v$; $II$ has first derivatives of metric and second derivatives of $v$; $III$ contains at most one derivative of the metric or $v$.  
 	
 	\subsubsection{Commuting the wave operator with $\rd_l$} \label{waveopcommspatial}
 	
	\begin{lem} \label{waveopcommspatiallemma} For any $C^3$ function $v$ and $l=1,2$, 
 		$$ 	\left[ \partial_l, \Box_g\right] v =I(\partial_l)(v)+II(\partial_l)(v)+III(\partial_l)(v),$$
 		where
 		\begin{align} 
		\label{dlBox-Boxdlmain}
 		I(\partial_l)(v)=	&\: +\frac{e_0( \partial_l\beta ^j  )}{N^2} \partial_j v+ \frac{ \partial_l e_0\log(N) }{N^2} e_0 v+e^{-2 \gamma}  \delta^{i j} \partial_{i l}^2 \log(N)  \partial_{j} v, \\
 		\label{dlBox-BoxdlII} 
 		II(\partial_l)(v) =  &\: -2\partial_l \log N \cdot  \Box_g v  + \frac{2 \partial_l\beta ^j }{N^2} e_0 \partial_j v  + 2\partial_l \log N   e^{-2 \gamma} \delta^{i j} \partial^{2}_{i j} v  -2\partial_l \gamma \cdot e^{-2 \gamma} \delta^{i j} \partial^{2}_{i j} v, \\
 		\notag
 		III(\partial_l)(v) = &\: 2 \partial_l \log N \cdot \partial_{i} \log(N) e^{-2 \gamma}  \delta^{i j} \partial_{j} v - \frac{\partial_l\beta ^j \partial_j\beta ^q }{N^2} \partial_q v \\
		\label{dlBox-BoxdlIII}  
 		&\: - \frac{e_0 N}{N^3} \partial_l \beta^j \partial_j v-2\partial_l \gamma \cdot e^{-2 \gamma}  \delta^{i j} \partial_{i} \log(N)  \partial_{j} v. 
 		\end{align}
 		%where the term $II(\partial_l)$ does not second derivatives of the metric coefficients but involves second derivatives of $v$ and the term $III((\partial_l)$ does not involve second derivatives of the metric coefficients and involves only first derivatives of $v$.
 	
 	\end{lem}
 	\begin{proof}
 		First, from the definition of $e_0= \partial_t -\beta^j \partial_j $, we get the commutator identity $$[\partial_l,e_0]= -\partial_l \beta^j \partial_j .$$ 
 		
 		We now compute the commutator. By \eqref{Box2+1} there are four terms to control.

 	The first term is 
 		\begin{equation*}\begin{split}
 		\left[ \partial_l, \frac{-e_0^2}{N^2}\right] v= &\: \frac{2\partial_l \log N}{N^2} e_0^2 v - \frac{1}{N^2} \left[ \partial_l, e_0^2\right] v= \frac{2\partial_l \log N}{N^2} e_0^2 v - \frac{1}{N^2} \left[ \partial_l, e_0\right] e_0v  -  \frac{1}{N^2} e_0 (\left[ \partial_l, e_0\right]v) \\
 		= &\: \frac{2\partial_l\log N}{N^2} e_0^2 v + \frac{\partial_l\beta ^j }{N^2} \partial_j e_0v  +  \frac{e_0(\partial_l\beta ^j  )}{N^2} \partial_j v  + \frac{\partial_l\beta ^j }{N^2} e_0 \partial_j v\\
 		=&\:   \frac{2\partial_l \log N}{N^2} e_0^2 v  + \frac{2 \partial_l\beta ^j }{N^2} e_0 \partial_j v  + \frac{\partial_l\beta ^j }{N^2} \left[ \partial_j ,e_0 \right]v +  \frac{e_0(\partial_l\beta ^j  )}{N^2} \partial_jv \\ 	= &\: 2\partial_l \log N \left(-\Box_g v +  e^{-2 \gamma} \delta^{i j} \partial^{2}_{i j} v + \frac{e_0 N}{N^3} e_0 v + \frac{ e^{-2 \gamma}}{N}  \delta^{i j} \partial_{i} N \partial_{j} v\right) \\
 		&\: + \frac{2 \partial_l\beta ^j }{N^2} e_0 \partial_j v  - \frac{\partial_l\beta ^j \partial_j\beta ^q }{N^2} \partial_q v +  \frac{e_0( \partial_l\beta ^j  )}{N^2} \partial_jv,
 		\end{split}
 		\end{equation*} where in the last line we expressed $e_0^2 v$ in terms of the other derivatives and $\Box_g v$ using \eqref{Box2+1}. 
 		
 	The second term is $$ \left[ \partial_l ,e^{-2 \gamma} \delta^{i j} \partial^{2}_{i j} \right] v=-2\partial_l \gamma \cdot e^{-2 \gamma} \delta^{i j} \partial^{2}_{i j} v. $$ 
 		
 		The third term is $$\left[ \partial_l ,\frac{e_0 N}{N^3} e_0 \right] v = \partial_l(\frac{e_0 N}{N^3}) e_0 v - \frac{e_0 N}{N^3} \partial_l \beta^j \partial_j v= \frac{ \partial_l e_0\log(N) }{N^2} e_0 v -\frac{ 2\partial_l \log(N)e_0\log(N) }{N^2} e_0 v - \frac{e_0 N}{N^3} \partial_l \beta^j \partial_j v,$$ 
 		and the fourth term is 
 		$$\left[ \partial_l ,\frac{ e^{-2 \gamma}}{N}  \delta^{i j} \partial_{i} N \partial_{j}  \right] v  = \partial_l( \frac{ e^{-2 \gamma}}{N}   \partial_{i} N ) \delta^{i j} \partial_{j} v= -2\partial_l \gamma \cdot e^{-2 \gamma}  \delta^{i j} \partial_{i} \log(N)  \partial_{j} v + e^{-2 \gamma}  \delta^{i j} \partial_{i l}^2 \log(N)  \partial_{j} v .$$

 	Finally, we regroup according to our convention described above, 	noticing that the $\frac{ 2\partial_l \log(N)e_0\log(N) }{N^2} e_0 v$ terms cancel in $III(\partial_l)$.
 	\end{proof}
 	\subsubsection{Commuting the wave operator with $E_k$} \label{waveopcommE}
 	
 	\begin{lem} \label{waveopcommElemma} For any $C^3$ function $v$, 
 		$$ 	\left[ E_k, \Box_g\right] v = I(E_k)(v) + II(E_k)(v) +III(E_k)(v),$$
 		where
 		\begin{equation} \label{EBox-BoxEmain} \begin{split}
 		I(E_k)(v) =&\: ( - 2X_k \chi_k +2 E_k \eta_k - L_k \chi_k) \cdot  E_k v - E_k \chi_k \cdot X_k v\\
 		&\:  + (L_k E_k + 2X_k E_k -E_k X_k)  \log(N) \cdot   L_k v ,
 		\end{split}
 		\end{equation}	
 		\begin{equation} \label{EBox-BoxEII} \begin{split}
 		II(E_k)(v)  =  &\:  2E_k \log(N) \cdot  L^2_k v + 2   (E_k \log(N)+\eta_k -K(E_k,X_k)) \cdot  X_k L_k v \\ 
 		&\: - 2  \chi_k  \cdot X_k E_k v  -2K(E_k,E_k) \cdot  E_k L_k  v, 
 		\end{split}
 		\end{equation}
 		\begin{equation}  \label{EBox-BoxEIII}  \begin{split}
 		III(E_k)(v) = &\:\chi_k \cdot (\eta_k-K(E_k,X_k))  \cdot X_k v \\
 		&\: +\chi_k \cdot ( 2\chi_k  - X_k\log N - K(E_k,E_k) ) \cdot E_k v \\&\: + E_k \log(N) \cdot ( X_k \log(N)-\chi_k) \cdot L_k v.
 		\end{split}
 		\end{equation}
 		% 		where the term $II(E_k)$ does not involve derivatives of the Ricci coefficients but involves second derivatives of $v$ and the term $III(E_k)$ does not involve derivative of Ricci coefficients and involves only first derivatives of $v$.

 	\end{lem}
 	
 	\begin{proof}
 		\pfstep{Step~1: The main computation} From \eqref{waveop}, we see that \begin{equation}\label{eq:comm.E.main} \begin{split}
 		&\: \left[ E_k, \Box_g\right] v\\
 		= &\: L_k ([ L_k, E_k]v)	+[ L_k, E_k] L_k v	+2 X_k( [ L_k, E_k] v)+2 [ X_k, E_k]L_k v \\
 		&\: + 2 E_k \eta_k \cdot E_k v - E_k \chi_k \cdot X_k v - \chi_k  \cdot [ E_k, X_k] v
 		-( X_k \log(N))  \cdot   [ E_k, L_k] v 
 		  -E_kX_k \log(N)  \cdot  L_k v.
 		\end{split}
 		\end{equation}
 		
 		We deal with all the terms one by one. For the first commutator, notice using \eqref{EL-LE} that 
 		\begin{equation}\label{eq:comm.E.top.1}  L_k ([ L_k, E_k]v) = E_k \log(N) \cdot  L^2_k v - \chi_k  \cdot L_k E_k v - L_k \chi_k \cdot  E_k v +L_k E_k \log(N)  \cdot L_k v,
 		\end{equation}
 		Similarly, 
 		\begin{equation}\label{eq:comm.E.top.2}
 		[ L_k, E_k] L_k v = E_k \log(N) \cdot  L^2_k v - \chi_k  \cdot E_k L_k  v .
 		\end{equation}
 		
 		Now using \eqref{EL-LE} and \eqref{EX-XE}, we
 		obtain 
 		\begin{equation}\label{eq:comm.E.top.3}
 		X_k ([ L_k, E_k]v) = E_k \log(N) \cdot  X_k L_k v - \chi_k  \cdot X_k E_k v - X_k \chi_k \cdot  E_k v +X_k E_k \log(N)  \cdot L_k v,
 		\end{equation}
 		\begin{equation}\label{eq:comm.E.top.4}
 		[ X_k, E_k] L_k v = (\eta_k -K(E_k,X_k)) \cdot  X_k L_k  v+ (\chi_k -K(E_k,E_k)) \cdot  E_k L_k  v .
 		\end{equation}
 		
 		%Now we group the terms in three sub-terms: $I(E_k)$ gathers all the terms which have the form $\partial^2 g \cdot \partial v$,  $II(E_k)$ all the terms which have the form $\partial g \cdot \partial^2 v$, and $III(E_k)$ all the terms which have the form $(\partial g )^2\cdot \partial v$.
 		
 		%Before defining precisely $I(E_k)$, $II(E_k)$ and $III(E_k)$, we must work on the term 
 		
 	\pfstep{Step~2: Rewriting some terms} We rearrange the term $-\chi_k L_k E_k v$ from \eqref{eq:comm.E.top.1}, which we write by \eqref{EL-LE} as
 		\begin{equation}\label{eq:comm.E.final.comm}
 		-\chi_k L_k E_k v = -\chi_k E_k L_k v - \chi_k [L_k, E_k ] v = -\chi_k (E_k L_k v -\chi_k \cdot E_k v  + E_k \log N \cdot L_k v).
 		\end{equation}
 		Notice that all instances of $\chi_k \cdot E_k L_k v$ in $\mbox{\eqref{eq:comm.E.top.1}}+\mbox{\eqref{eq:comm.E.top.2}}+2\times\mbox{\eqref{eq:comm.E.top.3}}+2\times\mbox{\eqref{eq:comm.E.top.4}}$ then cancel.
 		
 		Finally, we conclude the proof by plugging \eqref{eq:comm.E.top.1}--\eqref{eq:comm.E.top.4} into \eqref{eq:comm.E.main}, expanding the remaining terms using Lemma~\ref{riccibarXEL}, and finally substituting in  \eqref{eq:comm.E.final.comm}. Note that we split into the terms $I$, $II$ and $III$ according to our convention described in Section~\ref{waveopsection} above. \qedhere

 	\end{proof}

 	\subsubsection{Commuting the wave operator with $L_k$ } \label{waveopcommL}

 	\begin{lem}\label{waveopcommLlemma} For any $C^3$ function $v$, %we have the commutation formula: 
 		$$\left[ L_k, \Box_g\right] v= I(L_k)(v) +II(L_k) (v) +III(L_k) (v),$$ 
		where 
 		\begin{equation} \label{LBox-BoxLmain} \begin{split}
 		I(L_k) (v) =&\:  -L_k \chi_k \cdot X_k v +\left(2L_k \eta_k -E_k \chi_k\right) \cdot E_k v
 		+\left( E_k^2 \log(N)  -  L_kX_k \log(N)\right) \cdot L_k v, \end{split}
 		\end{equation}	
 		\begin{equation} \label{LBox-BoxLII} \begin{split}
 		II(L_k) (v)  =  &\: 2\left( K(E_k,X_k)+\eta_k  \right) \cdot E_k L_k v  +2\left(K( X_k,X_k) -X_k\log(N)\right) \cdot  X_k L_k v \\
 		&\: - 2 X_k\log(N) \cdot L_k^2 v-2\chi_k \cdot E^2_k v, 
 		\end{split}
 		\end{equation}
 		\begin{equation}  \label{LBox-BoxLIII} \begin{split}
 		III(L_k) (v)  = &\: \chi_k \cdot(K(X_k,X_k)-X_k \log(N)) \cdot X_k v+ \chi_k \cdot ( K(E_k,X_k) -  2  E_k \log(N)-\eta_k) \cdot E_k v \\
 		&\: + (  2\eta_k \cdot E_k\log(N) + (E_k\log(N))^2 - \chi_k \cdot X_k \log(N)) \cdot L_k v. \end{split}
 		\end{equation}
 		%where the term $II(L_k)$ does not involve derivatives of the Ricci coefficients but involves second derivatives of $v$ and the term $III(L_k)$ does not involve derivatives of Ricci coefficients and involves only first derivatives of $v$.

 	\end{lem}

 	\begin{proof}
 		By \eqref{waveop}, we have
 		
 		\begin{equation*} \begin{split}
 		\left[ L_k, \Box_g\right] v = &\: \overbrace{ -2 [L_k,X_k] L_k v}^{A(L_k)} + \overbrace{[L_k,E_k] E_k v+  E_k ([L_k,E_k]  v)}^{B(L_k)} +  \overbrace{2L_k \eta_k \cdot E_k v     -L_k \chi_k \cdot  X_k v}^{C(L_k)}\\ 
 		&\:\underbrace{-L_k   X_k \log(N) L_k v }_{D_1(L_k)}+ \underbrace{2 \eta_k [L_k,E_k] v -\chi_k [L_k,X_k] v}_{D_2(L_k)}.
 		\end{split}
 		\end{equation*} 
 		We treat each term separately. We start with $A(L_k)$ and using \eqref{LX-XL} we obtain \begin{equation*} \label{AL}
 		A(L_k) = 2(K(E_k,X_k) -E_k\log N +\eta_k) \cdot E_k L_k v  +2(K( X_k,X_k) -X_k\log(N)) \cdot  X_k L_k v -2 X_k\log(N) \cdot L_k^2 v.
 		\end{equation*}
 		Now we handle $B(L_k)$ using \eqref{EL-LE}: 
 		\begin{equation*} \label{BL} \begin{split}
 		&\: B(L_k) \\
 		= &\: -2\chi_k \cdot E_k^2 v -E_k\chi_k \cdot E_k v +  E_k \log(N) \cdot L_k E_k v + E_k \log(N) \cdot E_k L_k v  + E_k^2 \log(N) \cdot L_k v \\
 		= &\:  -2\chi_k \cdot E_k^2 v -E_k\chi_k \cdot E_k v +  2 E_k \log(N) \cdot E_k L_k v  + E_k^2 \log(N) \cdot L_k v  -E_k\log N(\chi_k \cdot E_k v - E_k\log N \cdot L_k v). % +E_k \log(N) \chi_k E_k v-(E_k \log(N))^2 L_k v \\= &\:  -2\chi_k \cdot E_k^2 v -E_k\chi_k \cdot E_k v + 2E_k \log(N) \cdot E_k L_k v+ E_k^2 \log(N) \cdot L_k v\\ 
 		%&\: +E_k \log(N) \chi_k E_k v+(E_k \log(N))^2 L_k v - 2 \chi_k \cdot E_k \log(N) \cdot  E_k v, 
 		\end{split}
 		\end{equation*}

 		For $C(L_k)$ and $D_1(L_k)$, there is nothing to do.

 		Finally, for $D_2(L_k)$ we use \eqref{EL-LE} and \eqref{LX-XL} to get 
 		\begin{equation*} \label{D2L}
 		\begin{split}
 		D_2(L_k)= &\: \chi_k \cdot(-\eta_k+K(E_k,X_k) -E_k\log(N) ) \cdot E_k v + \chi_k \cdot(K(X_k,X_k)-X_k \log(N)) \cdot X_k v \\
 		&\: + (2\eta_k E_k\log(N)  - \chi_k \cdot X_k \log(N) ) \cdot L_k v .
 		\end{split}
 		\end{equation*}
 		
 		Rearranging the terms according to conventions in Section~\ref{waveopsection} for $I$, $II$ and $III$ yields the conclusion. \qedhere
 		%Finally, we obtain that $A(L_k)+B(L_k)+C(L_k)+D_1(L_k)+D_2(L_k)=I(L_k)+II(L_k)+III(L_k)$ as follows: \begin{equation*} \begin{split}
 		%I(L_k)=  -L_k \chi_k X_k v+(2L_k \eta_k -E_k \chi_k) \cdot E_k v\\+( E_k^2 \log(N) -\nabla_{L_k}K(X_k,X_k)   -\nabla_{L_k}K(E_k,E_k) -  L_kX_k \log(N)) \cdot L_k v
 		%\end{split},
 		%\end{equation*}
 		%\begin{equation*} \begin{split}
 		%II(L_k)=  2(K(E_k,X_k)+\eta_k +E_k\log(N)) \cdot E_k L_k v  +2(K( X_k,X_k) -X_k\log(N)) \cdot  X_k L_k v -2 X_k\log(N) \cdot L_k^2 v-2\chi_k \cdot E^2_k v
 		%\end{split},
 		%\end{equation*}	\begin{equation*} \begin{split}
 		%III(L_k)=  \chi_k \cdot(K(X_k,X_k)-X_k \log(N)) \cdot X_k v+ \chi_k \cdot ( K(E_k,X_k)-E_k \log(N)-\eta_k) \cdot E_k v \\+ (E_k\log(N) \cdot ( -2K(E_k,X_k) +2\eta_k +E_k\log(N))+\chi_k \cdot(K(X_k,X_k)-X_k \log(N))) \cdot L_k v.
 		%\end{split}
 		%\end{equation*}

 	\end{proof}

\subsection{Estimating the commutator $[\Box_g, \partial_i]$}\label{sec:Box.rdi.commutator}
In the remainder of this section, we bound the commutators of $\Box_g$ with different vector fields. Once we bound the commutators, we also obtain an energy estimate for the commuted quantity using Proposition~\ref{prop:EE}.

In this subsection, we begin with the commutator $[\Box_g, \rd_i]$. We recall that this commutator is computed in Lemma~\ref{waveopcommspatiallemma}.

\begin{prop} \label{prop:commute.with.nabla}
 For any $k\in \{1,2,3\}$, define $\mathcal D$ as one of the sets in \eqref{eq:D.half.space}. Then the following holds for all solutions $v$ to $\Box_g v = f$, with $\mathrm{supp}(v),\,\mathrm{supp}(f) \subseteq \{(t,x): |x|\leq R\}$:
	\begin{equation*} 
	\begin{split}
	  \| [\Box_g, \partial_i] v\|_{L^1_t( [0,T_B), L^2(\Sigma_t\cap\mathcal D))}
	 \ls  \ep^{\f 32} \cdot (\|\partial_x \rd v\|_{L^1_t( [0,T_B), L^2(\Sigma_t\cap\mathcal D))}+ \|f\|_{L^1_t([0,T_B), L^2(\Sigma_t\cap\mathcal D))}).
	 \end{split}
	\end{equation*} 
\end{prop}
	\begin{proof}
	We control each term in Lemma~\ref{waveopcommspatiallemma}. Controlling the metric terms using \eqref{eq:g.main}, we immediately obtain
	$$\| II(\partial_l)(v) \|_{L^2(\Sigma_t\cap\mathcal D)} \lesssim \epsilon^{\frac{3}{2}} \cdot (\| \partial_x \partial v \|_{L^2(\Sigma_t\cap\mathcal D)}+\|\Box_g v \|_{L^2(\Sigma_t\cap\mathcal D)}),$$ $$\| III(\partial_l)(v) \|_{L^2(\Sigma_t\cap\mathcal D)} \lesssim \epsilon^{3} \cdot \|  \partial v \|_{L^2(\Sigma_t\cap\mathcal D)} \lesssim \epsilon^{3} \cdot \|  \rd_x \partial v \|_{L^2(\Sigma_t\cap\mathcal D)},$$
where in the last inequality, we used $\mathrm{supp}(v) \subseteq B(0,R)$ and Poincar\'e's inequality.
	
	For $I(\partial_l)(v)$, notice that after using \eqref{eq:g.main} and the support properties, each term is bounded above by $|\rd\rd_x \mfg\cdot \rd v|$, where $\rd\rd_x \mfg$ is as in Proposition~\ref{holdertypeprop}. Thus, Proposition~\ref{holdertypeprop} implies
	\begin{equation*} 
	\begin{split} 
	\| I(\partial_l)(v) \|_{L^2(\Sigma_t\cap\mathcal D)}  \lesssim \epsilon^{\frac{3}{2}} \cdot \|  \rd  v \|_{H^1(\Sigma_t\cap\mathcal D)} \ls \epsilon^{\frac{3}{2}} \cdot \| \rd_x \rd v \|_{L^2(\Sigma_t\cap\mathcal D)},
	\end{split}
	\end{equation*}
	where we again used Poincar\'e's inequality in the last inequality.
	
	Taking the $L^1_t$ norm of these three inequalities, and using $\Box_g v =f$, yields \color{black}  the claimed estimate.
	\end{proof}

Now, we use the commutator estimate in Proposition~\ref{prop:commute.with.nabla} to control the energy for the commuted function: 
\begin{prop} \label{prop:commute.with.spatial.direct}
		Suppose $\Box_g v = f$ with $v$ and $f$ both smooth and compactly supported in $B(0,R)$ for every $t\in [0, T_B)$. Let $-\infty \leq U_0 \leq U_1 \leq +\infty$, with either $U_0 = -\infty$ or $U_1 = \infty$ (or both). Let $\mathcal D = \mathcal D^{(k),T_B}_{U_0,U_1}$, where $\mathcal D^{(k),T_B}_{U_0,U_1}$ is given by \eqref{def:D.domain}.

		Then, for any $k'$,
		\begin{equation*}
		\begin{split}
		&\: \sup_{0 \leq t < T_B}\|\rd \partial_x v\|_{ L^2(\Sigma_t\cap\mathcal D)} + \sup_{u_{k'} \in [U_0,U_1)} \sum_{Z_{k'} \in \{L_{k'},\,E_{k'}\}} \|Z_{k'} \partial_x v\|_{L^2(C^{k'}_{u_{k'}}\cap \mathcal D)}\\
		\ls &\: \|\rd \partial_x v \|_{L^2(\Sigma_0\cap \mathcal D)} + \sum_{Z_k \in \{L_k,\,E_k\}} \|Z_k \partial_x v\|_{L^2(C^k_{U_0})}   + \| \partial_x f\|_{L^1_t([0,T_B),L^2(\Sigma_t\cap\mathcal D))}.
		\end{split}
		\end{equation*}

\end{prop}
\begin{proof} 
 We apply the energy estimate in Proposition~\ref{prop:EE}, but to $\rd_i v$ instead of $v$. Notice that $\Box_g (\rd_i v) = [\Box_g, \rd_i ] v + \rd_i f$. Hence, combining the energy estimate in Proposition~\ref{prop:EE} with the bound for the commutator $[\Box_g, \rd_i ]$ in Proposition~\ref{prop:commute.with.nabla}, we obtain
\begin{equation*}
		\begin{split}
		&\: \sup_{0 \leq t < T_B}\|\rd \partial_x v\|_{ L^2(\Sigma_t\cap\mathcal D)} + \sup_{u_{k'} \in [U_0,U_1)} \sum_{Z_{k'} \in \{L_{k'},\,E_{k'}\}} \|Z_{k'} \partial_x v\|_{L^2(C^{k'}_{u_{k'}}\cap \mathcal D)}\\
		\ls &\: \|\rd \partial_x v \|_{L^2(\Sigma_0\cap \mathcal D)} + \sum_{Z_k \in \{L_k,\,E_k\}} \|Z_k \partial_x v\|_{L^2(C^k_{U_0})} + \ep^{\f 32} \cdot \|\partial_x \rd v\|_{L^1_t( [0,T_B), L^2(\Sigma_t\cap\mathcal D))}  \\
		&\: + \|f\|_{L^1_t([0,T_B), L^2(\Sigma_t\cap\mathcal D))} + \| \partial_x f\|_{L^1_t([0,T_B), L^2(\Sigma_t\cap\mathcal D))}.
		\end{split}
		\end{equation*}

Now notice that by $T_B \leq 1$, $\ep^{\f 32} \|\partial_x \rd v\|_{L^1_t( [0,T_B), L^2(\Sigma_t\cap\mathcal D))} \ls \ep^{\f 32} \sup_{0 \leq t < T_B}\|\rd \partial_x v\|_{ L^2(\Sigma_t\cap\mathcal D)}$, and hence this term can be absorbed by the left-hand side. Moreover, using that $\mathrm{supp}(f) \subseteq B(0,R)$, we have $\|f\|_{L^1_t([0,T_B), L^2(\Sigma_t\cap\mathcal D))} \ls \| \partial_x f\|_{L^1_t([0,T_B), L^2(\Sigma_t\cap\mathcal D))}$ by Poincar\'e's inequality. Combining all these observations yields the desired estimate. \qedhere
\end{proof}

\subsection{Estimating the commutator $[\Box_g,E_k]$}

Next, we turn to the commutator $[\Box_g,E_k]$. \blue{Unlike the estimates in Section~\ref{sec:Box.rdi.commutator}, when we bound $[\Box_g,E_k]v$, we will not assume $v$ to be compactly supported.} (We remark that such bounds for non-compactly supported $v$ are needed for the applications in Section~\ref{unlochighestfrac}.)

\begin{prop} \label{prop:commute.with.E}
	Fix $k$ and define $\mathcal D$ as one of the sets in \eqref{eq:D.half.space}. Let $v$ be smooth function which is in Schwartz class for every $t\in [0, T_B)$. Then, for all $r\geq 1$, the following holds for all $t\in [0,T_B)$:
	\begin{equation*}
	\begin{split}
	 \|\wo2 [\Box_g, E_k] v\|_{L^2(\Sigma_t \cap \mathcal D)}
	\ls &\: \ep^{\f 32} \cdot (\|\rd v\|_{L^2(\Sigma_t \cap \mathcal D)}   + \sum_{Z_k \in \{E_k, L_k\} }\|\rd Z_k v\|_{L^2(\Sigma_t \cap \mathcal D)}),
	\end{split}
	\end{equation*} 
	where the implicit constant is allowed to depend on $r$.
\end{prop}
\begin{proof}  
Recall the computation of $[\Box_g, E_k]$ in Lemma \ref{waveopcommElemma}. We now bound the terms $I(E_k)$, $II(E_k)$, $III(E_k)$ from Lemma \ref{waveopcommElemma}.

We first recall the definition of $I(E_k) (v)$ in \eqref{EBox-BoxEmain}. By the $L^\i$ bound for $L_k\chi_k$ in \eqref{eq:Lchi.Leta}, the $L^\i$ estimates for the geometric vector fields in Proposition~\ref{prop:main.frame.est}, and the $L^\i$ estimates for the metric coefficients and their derivatives in \eqref{eq:g.main}, we get (recalling $\alp=0.01$)
\begin{equation}\label{eq:Box.E.comm.I.0}
 | I(E_k) |(v) \lesssim \underbrace{\ep^{\f 32} |\rd v|}_{=:A} + \underbrace{\la x \ra^{\alp} \left( |\partial_x \chi_k|+| E_k  \eta_k|\right) |E_k v|}_{=:B}  + \underbrace{\la x\ra^{\nu\color{black}} |\partial_x \partial N| \cdot |L_k v|}_{=:D}.
 \end{equation}
We control the $L^2_{-\f r2}(\Sigma_t)$ norm of each term. The term $A$ obviously satisfies 
\begin{equation}\label{eq:Box.E.comm.I.1}
\|A \|_{L^2_{-\f r2} (\Sigma_t \cap \mathcal D)}\ls \ep^{\f 32} \|\rd v\|_{L^2(\Sigma_t \cap \mathcal D)}.
\end{equation}
For $B$, we use Proposition~\ref{prop:dxchi.f}, $r\geq 1$ and Proposition \ref{prop:main.frame.est} to obtain
\begin{equation}\label{eq:Box.E.comm.I.2}
\begin{split}
\|B\|_{L^2_{-\f r 2}(\Sigma_t \cap \mathcal D)} \ls \|  ( |\partial_x \chi_k|+| E_k  \eta_k| ) \cdot E_k v \|_{L^2_{-\f {1} 4}(\Sigma_t \cap \mathcal D)} \ls \ep^{\f 32} \|E_k v\|_{L^2_{-\f 18}(\Sigma_t \cap \mathcal D)} + \ep^{\f 32} \|E_k^2 v\|_{L^2_{-\f 18}(\Sigma_t \cap \mathcal D)}\\ 
\ls \ep^{\f 32} \|\rd  v\|_{L^2_{-\f 1 {16}}(\Sigma_t \cap \mathcal D)}+ \ep^{\f 32} \|\rd E_k v\|_{L^2_{-\f 1 {16}}(\Sigma_t \cap \mathcal D)} \ls \ep^{\f 32} \|\rd  v\|_{L^2(\Sigma_t)}+ \ep^{\f 32} \|\rd E_k v\|_{L^2(\Sigma_t \cap \mathcal D)}.
\end{split}
\end{equation}
The term $D$ can be handled by Proposition~\ref{holdertypeprop}, giving
\begin{equation}\label{eq:Box.E.comm.I.3}
\begin{split}
\| \la x \ra^{ -\frac{r}{2}+ \alp}  |\rd \partial_x N| \cdot |L_k v| \|_{L^2(\Sigma_t \cap \mathcal D)} \ls \ep^{\f 32} (\|\rd v \|_{L^2(\Sigma_t \cap \mathcal D)} + \|\rd L_k v\|_{L^2(\Sigma_t \cap \mathcal D)}).
\end{split}
\end{equation}
Combining \eqref{eq:Box.E.comm.I.0}--\eqref{eq:Box.E.comm.I.3}, we obtain
\begin{equation}\label{eq:Box.E.comm.I}
\begin{split}
 \|\la x\ra^{-\frac{r}{2}} I(E_k)(v) \|_{L^2(\Sigma_t \cap \mathcal D)} 
\ls &\:  \ep^{\f 32} \cdot \big(\|\rd v\|_{ L^2(\Sigma_t \cap \mathcal D)}   + \sum_{Z_k \in \{E_k, L_k\} }\|\rd Z_k v\|_{ L^2(\Sigma_t \cap \mathcal D)}\big).
\end{split}
\end{equation}

By \eqref{EBox-BoxEII} and the estimates in Propositions~\ref{prop:main.metric.est}, \ref{prop:main.frame.est} and \ref{prop:main.Ricci.est}, we have the pointwise estimate\footnote{Notice that one could even put in additional decaying weights of $\la x \ra$ in this estimate, but this will not be necessary.}
\begin{equation}\label{eq:Box.E.comm.II}
 |II(E_k)(v)| \lesssim \epsilon^{\frac{3}{2}} \left( |\rd L_k v|+ |\partial E_k v| \right).
 \end{equation}
In a similar manner, but starting with \eqref{EBox-BoxEIII}, we also obtain the pointwise estimate
\begin{equation}\label{eq:Box.E.comm.III}
 |III(E_k)(v)| \lesssim \epsilon^{3} |\partial v| . 
\end{equation}	
By \eqref{eq:Box.E.comm.II} and \eqref{eq:Box.E.comm.III}, it follows immediately that
\begin{equation}\label{eq:Box.E.comm.II.III}
\begin{split}
&\: \|\la x\ra^{-\frac{r}{2}} II(E_k)(v) \|_{L^2(\Sigma_t \cap \mathcal D)} + \|\la x\ra^{-\frac{r}{2}} III(E_k)(v) \|_{L^2(\Sigma_t \cap \mathcal D)} \\
\ls &\: \ep^{\f 32} \cdot (\|\rd v\|_{L^2(\Sigma_t \cap \mathcal D)}   + \sum_{Z_k \in \{E_k, L_k\} }\|\rd Z_k v\|_{ L^2(\Sigma_t \cap \mathcal D)}).
\end{split}
\end{equation}
	
	Combining Lemma \ref{waveopcommElemma}, \eqref{eq:Box.E.comm.I} and \eqref{eq:Box.E.comm.II.III} yields the conclusion. \qedhere
\end{proof}

In the next proposition, we are going to use Proposition~\ref{prop:commute.with.E} to estimate the energy commuted with the vector field $E_k$, this time for compactly supported functions (so that the spatial weights become irrelevant).

\begin{proposition}\label{prop:commute.with.E.direct}
	Suppose $\Box_g v = f$ with $v$ and $f$ both smooth and compactly supported in $B(0,R)$ for every $t\in [0, T_B)$. Let $-\infty \leq U_0 \leq U_1 \leq +\infty$, with either $U_0 = -\infty$ or $U_1 = \infty$ (or both). Let $\mathcal D = \mathcal D^{(k),T_B}_{U_0,U_1}$, where $\mathcal D^{(k),T_B}_{U_0,U_1}$ is given by \eqref{def:D.domain}.
	
	Then, for any $k'$,
	\begin{equation*}
	\begin{split}
	&\: \sup_{0 \leq t < T_B}\|\rd E_k v\|_{ L^2(\Sigma_t \cap \mathcal D)} + \sup_{u_{k'} \in \RR} \sum_{Z_{k'} \in \{L_{k'},\,E_{k'}\}} \|Z_{k'} E_k v\|_{L^2(C^{k'}_{u_{k'}} \cap \mathcal D )}\\
	\ls &\: \|\rd E_k v \|_{L^2(\Sigma_0 \cap \mathcal D)} + \sum_{Z_k \in \{L_k,\,E_k\}} \|Z_k E_k  v\|_{L^2(C^k_{U_0})}\\
	&\: + \ep^{\f 32} (\|\rd v\|_{L^1_t( [0,T_B), L^2(\Sigma_t \cap \mathcal D))}   + \sum_{Z_k \in \{E_k, L_k\} }\|\rd Z_k v\|_{L^1_t( [0,T_B), L^2(\Sigma_t \cap \mathcal D))}) + \|E_k f\|_{L^1_t( [0,T_B), L^2(\Sigma_t \cap \mathcal D))}.
	\end{split}
	\end{equation*}
\end{proposition}
\begin{proof}
	This is an immediate consequence of the combination of Proposition \ref{prop:EE} and Proposition \ref{prop:commute.with.E}, since $$ \Box_g(E_k v)= [\Box_g,E_k] v + E_k f.$$
\end{proof}

\subsection{Estimating the commutator $[\Box_g,L_k]$}
The final commutator to estimate is $[\Box_g, L_k]$. We will prove analogues of Propositions~\ref{prop:commute.with.E} and \ref{prop:commute.with.E.direct} with $E_k$ replaced by $L_k$.

\begin{prop} \label{prop:commute.with.L}
	Fix $k$ and define $\mathcal D$ as one of the sets in \eqref{eq:D.half.space}. Let $v$ be smooth function which is in Schwartz class for every $t\in [0, T_B)$. Then, for all $r\geq 1$, the following holds for all $t\in [0,T_B)$:
	\begin{equation*} 
	\begin{split}
	&\: \|\wo2 [\Box_g, L_k] v\|_{L^2(\Sigma_t \cap \mathcal D)} \ls  \ep^{\f 32} \cdot (\|\rd v\|_{L^2(\Sigma_t \cap \mathcal D)}   + \sum_{Z_k \in \{E_k, L_k\}} \|\rd Z_k v\|_{L^2(\Sigma_t \cap \mathcal D)}),
	\end{split}
	\end{equation*} 
	where the implicit constant is allowed to depend on $r$.
\end{prop}
\begin{proof}  
	We bound the terms $I(L_k)$, $II(L_k)$, $III(L_k)$ from Lemma \ref{waveopcommLlemma}, following the same lines of reasoning as for Proposition \ref{prop:commute.with.E}.  We get (recall $\alp=0.01$):  
	$$ |I(L_k)(v)| \lesssim  \ep^{\f 32} |\rd v| + \la x\ra^{2\alp} ( |\partial_x \chi_k|+  |\partial\partial_x N|) \cdot   |\partial v |, \quad  |II(L_k)(v)| \lesssim \epsilon^{\frac{3}{2}} |\rd L_k v| ,\quad |III(L_k)(v)| \lesssim \epsilon^{3}|\partial v|.$$

These terms are exactly those in Proposition~\ref{prop:commute.with.E}, and therefore can be treated in exactly the same manner. \qedhere
	
\end{proof} 

The next proposition is analogous to Proposition \ref{prop:commute.with.E.direct}:
\begin{proposition}\label{prop:commute.with.L.direct}
	Suppose $\Box_g v = f$ with $v$ and $f$ both smooth and compactly supported in $B(0,R)$ for every $t\in [0, T_B)$. Let $-\infty \leq U_0 \leq U_1 \leq +\infty$, with either $U_0 = -\infty$ or $U_1 = \infty$ (or both). Let $\mathcal D = \mathcal D^{(k),T_B}_{U_0,U_1}$, where $\mathcal D^{(k),T_B}_{U_0,U_1}$ is given by \eqref{def:D.domain}.
	
	Then, for any $k'$
	\begin{equation*}
	\begin{split}
	&\: \sup_{0 \leq t < T_B}\|\rd L_k v\|_{ L^2(\Sigma_t \cap \mathcal D)} + \sup_{u_{k'} \in \RR} \sum_{Z_{k'} \in \{L_{k'},\,E_{k'}\}} \|Z_{k'} L_k v\|_{L^2(C^{k'}_{u_{k'}} \cap \mathcal D)}\\
	\ls &\: \|\rd L_k v \|_{L^2(\Sigma_0)} + \ep^{\f 32} \cdot (\|\rd v\|_{L^1_t( [0,T_B), L^2(\Sigma_t \cap \mathcal D))}   + \sum_{Z_k \in \{E_k, L_k\} }\|\rd Z_k v\|_{L^1_t( [0,T_B), L^2(\Sigma_t \cap \mathcal D))}) \\
	&\: + \|L_k f\|_{L^1_t( [0,T_B), L^2(\Sigma_t \cap \mathcal D))}.
	\end{split}
	\end{equation*}
\end{proposition}
\begin{proof}
	Noting  $\Box_g(L_k v)= [\Box_g,L_k] v + L_k f$, this is an immediate consequence of Propositions~\ref{prop:EE} and \ref{prop:commute.with.L}. \qedhere
\end{proof}

%\section{Commutation with fractional derivatives}

\section{Energy estimates for $\tphi$ up to two derivatives I: the basic estimates} \label{firstcommutedsection}

The goal of this section is to obtain energy estimates, for the scalar field commuted with zero or one derivative on the whole of $\Sigma_t$. When there is no commutations or one commutation with a good derivative, we bound the energy uniformly in $\de$, while if there is one commutation with a general spatial derivative, we allow the energy to grow in $\de^{-1}$. 
	
	We note already that some of these estimates will be later improved in Section~\ref{exterior}, by localizing on different regions of the spacetime.
	
	The main result of this section is the next proposition:

\begin{prop}\label{prop:easy.energy}
	The following energy estimate holds for the lowest order energy:
	\begin{align} 
	\label{energyglobal} 
	\sup_{0 \leq t < T_B}\| \partial \tphi \|_{L^2(\Sigma_t)}+  \sup_{u_{k'} \in \RR} \sum_{Z_{k'} \in \{L_{k'},\,E_{k'}\}} \|Z_{k'} \tphi\|_{L^2(C^{k'}_{u_{k'}}([0,T_B)) )}  \lesssim \epsilon.
	\end{align}
	The following energy estimate holds after commutation with one good vector field:
	\begin{align}
	\label{energygoodcommutedglobal}
	\sum_{Z_k \in \{L_k,\,E_k\}} (\sup_{0 \leq t < T_B} \| \partial Z_k \tphi \|_{L^2(\Sigma_t)}+  \sup_{ u_k \in \RR} \sum_{Y_{k'} \in \{L_{k'},\,E_{k'} \}} \| Y_{k'} Z_k \tphi\|_{L^2(C^{k'}_{u_{k'}}( [0,T_B)))} ) \lesssim \epsilon.
	\end{align}
	Finally, the following energy estimate holds for more general second derivatives of $\tphi$:
	\begin{equation} \label{badlocenergyestimate}
	\sup_{0 \leq t < T_B} \| \partial^2 \tphi\|_{L^2(\Sigma_t)} +  \sup_{ u_{k'} \in \RR } \sum_{Z_{k'} \in \{L_{k'},\,E_{k'}\}} \|Z_{k'} \partial_x \tphi\|_{L^2(C^{k'}_{u_{k'}}([0,T_B)))}  \lesssim  \epsilon \cdot  \delta^{-\frac{1}{2}}.
	\end{equation}
\end{prop}

\begin{proof} \pfstep{Step~1: Proof of \eqref{energyglobal}} This is an immediate consequence of the energy estimate in Proposition~\ref{prop:EE} (with $v = \tphi$\blue{, $f = 0$}, $U_0 = -\infty$, and $U_1 = \infty$), and the initial data bound in \eqref{eq:assumption.rough.energy}.
	
	\pfstep{Step~2: Proof of \eqref{energygoodcommutedglobal}} Summing the estimates in Propositions~\ref{prop:commute.with.E.direct} and \ref{prop:commute.with.L.direct} with $v = \tphi$ (so that $f = \Box_g v = \Box_g \tphi =0$), we have
	\begin{equation}
	\begin{split}
	&\: \sum_{Z_k \in \{L_k,\,E_k\}} (\sup_{0 \leq t < T_B} \| \partial Z_k \tphi \|_{L^2(\Sigma_t)}+  \sup_{ u_k \in \RR} \sum_{Y_{k'} \in \{L_{k'},\,E_{k'} \}} \| Y_{k'} Z_k \tphi\|_{L^2(C^{k'}_{u_{k'}}( [0,T_B)))} ) \\
	\ls &\: \underbrace{\sum_{Z_k \in \{L_k,E_k\}} \|\rd Z_k \tphi \|_{L^2(\Sigma_0)}}_{=:I} + \underbrace{\ep^{\f 32} \|\rd \tphi \|_{L^1_t( [0,T_B), L^2(\Sigma_t))} }_{=:II}  + \underbrace{\ep^{\f 32} \sum_{Z_k \in \{E_k, L_k\} }\|\rd Z_k \tphi \|_{L^1_t( [0,T_B), L^2(\Sigma_t))})}_{=:III}.
	\end{split}
	\end{equation}
	The data term can be controlled using \eqref{eq:assumption.rough.energy} and \eqref{eq:assumption.rough.energy.commuted} by $I \ls \ep$. The term $II \ls \ep^{\f 32} \cdot \ep \ls \ep^{\f 52}$ by \eqref{energyglobal}. For the term $III$, we can absorb it by the first term on the left-hand side, after choosing $\ep_0$ smaller if necessary. Putting all these together gives \eqref{energygoodcommutedglobal}.

	\pfstep{Step~3: Proof of \eqref{badlocenergyestimate}} Using Proposition~\ref{prop:commute.with.spatial.direct} with $v= \tphi$, $U_0 = -\infty$ and $U_1 = \infty$, we obtain
		\begin{equation}\label{eq:badlocenergyestimate.spatial.only}
		\begin{split}
		\sup_{t\in [0, T_B)} \|\rd \rd_x \tphi \|_{L^2(\Sigma_t)} + \sup_{u_{k'} \in \RR} \sum_{Z_{k'} \in \{L_{k'},\,E_{k'}\}} \|Z_{k'} \rd_x \tphi \|_{L^2(C^{k'}_{u_{k'}}([0,T_B)))} 
		\ls &\: \|\rd \rd_x \tphi  \|_{L^2(\Sigma_0)} \ls \ep \de^{-\f 12},	\end{split}
		\end{equation}
		where we used \eqref{tphiH2bootstrap} in the final inequality. In particular, this controls every term in \eqref{badlocenergyestimate}, with the only exception  being \color{black} the term $\underset{t\in [0, T_B)} {\sup}\|\rd^2_{tt} \tphi \|_{L^2(\Sigma_t)}$.
		
		In order to bound $\underset{t\in [0, T_B)} {\sup} \|\rd^2_{tt} \tphi \|_{L^2(\Sigma_t)}$, we write $\rd^2_{tt} = \rd_t (\beta^i \partial_i + N \cdot L_k+N \cdot X_k)$. Then, using the bounds for the metric in Proposition \ref{prop:main.metric.est}, together with Proposition \ref{prop:main.frame.est}, we have
		$$\sup_{t\in [0, T_B)} \|\rd^2_{tt} \tphi \|_{L^2(\Sigma_t)} \ls \|\rd  \tphi \|_{L^2(\Sigma_t)} + \|\rd \rd_x \tphi \|_{L^2(\Sigma_t)} + \sum_{Z_k \in \{E_k, L_k\}} \|\rd Z_k \tphi \|_{L^2(\Sigma_t)} \ls \ep \de^{-\f 12},$$
		where at the end we used \eqref{energyglobal}, \eqref{energygoodcommutedglobal} and \eqref{eq:badlocenergyestimate.spatial.only}. Putting everything together gives \eqref{badlocenergyestimate}. \qedhere

\end{proof}

\section{Energy estimates for $\tphi$ up to two derivatives II: the improved estimates} \label{exterior}

In this section, we derive improved estimates for the first and second derivatives of $\tphi$. We will obtain two improvements:
	\begin{itemize}
		\item In the (slightly enlarged) singular region $S^k_{2\de}$, $\rd \tphi$ and $\rd Z_k \tphi$ satisfy \emph{smallness} (in terms of $\de$) bounds in energy. (See \eqref{eq:intro.small.energy.1} and \eqref{eq:intro.small.energy.2} in the introduction, and Proposition~\ref{prop:local.small.energy} below.)
		\item Away from the singular region, i.e.~in $\Sigma_t \setminus S^k_{2\de}$, the $L^2$ norm of $\rd^2\tphi$ is bounded independently of $\de^{-1}$, in contrast to the global bound in \eqref{badlocenergyestimate}. (See \eqref{eq:intro.after.energy} in the introduction, and Proposition~\ref{prop:phiext} below.)
	\end{itemize}
	
	These two improved bounds are highly related: indeed, in order to obtain the latter estimate, we use the former estimate together with a slice-picking argument.
	
	\medskip
	
	We begin with the localized estimate \emph{restricted to the initial data}.

\begin{proposition}\label{prop:improved.data.in.S}
	The following estimates hold on the initial hypersurface $\Sigma_0$:
	\begin{align} 
	\label{locenergyestimate.data}
	\|\partial \tphi\|_{L^2(\Sigma_0 \cap S^k_{2\de})} \lesssim \epsilon \cdot \sdelta, \\
	\label{ELlocenergyestimate.data}
	\sum_{Z_k \in \{L_k,\,E_k\}}  \| \partial Z_k \tphi \|_{L^2(\Sigma_0 \cap S^k_{2\de})}  \lesssim \epsilon \cdot \sdelta.
	\end{align}
\end{proposition}
\begin{proof}
	Recall that on the initial hypersurface $\Sigma_0$, $(u_k,\th_k)$ are affine functions of $(x^1,x^2)$; see \eqref{eikonalinit} and \eqref{thetainit}. Therefore, in all the following estimates, we can easily bound $|\rd_{u_k}f| \ls |\rd_x f|$, as well as pass between $L^2_{x^1,x^2}(\Sigma_0)$ and $L^2_{u_k,\th_k}(\Sigma_0)$.

	Given any $f:\Sigma_0 \to \RR$ such that $\mathrm{supp}(f) \subseteq B(0,R) \cap \{u_k \geq -\de\}$, the fundamental theorem of calculus, the Minkowski inequality and the Cauchy--Schwarz inequality imply that for every $u_k \geq -\de$,
	\begin{equation}\label{eq:very.easy.FTC}
	\| f \|_{L^2_{\th_k}(\Sigma_0 \cap C^k_{u_k})} \ls  \int_{-\de}^{u_k} \| \rd_{u_k} f \|_{L^2_{\th_k}(\Sigma_0 \cap C^k_{\tilde{u}_k})} \,d\tilde{u}_k \ls |u_k +\de|^{\f 12} \|\rd_x f\|_{L^2(\Sigma_0)}.
	\end{equation}
	We now apply \eqref{eq:very.easy.FTC} to $f = \rd\tphi$ and $f = \rd Z_k \tphi$ (for $Z_k \in \{ E_k , L_k \}$). First, by \eqref{eq:delta.waves.1}, we have $$\|\rd_x \rd \tphi \|_{L^2(\Sigma_0)} + \|\rd_x Z_k \rd \tphi \|_{L^2(\Sigma_0)} \ls \ep \cdot \de^{-\f 1 2}.$$ Hence, using \eqref{eq:very.easy.FTC}, we have
	\begin{equation}
	\sup_{u_k \in [-2\de, 2\de]} (\| \rd \tphi \|_{L^2_{\th_k}(\Sigma_0 \cap C^k_{u_k})} + \sum_{Z_k \in \{L_k,\,E_k\}} \| \partial Z_k \tphi \|_{L^2_\th(\Sigma_0 \cap C^k_{u_k})}) \ls \ep.
	\end{equation}
	Finally, H\"older's inequality implies that 
	\begin{equation}\label{eq:data:triv}
	\begin{split}
	&\: \| \rd \tphi \|_{L^2(\Sigma_0 \cap S^k_{2\de})} + \sum_{Z_k \in \{L_k,\,E_k\}}  \| \partial Z_k \tphi \|_{L^2(\Sigma_0 \cap S^k_{2\de})} \\
	\ls &\: \de^{\f 12} \sup_{u_k \in [-2\de, 2\de]} (\| \rd \tphi \|_{L^2_{\th_k}(\Sigma_0 \cap C^k_{u_k})} + \sum_{Z_k \in \{L_k,\,E_k\}} \| \partial Z_k \tphi \|_{L^2_{\th_k}(\Sigma_0 \cap C^k_{u_k})}) \ls \ep\cdot \de^{\f 12}.
	\end{split}
	\end{equation}

\end{proof}

It is now straightforward to use the energy estimates in Propositions~\ref{prop:EE}, \ref{prop:commute.with.E.direct} and \ref{prop:commute.with.L.direct} to propagate the initial data bounds \eqref{locenergyestimate.data} and \eqref{ELlocenergyestimate.data} to all future times. This gives our first improved energy estimate.
\begin{proposition}\label{prop:local.small.energy}
	\begin{align} 
	\label{locenergyestimate}
	\sup_{0 \leq t < T_B}\|\partial \tphi\|_{L^2(\Sigma_t \cap S^k_{2\de})} \lesssim \epsilon \cdot \sdelta, \\
	\label{ELlocenergyestimate}
	\sup_{0 \leq t < T_B} \sum_{Z_k \in \{L_k,\,E_k\}}  \| \partial Z_k \tphi \|_{L^2(\Sigma_t \cap S^k_{2\de})}  \lesssim \epsilon \cdot \sdelta.
	\end{align}
\end{proposition}
\begin{proof}
 Applying Proposition \ref{prop:EE} with $v = \tphi$, $f = 0$, $U_0=-2\delta$, $U_1=2\delta$, and bounding the initial data terms by Proposition~\ref{prop:improved.data.in.S} and Lemma~\ref{lem:support}, we obtain \eqref{locenergyestimate}. 
	
	Next, we apply Propositions~\ref{prop:commute.with.E.direct} and \ref{prop:commute.with.L.direct} with $v = \tphi$, $f=0$, $U_0=-\infty$, $U_1=2\delta$. (Note that even though we apply the propositions with $U_0=-\infty$, since $\tphi$ is supported only in $\{u_k \geq -\de\}$ (by Lemma~\ref{lem:support}), we indeed obtain an estimate which is integrated over $\Sigma_t \cap S^k_{2\de}$.) We thus obtain
	\begin{equation}\label{eq:proving.improved.commuted}
	\begin{split}
	&\: \sup_{0 \leq t < T_B} \sum_{Z_k \in \{L_k, E_k\}} \|\rd Z_k v\|_{ L^2(\Sigma_t \cap S^k_{2\de})}  \\
	\ls &\:  \sum_{Z_k \in \{L_k, E_k\}} \|\rd Z_k v \|_{L^2(\Sigma_0 \cap S^k_{2\de})} \\
	 &\: + \ep^{\f 32} (\|\rd v\|_{L^1_t( [0,T_B), L^2(\Sigma_t \cap \mathcal D))}   + \sum_{Z_k \in \{E_k, L_k\} }\|\rd Z_k v\|_{L^1_t( [0,T_B), L^2(\Sigma_t \cap S^k_{2\de}))}). 	
	\end{split}
	\end{equation}
	The first term in \eqref{eq:proving.improved.commuted} is bounded $\ls \ep \de^{\f 12}$ by Proposition~\ref{prop:improved.data.in.S}. The second term is $\ls \ep^{\f 52} \de^{\f 12}$ by the estimate \eqref{locenergyestimate} that we just proved. Finally, the last term obeys $$\ep^{\f 32} \sum_{Z_k \in \{E_k, L_k\} }\|\rd Z_k v\|_{L^1_t( [0,T_B), L^2(\Sigma_t \cap S^k_{2\de}))} \ls \ep^{\f 32} \sup_{0 \leq t < T_B} \sum_{Z_k \in \{L_k, E_k\}} \|\rd Z_k v\|_{ L^2(\Sigma_t \cap S^k_{2\de})}.$$ This can thus be absorbed by the left-hand side. This concludes the proof of \eqref{ELlocenergyestimate}. \qedhere

\end{proof}

We now turn to the second improved energy estimate, which is an improved estimate after the singular zone.
\begin{prop}\label{prop:phiext}
	The following away-from-the-singular-zone estimate holds:
	\begin{equation} \label{phiextestimate}
	\sup_{t\in [0, T_B)}	\| \partial^2  \tphi \|_{ L^2(\Sigma_t \cap C^k_{ \geq  \delta})} + \sup_{u_{k} \in [\delta,+\infty)} \sum_{Z_{k'} \in \{L_{k'},\,E_{k'}\}} \|Z_{k'} \partial_x \tphi\|_{L^2(C^{k'}_{u_{k'}} ([0,T_B)) \setminus S^k_\de)}   \lesssim \epsilon.
	\end{equation}
	
	%	In particular, we also have 	
	%	\begin{equation} \label{tphiH2}
	%\sup_{t\in [0, T_B)}	\| \partial^2  \tphi \|_{ L^2(\Sigma_t )}  \lesssim \epsilon \cdot \delta^{-\frac{1}{2}}.
	%	\end{equation}
	
\end{prop}

\begin{proof}
	\pfstep{Step~1: Finding a good slice} We square \eqref{ELlocenergyestimate} and we integrate on $[0,T_B)$ to obtain on $\mathcal{D}:=\{ 0 \leq t < T_B, 0 \leq u_k\leq \delta \}$: $$  \sum_{Z_k \in \{L_k,\,E_k\}}  \int_{\mathcal{D}}| \partial Z_k \tphi|^2 dx^1 dx^2 dt \lesssim \epsilon^2 \cdot  \delta.$$
	Controlling the commutator $[\rd, Z_k]\tphi$ using Proposition~\ref{prop:main.frame.est} and \eqref{locenergyestimate}, we obtain
	$$  \sum_{Z_k \in \{L_k,\,E_k\}}  \int_{\mathcal{D}}| Z_k \partial  \tphi|^2 dx^1 dx^2 dt \lesssim \epsilon^2 \cdot  \delta.$$
	
	Since the volume measures $dx^1\,dx^2\,dt$ and $du_k\, d\th_k\,dt_k$ are comparable (by \eqref{jacobian}), it follows that 
	$$\sum_{Z_k \in \{L_k,\,E_k\}}   \int_0^\de \int_0^{T_B} \int_{-\infty}^\infty |Z_k \partial \tphi|^2(u_k,\th_k,t_k) \, d\th_k\, dt_k\, du_k \lesssim \epsilon^2 \cdot  \delta. $$
	By the mean value theorem, there exists $u_k^* \in [0,\de]$ such that
	\begin{equation}\label{eq:ext.new.data}
	\sum_{Z_k \in \{L_k,\,E_k\}}  \|Z_k \partial \tphi \|_{L^2(C^k_{u_k^*}( [0,T_B)))}^2 \ls  \sum_{Z_k \in \{L_k,\,E_k\}}   \int_0^{T_B} \int_{-\infty}^\infty |Z_k \partial \tphi|^2(u_k^*,\th_k,t_k) \, d\th_k\, dt_k \lesssim \epsilon^2.
	\end{equation}
	
	Notice that for this special value $u_k^*$, the estimate \eqref{eq:ext.new.data} is better than the bound provided by Proposition~\ref{prop:easy.energy}, which would have $\ep^2\de^{-1}$ on the right-hand side instead of $\ep^2$.

	%	We start integrate \eqref{ELbadunlocenergyestimate} on $\{-\delta \leq u_k \leq 0\}$ for  $t\in [0,1]$ and, for all $Y_k \in \{ E_k, L_k, \barL\}$, we obtain $$   \int_{ -2\delta \leq  u_k  \leq -\delta} \int_{C^k_{u_k}\cap [0,t]} |Y_k L_k \tphi|^2+ |Y_k E_k \tphi|^2 \leq \int_{\mathcal{D}(u_k=-2\delta,t=1)} |Y_k L_k \tphi|^2+ |Y_k E_k \tphi|^2  \lesssim \epsilon^2 \cdot \delta.$$
	
	%Then, we can apply the mean-value Theorem in $u_k$ to the left-hand-side and take a square-root: thus, there exists $v_k \in [0, \delta]$ such that $$    \sum_{Z_k \in \{L_k,\,E_k\}}  \|Z_k \partial \tphi \|_{L^2(C^k_{v_k} \cap [0,T_B)}  \lesssim \epsilon .$$
	
	\pfstep{Step~2: Applying  an \color{black} energy estimate in the regular region} The key point now is that we can apply \magenta{an} energy estimate again, but only in the region where $u_k \geq u_k*$. The initial data for this new problem has two parts: the energy on the hypersurface $C^k_{u_k^*}$ is good (i.e.~$\de$-independent) thanks to \eqref{eq:ext.new.data}, while the energy on the restriction of the initial hypersurface $\Sigma_0 \cap \{u_k \geq u_*\}$ is good by assumption on the data since $u_k^* \geq 0$.
	
	More precisely, we apply Proposition~\ref{prop:commute.with.spatial.direct} with $v = \tphi$, $f = 0$, $U_0=u_k^*$ and $U_1=+\infty$. Note in particular that $\mathcal D$ corresponds to $C^k_{\geq u_k^*}$.
	\begin{equation}\label{phiextestimate.almost}
	\begin{split}
	&\: \sup_{t\in [0, T_B)} \|\rd \rd_x \tphi \|_{L^2(\Sigma_t\cap  C^k_{\geq u_k^*} )}+ \sup_{u_{k'} \in \RR} \sum_{Z_{k'} \in \{L_{k'},\,E_{k'}\}} \|Z_{k} \rd_x \tphi\|_{L^2(C^{k'}_{u_{k'}}([0,T_B))\cap C^k_{\geq u_k^*})} \\
	\ls &\: \|\rd^2 \tphi \|_{L^2(\Sigma_0\cap  C^k_{\geq u_k^*} )} +   \sum_{Z_k \in \{L_k,\,E_k\}}  \|Z_k \partial \tphi \|_{L^2(C^k_{u_k^*}( [0,T_B)))} \ls \ep,
	\end{split}
	\end{equation} 
	where in the last inequality we used \eqref{eq:delta.waves.2} and \eqref{eq:ext.new.data}.

	Notice that \eqref{phiextestimate.almost} bounds every term in \eqref{phiextestimate} except for $\|\rd^2_{tt} \tphi \|_{L^2(\Sigma_t\cap  C^k_{\geq u_k^*} )}$. In order to bound this term, we write $\rd^2_{tt} = \rd_t (\beta^i \partial_i + N \cdot L_k+N \cdot X_k)$ and use the estimates \eqref{eq:g.main}, \eqref{energyglobal}, \eqref{energygoodcommutedglobal} together with the bound \eqref{phiextestimate.almost} that we just established. \qedhere
	
\end{proof}

\section{Energy estimates for the third derivatives} \label{highest} 

 In this section, we prove energy estimates for the third derivatives of $\tphi$ and $\rphi$.  
 
 There are two different estimates that we prove\blue{. The first type are estimates that concern} $\tphi$\blue{. These are third derivative estimates} where among the three derivatives \blue{on $\tphi$}, there is at least one good derivative $L_k$ or $E_k$; see Proposition~\ref{prop:highest.everything} for a precise statement. As we discussed in Section~\ref{sec:intro.higher.regularity}, these derivatives will be proven using specially chosen commutators $E_k \rd_q$ and $L_k L_k$. It will be shown that the estimates for $\|\rd E_k \rd_q \tphi\|_{L^2(\Sigma_t)}$ and $\| \rd L_k^2 \tphi \|_{L^2(\Sigma_t)}$ will indeed be sufficient to deduce the remaining desired bounds for third derivatives for $\tphi$. This will occupy Sections~\ref{sec:commutation.of.three}--\ref{sec:three.everything}. (Notice that this type of  anisotropic  third derivative estimates can also be derived for $\rphi$, but it is unnecessary and will not be derived. The fact that this is unnecessary is because $\rphi \in H^{2+s'}$ uniformly in $\de$; see Section~\ref{sec:rphi}.)

\blue{The second type of estimates we derive in this section concerns third derivatives} for $\phi$, \blue{where none of the derivatives are required} to be good. \blue{This includes bounding both $\tphi$ and $\rphi$.} \blue{These estimates will be proven in Section~\ref{sec:general.three}; s}ee Proposition~\ref{prop:three.derivatives}. These estimates are easier to obtain because we allow the bound to be very large in terms of $\de^{-1}$.

\subsection{Commutations of the three derivatives}\label{sec:commutation.of.three}

We first show that it suffices to control specific combination of order of commutators, namely that we only have to bound $\|\rd E_k \rd_x \tphi\|_{L^2(\Sigma_t)}$ and $\| \rd L_k^2 \tphi\|_{L^2(\Sigma_t)}$; see Corollary~\ref{cor:change.order}. This is particularly important because $E_k \rd_q$ and $L_k^2$ have better properties when commuted with $\Box_g$, thus allowing us to obtain the desired estimate. 

We first prove the following commutation estimate.
\begin{lemma}\label{lem:move.vector.field.around}
	Let $\sigma \in S_3$ be a permutation, and let $Y^{(1)}$, $Y^{(2)}$ and $Y^{(3)}$ be three (possibly non-distinct) vector fields from the set $\{L_k, E_k, X_k, \n, \rd_1, \rd_2\}$. Then
	$$\| Y^{(1)} Y^{(2)} Y^{(3)} \tphi \|_{L^2(\Sigma_t)} \ls \ep\cdot \de^{-\f 12} + \| Y^{(\sigma(1))} Y^{(\sigma(2))} Y^{(\sigma(3))} \tphi \|_{L^2(\Sigma_t)}.$$
\end{lemma}
\begin{proof}
	Clearly, it suffices to control $\|[Y^{(i)}, Y^{(j)}] Y^{(l)} \tphi \|_{L^2(\Sigma_t)}$ and $\|Y^{(i)}  [Y^{(j)}, Y^{(l)}] \tphi \|_{L^2(\Sigma_t)}$. Observe that since $\n = L_k +  X_k$ (by \eqref{nXEL}), we can assume that $Y^{(1)},\,Y^{(2)},\,Y^{(3)} \in \{L_k, E_k, X_k, \rd_1, \rd_2\}$.
	
	We begin with $\|[Y^{(i)}, Y^{(j)}] Y^{(l)} \tphi \|_{L^2(\Sigma_t)}$.  Using Proposition~\ref{prop:main.frame.est}, we see that $L_k^\mu$, $E_k^i$ and $X_k^i$ obey $C^1$ bounds on $B(0,R)$. Hence, using H\"older's inequality and \eqref{energyglobal}, \eqref{badlocenergyestimate}, we obtain
	\begin{equation}\label{eq:commute.three.1}
	\begin{split}
	&\: \|[Y^{(i)}, Y^{(j)}] Y^{(l)} \tphi \|_{L^2(\Sigma_t)}\\
	\ls &\: (\qquad \smashoperator{\sum_{Y_k \in \{ L_k, E_k, X_k\}}}\quad  \|\rd Y_k^\mu \|_{L^\i(\Sigma_t \cap B(0,R))}) \Big[ \|\rd^2 \tphi \|_{L^2(\Sigma_t)} +(\qquad \smashoperator{\sum_{Y_k \in \{ L_k, E_k, X_k\}}}\quad \|\rd Y_k^\mu \|_{L^\i(\Sigma_t \cap B(0,R))}) \|\rd \tphi\|_{L^2(\Sigma_t)} \Big] \\
	\ls &\: \ep^{\f 54} (\|\rd^2 \tphi \|_{L^2(\Sigma_t)} + \|\rd \tphi\|_{L^2(\Sigma_t)}) \ls \ep^{\f 54} \cdot (\ep \de^{-\f 12}) \ls \ep \de^{-\f 12}.
	\end{split}
	\end{equation}
	
	To bound $\|Y^{(i)}  [Y^{(j)}, Y^{(l)}] \tphi \|_{L^2(\Sigma_t)}$, we first observe that Proposition~\ref{prop:main.frame.est} does \underline{not} give $L^2(\Sigma_t)$ control of all second (spacetime) derivatives of $L_k^\mu$, $E_k^i$ and $X_k^i$ on $B(0,R)$. Nonetheless, the only second derivative that is not controlled is the term $\rd_{tt}^2 L^t_k$.
	
	Next, observe that in the set $\{L_k, E_k, X_k, \rd_1, \rd_2\}$, the only vector field with a $\rd_t$ component in the $\{\rd_t, \rd_1, \rd_2\}$ basis is $L_k$. Since $[L_k, L_k] = 0$, $[Y^{(j)}, Y^{(l)}]$ cannot generate a $\rd_t L^t$ term. As a result, using H\"older's inequality, Proposition~\ref{prop:main.frame.est}, and \eqref{energyglobal}, \eqref{badlocenergyestimate} together with the bootstrap assumption \eqref{BA:Li}, we obtain 
	\begin{equation}\label{eq:commute.three.2}
	\begin{split}
	&\: \|Y^{(i)}  [Y^{(j)}, Y^{(l)}] \tphi \|_{L^2(\Sigma_t)} \\
	\ls &\: (\sum_{Y_k \in \{E_k, X_k\}} \|\rd^2 Y_k^i\|_{L^2(\Sigma_t\cap B(0,R))} + \|\rd \rd_x L_k^\mu \|_{L^2(\Sigma_t\cap B(0,R))}) \|\rd \tphi \|_{L^\i(\Sigma_t)} \\
	&\: + \sum_{Y_k \in \{E_k, X_k, L_k\}} \|\rd Y_k^\mu \|_{L^\i(\Sigma_t \cap B(0,R))} \|\rd^2 \tphi \|_{L^2(\Sigma_t)} \\
	\ls &\: \ep^{\f 54} (\|\rd \tphi\|_{L^\i(\Sigma_t)} + \|\rd^2 \tphi \|_{L^2(\Sigma_t)}) \ls \ep^{\f 54} \cdot (\ep \de^{-\f 12}) \ls \ep \de^{-\f 12}.
	\end{split}
	\end{equation}
	
	Combining \eqref{eq:commute.three.1} and \eqref{eq:commute.three.2} yields the conclusion. \qedhere
\end{proof}

\begin{proposition}\label{prop:change.order}
	The following holds for all $t \in [0,T_B)$:
	$$\sum_{ \substack{ Y_k^{(1)}, Y_k^{(2)}, Y_k^{(3)} \in \{ X_k, E_k, L_k\} \\ \exists i, Y_k^{(i)} \neq X_k} } \| Y_k^{(1)} Y_k^{(2)} Y_k^{(3)} \tphi \|_{L^2(\Sigma_t)} \ls \ep \cdot \de^{-\f 12} + \|\rd E_k \rd_x \tphi\|_{L^2(\Sigma_t)} + \| \rd L_k^2 \tphi\|_{L^2(\Sigma_t)}.$$
\end{proposition}
\begin{proof}
	By Lemma~\ref{lem:move.vector.field.around}, it suffices to control $Y_k^{(1)} Y_k^{(2)} Y_k^{(3)} \tphi$ with any order of $Y_k^{(1)}$,  $Y_k^{(2)}$ and $Y_k^{(3)}$. We consider all possible cases below. (We will silently use that $\mathrm{supp}(\tphi) \subseteq B(0,R)$ so that we do not need to be concerned about the weights at infinity.)
	
	\pfstep{Case~1: At least one of $Y_k^{(i)} = E_k$} \color{black} There are two subcases: 1(a) there is at least one other spatial vector field $E_k$ or $X_k$, and 1(b) $Y_k^{(1)} Y_k^{(2)} Y_k^{(3)}$ is some commutation of $E_k L_k L_k$. In case 1(a), we assume $Y^{(2)}_k = E_k$ and $Y^{(3)}_k \in \{E_k, X_k\}$. Expanding $Y^{(3)}_k$ in terms of $\rd_i$, and using the bounds in Proposition~\ref{prop:main.frame.est}, we can control the term by $\ep \cdot \de^{-\f 12} + \|\rd E_k \rd_x \tphi\|_{L^2(\Sigma_t)}$. In case 1(b), the term is trivially controlled by $\ep \cdot \de^{-\f 12} + \| \rd L_k^2 \tphi\|_{L^2(\Sigma_t)}$.
	
	\pfstep{Case~2: At least one of $Y_k^{(i)} = L_k$, and none of them is $E_k$} \color{black} The three vector fields must therefore be (commutations of) 2(a) $L_k L_k L_k$, 2(b) $X_k L_k L_k$, or 2(c) $X_k L_k X_k$. In cases 2(a) and 2(b), clearly we have (using Proposition~\ref{prop:main.frame.est})
	$$\|L_k^3 \tphi \|_{L^2(\Sigma_t)} + \|X_k L_k^2 \tphi \|_{L^2(\Sigma_t)} \ls \| \rd L_k^2 \tphi\|_{L^2(\Sigma_t)},$$
	which is acceptable. In case 2(c), we use the wave equation $\Box_g\tphi = 0$ and the expression \eqref{waveop}, as well as the bounds in Propositions~\ref{prop:main.metric.est}--\ref{prop:main.Ricci.est} to obtain
	$$\| X_k L_k  X_k \tphi \|_{L^2(\Sigma_t)} \ls \| X_k L_k^2 \tphi \|_{L^2(\Sigma_t)} + \| X_k E_k^2 \tphi \|_{L^2(\Sigma_t)} + \ep^{\f 32} (\|\rd^2 \tphi\|_{L^2(\Sigma_t)} + \|\rd \tphi \|_{L^\infty(\Sigma_t)}).$$
	The first two terms are other combinations of $Y_k^{(1)} Y_k^{(2)} Y_k^{(3)} \tphi$ which we have controlled above, while the last two terms are bounded above by $\ep\cdot \de^{-\f 12}$ using  \eqref{badlocenergyestimate} and the bootstrap assumption \eqref{BA:Li}. \qedhere
\end{proof}

In fact, we can slightly strengthen Proposition~\ref{prop:change.order} to include $\n$ and $\rd_q$ derivatives.
\begin{corollary}\label{cor:change.order}
	The following estimates hold for all $t \in [0,T_B)$:
	\begin{equation}\label{eq:cor.change.order.1}
	\sum_{ \substack{ Y_k^{(1)}, Y_k^{(2)}, Y_k^{(3)} \in \{ X_k, E_k, L_k, \rd_q, \n\} \\ \exists i, Y_k^{(i)} =  E_k \text{ or } Y_k^{(i)} = L_k} } \| Y_k^{(1)} Y_k^{(2)} Y_k^{(3)} \tphi \|_{L^2(\Sigma_t)} \ls \ep \cdot \de^{-\f 12} + \|\rd E_k \rd_x \tphi\|_{L^2(\Sigma_t)} + \| \rd L_k^2 \tphi\|_{L^2(\Sigma_t)},
	\end{equation}
	and
	\begin{equation}\label{eq:cor.change.order.2}
	\sum_{ \substack{ Y_k^{(1)}, Y_k^{(2)} \in \{ X_k, E_k, L_k, \rd_q, \n\} \\ \exists i, Y_k^{(i)} =  E_k \text{ or } Y_k^{(i)} = L_k} } \| \rd Y_k^{(1)} Y_k^{(2)} \tphi \|_{L^2(\Sigma_t)} \ls \ep \cdot \de^{-\f 12} + \|\rd E_k \rd_x \tphi\|_{L^2(\Sigma_t)} + \| \rd L_k^2 \tphi\|_{L^2(\Sigma_t)},
	\end{equation}
\end{corollary}
\begin{proof}
	Clearly it suffices to prove \eqref{eq:cor.change.order.1}, since \eqref{eq:cor.change.order.2} follows from using \eqref{eq:cor.change.order.1} together with \eqref{spatialintermsofEX} \eqref{timeintermsofLEX}.
	
	Comparing to Proposition~\ref{prop:change.order}, the only new vector fields in \eqref{eq:cor.change.order.1} are $\n$ and $\rd_q$. 
	\begin{itemize}
		\item Since $\n = L_k +  X_k$ (by \eqref{nXEL}), if we have the vector field $\n$ (but not $\rd_i$), we can reduce directly to Proposition~\ref{prop:change.order}.
		\item Suppose now among $Y_k^{(1)}, Y_k^{(2)}, Y_k^{(3)}$, there is at least one $\rd_q$ and one $Z_k \in \{L_k, E_k\}$. Then using Lemma~\ref{lem:move.vector.field.around} to commute the vector fields, it suffices to bound
		$$\sum_{Z_k \in \{L_k, E_k\} } \| \rd \rd_x Z_k \tphi \|_{L^2(\Sigma_t)}.$$
		After using \eqref{eq:rdrdx.in.terms.of.geometric}, this can in turn be reduced to a term as in Proposition~\ref{prop:change.order} and plus another term $\sum_{Z_k \in \{L_k, E_k\} } \| \rd_x Z_k \tphi \|_{L^2(\Sigma_t)}$. The latter term can be controlled using Proposition~\ref{prop:easy.energy}. \qedhere
	\end{itemize}
\end{proof}

\begin{rmk}
	Note that despite Corollary~\ref{cor:change.order}, we do not control a term such as $\| \rd_t L_k \rd_t \tphi\|_{L^2(\Sigma_t)}$. This is due to a lack of control of $\rd_t^2 L_k^{\nu\color{black}}$ from Proposition~\ref{prop:main.frame.est}.
\end{rmk}

\subsection{Controlling $\| \partial E_k \partial_x  \tphi \|_{L^2(\Sigma_t)}$} \label{EtphiH2spatialsection}

\begin{proposition} \label{prop:rdErdxphi}
	\begin{equation}
	\label{dEnablatphi} 
	\sup_{t\in [0,T_B)}\| \partial E_k \partial_x  \tphi \|_{L^2(\Sigma_t)} \lesssim \epsilon \cdot \delta^{-\frac{1}{2}} + \ep^{\f 32}\sup_{t\in [0,T_B)} \| \rd L_k^2 \tphi\|_{L^2(\Sigma_t)}.
	\end{equation}
\end{proposition}
\begin{proof}
	We apply energy estimates for $E_k \rd_q \tphi$. First, we write $$ \Box_g( E_k \partial_q \tphi) = [\Box_g, E_k] \partial_q \tphi + E_k ( [\Box_g, \partial_q] \tphi ).$$
	
	\pfstep{Step~1: Controlling $[\Box_g, E_k] \partial_q \tphi$} By Proposition \ref{prop:commute.with.E} and the support properties of $\tphi$, we obtain
	\begin{equation*}
	\begin{split}
	&\: \| [\Box_g, E_k] \partial_q \tphi \|_{L^2(\Sigma_t)} 
	\ls  \ep^{\f 32} \cdot ( \| \partial^2 \tphi \| _{L^2(\Sigma_t) }  + \sum_{Z_k \in \{E_k, L_k\}} \| \partial Z_k \rd_q \tphi \| _{L^2(\Sigma_t)} )  \\ 
	\ls &\: \ep \cdot \de^{-\f 12} + \ep^{\f 32} \|\rd E_k \rd_x \tphi\|_{L^2(\Sigma_t)} + \ep^{\f 32} \| \rd L_k^2 \tphi\|_{L^2(\Sigma_t)}  \ls \ep \cdot \de^{-\f 12} + \ep^{\f 32} \| \rd L_k^2 \tphi\|_{L^2(\Sigma_t)},
	\end{split}
	\end{equation*} 
	where in the last line we additionally used \eqref{badlocenergyestimate}, Corollary~\ref{cor:change.order} and bootstrap assumption \eqref{EtphiH2bootstrap}. 
	
	\pfstep{Step~2: Controlling $\| E_k ( [\Box_g, \partial_q] \tphi )\|_{L^2(\Sigma_t)}$} 	
	By Lemma \ref{waveopcommspatiallemma}, $[\Box_g, \partial_q] \tphi$ can be written as a sum of terms of the schematic form 
	$$\Omg(\mfg)\cdot \partial \partial_x g \cdot \partial \tphi, \quad \Omg(\mfg)\cdot\partial_x \mfg \cdot \rd_x \partial_x \tphi, \quad \Omg(\mfg)\cdot\partial_x \mfg \cdot \n \partial_x \tphi, \quad \Omg(\mfg)\cdot (\partial \mfg)^2 \cdot \partial \tphi,$$ 
	with  $\mfg \in \{ N, \beta, \gamma\}$ and $\Omg(\mfg)$ a smooth function of the metric coefficients. (The important feature to notice here\footnote{One may observe that there are also no terms with $\n^2 \tphi$, but this is irrelevant for the argument below.}, other than the number of derivatives, is that there are no terms with $\rd_t^2 \mfg$ or $\rd_t \rd_x \mfg$. It is also useful to note that there are no $\Omg(\mfg)\cdot\partial_t \mfg \cdot \rd_x \partial_x \tphi$ or $\Omg(\mfg)\cdot\partial_t \mfg \cdot \n \partial_x \tphi$ terms.)
	
	Therefore, using 
	\begin{itemize}
		\item that $E_k$ is a spatial derivatives, satisfying \eqref{eq:frame.1}, 
		\item that $|\Omg(\mfg)|\ls 1$, $|\rd \mfg|\ls \ep^{\f 32}$, $|\rd^2_x\mfg| \ls \ep^{\f 32}$ on the support of $\tphi$ (by \eqref{eq:g.main}), and 
		\item that $|\rd\tphi| \ls \ep^{\f 34} \varpi$ by \eqref{BA:Li} and Lemma~\ref{lem:support}, where $\varpi\in C^\i_c$ is a cutoff such that $\varpi \equiv 1$ on $B(0,2R)$ and $\mathrm{supp}(\varpi) \subseteq B(0,3R)$,
	\end{itemize}
	we obtain
	\begin{equation}\label{eq:Ek[Box_g,rd_q]}
	\begin{split}
	&\: | E_k ([\Box_g, \partial_q] \tphi)| \\
	%\lesssim &\:  |\partial \partial_x^2 g|  |\partial \tphi|+ | \partial_x^2 g|   |\partial_x^2 \tphi|+|\partial_x g| | \partial_x \partial g|  |\partial \tphi| + |\partial \partial_x g| |E_k \partial \tphi|+ |\partial_x g| |E_k \partial_x^2 \tphi|+(\partial_x g)^2  |E_k\partial \tphi| \\ 
	\lesssim  &\: \underbrace{\epsilon^{\frac{3}{4}} \varpi  (|\partial \partial_x^2 \mfg| + |\partial \partial_x \mfg|)}_{A} +  \underbrace{\ep^{\f 34}  |\partial^2 \tphi|}_{B}   + \underbrace{|\partial \partial_x \mfg| (|E_k \partial_x \tphi| + |E_k \n \tphi|)}_{D}+  \underbrace{\epsilon^{\frac{3}{2}} ( |E_k \partial_x^2 \tphi| + |E_k \n \partial_x \tphi|)}_{F}.
	\end{split}
	\end{equation}
	
	We now control each term in \eqref{eq:Ek[Box_g,rd_q]}. The term $A$ can be bounded using \eqref{eq:g.main} and \eqref{eq:g.top}:
	$$\| A \|_{L^2(\Sigma_t)} \ls \ep^{\f 34} \cdot \ep^{\f 32} \cdot \de^{-\f 12} \ls \ep^{\frac 94} \de^{-\f 12}.$$
	The term $B$ can be bounded using \eqref{badlocenergyestimate}:
	$$\| B \|_{L^2(\Sigma_t)} \ls \ep^{\f 34} \cdot \ep \cdot \de^{-\f 12} \ls \ep^{\f 74} \de^{-\f 12}.$$
	The term $D$ can be bounded by first using Proposition~\ref{holdertypeprop} and then using \eqref{energygoodcommutedglobal} and Corollary~\ref{cor:change.order}:
	\begin{equation*}
	\begin{split}
	\| D \|_{L^2(\Sigma_t)} \ls &\:  \ep^{\frac{3}{2}} \Big( \| E_k \partial_x \tphi \|_{L^2(\Sigma_t)} + \| E_k \n \tphi \|_{L^2(\Sigma_t)} + \| \rd_x E_k \partial_x \tphi \|_{L^2(\Sigma_t)} + \| \rd_x E_k \n \tphi \|_{L^2(\Sigma_t)} \Big) \\
	\ls &\: \ep^{\f 52} +  \ep^{\frac{5}{2}} \color{black} \cdot \de^{-\f 12} +  \ep^{\frac{3}{2}} \color{black} \|\rd E_k \rd_x \tphi\|_{L^2(\Sigma_t)} +  \ep^{\frac{3}{2}} \color{black} \| \rd L_k^2 \tphi\|_{L^2(\Sigma_t)}.
	\end{split}
	\end{equation*}
	Finally, for the term $F$, we use Corollary~\ref{cor:change.order} to obtain
	$$\| F \|_{L^2(\Sigma_t)} \ls \ep^{\f 52} \cdot \de^{-\f 12} + \ep^{\f 32} \|\rd E_k \rd_x \tphi\|_{L^2(\Sigma_t)} + \ep^{\f 32}\| \rd L_k^2 \tphi\|_{L^2(\Sigma_t)}.$$
	
	Putting all these together, we obtain
	\begin{equation*}
	\begin{split}
	\| E_k ([\Box_g, \partial_q] \tphi) \|_{L^1([0,T_B);L^2(\Sigma_t))} \ls \ep \cdot \de^{-\f 12} + \ep^{\f 32} \|\rd E_k \rd_x \tphi\|_{L^2(\Sigma_t)} + \ep^{\f 32} \| \rd L_k^2 \tphi\|_{L^2(\Sigma_t)}.
	\end{split}
	\end{equation*} 
	
	\pfstep{Step~3: Putting everything together} Combining the estimates in Steps~1 and 2, we obtain 
	$$ \|\Box_g (E_k \partial_q \tphi)\|_{L^2(\Sigma_t)} \lesssim \ep \cdot \de^{-\f 12} + \ep^{\f 32} \|\rd E_k \rd_x \tphi\|_{L^2(\Sigma_t)} + \ep^{\f 32} \| \rd L_k^2 \tphi\|_{L^2(\Sigma_t)}. $$

	Therefore, applying Proposition \ref{prop:EE} with $v = E_k \rd_q \tphi$, $U_0=-\infty$ and $U_1=+\infty$, and bounding the initial data by \eqref{eq:delta.waves.1}, we obtain
	$$\sup_{t\in [0,T_B)}\| \partial E_k \partial_x  \tphi \|_{L^2(\Sigma_t)} \lesssim \epsilon \cdot \delta^{-\frac{1}{2}} + \ep^{\f 32}\sup_{t\in [0,T_B)} (\|\rd E_k \rd_x \tphi\|_{L^2(\Sigma_t)} +  \| \rd L_k^2 \tphi\|_{L^2(\Sigma_t)}).$$
	Choosing $\ep$ sufficiently small, we can absorb $\ep^{\f 32}\sup_{t\in [0,T_B)} \|\rd E_k \rd_x \tphi\|_{L^2(\Sigma_t)}$ by the term on the left-hand side, thus concluding the proof. \qedhere

\end{proof} 

\subsection{Controlling $ \| \partial L_k^2  \tphi \|_{L^2(\Sigma_t)}$} 	
\begin{proposition}\label{prop:rdLL}
	\begin{equation}
	\label{dLLtphi} \| \partial L_k^2  \tphi \|_{L^2(\Sigma_t)} \lesssim \epsilon \cdot \delta^{-\frac{1}{2}}.
	\end{equation}
\end{proposition}

\begin{proof}
	We apply energy estimates to $L_k^2\tphi$. First, we expand
	\begin{equation}\label{eq:Box.L^2.tphi}
	\Box_g( L_k^2 \tphi) = [\Box_g, L_k] L_k \tphi + L_k ( [\Box_g, L_k] \tphi ).
	\end{equation}
	The first term will be controlled in Step~1 and the second term will be controlled in Steps~2--3 below, after which we carry out the energy estimates in Step~4\color{black}.
	
	\pfstep{Step~1: Controlling $[\Box_g, L_k] L_k \tphi$} We use Proposition \ref{prop:commute.with.L} and $\mathrm{supp}(\tphi) \subseteq B(0,R)$ to obtain	\begin{equation}\label{eq:[Box,L]Lphi}
	\begin{split}
	&\:\| [\Box_g, L_k] L_k\tphi \|_{L^2(\Sigma_t)} 
	\ls \ep^{\f 32} \cdot ( \| \partial^2 \tphi \| _{L^2(\Sigma_t) }  + \sum_{Z_k \in \{E_k, L_k\}} \| \partial Z_k L_k \tphi \| _{L^2(\Sigma_t)} )  \\ 
	\ls &\: \ep^{\f 52}  \cdot \de^{-\f 12} + \ep^{\f 32} \|\rd E_k \rd_x \tphi\|_{L^2(\Sigma_t)} + \ep^{\f 32} \| \rd L_k^2 \tphi\|_{L^2(\Sigma_t)} \ls \ep^{\f  9 4}  \cdot \de^{-\f 12} + \ep^{\f 32} \| \rd L_k^2 \tphi\|_{L^2(\Sigma_t)},
	\end{split}
	\end{equation} 
	where in the last line we additionally used \eqref{badlocenergyestimate}, Corollary~\ref{cor:change.order} and the bootstrap assumption \eqref{EtphiH2bootstrap}.

	\pfstep{Step~2: Controlling $L_k[\Box_g, L_k] \tphi$ except for one term} Recalling the notations from Lemma~\ref{waveopcommLlemma}, we need to handle $L_k[\Box_g, L_k] \tphi= L_k(I(L_k)\tphi)+L_k(II(L_k)\tphi)+L_k(III(L_k)\tphi)$. 
	
	Using the $L^\i$ bounds in \eqref{eq:g.main}, \eqref{eq:frame.1},  \eqref{eq:frame.2}, \eqref{eq:Lchi.Leta}, and \eqref{BA:Li}, it is easy to deduce from \eqref{LBox-BoxLmain}, that
	\begin{equation}\label{eq:LIL}
	\begin{split}
	&\:|L_k(I(L_k\tphi) +  L_kX_k \log(N) \cdot L_k \tphi)| \\
	\ls &\: \epsilon^{\f 94} \varpi+\epsilon^{\f 32} \cdot |\partial^2 \tphi| + \ep^{\f 34}\cdot( |\rd^2 E_k^i| + |\rd^2 X_k^i|)\cdot \varpi +(\epsilon^{\f 94} +| \partial L_k \tphi|) \cdot |\partial \partial_x \mfg|\varpi  + \ep^{\f 3 4} |\partial \partial^2_x \mfg| \varpi \\ 
	&\:+  |\partial_x \chi_k| \cdot( \ep^{\f 9 4} \varpi + |L_k E_k  \tphi|)+ \ep^{\f 3 4} \cdot (|L_k \partial_x \chi_k|+|L_k^2 \chi_k|+|L_k^2 \eta_k|) \varpi ,\\
	\end{split}
	\end{equation} 
	where, as in Proposition~\ref{prop:rdErdxphi}, $\mfg \in \{ \beta^i, N, \gamma\}$ and $\varpi$ is a smooth cutoff with $\mathrm{supp}(\varpi) \subseteq B(0,3R)$. We isolated the term $L_k(   L_kX_k \log(N) \cdot L_k \tphi)$ on the left-hand side of \eqref{eq:LIL}, which will be treated in later steps.
	
	Arguing similarly, but starting with \eqref{LBox-BoxLII}, \eqref{LBox-BoxLIII}, we obtain
	\begin{align}
	|	L_k(II(L_k)\tphi) | \ls &\: \epsilon^{\f 32} \cdot (|\partial L_k^2 \tphi|+ |L_k E_k^2 \tphi|+ |\partial^2 \tphi|)  +(\epsilon^3+ |\partial \partial_x \mfg|+ |\rd K| )\cdot  |\partial L_k \tphi|, \label{eq:LIIL} \\
	|L_k(III(L_k)\tphi)| \ls &\: \epsilon^3 \cdot  |\partial^2 \phi|+  \epsilon^{\f 94} \cdot (|\partial \partial_x \mfg|+ |\rd K|) \cdot \varpi, \label{eq:LIIIL}
	\end{align} 
	where $\mfg$ and $\varpi$ are as in \eqref{eq:LIL}. (Note that in \eqref{eq:LIIL} we have used $L_k X_k L_k \tphi = X^i_k \partial_i L_k^2 \tphi+ (X^{\nu\color{black}} \partial_{\sigma} L^{\sigma}-L^{\nu\color{black}} \partial_{\sigma} X^{\sigma}) \partial_{\nu\color{black}} L_k \tphi$ combined with \eqref{eq:frame.1}, \eqref{eq:frame.2}; similarly for $L_k E_k L_k \tphi$.)

	We now control the terms \eqref{eq:LIL}, \eqref{eq:LIIL} and \eqref{eq:LIIIL}. We begin with the right-hand side of \eqref{eq:LIL}. First, using  the bootstrap assumption \eqref{tphiH2bootstrap} and \eqref{eq:g.main}, \eqref{eq:g.top}, \eqref{eq:frame.3}, \eqref{eq:Lchi.Leta}, \eqref{eq:dxchi}, \eqref{eq:Lrdchi}, we handle all the linear terms to obtain
	\begin{equation}\label{eq:LIL.basic}
	\begin{split}
	&\: \|L_k(I(L_k\tphi+ L_kX_k \log(N)) \cdot L_k \tphi)\|_{L^2(\Sigma_t)} \\
	\ls &\: \epsilon^{\f 94}  \cdot \delta^{-\frac{1}{2}} + \|\partial L_k   \tphi\cdot \partial \partial_x \mfg\|_{L^2(\Sigma_t)}+  \| \partial_x \chi_k \cdot L_k E_k  \tphi\|_{L^2(\Sigma_t)}.
	\end{split}
	\end{equation} 
	 For the $\partial L_k \tphi \cdot \partial \partial_x \mfg$ term in \eqref{eq:LIL.basic}, we first use Lemma~\ref{lem:rd.in.terms.of.XEL} (for $|\rd L_k \tphi|$) and then apply \eqref{holderLpeq}, Corollary~\ref{cor:change.order}, Proposition~\ref{prop:main.frame.est}, \eqref{energyglobal}, \eqref{badlocenergyestimate} and the bootstrap assumption \eqref{EtphiH2bootstrap} to obtain
	\begin{equation}\label{eq:LIL.1}
	\begin{split}
	&\: \|\partial L_k   \tphi\cdot \partial \partial_x \mfg \|_{L^2(\Sigma_t)} \ls  \sum_{ Y_k \in \{ X_k, E_k, L_k \}  } \| Y_k L_k \tphi \cdot \partial \partial_x g \|_{L^2(\Sigma_t)} 
	\\
	\ls &\: \ep^{\f 32} \sum_{ Y_k \in \{ X_k, E_k, L_k \}  } (\| Y_k L_k \tphi \|_{L^2(\Sigma_t)} + \| \rd Y_k L_k \tphi \|_{L^2(\Sigma_t)}) \\
	\ls &\: \ep^{\f 52} \cdot \de^{-\f 12} + \ep^{\f 32}(\| \rd^2 \tphi \|_{L^2(\Sigma_t)} + \|\rd E_k \rd_x \tphi\|_{L^2(\Sigma_t)} + \| \rd L_k^2 \tphi\|_{L^2(\Sigma_t)})  \\
	\ls &\: \ep^{\f 94} \cdot \de^{-\f 12} + \ep^{\f 32} \| \rd L_k^2 \tphi\|_{L^2(\Sigma_t)}.
	\end{split}
	\end{equation}
	For the $\partial_x \chi_k \cdot L_k E_k  \tphi $ term in \eqref{eq:LIL.basic}, we use \eqref{L2Linftyestimate} and then Proposition~\ref{prop:main.frame.est}, Corollary~\ref{cor:change.order} together with \eqref{energyglobal}, \eqref{badlocenergyestimate} and the bootstrap assumption \eqref{EtphiH2bootstrap} to obtain
	\begin{equation}\label{eq:LIL.2}
	\begin{split}
	&\: \|\partial_x \chi_k \cdot L_k E_k  \tphi \|_{L^2(\Sigma_t)} \\
	\ls &\: \|L_k E_k \tphi \|_{L^2(\Sigma_t)} + \| E_k L_k E_k  \tphi \|_{L^2(\Sigma_t)} \\
	\ls &\: \ep^{\f 52} \cdot \de^{-\f 12} + \ep^{\f 32}(\| \rd^2 \tphi \|_{L^2(\Sigma_t)} + \|\rd E_k \rd_x \tphi\|_{L^2(\Sigma_t)} + \| \rd L_k^2 \tphi\|_{L^2(\Sigma_t)})  \\
	\ls &\: \ep^{\f 94} \cdot \de^{-\f 12} + \ep^{\f 32} \| \rd L_k^2 \tphi\|_{L^2(\Sigma_t)}.
	\end{split}
	\end{equation}
	Plugging \eqref{eq:LIL.1}--\eqref{eq:LIL.2} into \eqref{eq:LIL.basic}, we obtain 
	\begin{equation}\label{eq:L[Box,L]phi.1}
	\|L_k(I(L_k\tphi+  L_kX_k \log(N)) \cdot L_k \tphi)\|_{L^2(\Sigma_t)} \lesssim  \epsilon^{\f 9 4} \cdot \delta^{-\frac{1}{2}} + \epsilon^{\f 32}\|\partial L_k^2\tphi \|_{L^2(\Sigma_t)}. 
	\end{equation}
	
	For the term in \eqref{eq:LIIL}, we use the estimates \eqref{eq:g.main}, \eqref{eq:K}, \eqref{eq:g.top}, \eqref{energyglobal}, \eqref{badlocenergyestimate} and Corollary~\ref{cor:change.order} together with \eqref{holderLpeq} to get
	\begin{equation}\label{eq:L[Box,L]phi.2}
	\begin{split}
	\|L_k(II(L_k)\tphi)\|_{L^2(\Sigma_t)} \lesssim &\: \epsilon^{\f 9 4} \cdot \delta^{-\frac{1}{2}}+ \|\partial \partial_x \mfg \cdot \partial L_k \tphi \|_{L^2(\Sigma_t)} + \epsilon^{\f 32} \|\partial E_k \rd_x \tphi \|_{L^2(\Sigma_t)} + \epsilon^{\f 32} \|\partial L_k^2\tphi \|_{L^2(\Sigma_t)} \\  \lesssim &\: 	\epsilon^{\f 9 4} \cdot \delta^{-\frac{1}{2}}+ \ep^{\f 3 2}\| \partial L_k \tphi \|_{H^1(\Sigma_t)} + \epsilon^{\f 32} \|\partial E_k \rd_x \tphi \|_{L^2(\Sigma_t)} + \epsilon^{\f 32} \|\partial L_k^2\tphi \|_{L^2(\Sigma_t)} 
	\\  \lesssim &\:  \epsilon^{\f 9 4} \cdot \delta^{-\frac{1}{2}} + \epsilon^{\f 32}\|\partial L_k^2\tphi \|_{L^2(\Sigma_t)}.
	\end{split}
	\end{equation}
	where we have used \eqref{maximality3} to rewrite $\rd K$ in terms of $\rd \rd_x \mfg$, and in the last inequality we have used \eqref{eq:LIL.1} and the bootstrap assumption \eqref{EtphiH2bootstrap}.
	
	Finally, for the term in \eqref{eq:LIIIL}, we simply use \eqref{badlocenergyestimate} and the estimates \eqref{eq:g.main} and \eqref{eq:K} to obtain 
	\begin{equation}\label{eq:L[Box,L]phi.3}
	\|L_k(III(L_k)\tphi)\|_{L^2(\Sigma_t)} \lesssim   \epsilon^4 \cdot \delta^{-\frac{1}{2}} \ls \epsilon \cdot \delta^{-\frac{1}{2}}.
	\end{equation}

	\pfstep{Step~3: Controlling  $L_k X_k \log(N)  \cdot L_k \tphi$ and $\partial_x (L_k X_k \log(N)  \cdot L_k \tphi )$} Combining \eqref{eq:[Box,L]Lphi} in Step~1 and \eqref{eq:L[Box,L]phi.1}--\eqref{eq:L[Box,L]phi.3} in Step~2 (together with Lemma~\ref{waveopcommLlemma}), we have proven 
	\begin{equation} \label{eq:L2.Steps12}
	\|\Box_g(L_k^2 \tphi) +L_k(L_k X_k \log(N)  \cdot L_k \tphi) \|_{L^2(\Sigma_t )} \ls \epsilon^{\f 94}  \cdot \delta^{-\frac{1}{2}}+\epsilon^{\f 32}\|\partial L_k^2\tphi \|_{L^2(\Sigma_t)}.
	\end{equation}  
	
	We will not directly estimate the term $L_k(L_k X_k \log(N)  \cdot L_k \tphi)$. Instead, we rely an integration by parts argument using Corollary \ref{cor:main.weighted.energy}. In preparation of the integration by parts argument, we estimate  $L_k X_k \log(N)  \cdot L_k \tphi$ and $\partial_x (L_k X_k \log(N)  \cdot L_k \tphi)$. 
	
	First, for $L_k X_k \log(N)  \cdot L_k \tphi$, we use the bootstrap assumption \eqref{BA:Li} and the estimates \eqref{eq:g.main}, \eqref{eq:frame.1}, \eqref{eq:frame.2} to obtain 
	\begin{equation}\label{eq:LXlogN.easy}
	\| L_k X_k \log(N)  \cdot L_k \tphi \|_{L^2(\Sigma_t)} \ls \ep^{\f 94}.
	\end{equation}
	
	As for the derivative $\partial_x (L_k X_k \log(N)  \cdot L_k \tphi)$, the Leibniz rule generates two terms: if $\partial_x$ falls on $L_k \tphi$ we get 
	$$  \|L_k X_k \log(N)  \cdot \partial_x L_k \tphi\|_{L^2(\Sigma_t)} \lesssim \epsilon^{\f 94}  \cdot \delta^{-\frac{1}{2}} + \epsilon^{\f 32}\|\partial L_k^2\tphi \|_{L^2(\Sigma_t)},
	$$ 
	where we used \eqref{holderLpeq} and  Corollary~\ref{cor:change.order}  together with \eqref{eq:frame.1}, \eqref{eq:frame.2}, \eqref{eq:g.main} and the bootstrap assumption \eqref{EtphiH2bootstrap}. 
	
	If, instead, $\partial_x$ falls on $L_k X_k \log(N)$, 
	we get $$  \| \partial_x L_k X_k \log(N)  \cdot  L_k \tphi\|_{L^2(\Sigma_t)} \lesssim   \| \partial_x L_k X_k \log(N) \|_{L^2(\Sigma_t \cap B(0,R))}    \| \partial \tphi \|_{L^{\infty}(\Sigma_t )} \lesssim \epsilon^{\f 94}  \cdot \delta^{-\frac{1}{2}},
	$$ 
	where we used the bootstrap assumption \eqref{BA:Li} together with  \eqref{eq:g.main}, \eqref{eq:g.top}, \eqref{eq:frame.1}, \eqref{eq:frame.2}, \eqref{eq:frame.3}. Therefore, using \eqref{eq:frame.1} again, we prove that

	\begin{equation} \label{spatialLXlogN}
	\| X_k [L_k X_k \log(N)  \cdot L_k \tphi]\|_{L^2(\Sigma_t)} \ls \| \partial_x (L_k X_k \log(N)  \cdot L_k \tphi)\|_{L^2(\Sigma_t)}  \ls \epsilon^{\f 94}  \cdot \delta^{-\frac{1}{2}}+  \epsilon^{\f 32}\|\partial L_k^2\tphi \|_{L^2(\Sigma_t)}.
	\end{equation}

	\pfstep{Step~4: An integration by parts argument and putting everything together} Therefore, writing  the decomposition $L_k= -X_k + N^{-1} e_0$ (by \eqref{nXEL}, \eqref{def:e0}) and combining the estimates in \eqref{eq:L2.Steps12} and \eqref{spatialLXlogN}, we get 
	\begin{equation*} 
	\|\Box_g(L_k^2 \tphi) -N^{-1} e_0(L_k X_k \log(N)  \cdot L_k \tphi) \|_{L^2(\Sigma_t)} \ls \epsilon^{\f 94}  \cdot \delta^{-\frac{1}{2}}+\epsilon^{\f 32}\|\partial L_k^2\tphi \|_{L^2(\Sigma_t)}.
	\end{equation*}

	Writing $\Box_g(L^2_k \tphi)= f_1+ N^{-1} e_0 f_2$ with $f_2=  L_kX_k \log(N) \cdot L_k \tphi$, we have therefore proved that 
	\begin{equation} \label{f1estimate}
	\|\Box_g(L_k^2 \tphi) -N^{-1} e_0 f_2 \|_{L^2(\Sigma_t)}=\| f_1 \|_{L^2(\Sigma_t)} \lesssim \epsilon^{\f 32}  \cdot (\delta^{-\frac{1}{2}}+\|\partial L_k^2\tphi\|_{L^2(\Sigma_t)}).
	\end{equation} 
	On the other hand, \eqref{eq:LXlogN.easy} and \eqref{spatialLXlogN}  give
	\begin{equation} \label{f2estimate}
	\| f_2 \|_{L^2(\Sigma_t)}+	\| \partial_x f_2 \|_{L^2(\Sigma_t)} \lesssim \epsilon^{\f 32}  \cdot \delta^{-\frac{1}{2}} + \ep^{\frac{3}{2}}\|\partial L_k^2\tphi\|_{L^2(\Sigma_t)}.
	\end{equation} 
	By \eqref{f1estimate} and \eqref{f2estimate},   applying Corollary \ref{cor:main.weighted.energy} we get  
	\begin{equation*}
	\begin{split}
	&\: \| \partial L_k^2 \tphi\|_{L^2(\Sigma_t)}^2 \\
	\ls &\: \| \partial L_k^2 \tphi\|_{L^2(\Sigma_0)}^2+ \epsilon^{\f 3 2}  \cdot \delta^{-1} + \epsilon^{\f 3 2} \int_0^t  \|\partial L_k^2\tphi\|^2_{L^2(\Sigma_{\tau})}   d\tau +  \sup_{0 \leq \tau \leq t} \| \langle x \rangle^{-r}L_k X_k \log(N) \cdot L_k \tphi\|_{L^2(\Sigma_{\tau})}  \\ 
	\ls   &\: \epsilon^{\f 3 2} \cdot \delta^{-1} + \epsilon^{\f 3 2} \sup_{0\leq \tau \leq t} \|\partial L_k^2\tphi\|^2_{L^2(\Sigma_{\tau})} ,
	\end{split}
	\end{equation*}  where for the last inequality we have used 
	\begin{itemize}
	\item the assumption on the data \eqref{eq:delta.waves.1} (recall indeed that $\tphi' - X_k\tphi=L_k \tphi$ on $\Sigma_0$) together with \eqref{eq:frame.1}, \eqref{eq:frame.2} to control the data term; and
	\item \eqref{eq:g.main} and the bootstrap assumption \eqref{BA:Li} with the H\"older's inequality to bound the last term.
	\end{itemize}
	 Taking the supremum over $t\in [0,T_B)$, and absorbing $\sup_{0\leq \tau < T_B} \|\partial L_k^2\tphi\|^2_{L^2(\Sigma_{\tau})}$ to the left-hand side, we obtain \eqref{dLLtphi}. This concludes the proof of the proposition. \qedhere

\end{proof}

\subsection{Energy estimates for three derivatives of $\tphi$ with at least one good derivative}\label{sec:three.everything}

We finally obtain our main result regarding the energy estimates for three derivatives of $\tphi$, where at least one of the three derivatives is $E_k$ or $L_k$.
\begin{prop}\label{prop:highest.everything}
	The following estimate holds for all $t \in [0,T_B)$:
	\begin{equation} \label{highest.everything}
	\sum_{ \substack{ Y_k^{(1)}, Y_k^{(2)}, Y_k^{(3)} \in \{ X_k, E_k, L_k, \rd_q, \n\} \\ \exists i, Y_k^{(i)} =  E_k \text{ or } Y_k^{(i)} = L_k} } \| Y_k^{(1)} Y_k^{(2)} Y_k^{(3)} \tphi \|_{L^2(\Sigma_t)} \ls \ep \cdot \de^{-\f 12}.
	\end{equation}
	
	In particular, the bootstrap assumption \eqref{EtphiH2bootstrap} holds with $\ep^{\f 34}$ replaced by $C\ep$.
\end{prop}
\begin{proof}
	This follows immediately from the combination of Corollary~\ref{cor:change.order}, Proposition~\ref{prop:rdErdxphi} and Proposition~\ref{prop:rdLL}. \qedhere
\end{proof}

\subsection{General third derivatives of $\phi$}\label{sec:general.three}

We end this section with an estimate for general third derivatives of $\phi$. We will not require any derivative to be good. In fact, the estimate we prove applies to the full $\phi$, and not just $\tphi$. Notice that the proposition only bounds $ \|\rd \phi \|_{H^2(\Sigma_t)}$ in terms of the $H^3$ norm of the initial data. In particular, while the right-hand side is finite for each fixed $\de>0$, it is allowed to blow up very rapidly as $\de \to 0$.

\begin{proposition}\label{prop:three.derivatives}
	The following estimate holds for all $t\in [0,T_B)$:
	\begin{equation}\label{eq:three.derivatives}
	\|\rd \phi \|_{H^2(\Sigma_t)} \ls \ep \cdot \de^{-\f 12} + \|\phi\|_{H^3(\Sigma_0)} + \|\n \phi \|_{H^2(\Sigma_0)}.
	\end{equation}
\end{proposition}
\begin{proof}
	Letting $\Gamma^{\lambda}:=(g^{-1})^{\nu\color{black}\bt} \Gamma_{\nu\color{black}\bt}^{\lambda}$, we write $\Box_g f = (g^{-1})^{\nu\color{black}\bt} \partial^2_{\nu\color{black}\bt} f +  \Gamma^{\lambda} \partial_{\lambda}f$. Hence, using the support properties of $\phi$, the estimates for the derivatives for $g^{-1}$ and $\Gamma$ in Proposition~\ref{prop:main.metric.est}, and the bootstrap assumption \eqref{BA:Li}, we can bound the commutator $[\Box_g, \rd^2_{ij}]$ as follows:
	\begin{equation}\label{eq:three.derivatives.1}
	\begin{split}
	|[\Box_g, \rd^2_{ij}] \phi| 
	=&\: |(\rd^2_{ij} \gi^{\nu\color{black}\bt}) \rd^2_{\nu\color{black}\bt} \phi + 2 (\rd_{(i} \gi^{\nu\color{black}\bt}) \rd^3_{j)\nu\color{black}\bt} \phi + (\rd^2_{ij} \Gamma^\lambda) \rd_\lambda \phi + 2 (\rd_{(i} \Gamma^\lambda) \rd^2_{j)\lambda} \phi| \\
	\ls &\: \underbrace{\ep^{\f 32}|\rd^2 \phi| }_{=:I} + \underbrace{\ep^{\f 32} |\rd_x\rd^2 \phi|}_{=:II} + \underbrace{\ep^{\f 34} |\varpi \rd^2_x \Gamma^\lambda|}_{=:III} + \underbrace{ | \rd_x \Gamma^\lambda| |\rd \rd_x \phi|}_{=:IV} ,
	\end{split}
	\end{equation}
	where $\varpi \in C^\infty_c$ is a cutoff such that $\varpi = 1$ on $B(0,R)$ and $\mathrm{supp}(\varpi) \subseteq B(0,3R)$.
	
	The terms $I$ in \eqref{eq:three.derivatives.1} can be estimated directly using \eqref{BA:rphi} and \eqref{badlocenergyestimate} so that
	\begin{equation}\label{eq:three.derivatives.2}
	\|I \|_{L^2(\Sigma_t)} \ls \ep^{\f 32} \| \rd^2 \phi \|_{L^2(\Sigma_t)}  \ls \ep^{\f 94} \cdot \de^{-\f 12}.
	\end{equation}
	
	We turn to term $II$ in \eqref{eq:three.derivatives.1}. This term is bounded by $\ep^{\f 32}|\rd \rd_x^2\phi|$ unless we have two $\rd_t$ derivatives, in which case we need to control $\rd_x \rd_t \rd_t \phi$. Since $\gi^{tt} = -\f 1{N^2}$, we have $\rd_{tt }^2 \phi = N^2(2 \gi^{t j} \rd^2_{tj}\phi + \gi^{ij}\rd^2_{ij}\phi + \Gamma^\lambda \rd_\lambda \phi)$ (using $\Box_g \phi =0$). Therefore, using Proposition~\ref{prop:main.metric.est}, \eqref{badlocenergyestimate} and the bootstrap assumptions \eqref{BA:Li}, \eqref{BA:rphi}, we have
	$$ \|\rd_x\rd^2 \phi \|_{L^2(\Sigma_t)} \ls \|\rd \rd^2_x \phi \|_{L^2(\Sigma_t)} + \ep^{\f 34} \| \rd_x \Gamma^\lambda \|_{L^2(\Sigma_t)} + \ep^{\f 32} \|\rd^2\phi \|_{L^2(\Sigma_t)} \ls  \|\rd \rd^2_x \phi \|_{L^2(\Sigma_t)} + \ep^{\f 94} \de^{-\f 12}.$$
	Putting all these together, we thus obtain
	\begin{equation}\label{eq:three.derivatives.3}
\| II \|_{L^2(\Sigma_t)}	 \ls \ep^{\f {15}4} \cdot \de^{-\f 12} + \ep^{\f 32}\|\rd \rd^2_x \phi \|_{L^2(\Sigma_t)}.
	\end{equation} 
	
	We bound the remaining term in \eqref{eq:three.derivatives.1}. By \eqref{eq:Gamma}, $III$ in \eqref{eq:three.derivatives.1} can be bounded by
	\begin{equation}\label{eq:three.derivatives.4}
	\|III \|_{L^2(\Sigma_t)} \ls \ep^{\f 34} \|\rd^2_x \Gamma^\lambda \|_{L^2(\Sigma_t \cap B(0,3R))} \ls \ep^{\f 94} \cdot \de^{-\f 12}.
	\end{equation}
	To handle term $IV$ in \eqref{eq:three.derivatives.1}, first observe that we can use Proposition~\ref{prop:main.metric.est} to bound $|\rd_x \Gamma^\lambda| \ls \ep^{\f 32} + |\rd_x \rd_t \mfg|$ on $B(0,3R)$, where $\mfg$ is as in Proposition~\ref{holdertypeprop}. Hence, using Proposition~\ref{holdertypeprop} together with \eqref{BA:rphi} and \eqref{badlocenergyestimate}, we obtain
	\begin{equation}\label{eq:three.derivatives.5}
	\begin{split}
	\|IV \|_{L^2(\Sigma_t)} \ls &\: \ep^{\f 32} \| \rd^2 \phi \|_{L^2(\Sigma_t)} +  \| |\rd \rd_x \mfg |\cdot |\rd \rd_x \phi| \|_{L^2(\Sigma_t)} \ls \ep^{\f 94} \de^{-\f 12} +  \ep^{\f 32} \| \rd \rd_x \phi \|_{H^1(\Sigma_t)} \\
	\ls &\: \ep^{\f 94} \de^{-\f 12} + \ep^{\f 32} \|\rd \rd_x^2 \phi\|_{L^2(\Sigma_t)},
	\end{split}
	\end{equation}
	where in the last line we also used that $\mathrm{supp}(\phi)\subseteq B(0,R)$ and applied Poincar\'e's inequality.
	
	Combining \eqref{eq:three.derivatives.1}--\eqref{eq:three.derivatives.5}, we obtain
	\begin{equation}\label{eq:Box.rdij}
	\| [\Box_g, \rd^2_{ij}] \phi\|_{L^2(\Sigma_t)} \ls \ep^{\f 94} \cdot \de^{-\f 12} + \ep^{\f 32} \|\rd \rd_x^2 \phi\|_{L^2(\Sigma_t)}.
	\end{equation}
	
	Since $\Box_g \rd^2_{ij} \phi = [\Box_g, \rd^2_{ij}] \phi$, applying the energy estimates in Proposition~\ref{prop:EE} (for $v = \rd^2_{ij} \phi$, $U_0 = -\infty$, $U_1 = \infty$) with \eqref{eq:Box.rdij}, we have
	$$\sup_{t\in [0,T_B)} \|\rd \rd^2_x \phi \|_{L^2(\Sigma_t)} \ls \|\rd \rd^2_x \phi \|_{L^2(\Sigma_0)} + \ep^{\f 94} \cdot \de^{-\f 12} + \ep^{\f 32} \sup_{t\in [0,T_B)} \|\rd \rd^2_x \phi \|_{L^2(\Sigma_t)}.$$

	To obtain the desired estimate, we absorb the last term to the left-hand side, and bound the data term as follows: using \eqref{defnormal} we get  $$ \|\rd \rd^2_x \phi \|_{L^2(\Sigma_0)} \ls   \|\phi\|_{H^3(\Sigma_0)} + \| N(\n \phi+ \beta^{i}\partial_i \phi) \|_{H^2(\Sigma_0)} \ls\ep \de^{-\f12} + \|\phi\|_{H^3(\Sigma_0)} + \|\n \phi \|_{H^2(\Sigma_0)},$$ 
	where in the last inequality we have used \eqref{energyglobal}, \eqref{badlocenergyestimate}, and the metric estimates of Proposition \ref{prop:main.metric.est}. \qedhere
\end{proof}

\section{Fractional energy estimates for $\tphi$} \label{unlochighestfrac}

In this section, we prove the energy estimates for $\tphi$ that involve fractional derivatives. These include
\begin{itemize}
\item bounds for $\|\rd \Db^{s'} \tphi \|_{L^2(\Sigma_t)}$ to be proven in Section~\ref{tphiH3/2xsection}, and
\item bounds for $\|\rd E_k\Db^{s''} \tphi\|_{L^2(\Sigma_t)}$ and $\|\rd L_k\Db^{s''} \tphi\|_{L^2(\Sigma_t)}$ to be proven respectively in Section~\ref{EtphiH3/2xsection} and Section~\ref{LtphiH3/2xsection}. (See also some auxiliary estimates in Section~\ref{sec:frac.L.E.commuted}.)
\end{itemize}

These estimates are the most technical ones in this paper, as they involve simultaneously the geometric vector fields, the weights at spatial infinity and fractional derivatives. Some of the main preliminaries regarding the fractional derivatives and the weights can be found in Section~\ref{sec:commutator.frac} and Section~\ref{sec:weights.and.cutoffs},\color{black} respectively. We also refer the reader back to Section~\ref{sec:intro.higher.regularity} for some comments on the analysis.

\subsection{Fractional derivative commutator estimates}\label{sec:commutator.frac}

\begin{defn}\label{def:LP.xcoord} 
	
	Define the Fourier transform in the $x=(x^1,x^2)$ coordinates with the following normalization: for all $f \in L^2(\RR^2)$
	$$ \hat{f}(\xi)=(\mathcal F f)(\xi) :=  \iint_{\mathbb R^2} f(x) e^{-2\pi i x\cdot \xi} \, dx^1 \, dx^2,$$
	and denote by $\mathcal F^{-1}$ the corresponding inverse Fourier transform.
	
	Let $\varphi:\mathbb R^2 \to [0,1]$ be radial, smooth such that $\varphi(\xi) = \begin{cases} 1\quad \mbox{for $|\xi|\leq 1$} \\
	0\quad \mbox{for $|\xi|\geq 2$}
	\end{cases}$, where $|\xi| = \sqrt{|\xi_1|^2 + |\xi_{2}|^2}$. 
	
	Define $P_0$ by 
	$$P_0 f := (\mathcal F)^{-1} (\varphi(\xi) \mathcal F f),$$
	and for $q\geq 1$, define $P_q f$ by
	$$P_q f := ( \mathcal F)^{-1} ( \tilde{\varphi}_q(\xi)  \mathcal F f(\xi)),$$ where we introduced  $\tilde{\varphi}_q(\xi):=\varphi(2^{-q} \xi) - \varphi(2^{-q+1} \xi)$.

\end{defn}

\begin{definition} With the notations of Definition \ref{def:LP.xcoord}, we further define
	$$\varphi_{hh}(\sigma,\xi) := \sum_j \sum_{k: |k-j|\leq 1} \widetilde{\varphi}_k(\sigma)\widetilde{\varphi}_j(\xi),$$
	$$\varphi_{hl}(\sigma,\xi) := \sum_j \sum_{k: k > j+1} \widetilde{\varphi}_k(\sigma)\widetilde{\varphi}_j(\xi),$$
	$$\varphi_{l h \color{black}}(\sigma,\xi) := \sum_j \sum_{k: k < j-1} \widetilde{\varphi}_k(\sigma)\widetilde{\varphi}_j(\xi).$$
	Note that since $\sum_{k\geq 0} \widetilde{\varphi}_k = 1$, we also have $\varphi_{hh} + \varphi_{hl} + \varphi_{lh} = 1$.
\end{definition}

\begin{defn} \label{paraproddef}
	For any multiplier $m(\sigma,\xi)$ real-valued function on $\RR^2$, we define the para-product
	$$ T_m(f_1,f_2)(x)= \int_{\xi \in \RR^2, \sigma \in \RR^2}  e^{2\pi i (\xi+\sigma)\cdot x} m(\sigma,\xi) \hat{f_1}(\sigma) \hat{f_2}(\xi) d\sigma d\xi,$$
	
	We also define the high-high $m$-para-product of $f_1$ and $f_2$
	
	$$ \Pi_{h h}(m)(f_1,f_2)(x)=T_{ \varphi_{h h} \cdot m } (f_1,f_2)(x) =\int_{\xi \in \RR^2, \sigma \in \RR^2}  e^{2\pi i (\xi+\sigma)\cdot x}\varphi_{h h}(\sigma,\xi)m(\sigma,\xi) \hat{f_1}(\sigma) \hat{f_2}(\xi)   d\sigma d\xi,$$
	the high-low $m$-para-product of $f_1$ and $f_2$
	
	$$ \Pi_{h l}(m)(f_1,f_2)(x)=T_{ \varphi_{h l} \cdot m } (f_1,f_2)(x) =\int_{\xi \in \RR^2, \sigma \in \RR^2}  e^{2\pi i (\xi+\sigma)\cdot x}\varphi_{h l}(\sigma,\xi)m(\sigma,\xi) \hat{f_1}(\sigma) \hat{f_2}(\xi)   d\sigma d\xi,$$
	and finally the low-high $m$-para-product of $f_1$ and $f_2$
	
	$$ \Pi_{l h}(m)(f_1,f_2)(x)=T_{ \varphi_{l h} \cdot m } (f_1,f_2)(x) =\int_{\xi \in \RR^2, \sigma \in \RR^2}  e^{2\pi i (\xi+\sigma)\cdot x} \varphi_{l h}(\sigma,\xi)m(\sigma,\xi) \hat{f_1}(\sigma) \hat{f_2}(\xi)   d\sigma d\xi.$$
	Since $\varphi_{h l}$, $\varphi_{h h}$, $\varphi_{l h}$ form a partition of unity, note that $$ T_m(f_1,f_2)= \Pi_{h h}(m)(f_1,f_2)+\Pi_{h l}(m)(f_1,f_2)+\Pi_{l h}(m)(f_1,f_2).$$
	
	We also denote $\Pi_{h h}(f_1,f_2) := \Pi_{h h}(1)(f_1,f_2)$, $\Pi_{h l}(f_1,f_2) := \Pi_{h l}(1)(f_1,f_2)$ , $\Pi_{l h}(f_1,f_2) := \Pi_{l h}(1)(f_1,f_2) $.
\end{defn}

Next, we recall the Coifman--Meyer theorem. This can be found for instance in \cite{TaoNotes}:
\begin{thm} [Coifman--Meyer] \label{CMThm} 
	Let $m(\xi,\eta)$ be a smooth function on $\RR^2$ which obeys the following bounds for any multi-indices $\nu\color{black}$, $\bt$: 
	$$ |\partial_{\xi}^{\nu\color{black}} \partial_{\sigma}^{\bt} m(\sigma,\xi)| \lesssim_{\nu\color{black},\bt} (\la \xi \ra + \la \sigma\ra)^{-|\nu\color{black}|-|\bt|}.$$
	We say that $m$ is a Coifman--Meyer multiplier.

	Then for any $1 \leq p, q, r \leq \infty$ such that $\frac{1}{r}=\frac{1}{p}+\frac{1}{q}$ and $(p,q) \neq (\infty,\infty), (\infty,1), (1,\infty)$, we have,
	$$ \| T_m(f_1,f_2) \|_{L^{r}(\RR^2)} \lesssim  \|f_1 \|_{L^{p}(\RR^2)}  \| f_2 \|_{L^{q}(\RR^2)}.$$
	
	Moreover, if $m$ is a high-high Coifman--Meyer multiplier in the sense that $\mathrm{supp}\, m(\sigma, \xi)\subseteq \{(\sigma, \xi)\in \mathbb R^2 \times \mathbb R^2: 10^{-1} |\sigma|\leq |\xi|\leq 10 |\sigma|\}$, then the following end-point estimate holds:
	$$ \| T_m(f_1,f_2) \|_{L^{2}(\RR^2)} \ls \|f_1\|_{BMO(\RR^2)} \|f_2\|_{L^2(\RR^2)}.$$
	
\end{thm}

We will need several different Kato--Ponce type commutator estimates to estimate $\Db^{s} (fh) - f \Db^s h$; see Theorem~\ref{KatoPonce}--Corollary~\ref{cor:commute.3}. The difference among these propositions is essentially the number of derivatives that is put on $f$. 

\begin{thm}[{Li \cite[Theorem~1.1]{dL2019}}] \label{KatoPonce}
	For all $0<\th \leq 1 $ and $1<p<\infty$ and  $1<p_1,p_2 \leq \infty$ with $\frac{1}{p}=\frac{1}{p_1}+\frac{1}{p_2}$, the following holds for all $f,\,h \in \mathcal S(\RR^2)$ with a constant depending only on $p_1$, $p_2$ and $\th$:
	$$ \| \Db^\th (f h)-   f (\Db^\th h) \|_{L^p(\RR^2)} \lesssim  \| \Db^\th f \|_{L^{p_1}(\RR^2)} \|  h \|_{L^{p_2}(\RR^2)}   .$$	\end{thm} 		

An easy consequence of Theorem~\ref{KatoPonce} is the following estimate\footnote{Remark that this estimate can also be derived directly, and is in fact much easier than Theorem~\ref{KatoPonce}.}:
\begin{lem}\label{lem:frac.product}
	For any $0< \th\leq 1$ and $2\leq p_1,\, p_1',\, p_2,\, p_2'\leq +\infty$ such that $\f 1{p_1} + \f 1{p_2} = \f 12 = \f 1{p_1'} + \f 1{p_2'}$,
	$$\|\Db^\th (fh)\|_{L^2(\RR^2)}\ls \|\Db^\th f\|_{L^{p_1}(\Sigma_t)} \|h\|_{L^{p_2}(\RR^2)} + \|f\|_{L^{p_1'}(\RR^2)} \| \Db^\th h\|_{L^{p_2'}(\RR^2)}.$$
	%Moreover, if $h$ is compactly supported in $B(0,R)$, then for any $r\geq 0$, the following stronger estimate holds (with a constant depending on $r$ and $R$):
	%\begin{equation}\label{eq:frac.product.weighted}
	%\|\Db^s (fh)\|_{L^2(\RR^2)}\ls \|\la x\ra^{-r} \Db^s f\|_{L^{p_1}(\Sigma_t)} \|h\|_{L^{p_2}(\RR^2)} + \|f\|_{L^{p_1'}(\RR^2\cap B(0,2R))} \| \Db^s h\|_{L^{p_2'}(\RR^2)}.
	%\end{equation}
\end{lem}

\begin{proposition}\label{prop:commute.2}
	Let $\th\geq 0$ and $p\in [ 2,+\infty]$. Then the following holds for any $f,\,h  \in \mathcal{S}(\RR^2)$ with an implicit constant depending only on $\th$ and $p$: 
	$$\|\Db^{\th} (fh) - f (\Db^{\th} h)\|_{L^2(\RR^2)} \ls \|f\|_{W^{1,p}(\RR^2)} \|\Db^{\th-1} h\|_{L^{\f{2p}{p-2}}(\RR^2)}.$$
	%where in the case $p=+\infty$, we use the convention $\|\Db^{\th-1} h\|_{L^{\f{2p}{p-2}}(\RR^2)} = \|\Db^{\th-1} h\|_{L^{2}(\RR^2)}$.
\end{proposition}
\begin{proof}
 First, by the Plancherel formula we have  $$ \|\Db^{\th} (fh) - f (\Db^{\th} h)\|_{L^2(\RR^2)}= \Big\| \int_{\xi \in \RR^2, \sigma \in \RR^2} e^{2\pi i (\xi+\sigma) \cdot x} \big( \la 2\pi(\xi + \sigma) \ra^{\th} - \la 2\pi \xi \ra^{\th} \big)\hat{f}(\sigma)  \hat{h}(\xi) d\sigma d\xi \Big\|_{L^2(\RR^2_{x})}.  $$ We will now write $1= \varphi_{lh}(\sigma, \xi) +\varphi_{h h }(\sigma, \xi) +\varphi_{h l}(\sigma, \xi) $ and analyze each term separately.
	
	  \color{black}
	\pfstep{Step~1: The low-high term} To handle this term,  we need to exploit the commutator structure to obtain the sought estimate.\color{black}
	\begin{equation}\label{eq:commutator.1.expansion}
	\begin{split}
	\varphi_{lh}(\sigma, \xi) (\la 2\pi(\xi + \sigma) \ra^{\th} - \la 2\pi \xi \ra^{\th}) = (4\pi^2\th) \varphi_{lh}(\sigma,\xi)  \int_{t=0}^1 \la 2\pi (\xi + t\sigma) \ra^{\th-2} (\sigma \cdot (\xi + t\sigma))\,dt.
	\end{split}
	\end{equation}
 Note that $\la \xi\ra^{-\th+1} \varphi_{lh}(\sigma,\xi)  \int_{t=0}^1 \la 2\pi (\xi + t\sigma) \ra^{\th-2} (\xi + t\sigma)_i \,dt$ is a Coifman--Meyer multiplier (for each $i$).  Indeed, on the support of $\varphi_{lh}$, $\la \xi + \blue{t}\sigma \ra$ and $\la\xi \ra$ are comparable when $t\in [0,1]$. Hence it is enough to show that for any $\nu\color{black}$, $\beta$: 
$$ \Big| \partial_{\xi}^{\nu} \partial_{\sigma}^{\beta} \Big(  \la\xi\ra^{-\th+1}\int_{t=0}^1 \la   \xi + t\sigma  \ra^{\th-2} (\xi + t\sigma)_i \,dt \Big) \Big| \lesssim \la \xi\ra^{-|\nu|+|\beta|},$$ which is an elementary computation. \color{black} It follows from Theorem~\ref{CMThm} that 
	$$\| \Db^{\th}(\Pi_{lh} (f,h)) - \Pi_{lh}(f, \Db^{\th} h) \|_{L^2(\RR^2)} \ls  \|f\|_{W^{1,p}(\RR^2)} \|\Db^{\th-1} h\|_{L^{\f{2p}{p-2}}(\RR^2)}.$$
	
	\pfstep{Step~2: The high-high terms} We do not need the commutator structure. In other words, we bound $\Db^{\th}(\Pi_{hh} (f,h))$ and $\Pi_{hh} (f ,\Db^{\th} h)$ separately.
	
	We begin with the term $\Db^{\th}(\Pi_{hh} (f,h))$:
	\begin{equation*}
	\begin{split}
	&\: \| \Db^{\th}(\Pi_{hh} (f,h)) \|_{L^2(\RR^2)}^2 \ls \sum_k \|P_k \Db^{\th}(\Pi_{hh} (f,h)) \|_{L^2(\RR^2)}^2 \\
	\ls &\: \sum_{k} \sum_{k': k'\geq {k-3}} \sum_{k'': |k''-k'|\leq 3} 2^{2\th k}\|P_{k'}f\|_{L^p(\RR^2)}^2 \|P_{k''}h\|_{L^{\f{2p}{p-2}}(\RR^2)}^2 \\
	\ls &\:  \sum_{k'} \sum_{k'': |k''-k'|\leq 3} (\sum_{k: k\leq k'+3} 2^{2\th k}) \|P_{k'}  f\|_{L^p(\RR^2)}^2 \|P_{k''}h\|_{L^{\f{2p}{p-2}}(\RR^2)}^2 \\
	\ls &\:  \sum_{k'} \sum_{k'': |k''-k'|\leq 3} (2^{2k'}\|P_{k'}  f\|_{L^p(\RR^2)}^2) (2^{2(\th-1)k''}\|P_{k''}h\|_{L^{\f{2p}{p-2}}(\RR^2)}^2) \\
	\ls &\:  (\sup_{k'} 2^{2k'}\|P_{k'}  f\|_{L^p(\RR^2)}^2) (\sum_{k''} 2^{2(\th-1)k''}\|P_{k''}h\|_{L^{\f{2p}{p-2}}(\RR^2)}^2)
	\ls \|f\|_{W^{1,p}(\RR^2)}^2 \|\Db^{\th-1} h\|_{L^{\f{2p}{p-2}}(\RR^2)}^2.
	\end{split}
	\end{equation*}
	%Remark that $\sup_{k'} 2^{k'}\|P_{k'}  f\|_{L^p}^2 \ls \|f\|_{W^{1,p}(\RR^2)}^2$ also holds for $p = +\infty$.
 Note that \blue{for any fixed $k$, we sum $(k', k'')$ over} $\{(k',k''),\ k'\geq k-3,\ |k''-k'|\leq 3 \} $ because the support of the $\tilde{\varphi}_i$ \blue{(only)} overlaps with the supports of $\tilde{\varphi}_{i-1}$ and $\tilde{\varphi}_{i+1}$. \color{black}
	
	To handle the term  $\Pi_{hh} (f, \Db^{\th} h)$, we rely on the Coifman--Meyer theorem. First, it is easy to check that $\f{\la \xi \ra}{\la \sigma \ra}\varphi_{hh}(\sigma,\xi)$ is a Coifman--Meyer multiplier. Therefore, for $p\in [2,+\infty)$, the Coifman--Meyer theorem gives
	$$\| \Pi_{hh} (f,\Db^{\th}h) \|_{L^2(\RR^2)} \ls \|\Db f\|_{L^p(\RR^2)} \|\Db^{\th-1} h\|_{L^{\f{2p}{p-2}}(\RR^2)} \ls \|f\|_{W^{1,p}(\RR^2)} \|\Db^{\th-1} h\|_{L^{\f{2p}{p-2}}(\RR^2)}. $$
	When $p = +\infty$, we use moreover that we have a high-high multiplier so that the BMO endpoint holds\footnote{Note that we need to use the BMO estimate here because the Riesz transform $\Db^{-1} \rd_x$ is not bounded on $L^\infty$. Instead, we rely on the fact that $\|\Db f\|_{BMO(\RR^2)} \ls \|\rd_x f\|_{BMO(\RR^2)} \ls \| f \|_{W^{1,\infty(\RR^2)}}$.}. Combining this with the \blue{estimate} $\|\Db f\|_{BMO(\RR^2)}\ls \|f\|_{W^{1,\infty}(\RR^2)}$, we obtain
	$$\| \Pi_{hh} (f,\Db^{\th}h) \|_{L^2(\RR^2)} \ls \|\Db f\|_{BMO(\RR^2)} \|\Db^{\th-1} h\|_{L^{2}(\RR^2)} \ls \|f\|_{W^{1,\infty}(\RR^2)} \|\Db^{\th-1} h\|_{L^{2}(\RR^2)}. $$
	
	\pfstep{Step~3: The high-low terms} In a similar manner as in Step~2, we estimate $ \Db^{\th}( \Pi_{hl} (f,h)) $ and $ \Pi_{hl} (f,\Db^{\th} h) $ separately.
	
	The $ \Db^{\th}( \Pi_{hl} (f,h)) $ term can be estimated as follows:
	\begin{equation}\label{eq:commutator.hl}
	\begin{split}
	&\: \|  \Db^{\th}( \Pi_{hl} (f,h))  \|_{L^2(\RR^2)}^2 \ls \sum_k \|P_k  \Db^{\th}( \Pi_{hl} (f,h)) \|_{L^2(\RR^2)}^2 \\
	\ls &\: \sum_{k} \sum_{k': |k'-k|\leq 3} \sum_{k'': k''\leq k+3} 2^{2\th k} \|P_{k'}f\|_{L^p(\RR^2)}^2 \|P_{k''}h\|_{L^{\f{2p}{p-2}}(\RR^2)}^2\\
	\ls &\:  \sum_{k''} \sum_{k': k'\geq k''-6} 2^{2\th k'} \|P_{k'}  f\|_{L^p(\RR^2)}^2 \|P_{k''}h\|_{L^{\f{2p}{p-2}}(\RR^2)}^2 \\
	\ls &\:  \sum_{k''} (\sum_{k': k'\geq k''-6} 2^{2(\th-1)k'}) (\sup_{\tilde{k}'} 2^{2\tilde{k}'}\|P_{\tilde{k}'}  f\|_{L^p(\RR^2)}^2) \|(P_{k''}h)\|_{L^{\f{2p}{p-2}}(\RR^2)}^2 \\
	\ls &\:  (\sup_{\tilde{k}'} 2^{2\tilde{k}'}\|P_{\tilde{k}'}  f\|_{L^p(\RR^2)}^2) (\sum_{k''} 2^{2(\th-1)k''}\|P_{k''}h\|_{L^{\f{2p}{p-2}}(\RR^2)}^2) \ls \|f\|_{W^{1,p}(\RR^2)}^2 \|\Db^{\th-1} h\|_{L^{\f{2p}{p-2}}(\RR^2)}^2.
	\end{split}
	\end{equation}
	
	For the $\Pi_{hl} (f,\Db^{\th} h)$ term, we begin with the trivial estimate
	\begin{equation*}
	\begin{split}
	&\: \|  \Pi_{hl} (f,\Db^{\th} h)  \|_{L^2(\RR^2)}^2 \ls \sum_k \|P_k \Pi_{hl} (f,\Db^{\th} h) \|_{L^2(\RR^2)}^2 \\
	\ls &\: \sum_{k} \sum_{k': |k'-k|\leq 3} \sum_{k'': k''\leq k+3} 2^{2\th k''} \|P_{k'}f\|_{L^p(\RR^2)}^2 \|P_{k''}h\|_{L^{\f{2p}{p-2}}(\RR^2)}^2 \\
	\ls &\: \sum_{k} \sum_{k': |k'-k|\leq 3} \sum_{k'': k''\leq k+3} 2^{2\th k} \|P_{k'}f\|_{L^p(\RR^2)}^2 \|P_{k''}h\|_{L^{\f{2p}{p-2}}(\RR^2)}^2.
	\end{split}
	\end{equation*}
	This coincides with the second line of \eqref{eq:commutator.hl} and we can argue in exactly the same way. \qedhere
\end{proof}

\begin{cor} \label{cor:commute.2}
	Let $p\in (2,+\infty]$ and $0<\th_2 < \th_1 \leq 1$ such $p \geq \f{2}{\th_1 - \th_2}$. Then the following holds for any $f,\,h \in \mathcal{S}(\RR^2)$ with an implicit constant depending only on $p$, $\th_1$ and $\th_2$:
	$$\|\Db^{\th_2} (fh) - f (\Db^{\th_2} h)\|_{L^2(\RR^2)} \ls \|f\|_{W^{1,p}(\RR^2)} \|\Db^{\th_1 - 1} h\|_{L^2(\RR^2)}.$$
\end{cor}
\begin{proof}
	This follows from Proposition~\ref{prop:commute.2} and the Sobolev embedding $H^{\th_1-\th_2}(\mathbb R^2) \hookrightarrow L^{\f{2p}{p-2}}(\mathbb R^2)$. \qedhere
\end{proof}

Next, we need a more precise commutator estimate which essentially gives the ``main term'' of the commutator up to some residual error satisfying better estimates. To set up the notation, for any $f,h \in \mathcal{S}(\RR^2)$, define
\begin{equation}\label{def:Tres}
T^\th_{\mathrm{res}}(f,h) := \Db^{\th} (fh) - f (\Db^\th h) - \th	 \de^{ij} (\rd_i f)(\rd_j \Db^{\th-2} h),
\end{equation} 
Define also $\Pi_{l h}T^\th_{\mathrm{res}}(f,h)$ by 
$$ \Pi_{l h}T^\th_{\mathrm{res}}(f,h):= \Db^{\th} \Pi_{l h} (f,h) - \Pi_{l h}(f ,\Db^\th h) - \th \de^{ij} \Pi_{l h} (\rd_i f,\rd_j \Db^{\th-2} h)$$ 
and similarly for 
$\Pi_{h h}T^\th_{\mathrm{res}}(f,h)$ and $\Pi_{h l}T^\th_{\mathrm{res}}(f,h)$.
Then the following estimate holds.
\begin{proposition}\label{prop:commute.3}
	Let $\th\geq 0$ and $p\in (2,+\infty]$. Then the following holds for any $f,\,h \in \mathcal{S}(\RR^2)$ with an implicit constant depending only on $\th$ and $p$: 
	\begin{equation}\label{eq:higher.order.commute}
	\|T^\th_{\mathrm{res}}(f,h)\|_{L^2(\RR^2)} \ls \|f\|_{W^{2,p}(\RR^2)} \|\Db^{\th -2} h\|_{L^{\f{2p}{p-2}}(\RR^2)}.
	\end{equation}
\end{proposition}
\begin{proof}
	Note that the only \blue{difficulty} concerns the low-high term. For the high-high and high-low interactions, it is easy to check as in the proof of Proposition~\ref{prop:commute.2} that each of the terms $\| \Db^{\th}\Pi_{hh} (f,h)\|_{L^2(\RR^2)}$, $\|\Pi_{hh}(f, \Db^\th h)\|_{L^2(\RR^2)}$, $\|\Pi_{hh}(\rd_i f,\rd_j \Db^{\th-2} h)\|_{L^2(\RR^2)}$, $\| \Db^{\th}\Pi_{hl} (f,h)\|_{L^2(\RR^2)}$, $\|\Pi_{hl}(f ,\Db^\th h)\|_{L^2(\RR^2)}$ and $\|\Pi_{hl}(\rd_i f,\rd_j \Db^{\th-2} h)\|_{L^2(\RR^2)}$ is bounded by the right-hand side of \eqref{eq:higher.order.commute}; we omit the details.
	
	For the low-high term, we continue the computation of \eqref{eq:commutator.1.expansion}. More precisely, we use
	$$r(1) = r(0) + r'(0) + \int_{t=0}^1 (1-t) r''(t)\, \ud t$$
	with $r(t) = \la 2\pi(\xi + t\sigma) \ra^{\th}$ to obtain
	\begin{equation}\label{eq:commutator.1.further} 
	\begin{split}
	&\:   (\la 2\pi(\xi + \sigma) \ra^{\th} - \la 2\pi \xi \ra^{\th}) - 4\pi^2 \th \la 2\pi \xi \ra^{\th-2} (\xi \cdot \sigma) \\
	= &\: 4 \pi^2 \th  \int_{t=0}^{1} (1 - t) \{ 4\pi^2 (\th -2) \la 2\pi (\xi+t\sigma) \ra^{\th-4}  (\sigma\cdot (\xi+ t\sigma))^2 +  \la 2\pi (\xi+t\sigma) \ra^{\th-2}  |\sigma|^2 \} \, \ud t \\
	=: &\: m(\sigma,\xi).
	\end{split}
	\end{equation} 
	
	Notice that 
	$$m(\sigma,\xi) = |\sigma|^2 m_A(\sigma,\xi) + \sum_{i,j} \sigma_i \sigma_j (m_B)_{ij}(\sigma,\xi),$$
	where $m_A$ and $m_B$ are defined by
	$$m_A(\sigma,\xi) :=  4 \pi^2 \th  \int_{t=0}^{1} (1 - t) \{ 4\pi^2 (\th -2) \la 2\pi (\xi+t\sigma) \ra^{\th-4} (t^2 |\sigma|^2 + 2 t (\sigma\cdot \xi))+   \la 2\pi (\xi+t\sigma) \ra^{\th-2}  \}   \, \ud t,$$
	and
	$$(m_B)_{ij}(\sigma, \xi) := 16 \pi^4 \th (\th-2) \xi_i \xi_j \int_{t=0}^{1} (1 - t) \la 2\pi (\xi+t\sigma) \ra^{\th-4}  \, dt.$$
	It is easy to check that $\la \xi\ra^{-\th+2} \varphi_{lh}(\sigma,\xi) m_A(\sigma,\xi)$ and $\la \xi\ra^{-\th+2} \varphi_{lh}(\sigma,\xi) (m_B)_{ij}(\sigma,\xi)$ are both Coifman--Meyer multipliers.
	
	The computation \eqref{eq:commutator.1.further} implies that 
	\begin{equation*}
	\begin{split}
	\|\Pi_{lh} T^\th_{\mathrm{res}}(f,h)\|_{L^2(\RR^2)} = &\: \|\Db^{\th} \Pi_{l h} (f,h) - \Pi_{l h}(f ,\Db^\th h) - \th \de^{ij} \Pi_{l h} (\rd_i f,\rd_j \Db^{\th-2} h)\|_{L^2(\RR^2)}\\
	\ls &\: \|T_{\varphi_{lh} m_A}(\Delta f, h)\|_{L^2(\RR^2)} + \sum_{i,j}\|T_{\varphi_{lh} (m_B)_{ij}}(\rd^2_{ij} f, h)\|_{L^2(\RR^2)} \\
	\ls &\:  \|f \|_{W^{2,p}(\Sigma_t)} \|\Db^{\th-2} h\|_{L^{\f{2p}{p-2}}(\RR^2)},
	\end{split}
	\end{equation*}
	where in the last line we have used the Coifman--Meyer theorem (Theorem~\ref{CMThm}).
	
	This gives the desired estimates for the low-high interaction. As described in the beginning the high-high and high-low are easier, and we have therefore completed the proof of the proposition. \qedhere
\end{proof}

We record another easy but useful way to estimate the term in Proposition~\ref{prop:commute.3}:
\begin{cor}\label{cor:commute.3}
	Let $T_{\mathrm{res}}^{\th_2}$ be as in \eqref{def:Tres}. Let $p\in (2,+\infty]$ and $0<\th_2 < \th_1$ such $p \geq \f{2}{\th_1 - \th_2}$. Then the following hold for any $f,\,h \in \mathcal{S}(\RR^2)$ with implicit constant depending only on $p$, $\th_1$ and $\th_2$:
	\begin{equation*}
	\|T^{\th_2}_{\mathrm{res}}(f,h)\|_{L^2(\RR^2)} \ls \min\{ \|f\|_{W^{1,p}(\RR^2)} \|\Db^{\th_1 -1} h\|_{L^{2}(\RR^2)}, \,\|f\|_{W^{2,p}(\RR^2)} \|\Db^{\th_1 -2} h\|_{L^{2}(\RR^2)}\}.
	\end{equation*}
\end{cor}
\begin{proof}
	On the one hand, by the triangle inequality, Corollary~\ref{cor:commute.2} and H\"older's inequality,
	$$\|T^{\th_2}_{\mathrm{res}}(f,h)\|_{L^2(\RR^2)} \ls \|f\|_{W^{1,p}(\RR^2)} \|\Db^{\th_1 -1} h\|_{L^{2}(\RR^2)}.$$
	%I added this because I think the reader can mistakenly think the second estimate uses the first one, but maybe it's just me
	On the other hand, by the triangle inequality, Proposition~\ref{prop:commute.3} and the Sobolev embedding $H^{\th_1-\th_2}(\mathbb R^2) \hookrightarrow L^{\f{2p}{p-2}}(\mathbb R^2)$,
	$$\|T^{\th_2}_{\mathrm{res}}(f,h)\|_{L^2(\RR^2)} \ls \|f\|_{W^{2,p}(\RR^2)} \|\Db^{\th_1 -2} h\|_{L^{2}(\RR^2)}.$$
	Combining yields the result. \qedhere
\end{proof}

Finally, we need an auxiliary commutation lemma concerning the commutation of a vector field with the (inhomogeneous) Riesz transform.

\begin{lem}\label{lem:commute.Riesz}
	Let $Y^i\rd_i$ be a vector field on $\RR^2$ such that $ Y^i  \in W^{1,\infty}(\RR^2)$ and $f \in L^2(\RR^2)$. Denoting $R_j = \rd_j \Db^{-1}$, we have
	$$ \left\| \left[ Y, R_j \right] f \right\|_{L^2(\RR^2)} \ls \max_{ i=1,2}\|  Y^i \|_{W^{1,\infty}(\RR^2)} \cdot \|  f \|_{L^{2}(\RR^2)}.$$
\end{lem}
\begin{proof}
	By the Calder\'on commutator estimate (see \cite[Corollary~on~p.309]{Stein.book}), 
	$$\|\rd_i R_j (Y^i f) - Y^i \rd_i R_j f\|_{L^2(\RR^2)} \ls \max_{ i=1,2}\|  Y^i \|_{W^{1,\infty}(\RR^2)} \cdot \|  f \|_{L^{2}(\RR^2)}. $$
	Hence, by the triangle inequality and the $L^2$-boundedness of $R_j$,
	\begin{equation*}
	\begin{split}
	\left\| \left[ Y, R_j \right] f \right\|_{L^2(\RR^2)} 
	\ls &\: \|\rd_i R_j (Y^i f) - Y^i \rd_i R_j f\|_{L^2(\RR^2)} + \| R_j [(\rd_i Y^i) f]\|_{L^2(\RR^2)} \\
	\ls &\:  \max_{ i=1,2}\|  Y^i \|_{W^{1,\infty}(\RR^2)} \cdot \|  f \|_{L^{2}(\RR^2)}.
	\end{split}
	\end{equation*}
\end{proof}

\subsection{Notations for this section}\label{sec:frac.notation}

We now define some notations that will be useful for the remainder of the section.

From now on, fix a cutoff function $\varpi\in C^\i_c$ such that $\varpi \equiv 1$ on $B(0,2R)$ and $\mathrm{supp}(\varpi) \subseteq B(0,3R)$.

We also introduce the following notations for the wave equation.

Let $\tboxtwo$ and $\tboxone$ be operators defined by
\begin{equation} \label{decompositionbox}
\Box_g f = \overbrace{(g^{-1})^{\nu\color{black}\bt} \partial^2_{\nu\color{black}\bt} f}^{:=\tboxtwo(f)}- \overbrace{ \Gamma^{\lambda} \cdot \partial_{\lambda}f} ^{:=\tboxone(f)},
\end{equation}
where $(g^{-1})^{\nu\color{black}\bt}$ are the components of the inverse matrix of $g_{\nu\color{black}\bt}$, as expressed in \eqref{metric2+1} and 
\begin{equation}
\Gamma^{\lambda}:=\gi^{\nu\color{black}\bt} \Gamma_{\nu\color{black}\bt}^{\lambda},\quad \Gamma_{\nu\color{black}\bt}^{\lambda}:= \frac{1}{2}(g^{-1})^{\lambda \sigma} ( \partial_{\nu\color{black}}g_{\sigma \beta}+\partial_{\beta}g_{\sigma \nu\color{black}}-\partial_{\sigma}g_{\nu\color{black}\bt}) 
\end{equation}
(in the coordinate system $(t,x^1,x^2)$ of \eqref{metric2+1}).

Finally, define
\begin{equation}\label{def:bg}
\bg^{ij}:=\f{(g^{-1})^{ij}}{(g^{-1})^{tt}},\quad \bg^{it}:=\f{2(g^{-1})^{it}}{(g^{-1})^{tt}}
\end{equation}
so that
\begin{equation}\label{eq:wave.bg}
\rd^2_{tt} = \f{1}{(g^{-1})^{tt}} \Box_g -\bg^{i\lambda} \rd^2_{i\lambda} + \frac{\Gamma^\lambda}{(g^{-1})^{tt}} \rd_\lambda.
\end{equation}

\subsection{Preliminary estimates}

The following basic estimate will be repeatedly used. (Recall the notation for $\varpi$ defined in the beginning of Section~\ref{sec:frac.notation}.)
\begin{lem}\label{lem:stupid_generic_v}
	Let $v$ be a smooth, compactly supported function on $B(0,R)$ and $f$ be a smooth function. Then	
	$$\|\Db^{s'} (f  v)\|_{L^2(\Sigma_t)}\ls \|\varpi f \|_{L^\i \cap W^{1,2}(\Sigma_t)} \| \Db^{s'} v\|_{L^2(\Sigma_t)}.$$
\end{lem}
\begin{proof}	
	Note that $f v = \varpi f  v$. Hence, by Theorem~\ref{KatoPonce} (with $p=2$, $p_1 = \f 2{s'}$, $p_2 = \f{2}{1-s'}$) and H\"older's inequality,	
	\begin{equation*}
	\begin{split}
	\|\Db^{s'} (f  v)\|_{L^2(\Sigma_t)} \ls &\: \| \varpi f\|_{L^\i (\Sigma_t)} \| v\|_{H^{s'}(\Sigma_t)} + \|\Db^{s'} (\varpi f)\|_{L^{\f{2}{s'}}(\Sigma_t)}\| v\|_{L^{\f 2{1-s'}}(\Sigma_t)} \\
	\ls &\:  \|\varpi f \|_{L^\i \cap W^{1,2}(\Sigma_t)}  \| v\|_{H^{s'}(\Sigma_t)},
	\end{split}
	\end{equation*}
	where in the last inequality we have used Sobolev embeddings $H^{1}(\Sigma_t) \hookrightarrow W^{s',\f 2{s'}}(\Sigma_t)$ and $H^{s'}(\Sigma_t) \hookrightarrow L^{\f 2{1-s'}}(\Sigma_t)$. \qedhere
\end{proof}

We apply Lemma~\ref{lem:stupid_generic_v} in the special case where $v= \rd_\lambda \tphi$.

\begin{lem}\label{lem:stupid}	Let $f$ be a smooth function satisfying	
	$$\| \varpi f\|_{L^\i \cap W^{1,2}(\Sigma_t)}\ls 1.$$	
	Then
	$$\|\Db^{s'} (f \rd_\lambda\tphi)\|_{L^2(\Sigma_t)}\ls \|\rd \Db^{s'} \tphi\|_{L^2(\Sigma_t)}.$$
\end{lem}
\begin{proof} Using the support properties in Lemma~\ref{lem:support}, the result follows from Lemma~\ref{lem:stupid_generic_v} with $v= \rd_\lambda \tphi$.
\end{proof}

We will derive a few of consequences of Lemma~\ref{lem:stupid}. 

\begin{lem}\label{lem:stupid.2}
	Let $f$ be a smooth function satisfying
	$$\| \varpi f\|_{L^\i \cap W^{1,2}(\Sigma_t)}\ls 1.$$
	Then (recall the notation in \eqref{def:bg}):
	\begin{equation}\label{eq:stupid.2.1}
	\|\Db^{s'}(f \bg^{j\lambda} \rd_\lambda \tphi)\|_{L^2(\Sigma_t)} \ls \|\rd \Db^{s'} \tphi\|_{L^2(\Sigma_t)},
	\end{equation}
	\begin{equation}\label{eq:stupid.2.2}
	\|\Db^{s'}[f (\rd_j \bg^{j\lambda}) (\rd_\lambda \tphi)]\|_{L^2(\Sigma_t)} \ls \|\rd \Db^{s'} \tphi\|_{L^2(\Sigma_t)},
	\end{equation}
	\begin{equation}\label{eq:stupid.2.3}
	\|\Db^{s'}[f \f{\Gamma^\lambda \rd_\lambda \tphi}{(g^{-1})^{tt}}]\|_{L^2(\Sigma_t)} \ls \|\rd \Db^{s'} \tphi\|_{L^2(\Sigma_t)}.
	\end{equation}
\end{lem}
\begin{proof}
	The estimates follow immediately from Lemma~\ref{lem:stupid} and the estimates \eqref{eq:g.main}. \qedhere

\end{proof}

Another consequence of Lemma~\ref{lem:stupid} is that we can control negative fractional derivatives of $\rd^2\tphi$, i.e.~terms of the form $ \Db^{s'-1} \rd^2\tphi$\color{black}. We start with a more general lemma, before turning to $\Db^{s'-1} \rd^2 \tphi$ in Lemma~\ref{lem:invert.tt}.

\begin{lem}\label{lem:invert.tt_generic_v}
	Let $f$ be a smooth function satisfying\footnote{It should be noted that the $W^{1,2}$ bound on $\varpi f$ \magenta{in \eqref{eq:invert.tt.f_generic_v}} is extraneous, as it is implied by the other bounds. We state the assumption as in \eqref{eq:invert.tt.f_generic_v} so as to make the application of Lemma~\ref{lem:stupid.2} more obvious.}
	\begin{equation}\label{eq:invert.tt.f_generic_v}
	\|\varpi f\|_{L^\i \cap W^{1,2}(\Sigma_t)} + \|\varpi \rd_i f\|_{L^4(\Sigma_t)} \ls 1,
	\end{equation}
	and $v$ be a smooth, compactly supported function on $B(0,R)$. Then
	\begin{equation}\label{eq:stupid.2.4_generic_v}
	\|\Db^{s'-1} (f \rd^2_{i\lambda} v) \|_{L^2(\Sigma_t)} \ls \|\rd \Db^{s'} v \|_{L^2(\Sigma_t)}
	\end{equation}
	and
	\begin{equation}\label{eq:stupid.2.5_generic_v}
	\|\Db^{s'-1} (f \rd^2_{tt} v)\|_{L^2(\Sigma_t)} \ls \|\rd \Db^{s'} v\|_{L^2(\Sigma_t)} +  \|\Box_g v\|_{L^2(\Sigma_t)}.
	\end{equation}
\end{lem}
\begin{proof}
	For \eqref{eq:stupid.2.4_generic_v}, note that 	
	\begin{equation*}
	\begin{split}
	&\: \|\Db^{s'-1} (f \rd^2_{i\lambda} v) \|_{L^2(\Sigma_t)} \\
	\ls &\: \|\Db^{s'-1} \rd_i (f \rd_{\lambda} v) \|_{L^2(\Sigma_t)} + \|\Db^{s'-1} [(\rd_i f) (\rd_{\lambda} v)] \|_{L^2(\Sigma_t)} \\
	\ls &\: \|\Db^{s'} (f \rd_{\lambda} v) \|_{L^2(\Sigma_t)}  + \| \varpi \rd_i f\|_{L^4(\Sigma_t)} \|\rd_\lambda v\|_{L^2(\Sigma_t)} \ls \|\rd \Db^{s'} v \|_{L^2(\Sigma_t)},
	\end{split}
	\end{equation*}
	where in the penultimate estimate we have used Sobolev embedding $\Db^{s'-1}: L^{\f 43}(\Sigma_t) \to L^2(\Sigma_t)$ (which is true since $s' < \frac{1}{2}$) and H\"older's inequality, and in the last estimate we have used Lemma~\ref{lem:stupid}.
	
	We now prove \eqref{eq:stupid.2.5_generic_v}. We rewrite $\rd^2_{tt} v$ in terms of $\Box_g v$ using \eqref{eq:wave.bg} and apply the triangle inequality to obtain
	\begin{equation*}
	\begin{split}
	&\: \|\Db^{s'-1} (f\rd^2_{tt} v)\|_{L^2(\Sigma_t)}  \\
	\ls &\: \|\Db^{s'-1} [f(\frac{\Box_g v}{(g^{-1})^{t t}} -\bg^{i\lambda}\rd^2_{i\lambda}v + \f{\Gamma^\lambda \rd_\lambda v}{(g^{-1})^{tt}})] \|_{L^2(\Sigma_t)}  \\
	\ls &\: \|R_i \Db^{s'} (f \bg^{i\lambda} \rd_\lambda v)\|_{L^2(\Sigma_t)} + \| \Db^{s'-1} [(\rd_i(f\bg^{i\lambda})) (\rd_\lambda v)]\|_{L^2(\Sigma_t)} \\
	&\: + \|\Db^{s'-1} ( \f{f \Gamma^\lambda \rd_\lambda v}{(g^{-1})^{tt}} )\|_{L^2(\Sigma_t)}+ \| \Db^{s'-1} ( f\frac{\Box_g v}{(g^{-1})^{t t}} ) \|_{L^2(\Sigma_t)}, 
	\end{split}
	\end{equation*}
	where $R_i = \rd_i \Db^{-1}$ as before.
	
	For the first term, we use that $R_i:L^2(\Sigma_t) \to L^2(\Sigma_t)$ is bounded and then use Lemma~\ref{lem:stupid_generic_v}, \eqref{eq:invert.tt.f_generic_v} and \eqref{eq:g.main} to bound it by $\ls \|\rd \Db^{s'} v \|_{L^2(\Sigma_t)}$. For the second and third terms, we use in addition the Sobolev embedding $\Db^{s'-1}: L^{\f 43}(\Sigma_t) \to L^2(\Sigma_t)$. For instance, for the second term we have (using $v = \varpi v$ and H\"older's inequality)
	\begin{equation*}
	\begin{split}
	&\:  \| \Db^{s'-1} [(\rd_i(f\bg^{i\lambda})) (\rd_\lambda v)]\|_{L^2(\Sigma_t)} \ls  \| [(\rd_i(f\bg^{i\lambda})) (\rd_\lambda v)]\|_{L^{\f 43}(\Sigma_t)}\\
	\ls &\: \left( \|\varpi \rd_i  f \|_{L^4(\Sigma_t)} \|\varpi \bg^\lambda\|_{L^\i(\Sigma_t)} + \|\varpi f\|_{L^\i(\Sigma_t)} \|\varpi \rd_i \bg^\lambda \|_{L^4(\Sigma_t) }  \right) \|\rd_\lambda v\|_{L^2(\Sigma_t)} \ls \|\rd v\|_{L^2(\Sigma_t)},
	\end{split}
	\end{equation*}
	where in the last estimate we used \eqref{eq:invert.tt.f_generic_v} together with \eqref{eq:g.main}. The third term is similar and omitted.
	For the fourth term, we use the  Sobolev embedding $\Db^{s'-1}: L^{\f 43}(\Sigma_t) \to L^2(\Sigma_t)$ and the estimates \eqref{eq:invert.tt.f_generic_v}, \eqref{eq:g.main} to get
	$$  \| \Db^{s'-1} ( f(\frac{\Box_g v}{(g^{-1})^{t t}} ) ) \|_{L^2(\Sigma_t)}  \ls  \| \Box_g v \|_{L^2(\Sigma_t)}. $$

	We have thus obtained \eqref{eq:stupid.2.5_generic_v}.
	\qedhere
\end{proof}
\begin{lem}\label{lem:invert.tt}
	Let $f$ be a smooth function satisfying
	\begin{equation}\label{eq:invert.tt.f}
	\|\varpi f\|_{L^\i \cap W^{1,2}(\Sigma_t)} + \|\varpi \rd_i f\|_{L^4(\Sigma_t)} \ls 1.
	\end{equation}
	Then
	\begin{equation}\label{eq:stupid.2.4}
	\blue{ \|\Db^{s'-1} (f \rd^2_{\bt\lambda} \tphi) \|_{L^2(\Sigma_t)} \ls \|\rd \Db^{s'} \tphi \|_{L^2(\Sigma_t)}.}
	\end{equation}
	%and
	%\begin{equation}\label{eq:stupid.2.5}
	%\|\Db^{s'-1} (f \rd^2_{tt} \tphi)\|_{L^2(\Sigma_t)} \ls \|\rd \Db^{s'} %\tphi\|_{L^2(\Sigma_t)}.
	%\end{equation}
\end{lem}
\begin{proof} This follows immediately from an application of Lemma~\ref{lem:invert.tt_generic_v}, with $v=\tphi$ since $\tphi$ is smooth compactly supported in $B(0,R)$ and $\Box_g \tphi=0$. \qedhere
	
\end{proof}

\subsection{Weighted estimates and cutoffs}\label{sec:weights.and.cutoffs}

\subsubsection{Gaining weights in the estimate for $\rd \Db^{s'} \tphi$}

\begin{lem}\label{lem:cutoff.commute}
	Let $\varpi\in C^\i_c$ be a cutoff function such that $\varpi \equiv 1$ on $B(0,2R)$ and $\mathrm{supp}(\varpi) \subset B(0,3R)$; and $\varpi'\in C^\infty_c$ such that $\varpi' \equiv 1$ on $B(0,R)$ and $\mathrm{supp}(\varpi')\subset B(0,2R)$. \blue{Let $P$ be a fixed pseudo-differential operator (of arbitrary order)}. Then
	$$[\varpi,  P \color{black}]\varpi' \mbox{ is a pseudo-differential operator of order }-\infty.$$
	
	In particular, for any \blue{$\sigma\in \mathbb R$, the following estimate holds:}
	\begin{equation}\label{eq:cutoff.commute}
\|[ P, \varpi] f\|_{H^{\sigma}(\RR^2)} \ls \| f\|_{H^{-2}(\RR^2)},
	\end{equation}
	\blue{where the implicit constant depends} only on $P$, $R$ and $\sigma$\magenta{.}
\end{lem}
\begin{proof}
    \blue{Since $\rd_x\varpi$ and $\varpi'$ have disjoint support, the desired conclusion follows from the usual symbolic calculus for pseudo-differential operators; see for instance \cite[Theorem~2~on~p.237]{Stein.book}.} \qedhere
	%Standard calculus for pseudo-differential operators gives that $[\varpi, \Db^\th] - B$ is a pseudodifferential operator of order $\th-2$, where $B$ is a pseudodifferential operator of order $\th-1$ with principal symbol $b=[\rd_{\xi_i} (1+4\pi^2|\xi|^2)^{\frac{\theta}{2}}] [\rd_{x^i} \varpi]$ (up to a constant factor). Because of the support properties $b \varpi' = [\rd_{\xi_i} (1+4\pi^2|\xi|^2)^{\frac{{\th}}{2}} ] [\rd_{x^i} \varpi] \varpi' = 0$. Therefore, $[\varpi, \Db^\th]\varpi'$ is a pseudo-differential operator of order $\th - 2$.
	
	%Finally, to show \eqref{eq:cutoff.commute}, suffices to note that $f = \varpi' f$ so that $[ \Db^{\th}, \varpi] f = [ \Db^{\th}, \varpi] \varpi' f$. \qedhere
\end{proof}

%\begin{lem}\label{lem:cutoff.commute.2}
%	Let $\varpi$, $\varpi'$ be as in Lemma~\ref{lem:cutoff.commute}, $\th\in (0,1)$ and $P$ be a pseudo-differential operator of order $1+\th$. Then $[\varpi, P] \varpi'$ is a pseudo-differential operator of order $-1+\th$. In particular, for any $\th\in (0,1)$ and any $\sigma\in \mathbb R$, there is a constant depending only on $R$, $\th$ and $\sigma$ such that the following holds for all smooth functions $f $ supported in $B(0,R)$:
%	\begin{equation*}
%	\|[ P, \varpi] f\|_{H^{\sigma}(\RR^2)} \ls \| f\|_{H^{\sigma+\th-1}(\RR^2)}.
%	\end{equation*}
	
%\end{lem}
%\begin{proof}
%	This can be proven in exactly the same manner as Lemma~\ref{lem:cutoff.commute}; we omit the details. \qedhere
%\end{proof} 
\begin{proposition}\label{prop:weight.gain.easy.generic_v} Let $v$ be a smooth, compactly supported function on $B(0,R)$. Then $v$ satisfies the following
	\begin{equation}\label{eq:weight.gain.generic_v}
	\| \rd \Db^{s'} v\|_{L^2(\Sigma_t)} \ls \|\la x\ra^{-r} \rd \Db^{s'} v\|_{L^2(\Sigma_t)} + \|\rd v\|_{L^2(\Sigma_t)}. 
	\end{equation}
\end{proposition}
\begin{proof}
	Let $\varpi\in C^\infty_c(\mathbb R^2;\mathbb R)$ be as in Lemma~\ref{lem:cutoff.commute}.
	
	Using the fact $v = \varpi v$, we compute
	$$\rd \Db^{s'} v = \varpi \rd \Db^{s'} v - [\varpi,\Db^{s'}] \rd v .$$
	Since $\varpi$ is compactly supported, we can bound $\|\varpi \rd \Db^{s'} v\|_{L^2(\Sigma_t)}\ls \|\la x\ra^{-r} \rd \Db^{s'} v\|_{L^2(\Sigma_t)}$. On the other hand, $[\varpi,\Db^{s'}]: L^2 \to L^2$ is bounded by Lemma~\ref{lem:cutoff.commute}. Hence  \eqref{eq:weight.gain.generic_v} follows. \color{black} \qedhere
%	$$\| \rd \Db^{s'} v\|_{L^2(\Sigma_t)} \ls \|\la x\ra^{-r} \rd \Db^{s'} v\|_{L^2(\Sigma_t)} + \|\rd v\|_{L^2(\Sigma_t)}.$$
	
\end{proof}
\begin{proposition}\label{prop:weight.gain.easy}
	\begin{equation}\label{eq:weight.gain}
	\| \rd \Db^{s'} \tphi\|_{L^2(\Sigma_t)} \ls \|\la x\ra^{-r} \rd \Db^{s'} \tphi\|_{L^2(\Sigma_t)} + \|\rd\tphi\|_{L^2(\Sigma_t)}. 
	\end{equation}
\end{proposition}
\begin{proof}
	After we recall that $\tphi$ is supported in $B(0,R)$ for every $t$,  \eqref{eq:weight.gain} is obtained as an immediate application of Proposition \ref{prop:weight.gain.easy.generic_v} with $v=\tphi$. \qedhere
	
\end{proof}

\subsubsection{Gaining weights in the estimates for $\rd E_k\Db^{s''} \tphi$ and $\rd L_k\Db^{s''} \tphi$}

Our next goal will be to prove an analogue of Proposition~\ref{prop:weight.gain.easy}, but with also commutation\magenta{s} with $E_k$ and $L_k$; see already Proposition~\ref{prop:weight.gain.EL}. In order to achieve this, we need to understand weighted bounds and derivative bounds involving $[\varpi,\Db^{s''}]$. This will be achieved in the next three lemmas, before we finally turn to Proposition~\ref{prop:weight.gain.EL}.

  Note that we use $\Db^{s''}$ here instead of $\Db^{s'}$ since we will only estimate  $\rd E_k\Db^{s''} \tphi$ and  $\rd L_k\Db^{s''} \tphi$ (in terms of $\rd \Db^{s'} \tphi$) where we recall that $0<s''<s'<\f 12$. It is important to comment that we cannot estimate $\rd E_k\Db^{s'} \tphi$ and  $\rd L_k\Db^{s'} \tphi$ due to the low regularity of the metric\blue{; see the second bullet point in the explanation of \eqref{eq:intro.top.order.1} in Section~\ref{sec:intro.higher.regularity}}. \color{black}
\begin{lem}\label{lem:commute.weight.1}
	Let $f$ be a smooth function which is supported in $B(0,R)$ for each $t$. Then
	$$\|\la x \ra [\varpi,\Db^{s''}] f\|_{L^2(\Sigma_t)}  \ls \|f\|_{L^2(\Sigma_t)}.$$
\end{lem}
\begin{proof}
	\pfstep{Step~1: An easy reduction} Obviously,
	$$\| \la x \ra  [\varpi,\Db^{s''}] f\|_{L^2(\Sigma_t)}^2 = \|  [\varpi,\Db^{s''}] f\|_{L^2(\Sigma_t)}^2 + \sum_{\ell=1}^2 \| x^\ell [\varpi,\Db^{s''}] f\|_{L^2(\Sigma_t)}^2  .$$
	Since $\|  [\varpi,\Db^{s''}] f\|_{L^2(\Sigma_t)} \ls  \| \Db^{s''-2} f\|_{L^2(\Sigma_t)}$ by Lemma~\ref{lem:cutoff.commute}, it suffices to show that for $\ell = 1,2$,
	\begin{equation}\label{eq:commute.with.weight.1.main}
	\| x^\ell [\varpi,\Db^{s''}] f\|_{L^2(\Sigma_t)} \ls \|\Db^{s''-1} f\|_{L^2(\Sigma_t)}  \ls \| f\|_{L^2(\Sigma_t)}.
	\end{equation}
	
	\pfstep{Step~2: Proof of \eqref{eq:commute.with.weight.1.main}} We compute
	\begin{equation}\label{comm.compute}
x^{\ell} [\varpi,\Db^{s''}]f = \underbrace{[\varpi,\Db^{s''}] (x^{\ell} f)}_{=:I} + \underbrace{\varpi [x^{\ell}, \Db^{s''}] f}_{=:II} \underbrace{- [x^{\ell}, \Db^{s''}](\varpi f)}_{=:III} 
	\end{equation}

	We have by Lemma~\ref{lem:cutoff.commute} and the support property of $f$ that
	$$\|I\|_{L^2(\Sigma_t)} \ls \|\Db^{s''-2} (x^\ell f)\|_{L^2(\Sigma_t)} \ls \|x^\ell f \|_{L^2(\Sigma_t)} \ls \| f \|_{L^2(\Sigma_t)}.$$
	%Now $x^\ell\varpi$ is a pseudo-differential operator of order $0$ so that $\|[\Db^{s''-2}, x^\ell \varpi] f \|_{L^2(\Sigma_t)} \ls \| \Db^{s''-3} f \|_{L^2(\Sigma_t)}$. Hence,
	%$$\|\Db^{s''-2} (x^\ell \varpi f)\|_{L^2(\Sigma_t)} \leq \|x^\ell \varpi \Db^{s''-2} f\|_{L^2(\Sigma_t)} + \|[\Db^{s''-2}, x^\ell \varpi] f \|_{L^2(\Sigma_t)} \ls \|\Db^{s''-2} f\|_{L^2(\Sigma_t)},$$
	%where in the final estimate we used that $\| x^\ell \varpi\|_{L^\i(\Sigma_t)} \ls 1$.
	
	Before handling $II$ and $III$, first note that using the Fourier transform, it can easily be checked that for any Schwartz function $f$,
	%\begin{equation*}
	%\begin{split}
	%\mathcal F[x^{\ell} , \Db^{s''}] f = &\: \f{i}{2\pi} \rd_{\xi_k} [(1+4\pi^2|\xi|^2)^{s''} \widehat{f}(\xi)]- \f{i}{2\pi}  (1+4\pi^2|\xi|^2)^{s''} \rd_{\xi_k} \widehat{f}(\xi)\\
	%= &\: s''(1+4\pi|\xi|^2)^{\f{s''}{2}-1} (2\pi i\xi_j) \widehat{f} = \mathcal F(s'' \Db^{s''-2} \rd_{k} f).
	%\end{split}
	%\end{equation*}
	\begin{equation}\label{eq:commute.one.weight}
	[x^{\ell}, \Db^{s''}] f = s'' \Db^{s''-2} \rd_{\ell} f.
	\end{equation}
	From \eqref{eq:commute.one.weight}, it immediately follows that
	$$\|II\|_{L^2(\Sigma_t)} + \|III\|_{L^2(\Sigma_t)}\ls \| \Db^{s''-1} f\|_{L^2(\Sigma_t)}  \ls \| f\|_{L^2(\Sigma_t)}.$$
	
	Combining the estimates for $I$, $II$ and $III$, we have thus proven \eqref{eq:commute.with.weight.1.main}, which by Step~1 implies the desired estimate. \qedhere
\end{proof}

\begin{lem}\label{lem:commute.weight.2}
	Let $f$ be a smooth function which is supported in $B(0,R)$ for each $t$. Then, for every index $\nu\color{black}$,
	$$\| \la x \ra \rd_{\nu\color{black}} [\varpi,\Db^{s''}]f\|_{L^2(\Sigma_t)}  \ls \| \rd f\|_{L^2(\Sigma_t)} .$$
	Moreover, if $\nu\color{black} = i$ is a spatial index, we have the improved estimate\footnote{Note that this is indeed an improvement since $f$ is compactly supported in $B(0,R)$ and we can apply the Poincar\'e inequality.}
	\begin{equation}\label{eq:commute.weight.2.2}
	\| \la x \ra \rd_i [\varpi,\Db^{s''}]f\|_{L^2(\Sigma_t)} \ls \|\Db^{s''}f\|_{L^2(\Sigma_t)}.
	\end{equation}
\end{lem}
\begin{proof}
	\pfstep{Step~0: An easy reduction} 
	
	By Lemma~\ref{lem:cutoff.commute} and the Poincar\'e inequality (since $f$ is compactly supported in $B(0,R)$),
	$$\| \rd_{\nu\color{black}} [\varpi,\Db^{s''}]f\|_{L^2(\Sigma_t)} \ls \|f\|_{L^2(\Sigma_t)} + \|\rd_t f\|_{L^2(\Sigma_t)} \ls \| \rd f\|_{L^2(\Sigma_t)} .$$
	
	Hence, by a reduction similar to Step~1 of Lemma~\ref{lem:commute.weight.1}, it suffices to prove that for $\ell = 1,2$,
	\begin{equation}\label{eq:commute.with.weight.2.main}
	\| x^\ell \rd_t [\varpi,\Db^{s''}] f\|_{L^2(\Sigma_t)} \ls \|\rd f\|_{L^2(\Sigma_t)},\quad \| x^\ell \rd_{ i \color{black}} [\varpi,\Db^{s''}]f\|_{L^2(\Sigma_t)} \ls \|\Db^{s''}f\|_{L^2(\Sigma_t)}.
	\end{equation}
	
	\pfstep{Step~1: Estimates for general $\nu\color{black}$} We compute
	\begin{equation}\label{eq:commute.with.weight.2.1}
	\begin{split}
	&\: x^\ell \rd_{\nu\color{black}} [\varpi, \Db^{s''}] f = x^\ell \rd_{\nu\color{black}} (\varpi\Db^{s''} f) - x^\ell \rd_{\nu\color{black}} \Db^{s''}(\varpi f) \\
	= &\: \underbrace{\rd_{\nu\color{black}} [\varpi,\Db^{s''}] (x^\ell f)}_{=:I} + \underbrace{\rd_{\nu\color{black}} (\varpi [x^{\ell}, \Db^{s''}]  f)}_{=:II} \underbrace{- \rd_{\nu\color{black}} ([x^{\ell}, \Db^{s''}](\varpi f))}_{=:III} \underbrace{- (\rd_{\nu\color{black}} x^{\ell})[\varpi, \Db^{s''}]  f}_{=:IV}.
	\end{split}
	\end{equation}

	The term $I$ needs to be treated differently for $\nu\color{black} = i$ and $\nu\color{black} = t$; see Steps~2 and 3 below.
	
	For the terms $II$ and $III$, %we first note that using the Fourier transform, it can easily be checked that for any Schwartz function $f$,
	%\begin{equation*}
	%\begin{split}
	%\mathcal F[x^{\ell} , \Db^{s''}] f = &\: \f{i}{2\pi} \rd_{\xi_k} [(1+4\pi^2|\xi|^2)^{s''} \widehat{f}(\xi)]- \f{i}{2\pi}  (1+4\pi^2|\xi|^2)^{s''} \rd_{\xi_k} \widehat{f}(\xi)\\
	%= &\: s''(1+4\pi|\xi|^2)^{\f{s''}{2}-1} (2\pi i\xi_j) \widehat{f} = \mathcal F(s'' \Db^{s''-2} \rd_{k} f).
	%\end{split}
	%\end{equation*}
	%\begin{equation}\label{eq:commute.one.weight}
	%[x^{\ell}, \Db^{s''}] f = s'' \Db^{s''-2} \rd_{\ell} f.
	%\end{equation}
	we use \eqref{eq:commute.one.weight} and the $L^2$-boundedness of the $\langle D_x \rangle^{-1} \partial_k$ to obtain
	\begin{equation}\label{eq:commute.with.weight.2.2}
	\|II\|_{L^2(\Sigma_t)} \leq  \|\rd_{\nu\color{black}} (\varpi\Db^{s''-2}\rd_k  f)\|_{L^2(\Sigma_t)} \ls \|\Db^{s''-1}  f\|_{L^2(\Sigma_t)} + \|\Db^{s''-1} \rd_{\nu\color{black}} f\|_{L^2(\Sigma_t)}
	\end{equation}
	and 
	\begin{equation}\label{eq:commute.with.weight.2.3}
	\| III\|_{L^2(\Sigma_t)} \leq \|\Db^{s''-2} \rd^2_{k \nu\color{black}} (\varpi f)\|_{L^2(\Sigma_t)} \ls \|\Db^{s''-1} \rd_{\nu\color{black}}  f\|_{L^2(\Sigma_t)}.
	\end{equation}

	The term $IV$ is the simplest. Since $\rd_{\nu\color{black}} x^\ell$ is bounded (in $L^\i$), we have by Lemma \ref{lem:cutoff.commute}
	\begin{equation}\label{eq:commute.with.weight.2.4}
	\|IV\|_{L^2(\Sigma_t)} \ls \| \Db^{s''-2} f \|_{L^2(\Sigma_t)} \ls \| f\|_{L^2(\Sigma_t)}.
	\end{equation}

	\pfstep{Step~2: Estimates when $\nu\color{black} = t$} We handle term $I$ in \eqref{eq:commute.with.weight.2.1} when $\nu\color{black} = t$. Note that $[\rd_t, [\varpi, \Db^{s''}]x^{\ell}] = 0$. Hence, using Lemma~\ref{lem:cutoff.commute}, we obtain
	\begin{equation}\label{eq:commute.with.weight.2.5}
	\begin{split}
	\|I\|_{L^2(\Sigma_t)} = \|\rd_t[\varpi,\Db^{s''}] (x^{\ell} f)\|_{L^2(\Sigma_t)} = &\: \|[\varpi,\Db^{s''}] (x^{\ell}\rd_t  f)\|_{L^2(\Sigma_t)} \\
	\ls &\: \|\Db^{s''-2}(x^{\ell} \rd_t  f)\|_{L^2(\Sigma_t)} 
	\ls \|x^{\ell} \rd_t  f\|_{L^2(\Sigma_t)} \ls \|\rd  f\|_{L^2(\Sigma_t)},
	\end{split}
	\end{equation}
	where we have used the support property of $ f$.
	
	On the other hand, by \eqref{eq:commute.with.weight.2.2}--\eqref{eq:commute.with.weight.2.4} (and the $L^2$-boundedness of $\Db^{s''-1}$ as well as the Poincar\'e inequality), we have
	\begin{equation}\label{eq:commute.with.weight.2.6}
	\|II\|_{L^2(\Sigma_t)} + \|III\|_{L^2(\Sigma_t)} + \|IV\|_{L^2(\Sigma_t)} \ls \|\rd f\|_{L^2(\Sigma_t)}.
	\end{equation}
	
	Combining \eqref{eq:commute.with.weight.2.5} and \eqref{eq:commute.with.weight.2.6}, we have thus proven the first estimate in \eqref{eq:commute.with.weight.2.main}.
	
	\pfstep{Step~3: Estimates when $\nu\color{black} = i$} If $\nu\color{black} = i$, by Lemma~\ref{lem:cutoff.commute}, we have
	\begin{equation}\label{eq:commute.with.weight.2.7}
	\begin{split}
	\|I\|_{L^2(\Sigma_t)} = &\: \|\rd_i[\varpi,\Db^{s''}] (x^{\ell} f)\|_{L^2(\Sigma_t)} \\
	\ls &\: \|\Db^{s''-1}(x^{\ell}  f)\|_{L^2(\Sigma_t)} \ls \|x^{\ell}  f\|_{L^2(\Sigma_t)} \ls \| f\|_{L^2(\Sigma_t)} \ls \|\Db^{s''} f\|_{L^2(\Sigma_t)},
	\end{split}
	\end{equation}
	where we have used the support property of $ f$.
	
	On the other hand, by \eqref{eq:commute.with.weight.2.2}--\eqref{eq:commute.with.weight.2.4} (and the $L^2$-boundedness of $\Db^{-1}\rd_i$), we have
	\begin{equation*}%\label{eq:commute.with.weight.2.6}
	\|II\|_{L^2(\Sigma_t)} + \|III\|_{L^2(\Sigma_t)} + \|IV\|_{L^2(\Sigma_t)} \ls \|\Db^{s''} f\|_{L^2(\Sigma_t)}.
	\end{equation*}
	
	Together with Step~2, we have thus completed the proof of \eqref{eq:commute.with.weight.2.main}, which then implies the lemma.	\qedhere
\end{proof}

%Our final lemma goes beyond Lemma~\ref{lem:commute.weight.3} and allow for an additional derivative. Unlike Lemma~\ref{lem:commute.weight.3}, Lemma~\ref{lem:commute.weight.4} will apply for $\tphi$ but not more general functions as we will need to use the wave equation.
\begin{lem}\label{lem:commute.weight.4}
	For every index $\nu\color{black}$, $\bt$,
	$$\| \la x\ra \rd^2_{\nu\color{black}\bt} [\varpi, \Db^{s''}] \tphi\|_{L^2(\Sigma_t)} \ls \|\rd \Db^{s'} \tphi\|_{L^2(\Sigma_t)}.$$
\end{lem}
\begin{proof}
	\pfstep{Step~0: Preliminary computations} Using Lemma~\ref{lem:cutoff.commute}, if $(\nu\color{black},\bt) \neq (t,t)$, we have
	\begin{equation}\label{eq:commute.weight.4.prelim.1}
	\| \rd^2_{\nu\color{black}\bt} [\varpi, \Db^{s''}] \tphi \|_{L^2(\Sigma_t)} \ls \| \rd \Db^{s''} \tphi\|_{L^2(\Sigma_t)}.
	\end{equation}
	In the case $(\nu\color{black},\bt) = (t,t)$, we know that $[\rd^2_{tt}, [\varpi, \Db^{s''}]] = 0$ and hence by Lemma~\ref{lem:cutoff.commute} and \blue{Lemma~\ref{lem:invert.tt}}, we have
	\begin{equation}\label{eq:commute.weight.4.prelim.2}
	\| \rd^2_{tt} [\varpi, \Db^{s''}] \tphi \|_{L^2(\Sigma_t)} \ls \| \Db^{s''-2} \rd^2_{tt} \tphi\|_{L^2(\Sigma_t)} \ls \|\rd \Db^{s'} \tphi\|_{L^2(\Sigma_t)}.
	\end{equation}
	Using \eqref{eq:commute.weight.4.prelim.1} and \eqref{eq:commute.weight.4.prelim.2} and arguing as in Step~1 of Lemma~\ref{lem:commute.weight.2}, it suffices to prove
	\begin{equation}\label{eq:commute.with.weight.3.main}
	\| x^\ell \rd^2_{\nu\color{black}\bt} [\varpi, \Db^{s''}] \tphi\|_{L^2(\Sigma_t)} \ls \| \rd \Db^{s'} \tphi\|_{L^2(\Sigma_t)}.
	\end{equation}
	
	We then compute \magenta{using \eqref{comm.compute} that}
	\begin{equation}\label{eq:commute.with.weight.3.main.1}
	\begin{split}
	&\: x^\ell \rd^2_{\nu\color{black}\bt} [\varpi, \Db^{s''}] \tphi \\
	= &\: \rd^2_{\nu \bt} (x^\ell [\varpi, \Db^{s''}] \tphi) - (\rd_{\nu} x^{\ell}) (\rd_\bt ([\varpi, \Db^{s''}] \tphi))- (\rd_\bt x^{\ell}) (\rd_{\nu} ([\varpi, \Db^{s''}] \tphi)) \\= &\:  \underbrace{\rd^2_{\nu\color{black}\bt} ([\varpi,  \Db^{s''}] (x^{\ell}\tphi))}_{=:I} \underbrace{+ \rd^2_{\nu\color{black}\bt} (\varpi [x^{\ell}, \Db^{s''}] \tphi)}_{=:II} \underbrace{- \rd^2_{\nu\color{black}\bt} ([x^{\ell},\Db^{s''}] \varpi \tphi)}_{=:III} \\
	&\: \underbrace{- (\rd_{\nu\color{black}} x^{\ell}) (\rd_\bt ([\varpi, \Db^{s''}] \tphi))}_{=:IV} \underbrace{- (\rd_\bt x^{\ell}) (\rd_{\nu\color{black}} ([\varpi, \Db^{s''}] \tphi))}_{=:V}.
	\end{split}
	\end{equation}
	We control each term in \eqref{eq:commute.with.weight.3.main.1} in the steps below.
	
	\pfstep{Step~1: Term $I$} We separate into three cases. When $(\nu\color{black},\bt) = (i,j)$, by Lemma~\ref{lem:cutoff.commute} and Poincar\'e's inequality (since $\mathrm{supp}(\tphi) \subset B(0,R)$),
	\begin{equation*}
	\begin{split}
	&\: \| \rd^2_{ij} ([\varpi,  \Db^{s''}] (x^{\ell}\tphi))\|_{L^2(\Sigma_t)}\ls \| \Db^{s''} (x^{\ell}\tphi))\|_{L^2(\Sigma_t)} \ls \|\rd \tphi\|_{L^2(\Sigma_t)}.
	\end{split}
	\end{equation*}
	When $(\nu\color{black},\bt) = (i,t)$, since $[\rd_t , [\varpi,  \Db^{s''}] x^{\ell}] = 0$, we use Lemma~\ref{lem:cutoff.commute} and then Lemma \ref{lem:stupid} to obtain
	$$\| \rd^2_{it} ([\varpi,  \Db^{s''}] (x^{\ell}\tphi))\|_{L^2(\Sigma_t)} \ls \|\Db^{s''-1}  (x^{\ell}\rd_t \tphi)\|_{L^2(\Sigma_t)} \ls \|\rd \Db^{s'} \tphi\|_{L^2(\Sigma_t)}.$$
	Finally, when $(\nu\color{black},\bt) = (t,t)$, since $[\rd_{tt}^2 , [\varpi,  \Db^{s''}] x^{\ell}] = 0$, we use Lemma~\ref{lem:cutoff.commute} and then \blue{Lemma~\ref{lem:invert.tt}} to get
	$$\| \rd^2_{tt} ([\varpi,  \Db^{s''}] (x^{\ell}\tphi))\|_{L^2(\Sigma_t)} \ls \|\Db^{s''-2}  (x^{\ell}\rd^2_{tt} \tphi)\|_{L^2(\Sigma_t)} \ls \|\rd \Db^{s'} \tphi\|_{L^2(\Sigma_t)}.$$
	
	Thus in all cases
	$$\| I \|_{L^2(\Sigma_t)} \ls \|\rd \Db^{s'} \tphi\|_{L^2(\Sigma_t)}.$$
	
	\pfstep{Step~2: Terms $II$ and $III$} Terms $II$ and $III$ are very similar --- we will only treat $II$. We use the formula \eqref{eq:commute.one.weight}. When one of $\nu\color{black}$ or $\bt$ is a spatial derivative, it follows easily from \eqref{eq:commute.one.weight} and the support properties of $\tphi$ that
	\begin{equation*}
	\begin{split}
	&\: \|\rd^2_{\nu\color{black}\bt} (\varpi [x^{\ell}, \Db^{s''}] \tphi)\|_{L^2(\Sigma_t)}  \ls \|\rd \tphi\|_{L^2(\Sigma_t)}. 
	\end{split}
	\end{equation*}
	
	If $(\nu\color{black},\bt) = (t,t)$, we use $[\rd^2_{tt}, \varpi [x^{\ell}, \Db^{s''}]] = 0$, the equation \eqref{eq:commute.one.weight} and then \blue{Lemma~\ref{lem:invert.tt}} to obtain
	\begin{equation*}
	\begin{split}
	&\: \|\rd^2_{tt} (\varpi [x^{\ell}, \Db^{s''}] \tphi)\|_{L^2(\Sigma_t)} \\
	\ls &\: \|  \varpi \Db^{s''-2} \rd_\ell \rd^2_{tt}\tphi\|_{L^2(\Sigma_t)} \ls \|\Db^{s''-1} \rd^2_{tt} \tphi\|_{L^2(\Sigma_t)}\ls  \|\rd \Db^{s'} \tphi\|_{L^2(\Sigma_t)}.
	\end{split}
	\end{equation*}
	
	In either case
	$$\| II \|_{L^2(\Sigma_t)} \ls \|\rd \Db^{s'} \tphi\|_{L^2(\Sigma_t)}.$$
	The same holds for $III$ in a similar manner; we omit the details.
	
	\pfstep{Step~3: Terms $IV$ and $V$} Since $\rd_{\nu\color{black}}x^{\ell}$ and $\rd_{\bt}x^{\ell}$ are both bounded, we have, by Lemma~\ref{lem:cutoff.commute},
	$$\|IV \|_{L^2(\Sigma_t)} + \|V \|_{L^2(\Sigma_t)} \ls \| \rd [\varpi, \Db^{s''}] \tphi \|_{L^2(\Sigma_t)} \ls \|\rd \Db^{s'} \tphi\|_{L^2(\Sigma_t)}.$$
	
	We have thus estimated every term on the right-hand side of \eqref{eq:commute.with.weight.3.main.1} and proven \eqref{eq:commute.with.weight.3.main}. As argued in Step~0, this gives the lemma.	 \qedhere
\end{proof}

\begin{proposition}\label{prop:weight.gain.EL}
	For any $r\geq 0$, the following holds (with implicit constants depending on $r$)\footnote{We remark that the estimates hold with $\| \rd \Db^{s'} \tphi\|_{L^2(\Sigma_t)}$ replaced by $\| \rd \Db^{s''} \tphi\|_{L^2(\Sigma_t)}$ on the right-hand side, but the exposition becomes slightly less convenient when citing earlier lemmas.}:
	\begin{align}
	\| \rd E_k \Db^{s''} \tphi\|_{L^2(\Sigma_t)} \ls &\: \|\la x\ra^{-r} \rd E_k \Db^{s''} \tphi\|_{L^2(\Sigma_t)} + \| \rd \Db^{s'} \tphi\|_{L^2(\Sigma_t)}, \label{eq:weight.gain.E}\\
	\| \rd L_k \Db^{s''} \tphi\|_{L^2(\Sigma_t)} \ls &\: \|\la x\ra^{-r} \rd L_k \Db^{s''} \tphi\|_{L^2(\Sigma_t)} + \| \rd \Db^{s'} \tphi\|_{L^2(\Sigma_t)}. \label{eq:weight.gain.L}
	\end{align}
\end{proposition}
\begin{proof}
	We will only prove \eqref{eq:weight.gain.E} in detail; \eqref{eq:weight.gain.L} is similar.

	Since $\rd E_k \Db^{s''} \tphi = \rd E_k \Db^{s''} (\varpi\tphi)$, we compute
	$$\rd E_k \Db^{s''} \tphi = \varpi \rd E_k \Db^{s''} \tphi + \rd [(E_k \varpi) \Db^{s''}\tphi] - \rd E_k [\varpi,\Db^{s''}] \tphi + (\rd \varpi) E_k \Db^{s''} \tphi.$$
	
	In particular, by the bounds of $E_k$ in \eqref{eq:frame.1}, \eqref{eq:frame.2} and the support properties of $\varpi$,
	\begin{equation*}
	\begin{split}
	\|\rd E_k \Db^{s''} \tphi\|_{L^2(\Sigma_t)} \ls &\: \underbrace{\|\rd E_k \Db^{s''} \tphi\|_{L^2(\Sigma_t\cap B(0,3R))}}_{=:I} + \underbrace{\|\rd \Db^{s''} \tphi\|_{L^2(\Sigma_t\cap B(0,3R))}}_{=:II} \\
	&\: + \underbrace{\|\Db^{s''} \tphi\|_{L^2(\Sigma_t)}}_{=:III} + \underbrace{\|\rd E_k [\varpi,\Db^{s''}] \tphi\|_{L^2(\Sigma_t)}}_{=:IV}.
	\end{split}
	\end{equation*}
	
	The terms $I$ and $II$ are obviously bounded by the right-hand side of \eqref{eq:weight.gain.E}.
	
	For the term $III$, we use interpolation and then Poincar\'e's inequality to obtain
	\begin{equation}\label{eq:est.for.Ds''}
	\|\Db^{s''} \tphi\|_{L^2(\Sigma_t)} \ls \| \tphi\|_{L^2(\Sigma_t)} + \|\rd \tphi \|_{L^2(\Sigma_t)} \ls \|\rd \tphi \|_{L^2(\Sigma_t)} \ls \| \rd \Db^{s'} \tphi\|_{L^2(\Sigma_t)}.
	\end{equation}
	
	For the term $IV$, we use the bounds for $E^\mu_k$ in  \eqref{eq:frame.1}, \eqref{eq:frame.2}, H\"older's inequality and Lemmas~\ref{lem:commute.weight.2} and \ref{lem:commute.weight.4} to obtain
	\begin{equation*}
	\begin{split}
	&\: \| \rd_{\nu\color{black}} E_k [\varpi,\Db^{s''}] \tphi\|_{L^2(\Sigma_t)} \\
	\ls &\: \|\la x\ra^{-1} \rd_{\nu\color{black}} E_k^\mu\|_{L^\i(\Sigma_t)} \|\la x\ra \rd_\mu [\varpi,\Db^{s''}] \tphi\|_{L^2(\Sigma_t)} + \|\la x\ra^{-1} E_k^\mu\|_{L^\i(\Sigma_t)} \|\la x\ra \rd^2_{\nu\color{black}\mu} [\varpi,\Db^{s''}] \tphi\|_{L^2(\Sigma_t)} \\
	\ls &\: \| \rd \Db^{s'} \tphi\|_{L^2(\Sigma_t)} + \|\rd\tphi\|_{L^2(\Sigma_t)} \ls \|\rd \Db^{s'} \tphi\|_{L^2(\Sigma_t)}.
	\end{split}
	\end{equation*}
	
\end{proof}

\subsubsection{Auxiliary weighted estimates for commutator of a vector field and Riesz transform}

\begin{lem}\label{lem:commute.Riesz.phi}
	Let $f$, $h$ be smooth functions such that 
	$$\|\la x\ra^{-1} h\|_{L^\i(\Sigma_t)} + \|\varpi h\|_{W^{1,\i}\cap W^{2,4}(\Sigma_t)} \ls 1,$$
	$$\|\varpi f\|_{L^{\i}\cap W^{1,4}(\Sigma_t)} + \|\varpi \rd_i f \|_{L^4(\Sigma_t)} \ls 1.$$
	Then, for $R_j = \rd_j \Db^{-1}$, we have
	\begin{equation}\label{eq:commute.Riesz.phi.1}
	\|\la x\ra^{-1} [h\rd_i , R_l R_q] \Db^{s''} (f \rd_\lambda \tphi)\|_{L^2(\Sigma_t)} \ls \| \rd \Db^{s'} \tphi \|_{L^2(\Sigma_t)},
	\end{equation}
	and
	\begin{equation}\label{eq:commute.Riesz.phi.2}
	\|\la x\ra^{-1} [h\rd_i , R_l R_q \Db^{s''}] (f \rd_\lambda \tphi)\|_{L^2(\Sigma_t)} \ls \| \rd \Db^{s'} \tphi \|_{L^2(\Sigma_t)}.
	\end{equation}
\end{lem}
\begin{proof}
	\pfstep{Step~1: Proof of \eqref{eq:commute.Riesz.phi.1}} We first compute
	\begin{equation*}
	\begin{split}
	&\: [h\rd_i , R_l R_q] \Db^{s''} (f \rd_\lambda \tphi) \\%= &\: h\rd_i R_l R_q \Db^{s''} (f \rd_\lambda \tphi) - R_l R_q h\rd_i \Db^{s''} (f \rd_\lambda \tphi) \\
	= &\: \underbrace{[(\varpi h)\rd_i, R_l R_q] \Db^{s''} (f \rd_\lambda \tphi)}_{=:I} + \underbrace{h [\rd_i R_l R_q \Db^{s''}, \varpi] (f \rd_\lambda \tphi)}_{=:II} \\
	&\: \underbrace{- R_l R_q \{ h \rd_i [\Db^{s''}, \varpi] (f \rd_\lambda \tphi) \}}_{=:III} \underbrace{- R_l R_q \{ h (\rd_i \varpi) \Db^{s''} (f \rd_\lambda \tphi) \}}_{=:IV} .
	\end{split}
	\end{equation*}
	
	Term $I$ can be bounded using Lemmas~\ref{lem:commute.Riesz} (with $Y=h\partial_i$), Lemma~\ref{lem:stupid} and the identity $[Y, R_l R_q]= [Y,R_l] R_q+ R_l [Y,R_q]$ so that 
	$$\|I\|_{L^2(\Sigma_t)} \ls \|\varpi h\|_{W^{1,\infty}(\Sigma_t)} \|\Db^{s''} (f \rd_\lambda \tphi)\|_{L^2(\Sigma_t)} \ls  \| \rd \Db^{s'} \tphi\|_{L^2(\Sigma_t)} \color{black}.$$
	
	For term $II$, we use H\"older's inequality, \blue{Lemmas~\ref{lem:cutoff.commute}} and \ref{lem:stupid} to obtain
	\begin{equation*}%\label{eq:commute.Riesz.phi.3}
	\begin{split}
	\|\la x\ra^{-1} II\|_{L^2(\Sigma_t)} \ls&\:  \|\la x\ra^{-1} h\|_{L^\i} \|[\rd_i R_l R_q \Db^{s''}, \varpi] (f \rd_\lambda \tphi)\|_{L^2(\Sigma_t)}\\
	\ls &\:  \|\Db^{s''-1} (f \rd_\lambda \tphi)\|_{L^2(\Sigma_t)} \ls \| \rd \Db^{s'} \tphi\|_{L^2(\Sigma_t)}.
	\end{split}
	\end{equation*}
	
	For the term $III$, we use the boundedness of the Riesz transform, H\"older's inequality, the improved estimate \eqref{eq:commute.weight.2.2} in Lemma~\ref{lem:commute.weight.2}  and then Lemma~\ref{lem:stupid} to obtain
	\begin{equation*}
	\begin{split}
	\|III\|_{L^2(\Sigma_t)} \ls &\: \|h \rd_i[\Db^{s''}, \varpi] (f \rd_\lambda \tphi)\|_{L^2(\Sigma_t)} \\
	\ls &\: \|\la x \ra^{-1} h\|_{L^\i(\Sigma_t)} \|\la x \ra \rd_i[\Db^{s''}, \varpi] (f \rd_\lambda \tphi)\|_{L^2(\Sigma_t)} \\
	\ls &\: \|\Db^{s''} (f \rd_\lambda \tphi)\|_{L^2(\Sigma_t)} \ls \|\rd \Db^{s'} \tphi\|_{L^2(\Sigma_t)}.
	\end{split}
	\end{equation*}
	
	Finally, for the term $IV$, we use the $L^2$-boundedness of the Riesz transform, H\"older's inequality and Lemma~\ref{lem:stupid} to obtain
	$$\|IV\|_{L^2(\Sigma_t)} \ls \|h\rd_i \varpi\|_{L^\i(\Sigma_t)} \|\Db^{s''} (f \rd_\lambda \tphi)\|_{L^2(\Sigma_t)} \ls  \| \rd \Db^{s'} \tphi\|_{L^2(\Sigma_t)}.$$
	%$$[\rd_i \Db^{s''}, \varpi] = \rd_i \Db^{s''} \varpi - \varpi \rd_i \Db^{s''} = \rd_i [\Db^{s''}, \varpi] + \rd_i (\varpi \Db^{s''} - \varpi \rd_i \Db^{s''}$$
	
	\pfstep{Step~2: Proof of \eqref{eq:commute.Riesz.phi.2}} We first notice that
	$$[h\rd_i , R_l R_q \Db^{s''}] (f\rd_\lambda\tphi) = [h\rd_i , R_l R_q ]\Db^{s''} (f\rd_\lambda\tphi) + R_l R_q [h \rd_i, \Db^{s''}] (f\rd_\lambda\tphi).$$
	In view of \eqref{eq:commute.Riesz.phi.1} (which controls the first term) and the boundedness of the Riesz transform, it therefore suffices to prove
	$$\| [h \rd_i, \Db^{s''}] (f\rd_\lambda\tphi)\|_{L^2(\Sigma_t)}\ls  \| \rd \Db^{s'} \tphi\|_{L^2(\Sigma_t)} \color{black}.$$
	On the other hand,
	\begin{equation*}
	\begin{split}
	&\: [h \rd_i, \Db^{s''}] (f\rd_\lambda\tphi)  = [h \rd_i, \Db^{s''}] (\varpi f\rd_\lambda\tphi)  \\
	= &\: \underbrace{ h \rd_i [\Db^{s''} ,\varpi](f\rd_\lambda\tphi)}_{=:I} + \underbrace{[(\varpi h) \rd_i, \Db^{s''}] (f\rd_\lambda\tphi)}_{=:II} - \underbrace{\Db^{s''} [(\rd_i \varpi ) h f (\rd_\lambda \tphi)]}_{=:III} + \underbrace{h(\rd_i \varpi) \Db^{s''} (f\rd_\lambda \tphi)}_{=:IV}.
	\end{split}
	\end{equation*}
	
	For term $I$, we use H\"older's inequality, Lemmas~\ref{lem:cutoff.commute} and \ref{lem:stupid} to obtain
	\begin{equation*}
	\begin{split}
	\|\la x \ra^{-1} I\|_{L^2(\Sigma_t)} \ls &\: \|\la x \ra^{-1} h\|_{L^\i(\Sigma_t)} \|\rd_i [\Db^{s''} ,\varpi](f\rd_\lambda\tphi)\|_{L^2(\Sigma_t)} \\
	\ls &\: \|\Db^{s''-1}(f\rd_\lambda\tphi)\|_{L^2(\Sigma_t)} \ls \|\rd \Db^{s'} \tphi\|_{L^2(\Sigma_t)}.
	\end{split}
	\end{equation*}
	
	For term $II$, we use Proposition~\ref{prop:commute.2} (with $p = \infty$), the $L^2$-boundedness of $\rd_i \Db^{-1}$ and Lemma~\ref{lem:stupid} to obtain
	$$\|II\|_{L^2(\Sigma_t)} \ls \|\varpi h\|_{W^{1,\infty}} \|\Db^{s''-1} \rd_i (f\rd_\lambda\tphi)\|_{L^2(\Sigma_t)} \ls \|\Db^{s''} (f\rd_\lambda\tphi)\|_{L^2(\Sigma_t)} \ls \| \rd \Db^{s'} \tphi\|_{L^2(\Sigma_t)}.$$
	
	For term $III$, note that ${\varpi} \equiv 1$ on $\mathrm{supp}(\tphi)$, hence $( \partial_i {\varpi} )\partial_{\lambda} \tphi \equiv 0$ which implies $III=0$.% can be estimated by Proposition~\ref{prop:commute.2}, the Sobolev embedding $\Db^{-1}: L^2(\Sigma_t) \to L^4(\Sigma_t)$ and Lemmas~\ref{lem:stupid} as follows:	\begin{equation*}	\begin{split}	\|III\|_{L^2(\Sigma_t)} \ls &\: \| \rd_i (\varpi h) \|_{L^\i(\Sigma_t)} \| \|\Db^{s''} (f \rd_\lambda \tphi) \|_{L^2(\Sigma_t)} + \| \rd_{ij} (\varpi h) \|_{L^4(\Sigma_t)} \|\Db^{s''-1} (f \rd_\lambda \tphi) \|_{L^4(\Sigma_t)} \\	\ls &\: \|\Db^{s''} (f \rd_\lambda \tphi) \|_{L^2(\Sigma_t)} \ls  \| \rd \Db^{s'} \tphi\|_{L^2(\Sigma_t)}.	\end{split}	\end{equation*}
	
	Finally, term $IV$ can be treated exactly as term $IV$ in the proof of \eqref{eq:commute.Riesz.phi.1} so that 
	$$\|IV\|_{L^2(\Sigma_t)} \ls  \| \rd \Db^{s'} \tphi\|_{L^2(\Sigma_t)}.$$
\end{proof}

%\subsection{Proof of \eqref{tphiH3/2x}} 

\subsection{Estimates for $\rd \Db^{s'} \tphi$}\label{tphiH3/2xsection}
In this section, we bound $\|\rd \Db^{s'} \tphi\|_{L^2(\Sigma_t)}$. We will give the main result in Proposition~\ref{prop:Ds.est} and give a high level proof. The main estimates that are used in the proof will be proven in Propositions~\ref{prop:frac.2} and \ref{prop:frac.1} below.
\begin{proposition}\label{prop:Ds.est}
	\begin{equation} \label{tphiH3/2x}
	\sup_{t\in [0,T_B)} \|\rd \Db^{s'} \tphi\|_{L^2(\Sigma_t)}\ls \ep.
	\end{equation}
\end{proposition}
\begin{proof} 	
	By Corollary~\ref{cor:main.weighted.energy} \magenta{(with $v = \Db^{s'} \tphi$, $f_1 = \Box_g \Db^{s'} \tphi$, $f_2 = 0$)} and  Proposition~\ref{prop:weight.gain.easy}\color{black}, for every $T\in [0,T_B)$,
	\begin{equation*}
	\begin{split}
	&\: \sup_{t\in [0,T)} \|\rd \Db^{s'} \tphi\|_{L^2(\Sigma_t)}^2\\
	\ls &\: \|\la x \ra^{-\f r2} \rd \Db^{s'} \tphi\|_{L^2(\Sigma_0)}^2 + \sup_{t\in [0,T)} \| \rd \tphi \|_{L^2(\Sigma_t)}^2  \\
	&\: + \int_0^{T} \|\rd \Db^{s'} \tphi \|_{L^2(\Sigma_\tau)}^2 \, \ud \tau + \int_0^T \|\wo2 \Box_g \Db^{s'} \tphi\|_{L^2(\Sigma_\tau)}^2 \, \ud \tau \\
	\ls &\: \ep^2 + \int_0^{T} \|\rd \Db^{s'} \tphi \|_{L^2(\Sigma_\tau)}^2 \, \ud\tau + \int_0^T \|\wo2 \Box_g \Db^{s'} \tphi\|_{L^2(\Sigma_\tau)}^2 \,\ud \tau,
	\end{split}
	\end{equation*}
	where in the last line we have used the data bound \eqref{eq:assumption.rough.energy} and the estimate \eqref{energyglobal}.
	
	Clearly, $\Box_g \Db^{s'} \tphi = [\Box_g,\Db^{s'}]\tphi$. Using \eqref{decompositionbox}, we in fact have
	$$\Box_g \Db^{s'} \tphi = [\Box^2,\Db^{s'}]\tphi - [\Box^1,\Db^{s'}]\tphi.$$
	In Propositions~\ref{prop:frac.2} and \ref{prop:frac.1} below, we will prove respectively that for $r\geq 1$,
	\begin{equation}\label{eq:frac.commute.2}
	\|\wo2 [ \Box^2, \Db^{s'}]\tphi\|_{L^2(\Sigma_t)}^2 \ls \|\rd \Db^{s'} \tphi\|_{L^2(\Sigma_t)}^2,
	\end{equation}
	and 
	\begin{equation}\label{eq:frac.commute.1}
	\|\wo2 [ \Box^1, \Db^{s'}]\tphi\|_{L^2(\Sigma_t)}^2 \ls \|\rd \Db^{s'} \tphi\|_{L^2(\Sigma_t)}^2.
	\end{equation}
	Hence,
	$$ \sup_{t\in [0,T)} \|\rd \Db^{s'} \tphi\|_{L^2(\Sigma_t)}^2 \ls \ep^2 + \int_0^T \|\rd \Db^{s'} \tphi\|_{L^2(\Sigma_t)}^2 \,dt.$$
	The desired estimate therefore follows from Gr\"onwall's inequality. \qedhere
\end{proof}

Given the proof of Proposition~\ref{prop:Ds.est} above, we need to prove the commutator estimates \blue{\eqref{eq:frac.commute.2} and \eqref{eq:frac.commute.1}}. For both of these bounds, we \blue{also} prove corresponding commutator estimates for more general functions\blue{.} (These more general commutator bounds will be useful later in Section~\ref{sec:rphi}.)

\begin{proposition}\label{prop:frac.2}
Let $v$ be a smooth, compactly supported function on $B(0,R)$. Then for $r\geq 1$,
	\begin{equation}\label{eq:frac.2_generic_v}
	\|\la x\ra^{-\f r2} [ \Box^2, \Db^{s'}]v\|_{L^2(\Sigma_t)} \ls  \| \partial \Db^{s'} v\|_{L^2(\Sigma_t)}+ \| \Box_g v\|_{L^2(\Sigma_t)}.
	\end{equation}
	\blue{As a result, \eqref{eq:frac.commute.2} holds.}
\end{proposition}
\begin{proof}
	
	%\pfstep{Step~1: Preliminary reduction} 
	Recall from \eqref{decompositionbox} that $\Box^2 = (g^{-1})^{\nu\color{black}\bt} \partial^2_{\nu\color{black}\bt}$. By the support properties of $v$, we have $\rd^2_{\nu\color{black}\bt} v = \varpi \rd^2_{\nu\color{black}\bt} v$. Hence,
	\begin{equation}\label{eq:frac.2.reduction}
	\begin{split}
	&\: \|\la x\ra^{-\f r2} [ \Box^2, \Db^{s'}]v\|_{L^2(\Sigma_t)} = \| \la x \ra^{-\f r2} [\Db^{s'} ((g^{-1})^{\nu\color{black}\bt} \rd^2_{\nu\color{black}\bt} v) - (g^{-1})^{\nu\color{black}\bt} \Db^{s'} \rd^2_{\nu\color{black}\bt} v] \|_{L^2(\Sigma_t)} \\
	\ls &\: \underbrace{\| \Db^{s'} (\varpi (g^{-1})^{\nu\color{black}\bt} \rd^2_{\nu\color{black}\bt} v) - \varpi (g^{-1})^{\nu\color{black}\bt} \Db^{s'} \rd^2_{\nu\color{black}\bt} v \|_{L^2(\Sigma_t)}}_{=:I} + \underbrace{\| \la x \ra^{-\f r2} (g^{-1})^{\nu\color{black}\bt} [\Db^{s'},\varpi] \rd^2_{\nu\color{black}\bt} v\|_{L^2(\Sigma_t)}}_{=:II}.
	\end{split}
	\end{equation}
	
	By Proposition~\ref{prop:commute.2} with $p=\infty$, the estimates for the metric components in \eqref{eq:g.main}, and Lemma~\ref{lem:invert.tt_generic_v}, $I$ in \eqref{eq:frac.2.reduction} is bounded by
	\begin{equation*}
	\begin{split}
	&\: \| \varpi(g^{-1})^{\nu\color{black}\bt} \partial^2_{\nu\color{black}\bt} \Db^{s'} v - \Db^{s'} (\varpi(g^{-1})^{i \lambda} \partial^2_{i \lambda} v)\|_{L^2(\Sigma_t)} \\
	\ls &\: \|\varpi (g^{-1})^{\nu\color{black}\bt}\|_{W^{1,\infty}(\Sigma_t)} \|\Db^{s'-1} \partial^2_{\nu\color{black}\bt} v\|_{L^2(\Sigma_t)} \ls \| \partial \Db^{s'} v\|_{L^2(\Sigma_t)}+ \| \Box_g v\|_{L^2(\Sigma_t)}. 
	\end{split}
	\end{equation*}
	
	By H\"older's inequality, \eqref{eq:g.main}, Lemmas~\ref{lem:cutoff.commute} and \ref{lem:invert.tt_generic_v}, $II$ in \eqref{eq:frac.2.reduction} can be controlled by
	\begin{equation*}
	\begin{split}
	II \ls &\: \|\wo2 (g^{-1})^{\nu\color{black}\bt}\|_{L^\i(\Sigma_t)}\|[\Db^{s'},\varpi] \rd^2_{\nu\color{black}\bt} v\|_{L^2(\Sigma_t)} \\
	\ls &\: \|\Db^{s'-2} \rd^2_{\nu\color{black}\bt} v\|_{L^2(\Sigma_t)} \ls \|\rd \Db^{s'} v\|_{L^2(\Sigma_t)} + \| \Box_g v\|_{L^2(\Sigma_t)}. 
	\end{split}
	\end{equation*}
	
	Combining the above estimates gives \blue{\eqref{eq:frac.2_generic_v}}.
	
	\blue{Finally, since $\tphi$ is supported in $B(0,R)$ for every $t$, and that $\Box_g \tphi=0$, \eqref{eq:frac.commute.2} follows from \eqref{eq:frac.2_generic_v}.} \qedhere
	
\end{proof}
%\begin{proposition}\label{prop:frac.2}
%	For $r\geq 1$,
%	$$\|\la x\ra^{-\f r2} [ \Box^2, \Db^{s'}]\tphi\|_{L^2(\Sigma_t)} \ls  \| \partial \Db^{s'} \tphi\|_{L^2(\Sigma_t)}.$$
%\end{proposition}
%\begin{proof} After we recall that $\tphi$ is supported in $B(0,R)$ for every $t$, and that $\Box_g \tphi=0$ the result is obtained as an immediate application of Proposition \ref{prop:frac.2_generic_v} with $v=\tphi$.\end{proof}

\begin{proposition}\label{prop:frac.1} Let $v$ be a smooth, compactly supported function on $B(0,R)$. Then for $r\geq 1$,
	\begin{equation}\label{eq:frac.1_generic_v}
	\|\wo2 [ \Box^1, \Db^{s'}]v\|_{L^2(\Sigma_t)} \ls \|\rd \Db^{s'} v\|_{L^2(\Sigma_t)}.
	\end{equation}
	\blue{As a result, \eqref{eq:frac.commute.1} holds.}
\end{proposition}
\begin{proof}
	Recall that $\Box^1 = \Gamma^\lambda \rd_\lambda$. We will in fact not need the commutator structure and bound each term separately. By H\"older's inequality and the estimates for $\Gamma^\lambda$ in \eqref{eq:Gamma}, we have
	$$\|\wo2 \Box^1 (\Db^{s'} v)\|_{L^2(\Sigma_t)} = \|\wo2 \Gamma^\lambda \rd_\lambda \Db^{s'} v\|_{L^2(\Sigma_t)} \ls \|\wo2 \Gamma^\lambda\|_{L^\i(\Sigma_t)} \|\rd \Db^{s'} v\|_{L^2(\Sigma_t)}.$$
	
	On the other hand, by \eqref{eq:Gamma}, we have $\|\varpi \Gamma^{\lambda} \|_{L^\i \cap W^{1,2}(\Sigma_t)} \ls \ep^{\f 32}$. Hence, by Lemma~\ref{lem:stupid_generic_v}, we obtain
	\begin{equation*}
	\begin{split}
	\|\wo2 \Db^{s'} (\Box^1 v)\|_{L^2(\Sigma_t)} = &\: \|\Db^{s'} (\Gamma^\lambda \rd_\lambda v)\|_{L^2(\Sigma_t)} \ls  \|\rd \Db^{s'} v\|_{L^2(\Sigma_t)} .
	\end{split}
	\end{equation*}
	
	Combining the above estimates gives \blue{\eqref{eq:frac.1_generic_v}. We then conclude \eqref{eq:frac.commute.1} by using \eqref{eq:frac.1_generic_v}, $\Box_g \tphi=0$ and the support property of $\tphi$.}\qedhere
\end{proof}

%\begin{proposition}\label{prop:frac.1}
%	For $r\geq 1$,
%	$$\|\wo2 [ \Box^1, \Db^{s'}]\tphi\|_{L^2(\Sigma_t)} \ls \|\rd \Db^{s'} \tphi\|_{L^2(\Sigma_t)}.$$
%\end{proposition}
%\begin{proof} After we recall that $\tphi$ is supported in $B(0,R)$ for every $t$, and that $\Box_g \tphi=0$ the result is obtained as an immediate application of Proposition \ref{prop:frac.1_generic_v} with $v=\tphi$.
%\end{proof}

\subsection{Estimates for $\rd E_k \Db^{s''} \tphi$}
%\subsection{Proof of \eqref{EtphiH3/2x}} 
\label{EtphiH3/2xsection}

Similarly as in Section~\ref{tphiH3/2xsection}, let us first give a high level proof of our main estimate. The main steps will be postponed to a number of propositions below.
\begin{proposition}\label{prop:EkDsphi.main}
	The following estimate holds for all $t\in [0,T_B)$:
	\begin{equation} \label{EtphiH3/2x.old}
	\|\rd E_k \Db^{s''} \tphi \|_{L^2(\Sigma_t)}  \ls \ep + \ep^{\f 3 2}  \|\rd L_k \Db^{s''} \tphi \|_{L^2(\Sigma_t)}. 
	\end{equation}
\end{proposition}
\begin{proof}
	Take $r\geq 2$. By Corollary~\ref{cor:main.weighted.energy} (with $v = E_k \Db^{s''} \tphi$, $f_1= \Box_g(E_k\Db^{s''}\tphi\blue{)}$ and $f_2=0$), 
	\begin{equation}\label{eq:rdEDstphi.1}
	\begin{split}
	&\: \sup_{t\in [0, T)} \|\la x \ra^{-r-2\nu\color{black}} \rd E_k \Db^{s''} \tphi  \|_{L^2(\Sigma_t)}^2  \\
	\ls  &\: \|\la x\ra^{-\f r2} \rd E_k \Db^{s''} \tphi  \|_{L^2(\Sigma_0)}^2 + \int_0^{T} \|\la x\ra^{-\f{r}{2}} \Box_g E_k \Db^{s''} \tphi \|_{L^2(\Sigma_\tau)}^2 \, \ud\tau \\
	\ls &\: \ep^2 +  \int_0^{T} \|\la x\ra^{-\f{r}{2}} \Box_g E_k \Db^{s''} \tphi \|_{L^2(\Sigma_\tau)}^2 \, \ud\tau,
	\end{split}
	\end{equation}
	where we have bound the initial data term as $\|\la x\ra^{-\f r2} \rd E_k \Db^{s''} \tphi  \|_{L^2(\Sigma_0)}^2 \ls \ep^2$ using the assumptions \eqref{eq:assumption.rough.energy} and \eqref{eq:assumption.rough.energy.commuted} on $\Sigma_0$ (using additionally \eqref{defnormal}, \eqref{eq:frame.1}, \eqref{eq:frame.2}, \eqref{eq:g.main} and Proposition \ref{prop:commute.2} to address commutator terms involving fractional derivatives).

	Using \eqref{eq:weight.gain.E} and \eqref{tphiH3/2x}, we can gain weights on the left-hand side of \eqref{eq:rdEDstphi.1} as follows
	\begin{equation*}
	\begin{split}
	&\: \sup_{t\in [0, T)} \|\rd E_k \Db^{s''} \tphi  \|_{L^2(\Sigma_t)}^2  \\
	\ls  &\: \ep^2 + \sup_{t \in [0,T)} \|\rd \Db^{s'} \tphi\|_{L^2(\Sigma_t)}^2 + \int_0^{T} \|\la x\ra^{-\f{r}{2}} \Box_g E_k \Db^{s''} \tphi \|_{L^2(\Sigma_\tau)}^2 \, \ud\tau \\
	\ls &\: \ep^2 + \int_0^{T} \|\la x\ra^{-\f{r}{2}} \Box_g E_k \Db^{s''} \tphi \|_{L^2(\Sigma_\tau)}^2 \, \ud\tau.
	\end{split}
	\end{equation*}
	Thus, in order to prove the bound \eqref{EtphiH3/2x.old}, it suffices to show 
	\begin{equation}\label{eq:commute.with.E.goal}
	\sup_{t\in [0,T_B)} \|\wo2 \Box_g E_k \Db^{s''} \tphi \|_{L^2(\Sigma_t)} \ls \ep + \ep^{\f 3 2}  \sum_{Z_k \in \{E_k,L_k\}}\|\rd Z_k \Db^{s''} \tphi \|_{L^2(\Sigma_t)}, 
	\end{equation}
	since we can then absorb the  $\ep^{\f 3 2}\|\rd E_k \Db^{s''} \tphi \|_{L^2(\Sigma_t)}$ terms into the left-hand side of \eqref{eq:commute.with.E.goal} using the smallness of $\ep$.
	
	To prove \eqref{eq:commute.with.E.goal}, we need to control (recall \blue{the} notation in \eqref{decompositionbox})
	\begin{equation}\label{eq:commute.with.E.strategy}
	\Box_g E_k \Db^{s''} \tphi = [\Box_g,E_k \Db^{s''}]\tphi = [\Box_g,E_k ]  \Db^{s''}\tphi + E_k [\Box^1, \Db^{s''}]\tphi + E_k [\Box^2, \Db^{s''}]\tphi.
	\end{equation}
	We further expand the last term in \eqref{eq:commute.with.E.strategy} using the product rule. First, we will introduce the notation \begin{equation} \label{def.barg}
	\barg:= (g^{-1})^{\nu\color{black}\bt}- {m}^{\nu\color{black}\bt},
	\end{equation} where $m$ is the Minkowski metric. Note moreover that 
	\begin{equation}\label{eq:g.and.gb}
	[\Box^2,\Db^{s''}]\tphi = [\gi^{\nu\color{black}\bt}\rd^2_{\nu\color{black}\bt}, \Db^{s''}] \tphi = [\barg \rd^2_{\nu\color{black}\bt}, \Db^{s''}] \tphi,\quad \partial_{\lambda} [\barg]=  \partial_{\lambda}[(g^{-1})^{\nu\color{black}\bt}].
	\end{equation} 
	
	Now we compute the last term in \eqref{eq:commute.with.E.strategy} using the product rule and \eqref{eq:g.and.gb} (recall here \eqref{def:Tres}):
	\begin{equation}\label{eq:commute.with.E.strategy.2}
	\begin{split}
	&\: E_k [\Box^2, \Db^{s''}]\tphi = E_k^i \rd_i [\Box^2, \Db^{s''}]\tphi\\
	= &\: E_k^i  (\rd_i(g^{-1})^{\nu\color{black}\bt})[\Db^{s''},\varpi]\rd^2_{\nu\color{black}\bt} \tphi + \varpi E_k^i  [(\rd_i(g^{-1})^{\nu\color{black}\bt})\rd^2_{\nu\color{black}\bt} \Db^{s''}\tphi] - E_k^i \Db^{s''} [\varpi (\rd_i(g^{-1})^{\nu\color{black}\bt})\rd^2_{\nu\color{black}\bt} \tphi] \\	
	&\: + E_k^i [\barg [\Db^{s''},\varpi]\rd^3_{i \nu\color{black}\bt} \tphi] + \varpi E_k^i [\barg\rd^3_{i\nu\color{black}\bt} \Db^{s''}\tphi]  - E_k^i \Db^{s''} [\varpi \barg \rd^3_{i \nu\color{black}\bt} \tphi] \\
	= &\: E_k^i  [\varpi(\rd_i(g^{-1})^{\nu\color{black}\bt})\rd^2_{\nu\color{black}\bt} \Db^{s''}\tphi] - E_k^i \Db^{s''} [\varpi(\rd_i(g^{-1})^{\nu\color{black}\bt})\rd^2_{\nu\color{black}\bt} \tphi] \\
	&\: + E_k^i  (\rd_i(g^{-1})^{\nu\color{black}\bt})[\Db^{s''},\varpi]\rd^2_{\nu\color{black}\bt} \tphi + E_k^i [ \barg [\Db^{s''},\varpi]\rd^3_{i \nu\color{black}\bt} \tphi]\\
	&\: - s'' \de^{jq} [(\rd_j (\varpi\barg)) E_k \rd^3_{q \nu\color{black}\bt} \Db^{s''-2}\tphi] - E_k^i T_{\mathrm{res}}^{s''} (\varpi\barg, \rd^3_{i \nu\color{black}\bt} \Db^{s''}\tphi).
	\end{split}
	\end{equation}

	Therefore, by \eqref{eq:commute.with.E.strategy} and \eqref{eq:commute.with.E.strategy.2}, in order to obtain \eqref{eq:commute.with.E.goal}, it suffices to prove
	\begin{equation}\label{eq:commute.with.E.1}
	\|\wo2 [\Box_g,E_k ]  \Db^{s''}\tphi\|_{L^2(\Sigma_t)}\ls \ep+ \ep^{\f 3 2}  \sum_{Z_k \in \{E_k,L_k\}}\|\rd Z_k \Db^{s''} \tphi \|_{L^2(\Sigma_t)}, 
	\end{equation}
	\begin{equation}\label{eq:commute.with.E.2}
	\|\wo2 E_k [\Box^1, \Db^{s''}]\tphi\|_{L^2(\Sigma_t)}\ls \ep, 
	\end{equation}
	\begin{equation}\label{eq:commute.with.E.weight.comm}
	\|\la x \ra^{-\f r2} E_k^i  (\rd_i(g^{-1})^{\nu\color{black}\bt})[\Db^{s''},\varpi]\rd^2_{\nu\color{black}\bt} \tphi \|_{L^2(\Sigma_t)} + \|\la x\ra^{-\f r2} E_k^i \barg [\Db^{s''},\varpi]\rd^3_{i\nu\color{black}\bt} \tphi\|_{L^2(\Sigma_t)} \ls \ep,
	\end{equation}
	\begin{equation}\label{eq:commute.with.E.3}
	\|\wo2 \{ E_k^i  [\varpi(\rd_i(g^{-1})^{nu\color{black}\bt})\rd^2_{\nu\color{black}\bt} \Db^{s''}\tphi] - E_k^i \Db^{s''} [\varpi(\rd_i(g^{-1})^{\nu\color{black}\bt})\rd^2_{\nu\color{black}\bt} \tphi] \} \|_{L^2(\Sigma_t)}\ls \ep, 
	\end{equation}
	\begin{equation}\label{eq:commute.with.E.4}
	\|\wo2 s'' \de^{jq} [(\rd_j (\varpi \barg) E_k \rd^3_{q\nu\color{black}\bt} \Db^{s''-2}\tphi] \|_{L^2(\Sigma_t)}\ls \ep+ \ep^{\f 3 2}\|\rd E_k \Db^{s''} \tphi \|_{L^2(\Sigma_t)}, 
	\end{equation}
	\begin{equation}\label{eq:commute.with.E.5}
	\|E_k^i T_{\mathrm{res}}^{s''} (\varpi \barg, \rd^3_{i\nu\color{black}\bt} \Db^{s''}\tphi) \|_{L^2(\Sigma_t)}\ls \ep.
	\end{equation}
	
	The above six estimates will respectively be proven in Propositions~\ref{prop:comm.ED.0}, \ref{prop:comm.ED.1}, \ref{prop:commute.with.E.weight.comm}, \ref{prop:commute.with.E.3}, \ref{prop:commute.with.E.4} and \ref{prop:commute.with.E.5} below, for $r \geq 2$. \qedhere
\end{proof}

\begin{proposition}\label{prop:comm.ED.0}
	For $r \geq 2$, the estimate \eqref{eq:commute.with.E.1} holds, i.e.
	$$\|\wo2 \left[\Box_g,E_k \right]  \Db^{s''}\tphi\|_{L^2(\Sigma_t)} \ls \ep^{\f 5 2} + \ep^{\f 3 2}  \sum_{Z_k \in \{E_k,L_k\}}\|\rd Z_k \Db^{s''} \tphi \|_{L^2(\Sigma_t)}.$$
\end{proposition}
\begin{proof}
	By Proposition~\ref{prop:commute.with.E}, we obtain  	
	$$\|\wo2 \left[\Box_g,E_k \right]  \Db^{s''}\tphi\|_{L^2(\Sigma_t)} \ls \ep^{\f 3 2} (\|\rd  \Db^{s''} \tphi \|_{L^2(\Sigma_t)} + \ep^{\f 3 2}  \sum_{Z_k \in \{E_k,L_k\}}\|\rd Z_k \Db^{s''} \tphi \|_{L^2(\Sigma_t)}).$$ 
	To conclude, we control the first term by \eqref{tphiH3/2x}. \qedhere
\end{proof}

\begin{proposition}\label{prop:comm.ED.1}
	Let $\Box^1:= \Gamma^\lambda \rd_\lambda$ as in \eqref{decompositionbox}. Then, for $r \geq 2$, 
	\begin{equation}\label{eq:comm.ED.1.main}
	\left\|\wo2 \rd_i \left[\Box^1, \Db^{s''}\right] \tphi \right\|_{L^2(\Sigma_t)} \ls \ep^{\f 94}.
	\end{equation}
	In particular, \eqref{eq:commute.with.E.2} holds.
\end{proposition}
\begin{proof}
	That \eqref{eq:commute.with.E.2} holds is immediate from \eqref{eq:comm.ED.1.main} and the estimate \eqref{eq:frame.1} for $E_k^i$. From now on we focus on \eqref{eq:comm.ED.1.main}.
	
	Using the product rule and the fact that $\varpi \tphi = \tphi$,
	\begin{equation*}
	\begin{split}
	& \rd_i \left[ \Db^{s''}, \Box^1\right] \tphi  =  \Db^{s''} ( \varpi \Gamma^{\lambda}  \partial_{\lambda} \partial_i \tphi) - \varpi \Gamma^{\lambda}  \partial_{\lambda} \Db ^{s''} \partial_i \tphi+ \Gamma^{\lambda}  [\varpi, \Db ^{s''}]\partial_{\lambda}  \partial_i \tphi \\ &\:+   \Db^{s''} ( \varpi (\partial_i \Gamma^{\lambda})  \partial_{\lambda} \tphi) - \varpi(\partial_i \Gamma^{\lambda}) \Db^{s''} \partial_{\lambda} \tphi+  (\partial_i \Gamma^{\lambda}) [\varpi,\Db^{s''}] \partial_{\lambda} \tphi\\ =  &\: \underbrace{ [\Db^{s''} , \varpi \Gamma^{\lambda} \partial_{\lambda}] \partial_i \tphi \color{black} }_{=:I} + \underbrace{\Gamma^{\lambda}  [\varpi, \Db ^{s''}]\partial_{\lambda}  \partial_i \tphi}_{=:II}
	 + \underbrace{ [\Db^{s''} , \varpi \partial_i \Gamma^{\lambda}]  \partial_{\lambda} \tphi \color{black}}_{=:III} + \underbrace{(\partial_i \Gamma^{\lambda}) [\varpi,\Db^{s''}] \partial_{\lambda} \tphi}_{=:IV}.
	\end{split}
	\end{equation*}
	
	By Corollary~\ref{cor:commute.2} we have
	\begin{equation}\label{eq:E.Box1.com.1}
	\|I \|_{L^2(\Sigma_t)} \lesssim  \| \varpi \Gamma \|_{W^{1,\f{2}{s'-s''}}(\Sigma_t)} \| \partial_\lambda \tphi \|_{ H^{s'}(\Sigma_t)} \lesssim \epsilon^{\frac 94}
	\end{equation} 
	using \eqref{tphiH3/2x} and the estimate \eqref{eq:Gamma}.
	
	For $II$, we use Lemma~\ref{lem:cutoff.commute}, H\"older's inequality and \eqref{eq:Gamma}, \eqref{energyglobal} to obtain
	\begin{equation}\label{eq:E.Box1.com.2} 
	\begin{split}
	\|\wo2 II\|_{L^2(\Sigma_t)} \ls &\: \|\wo2 \Gamma^\lambda\|_{L^\i(\Sigma_t)} \|[\Db^{s''}, \varpi] \rd^2_{i\lambda}\tphi\|_{L^2(\Sigma_t)} \\
	\ls &\: \ep^{\f 32} \|\Db^{s''-2} \rd^2_{i\lambda} \tphi\|_{L^2(\Sigma_t)} \ls \ep^{\f 32} \|\rd_\lambda \tphi\|_{L^2(\Sigma_t)} \ls \ep^{\f 5 2}.
	\end{split}
	\end{equation}
	
	For $III$, we apply the commutator estimate in Theorem~\ref{KatoPonce} with $p_1=\infty$, $p=p_2=2$ to obtain 
	\begin{equation}\label{eq:E.Box1.com.3}
	\|III \|_{L^2(\Sigma_t)} \lesssim \| \Db^{s''}(\varpi \rd_i\Gamma) \|_{L^2(\Sigma_t)} \cdot \| \partial \tphi \|_{ L^{\infty}(\Sigma_t)}  \lesssim \epsilon^{\frac{9}{4}},
	\end{equation} using \eqref{eq:Gamma} to control $\| \Db^{s''}(\varpi \rd_i\Gamma) \|_{L^2(\Sigma_t)}$ and the bootstrap assumption \eqref{BA:Li}.
	
	For $IV$, we use Lemma~\ref{lem:cutoff.commute}, Sobolev embedding ($H^{\f 32-s''}(\mathbb R^2) \hookrightarrow L^\i(\mathbb R^2)$) and H\"older's inequality to obtain
	\begin{equation}\label{eq:E.Box1.com.4}
	\begin{split}
	\|\wo2 IV\|_{L^2(\Sigma_t)} \ls &\: \|\wo2 \rd_i \Gamma^\lambda\|_{L^2(\Sigma_t)} \|[\Db^{s''}, \varpi]\rd\tphi\|_{L^\i(\Sigma_t)} \\
	\ls &\: \ep^{\f 32} \|[\Db^{s''}, \varpi]\rd \phi\|_{H^{\f 32-s''}(\Sigma_t)} \ls \ep^{\f 32} \|\rd\phi\|_{H^{-\f 12}(\Sigma_t)} \ls  \ep^{\f 32} \|\rd\phi\|_{L^2(\Sigma_t)} \ls \ep^{\f 52},
	\end{split}
	\end{equation} where in the second line we have used \eqref{eq:Gamma} and \eqref{energyglobal}. Combining \eqref{eq:E.Box1.com.1}--\eqref{eq:E.Box1.com.4} yields the proposition. \qedhere
	
\end{proof}

\begin{proposition}\label{prop:commute.with.E.weight.comm}
	For any indices $(\nu\color{black},\bt)$,
	\begin{equation}\label{eq:commute.with.E.weight.comm.reduced}
	\|[\Db^{s''},\varpi]\rd^3_{i\nu\color{black}\bt} \tphi\|_{L^2(\Sigma_t)} \ls \ep
	\end{equation}
	and 
	\begin{equation}\label{eq:commute.with.E.weight.comm.reduced.trivial}
	\|[\Db^{s''},\varpi]\rd^2_{\nu\color{black}\bt} \tphi\|_{L^2(\Sigma_t)} \ls \ep.
	\end{equation}
	As a consequence, \eqref{eq:commute.with.E.weight.comm} holds.
\end{proposition}
\begin{proof}
	Assuming \eqref{eq:commute.with.E.weight.comm.reduced} and \eqref{eq:commute.with.E.weight.comm.reduced.trivial}, it follows from H\"older's inequality, and the estimates \eqref{eq:frame.1}, \eqref{eq:g.main} that
	\begin{equation*}
	\begin{split}
	&\: \|\la x \ra^{-\f r2} E_k^i  (\rd_i(g^{-1})^{\nu\color{black}\bt})[\Db^{s''},\varpi]\rd^2_{\nu\color{black}\bt} \tphi \|_{L^2(\Sigma_t)} + \|\la x\ra^{-\f r2} E_k^i [\barg [\Db^{s''},\varpi]\rd^3_{i\nu\color{black}\bt} \tphi]\|_{L^2(\Sigma_t)} \\
	\ls &\: \|\la x\ra^{-\f r2} E_k^i (\rd_i(g^{-1})^{\nu\color{black}\bt})\|_{L^\i(\Sigma_t)} \|[\Db^{s''},\varpi]\rd^2_{\nu\color{black}\bt} \tphi\|_{L^2(\Sigma_t)}\\
	&\: + \|\la x\ra^{-\f r2} E_k^i \barg\|_{L^\i(\Sigma_t)} \|[\Db^{s''},\varpi]\rd^3_{i\nu\color{black}\bt} \tphi\|_{L^2(\Sigma_t)} \ls \ep,
	\end{split}
	\end{equation*}
	i.e.~\eqref{eq:commute.with.E.weight.comm} holds (recall the notation $\barg$ from \eqref{def.barg}).
	
	To obtain \eqref{eq:commute.with.E.weight.comm.reduced}, we note that by Lemma~\ref{lem:cutoff.commute}, the $L^2$-boundedness of $\rd_i \Db^{-1}$, Lemma~\ref{lem:invert.tt}, and Proposition~\ref{prop:Ds.est}, 
	$$\|[\Db^{s''},\varpi]\rd^3_{i\nu\color{black}\bt} \tphi\|_{L^2(\Sigma_t)} \ls  \|\Db^{s'-1} \rd^2_{\nu\color{black}\bt} \tphi\|_{L^2(\Sigma_t)} \ls \| \rd \Db^{s'} \tphi\|_{L^2(\Sigma_t)} \ls \ep.$$
	%\blue{If} $(\color{red}\nu\color{black},\bt) = (j,\lambda)$, then by boundedness of Riesz transform and Proposition~\ref{prop:Ds.est}, $\|\Db^{s''-1} \rd^2_{j\lambda} \tphi\|_{L^2(\Sigma_t)}\ls \ep$. If $(\color{red}\nu\color{black},\bt) = (t,t)$, then by Lemma~\ref{lem:invert.tt} and Proposition~\ref{prop:Ds.est}, $\|\Db^{s''-1} \rd^2_{tt} \tphi\|_{L^2(\Sigma_t)}\ls \ep$. In other words, we have
	%\begin{equation}\label{eq:s-1.alpbt}
 	%\|\Db^{s''-1} \rd^2_{\color{red}\nu\color{black}\bt} \tphi\|_{L^2(\Sigma_t)}\lesssim\color{red}	\|\Db^{s'-1} \rd^2_{\color{red}\nu\color{black}\bt} \tphi\|_{L^2(\Sigma_t)} \color{black}\ls \ep,
	%\end{equation}
    The estimate \eqref{eq:commute.with.E.weight.comm.reduced.trivial} is even simpler and can be obtained similarly\color{black}.  \qedhere
	\qedhere
\end{proof}

\begin{proposition}\label{prop:commute.with.E.3}
	For any index $(\nu\color{black},\bt)$ and any $f$ satisfying
	\begin{equation}\label{eq:commute.with.E.3.assumption}
	\|f\|_{W^{1,\frac{2}{s'-s''}}(\Sigma_t)} \ls 1,
	\end{equation}
	we have
	\begin{equation}\label{eq:commute.with.E.3.f}
	\|f \Db^{s''}\rd^2_{\nu\color{black}\bt} \tphi - \Db^{s''} (f \rd^2_{\nu\color{black}\bt}\tphi)\|_{L^2(\Sigma_t)} \ls \ep. 
	\end{equation}
	As a consequence, for any index $\sigma$, 
	\begin{equation}\label{eq:commute.with.E.3.prelim}
	\|\varpi(\rd_\sigma(g^{-1})^{\nu\color{black}\bt})\rd^2_{\nu\color{black}\bt} \Db^{s''}\tphi - \Db^{s''} [\varpi(\rd_\sigma(g^{-1})^{\nu\color{black}\bt})\rd^2_{\nu\color{black}\bt} \tphi] \|_{L^2(\Sigma_t)}\ls \ep.
	\end{equation}
	In particular, \eqref{eq:commute.with.E.3} holds.
\end{proposition}
\begin{proof}
	\pfstep{Step~1: Proof of \eqref{eq:commute.with.E.3.f}} By Corollary~\ref{cor:commute.2} and \eqref{eq:commute.with.E.3.assumption}, 
	\begin{equation}\label{eq:commute.with.E.3.main}
	\mbox{LHS of \eqref{eq:commute.with.E.3.f}} \ls \| f\|_{W^{1,\frac{2}{s'-s''}}(\Sigma_t)} \|\Db^{s'-1}\rd^2_{\nu\color{black}\bt} \tphi \|_{L^2(\Sigma_t)} \ls \|\Db^{s'-1}\rd^2_{\nu\color{black}\bt} \tphi \|_{L^2(\Sigma_t)}.
	\end{equation}
	The estimate \eqref{eq:commute.with.E.3.f} thus follows from \blue{Lemma~\ref{lem:invert.tt} and Proposition~\ref{prop:Ds.est}}. 
	
	\pfstep{Step~2: Proof of \eqref{eq:commute.with.E.3.prelim} and \eqref{eq:commute.with.E.3}} By \eqref{eq:commute.with.E.3.f}, to establish \eqref{eq:commute.with.E.3.prelim} requires only that $$\|\varpi(\rd_\sigma(g^{-1})^{\nu\color{black}\bt})\|_{W^{1,\frac{2}{s'-s''}}(\Sigma_t)} \ls 1,$$ which in turn follows from \eqref{eq:g.main}.
	
	Using \eqref{eq:commute.with.E.3.prelim}, H\"older's inequality and \eqref{eq:frame.1}, we obtain
	\begin{equation*}
	\begin{split}
	&\: \|\wo2 \{ E_k^i  [\varpi(\rd_i(g^{-1})^{\nu\color{black}\bt})\rd^2_{\nu\color{black}\bt} \Db^{s''}\tphi] - E_k^i \Db^{s''} [\varpi(\rd_i(g^{-1})^{\nu\color{black}\bt})\rd^2_{\nu\color{black}\bt} \tphi] \}\|_{L^2(\Sigma_t)} \\
	\ls &\: \|\wo2 E_k^i\|_{L^\i(\Sigma_t)} \|[\varpi(\rd_i(g^{-1})^{\nu\color{black}\bt})\rd^2_{\nu\color{black}\bt} \Db^{s''}\tphi] - E_k^i \Db^{s''} [\varpi(\rd_i(g^{-1})^{\nu\color{black}\bt})\rd^2_{\nu\color{black}\bt} \tphi] \|_{L^2(\Sigma_t)} \ls \ep,
	\end{split}
	\end{equation*}
	which establishes \eqref{eq:commute.with.E.3}. \qedhere
	\qedhere
\end{proof}

Next, we consider the term \eqref{eq:commute.with.E.4}; the main estimate will be obtained in Proposition~\ref{prop:commute.with.E.4} below. To ease the exposition, we prove an important but technically involved commutator estimate in the following lemma:
\begin{lem}\label{lem:Yt.aux}
	Let $Y^t$ be a smooth function satisfying 
	$$\|\la x \ra^{-1} Y^t\|_{L^\i(\Sigma_t)} \ls 1,\quad \|\varpi Y^t\|_{W^{1,\i}(\Sigma_t)\cap W^{2,4}(\Sigma_t)}\ls 1.$$ 
	Then 
	\begin{equation}\label{eq:Yt.aux.main.in.prop}
	\| \la x\ra^{-1} \{ Y^t R_q R_i \Db^{s''} (\bg^{i\lambda} \rd^2_{t\lambda} \tphi) - R_i R_q \Db^{s''} (\bg^{i\lambda} Y^t\rd^2_{t\lambda} \tphi) \} \|_{L^2(\Sigma_t)}\ls \ep.
	\end{equation}
\end{lem}

\begin{proof}
	%It suffices to separately prove
	%\begin{equation}\label{eq:Yt.aux.main.1}
	%\| \la x\ra^{-1} \{ Y^t R_q R_i \Db^{s''} (\bg^{ij} \rd^2_{jt} \tphi) - R_i R_q \Db^{s''} (\bg^{ij} Y^t\rd^2_{jt} \tphi) \}\|_{L^2(\Sigma_t)}\ls \ep,
	%\end{equation}
	%and
	%\begin{equation}\label{eq:Yt.aux.main.2}
	%\| \la x\ra^{-1} \{ Y^t R_q R_i \Db^{s''} (\bg^{it} \rd^2_{tt} \tphi) - R_i R_q \Db^{s''} (\bg^{it} Y^t\rd^2_{tt} \tphi) \}\|_{L^2(\Sigma_t)}\ls \ep.
	%\end{equation}
	
	\magenta{We consider separately the cases when we sum $\lambda$ over the spatial indices and when $\lambda = t$.}
	
	\pfstep{Step~1: \magenta{Summing $\lambda$ over spatial indices}} We compute\begin{equation}\label{eq:Y.commute.2.1.3}
	\begin{split}
	&\: Y^t R_q R_i \Db^{s''}(\bg^{ij} \rd^2_{j t} \tphi) - R_q R_i \Db^{s''}(\bg^{ij} Y^t \rd^2_{tj}\tphi)\\
	= &\: \underbrace{[Y^t \rd_j, R_q R_i \Db^{s''}](\bg^{ij} \rd_{t} \tphi)}_{=:I} \underbrace{- Y^t R_q R_i \Db^{s''}[(\rd_j\bg^{ij}) (\rd_{t} \tphi)]}_{=:II} + \underbrace{R_q R_i \Db^{s''}[Y^t (\rd_j \bg^{ij})(\rd_t\tphi)]}_{=:III}.
	\end{split}
	\end{equation}
	The term $\langle x \rangle^{-1} I$ is bounded in $L^2(\Sigma_t)$ by $\ep$ using \eqref{eq:commute.Riesz.phi.2} (with $h=Y^t$, $f=\bg^{ij}$), \eqref{eq:g.main} and Proposition~\ref{prop:Ds.est}. $\langle x \rangle^{-1} II$ and $\langle x \rangle^{-1} III$ are 
	both bounded in $L^2(\Sigma_t)$ by $\ep$ by the assumed estimates on $Y^t$, the $L^2$-boundedness of $R_i$, \eqref{eq:g.main}, Lemma~\ref{lem:stupid} and Proposition~\ref{prop:Ds.est}. Combining these observations give \magenta{\eqref{eq:Yt.aux.main.in.prop} when $\lambda$ is only summed over the spatial indices}.

	\pfstep{Step~2: \magenta{The $\lambda = t$ case}} It is notationally more convenient to prove more generally that for $f$ satisfying $\|\varpi f\|_{W^{2,\infty}(\Sigma_t)} \ls 1$, we have
	
	\begin{equation}\label{eq:Yt.aux.main.3}
	\| \la x\ra^{-1} \{ Y^t R_q R_i \Db^{s''} (f \rd^2_{tt} \tphi) - R_i R_q \Db^{s''} (f Y^t\rd^2_{tt} \tphi) \}\|_{L^2(\Sigma_t)}\ls \ep.
	\end{equation}
	
	First we compute, using the wave equation for $\tphi$ (see \eqref{eq:wave.bg}), that
	\begin{equation}\label{eq:Yt.aux.main}
	\begin{split}
	&\: - Y^t R_q R_i \Db^{s''} (f \rd^2_{tt} \tphi) \\
	= &\: Y^t \rd_q R_i \Db^{s''-1} (f \bg^{j\lambda} \rd^2_{j\lambda} \tphi + \f{f \Gamma^\lambda \rd_\lambda \tphi}{(g^{-1})^{tt}}) \\
	%= &\: Y^t \rd_q R_i R_j \Db^{s''} (f \bg^{j\lambda} \rd_{\lambda} \tphi) - Y^t R_q R_i \Db^{s''} [(\rd_j (f\bg^{j\lambda})) (\rd_{\lambda} \tphi)] + Y^t \rd_q R_i \Db^{s''-1} (\f{f \Gamma^\lambda \rd_\lambda \tphi}{(g^{-1})^{tt}}) \\
	= &\: \underbrace{[Y^t \rd_q, R_i R_j \Db^{s''}](f \bg^{j\lambda} \rd_{\lambda} \tphi)}_{=:I} +  \underbrace{R_i R_j \Db^{s''} [Y^t (\rd_q (f\bg^{j\lambda}) )(\rd_{\lambda} \tphi)]}_{=:II} +  \underbrace{ R_i R_j \Db^{s''} (Y^t f \bg^{j\lambda}  \rd^2_{q \lambda} \tphi)}_{=:III} \\
	&\: \underbrace{- Y^t R_q R_i \Db^{s''} [(\rd_j (f\bg^{j\lambda}))(\rd_{\lambda} \tphi)]}_{=:IV} + \underbrace{ Y^t R_q R_i \Db^{s''} (\f{f \Gamma^\lambda \rd_\lambda \tphi}{(g^{-1})^{tt}})}_{=:V},
	\end{split}
	\end{equation} where the term $Y^t \rd_q R_i \Db^{s''-1} (f \bg^{j\lambda} \rd^2_{j\lambda} \tphi)$ is computed in a way similar to \eqref{eq:Y.commute.2.1.3}. \color{black}
	The main term in \eqref{eq:Yt.aux.main} is $III$. \magenta{Indeed, t}erm $I$ can be controlled by \eqref{eq:commute.Riesz.phi.2} in Lemma~\ref{lem:commute.Riesz.phi} (combined with \eqref{eq:g.main}), while terms $II$, $IV$ and $V$ can be handled using Lemma~\ref{lem:stupid} (again combined with \eqref{eq:g.main}) so that
	\begin{equation}\label{eq:Yt.aux.1}
	\begin{split}
	\|I\|_{L^2(\Sigma_t)} + \|II\|_{L^2(\Sigma_t)} + \|IV\|_{L^2(\Sigma_t)} + \|V\|_{L^2(\Sigma_t)} \ls \ep^{\f 34}.
	\end{split}
	\end{equation}
	
	Finally, for the term $III$ in \eqref{eq:Yt.aux.main}, we shuffle the $\rd_j$ and $\rd_q$ derivatives and once again use the wave equation for $\tphi$ (recall again \eqref{eq:wave.bg}) to obtain
	\begin{equation}\label{eq:Yt.aux.2}
	\begin{split}
	III = &\: R_i R_j \Db^{s''} (Y^t f \bg^{j\lambda}  \rd^2_{q \lambda} \tphi) \\
	= &\: R_i R_q \Db^{s''} (f \bg^{j\lambda} Y^t\rd^2_{j\lambda} \tphi) + R_i R_q \Db^{s''} [(\rd_j (f Y^t \bg^{j\lambda}))(\rd_{\lambda} \tphi)] - R_i R_j \Db^{s''} [(\rd_q (f Y^t\bg^{j\lambda}))(\rd_{\lambda} \tphi)] \\
	= &\: \underbrace{- R_i R_q \Db^{s''} (f Y^t\rd^2_{tt} \tphi)}_{=:III_1} \underbrace{- R_i R_q \Db^{s''}(\f{f Y^t \Gamma^\lambda \rd_\lambda \tphi}{(g^{-1})^{tt}})}_{=:III_2}\\
	&\:+ \underbrace{R_i R_q \Db^{s''} [(\rd_j(f Y^t\bg^{j\lambda}))(\rd_{\lambda} \tphi)]}_{=:III_3} \underbrace{- R_i R_j \Db^{s''} [(\rd_q (f Y^t\bg^{j\lambda}))(\rd_{\lambda} \tphi)]}_{=:III_4}.
	\end{split}
	\end{equation}
	The main term here is $III_1$ (i.e.~it is included in the main term on the left-hand side of \eqref{eq:Yt.aux.main.3}). Using again Lemma~\ref{lem:stupid} (or obvious modifications), we have
	\begin{equation}\label{eq:Yt.aux.3}
	\begin{split}
	\|III_2\|_{L^2(\Sigma_t)} + \|III_3\|_{L^2(\Sigma_t)} + \|III_4\|_{L^2(\Sigma_t)} \ls \ep^{\f 34}.
	\end{split}
	\end{equation}
	
	Combining \eqref{eq:Yt.aux.main}--\eqref{eq:Yt.aux.3}, we obtain the desired estimate. \qedhere
\end{proof}

\begin{proposition}\label{prop:commute.with.E.4}
	Let $Y$ be a spacetime vector field satisfying 
	\begin{equation}\label{eq:Y.est.assume}
	\|\la x \ra^{-1} Y^{\nu\color{black}} \|_{L^\i(\Sigma_t)} + \| \rd_x Y^{\nu\color{black}}\|_{L^\i(\Sigma_t)} + \|\la x \ra^{-1} \partial_t Y^{\nu\color{black}} \|_{L^\i(\Sigma_t)} + \|\varpi Y^{\nu\color{black}} \|_{W^{1,\i}\cap W^{2,4}(\Sigma_t)} \ls 1. 
	\end{equation}
	
	Then for any spacetime index $(\nu\color{black},\bt)\neq (t,t)$ and $r\geq 2$, 
	\begin{equation}\label{eq:Y.commute.gen}
	\begin{split}
	\|\wo2 Y \Db^{s''-2} \rd^3_{q\nu\color{black}\bt} \tphi \|_{L^2(\Sigma_t)}
	\ls &\: \| \rd Y \Db^{s''} \tphi\|_{L^2(\Sigma_t)}  + \ep.
	\end{split}
	\end{equation}
	Also, 
	\begin{equation}\label{eq:Y.commute.gen.2}
	\begin{split}
	\|\wo2 \{ Y \Db^{s''-2} \rd^3_{q tt } \tphi + Y^t \rd_t \Db^{s''-2} \rd_q (\f{\Gamma^\lambda \rd_\lambda \tphi}{(g^{-1})^{tt}}) \}\|_{L^2(\Sigma_t)}
	\ls &\: \| \rd Y \Db^{s''} \tphi\|_{L^2(\Sigma_t)}  + \ep.
	\end{split}
	\end{equation}
	
	In particular, \eqref{eq:commute.with.E.4} holds. 
\end{proposition}
\begin{proof}
	Once we obtain \eqref{eq:Y.commute.gen}, \eqref{eq:commute.with.E.4} follows easily from the fact that $E_k$ has only spatial components, the estimates \eqref{eq:g.main}  and  \eqref{eq:frame.1}, \eqref{eq:frame.2}. More precisely, taking advantage of the compact support of $\varpi$, we \blue{obtain}\color{black} 
	\begin{equation*}
	\begin{split}
	&\: \|s'' \de^{jq} [(\rd_j (\varpi(\bar g^{-1})^{\nu\color{black}\bt})) E_k \rd^3_{q\nu\color{black}\bt} \Db^{s''-2}\tphi]\color{black} \|_{L^2(\Sigma_t)}   \\
	\ls &\:   \| \la x\ra^{\f r2} \rd_j (\varpi \barg)\|_{L^\i(\Sigma_t)} \|\wo2 E_k \rd^3_{q\nu\color{black}\bt} \Db^{s''-2}\tphi\color{black} \|_{L^2(\Sigma_t)} 
\\ 	\ls &\:\|\rd_j (\varpi \barg)\|_{L^\i(\Sigma_t)} \|\wo2 E_k \rd^3_{q\nu\color{black}\bt} \Db^{s''-2}\tphi\color{black} \|_{L^2(\Sigma_t)} \ls \ep^{\f 32} (\|\rd E_k \Db^{s''}\tphi\|_{L^2(\Sigma_t)} + \ep).
	\end{split}
	\end{equation*}
	
	In the remainder of the proof we focus on proving \eqref{eq:Y.commute.gen} and \eqref{eq:Y.commute.gen.2}: \eqref{eq:Y.commute.gen} will be proven in Steps~1--2 and \eqref{eq:Y.commute.gen.2} will be established in Step~3. The proof of \eqref{eq:Y.commute.gen} will be further split into two cases: the $(\nu\color{black},\bt) = (i,j)$ case will be treated in Step~1; the $(\nu\color{black},\bt) = (i,t)$ case will be treated in Step~2.
	
	\pfstep{Step~1: Proof of \eqref{eq:Y.commute.gen} when $(\nu\color{black},\bt) = (i,j)$} In this case,
	\begin{equation*}
	\begin{split}
	Y \Db^{s''-2} \rd^3_{q i j} \tphi = &\: Y^\sigma \rd_q R_i R_j \rd_\sigma \Db^{s''}  \tphi \\
	= &\: [Y^\sigma \rd_q, R_i R_j ] \rd_\sigma \Db^{s''}  \tphi - R_i R_j [(\rd_q Y^\sigma)(\rd_\sigma \Db^{s''} \tphi)] + R_i R_j (\rd_q (Y  \Db^{s''} \tphi)).
	\end{split}
	\end{equation*}
	
	Using Lemma~\ref{lem:commute.Riesz.phi} for the first term, and using the $L^2$-boundedness of $R_i$ (and H\"older's inequality) for the second and third terms, we thus obtain
	\begin{equation*}
	\begin{split}
	&\: \|\la x\ra^{-\f r2} Y \Db^{s''-2} \rd^3_{q i j} \tphi\|_{L^2(\Sigma_t)} \\
	\ls &\:  \|\rd \Db^{s''} \tphi\|_{L^2(\Sigma_t)} + \|\rd_\ell Y^\sigma \|_{L^\i(\Sigma_t)} \|\rd \Db^{s''} \tphi\|_{L^2(\Sigma_t)} + \|\rd Y\Db^{s''} \tphi\|_{L^2(\Sigma_t)} \\
	\ls &\:  \|\rd Y\Db^{s''} \tphi\|_{L^2(\Sigma_t)} + \ep,
	\end{split}
	\end{equation*}
	where we have used \eqref{eq:Y.est.assume} and Proposition~\ref{prop:Ds.est}.

	\pfstep{Step~2: Proof of \eqref{eq:Y.commute.gen} when $(\nu\color{black},\bt) = (i,t)$} Decompose $Y = Y^\ell \rd_\ell + Y^t \rd_t$ and commute $Y^\ell \rd_\ell$ with $R_q R_i \Db^{s''}$. We then obtain
	\begin{equation}\label{eq:commute.Y.1}
	\begin{split}
	Y \Db^{s''-2} \rd^3_{q i t} \tphi %= &\: Y^\ell \rd_\ell R_q R_i \Db^{s''} \rd_t \tphi + Y^t \rd_q R_i \Db^{s''-1} \rd^2_{tt} \tphi \\
	= &\: [Y^\ell \rd_\ell, R_q R_i] \Db^{s''} \rd_t \tphi + R_q R_i  (Y^\ell \rd_\ell \Db^{s''} \rd_t \tphi) + Y^t \rd_q R_i \Db^{s''-1} \rd^2_{tt} \tphi.
	\end{split}
	\end{equation}
	Rearranging, this implies
	\begin{equation}\label{eq:commute.Y.2}
	\begin{split}
	&\: Y \Db^{s''-2} \rd^3_{q i t} \tphi - R_q R_i (Y \Db^{s''} \rd_t \tphi) \\
	= &\: \underbrace{[Y^\ell \rd_\ell, R_q R_i] \Db^{s''} \rd_t \tphi}_{=:I} + \underbrace{(Y^t R_q R_i \Db^{s''} \rd^2_{tt} \tphi -  R_q R_i (Y^t \Db^{s''} \rd^2_{tt} \tphi))}_{=:II}.
	\end{split}
	\end{equation}
	
	We begin with the main term on the left-hand side of \eqref{eq:commute.Y.2}. We compute
	$$ Y \Db^{s''} \rd_t \tphi = \rd_t (Y \Db^{s''}\tphi) - \varpi (\rd_t Y^{\nu\color{black}}) \rd_{\nu\color{black}} \Db^{s''}   \tphi - (\rd_t Y^{\nu\color{black}}) \rd_{\nu\color{black}} [\Db^{s''}, \varpi]  \tphi - (\rd_t Y^{\nu\color{black}}) (\rd_{\nu\color{black}} \varpi) (\Db^{s''} \tphi).$$
	Hence, by $L^2$-boundedness of $R_qR_i$, H\"older's inequality, \eqref{eq:Y.est.assume}, \eqref{energyglobal},  Lemma~\ref{lem:commute.weight.2}, \eqref{eq:est.for.Ds''} and Proposition~\ref{prop:Ds.est},
	\begin{equation}\label{eq:commute.Y.main.manipulation}
	\begin{split}
	&\: \| R_q R_i (Y \Db^{s''} \rd_t \tphi)\|_{L^2(\Sigma_t)} \\
	\ls &\: \|\rd_t Y \Db^{s''}  \tphi\|_{L^2(\Sigma_t)} + \|\varpi \rd_t Y^{\nu\color{black}} \|_{L^\i(\Sigma_t)} \|\rd_{\nu\color{black}} \Db^{s''} \tphi\|_{L^2(\Sigma_t)} \\
	&\: + \|\la x\ra^{-1} \rd_t Y^{\nu\color{black}}\|_{L^\i(\Sigma_t)} (\|\la x\ra \rd [\Db^{s''}, \varpi] \tphi \|_{L^2(\Sigma_t)} + \|\la x \ra \rd_{\nu\color{black}} \varpi\|_{L^\i(\Sigma_t)} \|\Db^{s''} \tphi\|_{L^2(\Sigma_t)}) \\
	\ls &\:  \|\rd_t Y \Db^{s''}  \tphi\|_{L^2(\Sigma_t)} + \ep.
	\end{split}
	\end{equation}
	
	By Lemma~\ref{lem:commute.Riesz.phi} and Proposition~\ref{prop:Ds.est}, the term $I$ in \eqref{eq:commute.Y.2} can be bounded as follows, 
	\begin{equation}\label{eq:commute.Y.3}
	\|\langle x \rangle^{-\frac{r}{2}}I\|_{L^2(\Sigma_t)} \ls \|\Db^{s''} \rd_t \tphi\|_{L^2(\Sigma_t)} \ls \ep.
	\end{equation}
	
	For the commutator term $II$, since in general $[Y^t, R_q R_i]$ is only bounded $L^2(\Sigma_t)\to L^2(\Sigma_t)$ (instead of gaining one derivative), we need to use the wave equation for $\tphi$ and then exploit the gain given by Lemma~\ref{lem:commute.Riesz.phi}. More precisely, we compute using \eqref{eq:wave.bg} and $\rd_\lambda \tphi = \varpi \rd_\lambda \tphi$ that
	\begin{equation}\label{eq:commute.Y.4}
	\begin{split}
	&\: Y^t R_q R_i \Db^{s''} \rd^2_{tt} \tphi  \\
	= &\: - Y^t R_q R_i \Db^{s''} (\bg^{j\lambda} \rd^2_{j\lambda} \tphi) - Y^t R_q R_i \Db^{s''} ( \f{\Gamma^\lambda \rd_\lambda \tphi}{(g^{-1})^{tt}} ) \\
	%= &\: - Y^t \rd_j R_q R_i \Db^{s''} (\bg^{j\lambda} \rd_{\lambda} \tphi) + Y^t R_q R_i \Db^{s''} [(\rd_j \bg^{j\lambda}) (\rd_{\lambda} \tphi)]  - Y^t R_q R_i \Db^{s''} \f{\Gamma^\lambda \rd_\lambda \tphi}{(g^{-1})^{tt}} \\
	= &\: - R_q R_i Y^t \Db^{s''} (\bg^{j\lambda}\rd^2_{j\lambda} \tphi) - [Y^t \rd_j, R_q R_i] \Db^{s''} (\bg^{j\lambda} \rd_{\lambda} \tphi)  - R_q R_i Y^t \Db^{s''} [(\rd_j \bg^{j\lambda}) (\rd_{\lambda} \tphi)] \\
	&\: + Y^t R_q R_i \Db^{s''} [(\rd_j \bg^{j\lambda}) (\rd_{\lambda} \tphi)]  - Y^t R_q R_i \Db^{s''} ( \f{\Gamma^\lambda \rd_\lambda \tphi}{(g^{-1})^{tt}} ) \\
	= &\: R_q R_i Y^t  \Db^{s''} \rd^2_{tt} \tphi \underbrace{- [Y^t \rd_j, R_q R_i] \Db^{s''} (\bg^{j\lambda} \rd_{\lambda} \tphi)}_{=:II_1}  \underbrace{- R_q R_i \{ \varpi Y^t \Db^{s''} [(\rd_j \bg^{j\lambda}) (\rd_{\lambda} \tphi)] \} }_{=:II_2} \\
	&\: -  \underbrace{ R_q R_i \{ Y^t [\Db^{s''},\varpi] ((\rd_j \bg^{j\lambda}) (\rd_{\lambda} \tphi)) \}}_{=:II_3} + \underbrace{Y^t R_q R_i \Db^{s''} [(\rd_j \bg^{j\lambda}) (\rd_{\lambda} \tphi)]}_{=:II_4}  \underbrace{- Y^t R_q R_i \Db^{s''} ( \f{\Gamma^\lambda \rd_\lambda \tphi}{(g^{-1})^{tt}} )}_{=:II_5} \\
	&\: + \underbrace{R_q R_i [\varpi Y^t \Db^{s''} ( \f{\Gamma^\lambda \rd_\lambda \tphi}{(g^{-1})^{tt}}] )}_{=:II_6} 
	+  \underbrace{R_q R_i Y^t [\Db^{s''},\varpi] ( \f{\Gamma^\lambda \rd_\lambda \tphi}{(g^{-1})^{tt}} )}_{=:II_7}.
	\end{split}
	\end{equation}
	
	The first term on the right-hand side of \eqref{eq:commute.Y.4} is the main term (recall again term $II$ in \eqref{eq:commute.Y.2}). We will bound all the other terms. First, $II_1$ can be bounded by Lemmas~\ref{lem:commute.Riesz.phi}, and the estimates \eqref{eq:g.main} and Proposition \ref{prop:Ds.est},
	\begin{equation}\label{eq:commute.Y.II.1}
	\|II_1\|_{L^2(\Sigma_t)} \ls \|\rd \Db^{s'} \tphi\|_{L^2(\Sigma_t)} \ls \ep.
	\end{equation}
	The terms $II_2$, $II_4$, $II_5$ and $II_6$ are easier: We use the $L^2$-boundedness of $R_i$, H\"older's inequality, the assumption \eqref{eq:Y.est.assume}, Lemma~\ref{lem:stupid.2} and Proposition~\ref{prop:Ds.est} to obtain
	\begin{equation}\label{eq:commute.Y.II.2}
	\|II_2\|_{L^2(\Sigma_t)} +\|\wo2 II_4\|_{L^2(\Sigma_t)} + \|\wo2 II_5\|_{L^2(\Sigma_t)} + \|II_6\|_{L^2(\Sigma_t)} \ls \ep.
	\end{equation}
	For $II_3$ and $II_7$, we use the $L^2$-boundedness of $R_i$, H\"older's inequality, the assumption \eqref{eq:Y.est.assume}, Lemma~\ref{lem:commute.weight.1}, \eqref{energyglobal}, \eqref{eq:g.main}, \eqref{eq:Gamma} to obtain
	\begin{equation}\label{eq:commute.Y.II.3}
	\begin{split}	&\: \|II_3\|_{L^2(\Sigma_t)} + \|II_7\|_{L^2(\Sigma_t)} \\
	\ls &\: \|\la x \ra^{-1} Y^t \|_{L^{\infty}(\Sigma_t)} (\| \la x\ra [\Db^{s''},\varpi] ((\rd_j \bg^{j\lambda}) (\rd_{\lambda} \tphi)) \|_{L^2(\Sigma_t)} + \| \la x\ra [\Db^{s''},\varpi] (\f{\Gamma^\lambda \rd_\lambda \tphi}{(g^{-1})^{tt}} )\|_{L^2(\Sigma_t)}) \\
	\ls &\: \| (\rd_j \bg^{j\lambda}) (\rd_{\lambda} \tphi) \|_{L^2(\Sigma_t)} + \| \f{\Gamma^\lambda \rd_\lambda \tphi}{(g^{-1})^{tt}}\|_{L^2(\Sigma_t)} \ls \epsilon.
	\end{split}
	\end{equation}
	Plugging in the estimates \eqref{eq:commute.Y.II.1}--\eqref{eq:commute.Y.II.3} into \eqref{eq:commute.Y.4}, we obtain
	\begin{equation}\label{eq:commute.Y.8}
	\begin{split}
	\|II\|_{L^2(\Sigma_t)} = \| Y^t R_q R_i \Db^{s''} \rd^2_{tt} \tphi -  R_q R_i Y^t \Db^{s''} \rd^2_{tt} \tphi\|_{L^2(\Sigma_t)} \ls \ep.
	\end{split}
	\end{equation}
	
	Combining \eqref{eq:commute.Y.main.manipulation}, \eqref{eq:commute.Y.3} and \eqref{eq:commute.Y.8}, and returning to \eqref{eq:commute.Y.2}, we thus obtain
	\begin{equation*}
	\begin{split}
	\|Y \Db^{s''-2} \rd^3_{q i t} \tphi\|_{L^2(\Sigma_t)}  \ls &\: \| R_q R_i (Y \Db^{s''} \rd_t \tphi)\|_{L^2(\Sigma_t)} + \|I\|_{L^2(\Sigma_t)} + \|II\|_{L^2(\Sigma_t)} \\\ls &\: \|\rd Y\Db^{s''} \tphi\|_{L^2(\Sigma_t)} + \ep,
	\end{split}
	\end{equation*}
	as desired.
	
	\pfstep{Step~3: Proof of \eqref{eq:Y.commute.gen.2}} We begin with an application of the wave equation (recall \eqref{eq:wave.bg}):
	\begin{equation}\label{eq:Y.commute.tt}
	\begin{split}
	Y\Db^{s''-2} \partial^{3}_{q t t} \tphi 
	%= &\: -Y\Db^{s''-1} R_q \f{1}{(g^{-1})^{tt}}[(g^{-1})^{ij} \rd^2_{ij} \tphi + 2(g^{-1})^{ti} \rd^2_{it}\tphi + \Gamma^\lambda \rd_\lambda \tphi] \\
	= &\: \underbrace{-Y R_q R_i \Db^{s''}(\bg^{i\lambda} \rd_{\lambda} \tphi)}_{=:I} + \underbrace{Y R_q \Db^{s''-1} [ (\rd_i \bg^{i\lambda}) (\rd_\lambda\tphi)]}_{II} \underbrace{- Y^\ell R_\ell R_q \Db^{s''} (\f{\Gamma^\lambda \rd_\lambda \tphi}{(g^{-1})^{tt}})}_{III} \\
	&\: - Y^t \rd_t R_q \Db^{s''-1} (\f{\Gamma^\lambda \rd_\lambda \tphi}{(g^{-1})^{tt}}).
	\end{split}
	\end{equation}
	
	Note that the last term $- Y^t \rd_t R_q \Db^{s''-1} (\f{\Gamma^\lambda \rd_\lambda \tphi}{(g^{-1})^{tt}})$ is present in the statement of \eqref{eq:Y.commute.gen.2}. It is therefore sufficient to control each of $I$, $II$ and $III$ in \eqref{eq:Y.commute.tt}. This will be carried out in Steps~3(a)--3(c) below.
	
	\pfstep{Step~3(a): Term $I$ in \eqref{eq:Y.commute.tt}} As in Step~2, we write $Y = Y^\ell\rd_\ell + Y^t \rd_t$. We compute
	\begin{equation}\label{eq:Y.commute.3.1.1}
	\begin{split}
	&\: Y R_q R_i \Db^{s''}(\bg^{i\lambda} \rd_{\lambda} \tphi) \\
	= &\: \underbrace{[Y^\ell\rd_\ell, R_q R_i \Db^{s''}](\bg^{i\lambda} \rd_{\lambda} \tphi)}_{=:I_1} + \underbrace{R_q R_i \Db^{s''}[Y^\ell (\rd_\ell\bg^{i\lambda}) (\rd_{\lambda} \tphi)]}_{=:I_2} + \underbrace{R_q R_i \Db^{s''}[ \bg^{i\lambda} (Y^\ell \rd_\ell \rd_{\lambda} \tphi)]}_{=:I_3}\\
	&\: + \underbrace{Y^t R_q R_i \Db^{s''}[(\rd_t\bg^{i\lambda}) (\rd_{\lambda} \tphi)]}_{=:I_4} + \underbrace{Y^t R_q R_i \Db^{s''}(\bg^{i\lambda} \rd^2_{\lambda t} \tphi)}_{=:I_5}.
	\end{split}
	\end{equation}
	$I_1$, $I_2$, $I_4$ are easier error terms. Indeed, using Lemma~\ref{lem:commute.Riesz.phi}, \eqref{eq:Y.est.assume} and \eqref{eq:g.main}  for $I_1$, using Lemma~\ref{lem:stupid.2} for $I_2$ and $I_4$, and then applying Proposition~\ref{prop:Ds.est}, we have the estimates
	\begin{equation}\label{eq:Y.commute.3.1.2}
	\|I_1\|_{L^2(\Sigma_t)} + \|I_2\|_{L^2(\Sigma_t)} + \|I_4\|_{L^2(\Sigma_t)} \ls \ep,
	\end{equation}
	we skip the details.
	
	We will combine $I_3$ and $I_5$ in \eqref{eq:Y.commute.3.1.1}. First, by Lemma~\ref{lem:Yt.aux}, 
	\begin{equation}\label{eq:Y.commute.3.1.3}
	\|\langle x\rangle^{-1} (I_5 - R_q R_i \Db^{s''}(\bg^{i\lambda} Y^t \rd^2_{t\lambda}\tphi) )\|_{L^2(\Sigma_t)} \ls \ep.
	\end{equation}
	Therefore, recalling the definition of $I_3$ and $I_5$ in \eqref{eq:Y.commute.3.1.1} and then using \eqref{eq:Y.commute.3.1.3}, we obtain
	\begin{equation}\label{eq:Y.commute.3.1.4}
	\| I_3 + I_5 - R_q R_i \Db^{s''}[ \bg^{i\lambda} (Y \rd_{\lambda} \tphi)]\|_{L^2(\Sigma_t)} \ls \ep.
	\end{equation}
	
	We next estimate the term $R_q R_i \Db^{s''}[ \bg^{i\lambda} (Y \rd_{\lambda} \tphi)]$ (appearing in \eqref{eq:Y.commute.3.1.4}). By the $L^2$-boundedness of $R_q R_i$ and $Y \rd_\lambda \tphi = \varpi Y \rd_\lambda \tphi$, Proposition~\ref{prop:commute.2} and \eqref{eq:g.main}
	\begin{equation}\label{eq:Y.commute.3.1.5}
	\begin{split}
	&\: \| R_q R_i \Db^{s''}[ \bg^{i\lambda} (Y \rd_{\lambda} \tphi)] \|_{L^2(\Sigma_t)} \ls  \| \Db^{s''}[ \varpi \bg^{i\lambda} (Y \rd_{\lambda} \tphi)]\|_{L^2(\Sigma_t)} \\
	\ls &\: \| \varpi \bg^{i\lambda} \|_{L^\i(\Sigma_t)} \|\Db^{s''} Y \rd_\lambda \tphi\|_{L^2(\Sigma_t)} + \| \varpi \bg^{i\lambda} \|_{W^{1,\i}(\Sigma_t)} \|\Db^{s''-1} Y \rd_\lambda \tphi \|_{L^2(\Sigma_t)} \\
	\ls &\: \|\Db^{s''} Y \rd_\lambda \tphi\|_{L^2(\Sigma_t)} + \|\Db^{s''-1} Y^{\nu\color{black}} \rd^2_{\nu\color{black}\lambda} \tphi \|_{L^2(\Sigma_t)}
	\end{split}
	\end{equation}
	We then bound each term on the right-hand side of \eqref{eq:Y.commute.3.1.5}. For the first term, we use $\rd^2_{\nu\color{black}\lambda}\tphi = \varpi \rd^2_{\nu\color{black}\lambda}\tphi$, Proposition~\ref{prop:commute.2}, \eqref{eq:Y.est.assume}, Lemma~\ref{lem:invert.tt} and Proposition~\ref{prop:Ds.est} to obtain
	\begin{equation}\label{eq:Y.commute.3.1.6.1}
	\begin{split}
	&\: \|\Db^{s''} (Y^{\nu\color{black}} \rd^2_{\nu\color{black} \lambda} \tphi) \|_{L^2(\Sigma_t)} = \|\Db^{s''} (\varpi Y^{\nu\color{black}} \rd^2_{\nu\color{black} \lambda} \tphi) \|_{L^2(\Sigma_t)} \\
	\ls &\: \|\varpi Y^{\nu\color{black}} \rd^2_{\nu\color{black} \lambda} \Db^{s''} \tphi \|_{L^2(\Sigma_t)} + \| \varpi Y^{\nu\color{black}}\|_{W^{1,\infty}(\Sigma_t)} \|\Db^{s''-1} \rd^2_{\nu\color{black}\lambda} \tphi \|_{L^2(\Sigma_t)} \\
	\ls &\: \|\varpi  \rd_\lambda (Y \Db^{s''} \tphi) \|_{L^2(\Sigma_t)} +  \|  \rd_\lambda (\varpi Y^{\nu\color{black}}) \|_{L^\i(\Sigma_t)} \| \rd_{\nu\color{black}} \Db^{s''} \tphi  \|_{L^2(\Sigma_t)} + \|\Db^{s''-1} \rd^2_{\nu\color{black}\lambda} \tphi \|_{L^2(\Sigma_t)} \\
	\ls &\: \| \rd Y \Db^{s''} \tphi \|_{L^2(\Sigma_t)} + \ep.
	\end{split}
	\end{equation}
	For the second term on the right-hand side of \eqref{eq:Y.commute.3.1.5}, we directly use Lemma~\ref{lem:invert.tt} and Proposition~\ref{prop:Ds.est} to obtain
	\begin{equation}\label{eq:Y.commute.3.1.6.2}
	\begin{split}
	\|\Db^{s''-1} Y^{\nu\color{black}} \rd^2_{\nu\color{black}\lambda} \tphi \|_{L^2(\Sigma_t)} \ls \|\rd \Db^{s'} \tphi \|_{L^2(\Sigma_t)} \ls \ep.
	\end{split}
	\end{equation}
	
	Plugging \eqref{eq:Y.commute.3.1.6.1} and \eqref{eq:Y.commute.3.1.6.2} into \eqref{eq:Y.commute.3.1.5}, we thus obtain
	\begin{equation}\label{eq:Y.commute.3.1.7}
	\|R_q R_i \Db^{s''}[ \bg^{i\lambda} (Y \rd_{\lambda} \tphi)] \|_{L^2(\Sigma_t)} \ls \| \rd Y \Db^{s''} \tphi \|_{L^2(\Sigma_t)} + \ep.
	\end{equation}
	
	Finally, we combine \eqref{eq:Y.commute.3.1.2}, \eqref{eq:Y.commute.3.1.4} and \eqref{eq:Y.commute.3.1.7} to obtain
	\begin{equation}\label{eq:Y.commute.3.1.8}
	\begin{split}
	\| I \|_{L^2(\Sigma_t)} \ls &\:  \|R_q R_i \Db^{s''}[ \bg^{i\lambda} (Y \rd_{\lambda} \tphi)]\|_{L^2(\Sigma_t)} + \ep \ls \|\rd Y \Db^{s''} \tphi\|_{L^2(\Sigma_t)} + \ep.
	\end{split}
	\end{equation}
	
	\pfstep{Step~3(b): Term $II$ in \eqref{eq:Y.commute.tt}} We again write $Y = Y^\ell \rd_\ell + Y^t \rd_t$ and expand as follows:
	\begin{equation}\label{eq:Y.commute.tt.II}
	\begin{split}
	II= &\: Y R_q \Db^{s''-1} [ (\rd_i \bg^{i\lambda}) (\rd_\lambda\tphi)] \\
	=  &\: \underbrace{Y^\ell R_q R_\ell \Db^{s''}[ (\rd_i \bg^{i\lambda}) (\rd_\lambda\tphi)]}_{=:II_1} + \underbrace{Y^t R_q \Db^{s''-1} [ (\rd^2_{it} \bg^{i\lambda}) (\rd_\lambda\tphi)] }_{=:II_2} + \underbrace{ Y^t R_q \Db^{s''-1} [ (\rd_i \bg^{i\lambda}) (\rd^2_{t\lambda}\tphi)] }_{=:II_3}.
	\end{split}
	\end{equation}
	
	The term $II_1$ can be directly handled by \eqref{eq:Y.est.assume}, Lemma~\ref{lem:stupid.2} and Proposition~\ref{prop:Ds.est} so that
	\begin{equation}\label{eq:Y.commute.3.2.1}
	\|\wo2 II_1\|_{L^2(\Sigma_t)}  \ls  \ep.
	\end{equation} 
	
	$II_2$ is even easier: we use the $L^2$-boundedness of $R_q \Db^{s''-1}$, \eqref{eq:Y.est.assume}, \eqref{eq:g.main} and the bootstrap assumption \eqref{BA:Li} to obtain
	\begin{align*}
	\| \wo2 II_2\|_{L^2(\Sigma_t)}  \ls &  \| \langle x \rangle^{-1} Y^t\|_{L^\i(\Sigma_t)} \| (\rd^2_{it} \bg^{i\lambda}) (\rd_\lambda\tphi) \|_{L^2(\Sigma_t)}\\ \ls & \| (\rd^2_{it} \bg^{i\lambda}) (\rd_\lambda\tphi)\|_{L^2(\Sigma_t)} \ls \| \rd^2_{it} \bg^{i\lambda} \|_{L^2(\Sigma_t)} \|  \rd\tphi\|_{L^\i(\Sigma_t)} \ls \ep^{\f 9 4}.
	\end{align*}

	For the term $II_3$ in \eqref{eq:Y.commute.tt.II}, we use \eqref{eq:Y.est.assume}, Lemma~\ref{lem:invert.tt}, \eqref{eq:g.main} and Proposition~\ref{prop:Ds.est} to obtain	
	\begin{equation}\label{eq:Y.commute.3.2.2}
	\begin{split}
	\|\wo2 II_{3}\|_{L^2(\Sigma_t)} \ls \ep.
	\end{split}
	\end{equation}
	
	Finally, combining \eqref{eq:Y.commute.3.2.1} and \eqref{eq:Y.commute.3.2.2} yields
	\begin{equation}\label{eq:Y.commute.3.2.3}
	\|\wo2 II\|_{L^2(\Sigma_t)}\ls \ep.
	\end{equation}
	
	\pfstep{Step~3(c): Term $III$ in \eqref{eq:Y.commute.tt}} The very final term $III$ in \eqref{eq:Y.commute.tt} is simple. Indeed, by H\"older's inequality, \eqref{eq:Y.est.assume}, Lemma~\ref{lem:stupid.2} and Proposition~\ref{prop:Ds.est},
	\begin{equation}\label{eq:Y.commute.3.3}
	\begin{split}
	\|\wo2 III\|_{L^2(\Sigma_t)} \ls \|\langle x \rangle^{-1} Y^t\|_{L^\i(\Sigma_t)} \|\Db^{s''} (\f{\Gamma^\lambda \rd_\lambda \tphi}{(g^{-1})^{tt}})\|_{L^2(\Sigma_t)} \ls \ep.
	\end{split}
	\end{equation}
	
	Finally, we plug the estimates \eqref{eq:Y.commute.3.1.8}, \eqref{eq:Y.commute.3.2.3} and \eqref{eq:Y.commute.3.3} into \eqref{eq:Y.commute.tt} to obtain \eqref{eq:Y.commute.gen.2}.	\qedhere
\end{proof}

\begin{proposition}\label{prop:commute.with.E.5}
	When $(\nu\color{black},\bt,\sigma)\neq (t,t,t)$, 

	\begin{equation}\label{eq:the.long.term}
	\| T_{\mathrm{res}}^{s''} (\varpi \barg, \rd^3_{\sigma\nu\color{black}\bt} \Db^{s''}\tphi) \|_{L^2(\Sigma_t)} \ls \ep^{\f 52},
	\end{equation} where we recall the notation $\barg$ from \eqref{def.barg}.	
	
	In particular, \eqref{eq:commute.with.E.5} holds.
\end{proposition}
\begin{proof}
	That \eqref{eq:commute.with.E.5} holds is immediate from \eqref{eq:the.long.term} and the fact that $E_k$ does not have a $t$ component. From now on we focus on the proof of \eqref{eq:the.long.term}.
	
	By Corollary~\ref{cor:commute.3} and \eqref{eq:g.main}, the left-hand side of \eqref{eq:the.long.term} is bounded by 
	$$\mbox{LHS of \eqref{eq:the.long.term}} \ls \|\varpi \barg \|_{W^{2,\f 2{s'-s''}}(\Sigma_t)} \|\Db^{s'-2} \rd^3_{\sigma\nu\color{black}\bt}\tphi \|_{L^2(\Sigma_t)} \ls \ep^{\f 32} \|\Db^{s'-2} \rd^3_{\sigma\nu\color{black}\bt}\tphi\|_{L^2(\Sigma_t)}.$$
	
	Note that if two of $\sigma$, $\nu\color{black}$, $\bt$ are spatial, we use Proposition~\ref{prop:Ds.est} to get
	$$\|\Db^{s'-2} \rd^3_{\sigma\nu\color{black}\bt}\tphi \|_{L^2(\Sigma_t)} \ls \|\rd \Db^{s'} \tphi\|_{L^2(\Sigma_t)}\ls \ep,$$
	which gives the desired estimate.
	
	The only remaining case to consider (after relabeling) is $(\nu\color{black},\bt,\sigma) = (t, t, i)$. By the $L^2$-boundedness of $\Db^{-1} \rd_i$, Lemma~\ref{lem:invert.tt} and Proposition~\ref{prop:Ds.est},
	$$\|\Db^{s'-2} \rd^3_{itt}\tphi\|_{L^2(\Sigma_t)} \ls \|\Db^{s'-1} \rd^2_{tt}\tphi\|_{L^2(\Sigma_t)} \ls \ep,$$
	which again gives the desired estimate. \qedhere
	
\end{proof}

\subsection{Estimates for $\rd L_k \Db^{s''} \tphi$} \label{LtphiH3/2xsection}

As in Section~\ref{EtphiH3/2xsection}, we begin with a high level proof, leaving the main estimates in later propositions.

\begin{proposition}  \label{prop:LkEkDsphi.main}\begin{equation} \label{LtphiH3/2x}
	\|\rd L_k \Db^{s''} \tphi \|_{L^2(\Sigma_t)} \ls \ep. 
	\end{equation} 
	Moreover, combining \eqref{EtphiH3/2x.old} with \eqref{LtphiH3/2x}, we obtain
	\begin{equation} \label{EtphiH3/2x}
	\|\rd E_k \Db^{s''} \tphi \|_{L^2(\Sigma_t)} \ls \ep.
	\end{equation}
	
\end{proposition}
\begin{proof}
	In order to prove $\|\rd L_k \Db^{s''} \tphi \|_{L^2(\Sigma_t)} \ls \ep$, we will derive and use the wave equation for $L_k \Db^{s''} \tphi$.
	
	%\begin{equation}\label{eq:commute.with.L.goal}
	%\|\Box_g L_k \Db^{s''} \tphi \|_{L^2(\Sigma_t)} \ls \ep. 
	%\end{equation}
	
	Our main strategy is to decompose
	\begin{equation}\label{eq:commute.L.main.decomp}
	\Box_g L_k \Db^{s''} \tphi = L^t_k \rd_t C + F %= \f 1N(\rd_t C) + F
	\end{equation}
	for appropriate $C$ and $F$ \color{black} (Step~1). The term $F$ will be bounded in $L^2(\Sigma_t)$, while the term $L^t_k \rd_t C$ will be treated with an additional integration by parts in $t$. The relevant estimates will be treated in Step~2 below.
	
	\pfstep{Step~1: Achieving the decomposition \eqref{eq:commute.L.main.decomp}} We first write (recall \blue{the} notation in \eqref{decompositionbox})
	\begin{equation}\label{eq:commute.with.L.strategy}
	\Box_g L_k \Db^{s''} \tphi = [\Box_g,L_k \Db^{s''}]\tphi = [\Box_g,L_k ]  \Db^{s''}\tphi + L_k [\Box^1, \Db^{s''}]\tphi + L_k [\Box^2, \Db^{s''}]\tphi.
	\end{equation}
	
	Now, recalling the notation in \eqref{def:Tres} and \eqref{def.barg}, the last term in \eqref{eq:commute.with.L.strategy} can be further computed using the product rule as follows, by the same computation that led to \eqref{eq:commute.with.E.strategy.2}:
	\begin{equation}\label{eq:commute.with.L.strategy.2}
	\begin{split}
	&\: L_k [\Box^2, \Db^{s''}]\tphi = L_k^\mu \rd_\mu [\Box^2, \Db^{s''}]\tphi\\
	%	&\: + L_k^\mu [(g^{-1})^{\color{red}\nu\color{black}\bt}[\Db^{s''},\varpi]\rd^3_{\mu\color{red}\nu\color{black}\bt} \tphi] + \varpi L_k^\mu [(g^{-1})^{\color{red}\nu\color{black}\bt}\rd^3_{\mu\color{red}\nu\color{black}\bt} \Db^{s''}\tphi] - L_k^\mu \Db^{s''} [\varpi(g^{-1})^{\color{red}\nu\color{black}\bt}\rd^3_{\mu\color{red}\nu\color{black}\bt} \tphi] \\
	= &\: L_k^\mu  [\varpi (\rd_\mu(g^{-1})^{\nu\color{black}\bt})\rd^2_{\nu\color{black}\bt} \Db^{s''}\tphi] - L_k^\mu \Db^{s''} [\varpi (\rd_\mu(g^{-1})^{\nu\color{black}\bt})\rd^2_{\nu\color{black}\bt} \tphi] \\
	&\: + L_k^\mu  (\rd_\mu(g^{-1})^{\nu\color{black}\bt})[\Db^{s''},\varpi]\rd^2_{\nu\color{black}\bt} \tphi + L_k^\mu [ \barg [\Db^{s''},\varpi]\rd^3_{\mu\nu\color{black}\bt} \tphi] \\
	&\: - s'' \de^{jq} (\rd_j (\varpi \barg)) (L_k \rd^3_{q\nu\color{black}\bt} \Db^{s''-2}\tphi) - L_k^\mu T^{s''}_{\mathrm{res}}(\varpi \barg,\,\rd^3_{\mu\nu\color{black}\bt}\tphi).
	\end{split}
	\end{equation}
	
	There are a few terms in \eqref{eq:commute.with.L.strategy} and \eqref{eq:commute.with.L.strategy.2} which cannot be estimated directly and have to be separated out. First, there are the following terms:
	\begin{equation}\label{eq:L.cannot.handle}
	\begin{split}
	L_k^t \rd_t [\Box^1, \Db^{s''}]\tphi,\quad L_k^t (\bar g^{-1})^{tt} [\Db^{s''},\varpi]\rd^3_{ttt}\tphi,\quad - L_k^t T^{s''}_{\mathrm{res}} (\varpi (\bar{g}^{-1})^{tt}, \, \rd^3_{ttt} \tphi).
	\end{split}
	\end{equation}
	The term $- s'' \de^{jq} [(\rd_j (\varpi \barg))L_k \rd^3_{q\nu\color{black}\bt} \Db^{s''-2}\tphi]$ also cannot be controlled directly. It can be shown (see \eqref{eq:commute.with.L.4} below) that up to controllable error, this term is essentially the following:
	\begin{equation}\label{eq:L.cannot.handle.2}
	s'' \de^{jq} (\rd_j (\varpi(\bar{g}^{-1})^{tt}))L_k^t \rd_t \Db^{s''-2} \rd_q (\f{\Gamma^\lambda \rd_\lambda \tphi}{(g^{-1})^{tt}}).
	\end{equation}
	
	In order to handle the terms in \eqref{eq:L.cannot.handle} and \eqref{eq:L.cannot.handle.2}, we define
	\begin{equation}\label{eq:def.C}
	\begin{split}
	C:= &\: [\Box^1, \Db^{s''}]\tphi + s'' \de^{jq} (\rd_j (\varpi (\bar g^{-1})^{tt})) \Db^{s''-2} \rd_q (\f{\Gamma^\lambda \rd_\lambda \tphi}{(g^{-1})^{tt}})\\
	&\: + (\bar g^{-1})^{tt} [\Db^{s''},\varpi]\rd^2_{tt}\tphi - T^{s''}_{\mathrm{res}} (\varpi(\bar{g}^{-1})^{tt}, \, \rd^2_{tt} \tphi).
	\end{split}
	\end{equation} 
	
	It is easy to check that 
	\begin{equation}\label{eq:LtdtC}
	\begin{split}
	L_k^t \rd_t C 
	= &\: L_k^t \rd_t [\Box^1, \Db^{s''}]\tphi + s'' \de^{jq} (\rd_j (\varpi (\bar g^{-1})^{tt}))L_k^t \rd_t \Db^{s''-2} \rd_q (\f{\Gamma^\lambda \rd_\lambda \tphi}{(g^{-1})^{tt}}) \\
	&\: + L_k^t (g^{-1})^{tt} [\Db^{s''},\varpi]\rd^3_{ttt}\tphi - L_k^t T^{s''}_{\mathrm{res}} (\varpi(\bar{g}^{-1})^{tt}, \, \rd^3_{ttt} \tphi)\\
	&\:  + s'' \de^{jq} (\rd^2_{tj} ( \varpi (\bar g^{-1})^{tt})) \Db^{s''-2} \rd_q (\f{\Gamma^\lambda \rd_\lambda \tphi}{(g^{-1})^{tt}}) + (\rd_t(g^{-1})^{tt}) [\Db^{s''},\varpi]\rd^2_{tt}\tphi \\
	&\: - L_k^t T^{s''}_{\mathrm{res}} (\rd_t (\varpi(\bar{g}^{-1})^{tt}), \, \rd^2_{tt} \tphi)
	\end{split}
	\end{equation}
	so that the first four terms are exactly the uncontrollable terms in \eqref{eq:L.cannot.handle} and \eqref{eq:L.cannot.handle.2}, and the last three terms are error terms generated in this process.
	
	Finally, we define $F$ as follows so that \eqref{eq:commute.L.main.decomp} holds by \eqref{eq:commute.with.L.strategy}, \eqref{eq:commute.with.L.strategy.2}, \eqref{eq:def.C} and \eqref{eq:LtdtC}:
	\begin{equation}\label{eq:def.F}
	\begin{split}
	&\: F \\
	:=&\: [\Box_g,L_k ]  \Db^{s''}\tphi + L_k^i \rd_i [\Box^1, \Db^{s''}]\tphi 
	+ L_k^\mu  [(\rd_\mu(g^{-1})^{\nu\color{black}\bt})[\Db^{s''},\varpi]\rd^2_{\nu\color{black}\bt} \tphi] \\
	&\: + L_k^\mu [{\barg} [\Db^{s''},\varpi]\rd^3_{\mu\nu\color{black}\bt} \tphi] - L_k^t ({\bar g^{-1}})^{tt}[\Db^{s''},\varpi]\rd^3_{ttt} \tphi\\
	&\: + L_k^\mu  [\varpi (\rd_\mu(g^{-1})^{\nu\color{black}\bt})\rd^2_{\nu\color{black}\bt} \Db^{s''}\tphi] - L_k^\mu \Db^{s''} [ \varpi (\rd_\mu(g^{-1})^{\nu\color{black}\bt})\rd^2_{\nu\color{black}\bt} \tphi]\\ 
	&\: - s'' \de^{jq} [(\rd_j {(\varpi \barg)})L_k \rd^3_{q\nu\color{black}\bt} \Db^{s''-2}\tphi] - s'' \de^{jq} (\rd_j {(\varpi \barg)}) L_k^t \rd_t \Db^{s''-2} \rd_q (\f{\Gamma^\lambda \rd_\lambda \tphi}{(g^{-1})^{tt}}) \\
	&\: - L_k^\mu T^{s''}_{\mathrm{res}}(\varpi {\barg},\,\rd^3_{\mu\nu\color{black}\bt}\tphi) + L_k^t T^{s''}_{\mathrm{res}}(\varpi({\bar{g}^{-1}})^{tt},\,\rd^3_{ttt}\tphi) \\
	&\: - s'' \de^{jq} (\rd^2_{tj} {(\varpi (\bar g^{-1})^{tt})}) \Db^{s''-2} \rd_q (\f{\Gamma^\lambda \rd_\lambda \tphi}{(g^{-1})^{tt}}) - (\rd_t(g^{-1})^{tt}) [\Db^{s''},\varpi]\rd^2_{tt}\tphi \\
	&\: + L_k^t T^{s''}_{\mathrm{res}} (\rd_t (\varpi({\bar{g}^{-1}})^{tt}), \, \rd^2_{tt} \tphi).
	\end{split}
	\end{equation}
	
	\pfstep{Step~2: The estimates} We will handle the $F$ term and the $L^t \rd_t C$ term separately.
	
	For the $F$ term, we will prove that for all $r\geq 2$
	\begin{equation}\label{eq:F.main.est}
	\|\wo2 F\|_{L^2(\Sigma_t)} \ls \ep + \ep^{\f 3 2}\sum_{Z_k \in \{E_k,L_k\} }\| \partial Z_k \Db^{s''}\tphi\|_{L^2(\Sigma_t)}.
	\end{equation}
	In view of \eqref{eq:def.F}, the following estimates together imply \eqref{eq:F.main.est} (note in particular the similarity of \eqref{eq:commute.with.L.1}--\eqref{eq:commute.with.L.5} with \eqref{eq:commute.with.E.1}--\eqref{eq:commute.with.E.5}): 
	\begin{equation}\label{eq:commute.with.L.1}
	\|\wo2 [\Box_g, L_k ]  \Db^{s''}\tphi\|_{L^2(\Sigma_t)}\ls \ep + \ep^{\f 3 2}\sum_{Z_k \in \{E_k,L_k\} }\| \partial Z_k \Db^{s''}\tphi\|_{L^2(\Sigma_t)}, 
	\end{equation}
	\begin{equation}\label{eq:commute.with.L.2}
	\|\wo2 L_k^i \rd_i [\Box^1, \Db^{s''}]\tphi\|_{L^2(\Sigma_t)}\ls \ep, 
	\end{equation}
	\begin{equation}\label{eq:commute.with.L.weight.comm.0}
	\|\la x\ra^{-\f r2} L_k^\mu  [(\rd_\mu(g^{-1})^{\nu\color{black}\bt})[\Db^{s''},\varpi]\rd^2_{\nu\color{black}\bt} \tphi] \|_{L^2(\Sigma_t)} \ls \ep,
	\end{equation}
	\begin{equation}\label{eq:commute.with.L.weight.comm}
	\|\la x\ra^{-\f r2} \{ L_k^\mu [({\bar g^{-1}})^{\nu\color{black}\bt}[\Db^{s''},\varpi]\rd^3_{\mu\nu\color{black}\bt} \tphi] - L_k^t ({\bar g^{-1}})^{tt}[\Db^{s''},\varpi]\rd^3_{ttt} \tphi \}\|_{L^2(\Sigma_t)} \ls \ep,
	\end{equation}
	\begin{equation}\label{eq:commute.with.L.3}
	\|\wo2 \{L_k^\mu  [\varpi (\rd_\mu(g^{-1})^{\nu\color{black}\bt})\rd^2_{\nu\color{black}\bt} \Db^{s''}\tphi] - L_k^\mu \Db^{s''} [\varpi (\rd_\mu(g^{-1})^{\nu\color{black}\bt})\rd^2_{\nu\color{black}\bt} \tphi] \} \|_{L^2(\Sigma_t)}\ls \ep, 
	\end{equation}
	\begin{equation}\label{eq:commute.with.L.4}
	\|\la x \ra^{-\f r2} s'' \de^{jq}  \{ (\rd_j {(\varpi \barg)}) L_k \rd^3_{q\nu\color{black}\bt} \Db^{s''-2}\tphi + L_k^t (\rd_j {(\varpi (\bar g^{-1})^{tt})}) \rd_t \Db^{s''-2} \rd_q (\f{\Gamma^\lambda \rd_\lambda \tphi}{(g^{-1})^{tt}}) \} \|_{L^2(\Sigma_t)}\ls \ep, 
	\end{equation}
	\begin{equation}\label{eq:commute.with.L.5}
	\begin{split}
	\|\wo2 [L_k^\mu T^{s''}_{\mathrm{res}}({\varpi \barg},\,\rd^3_{\mu\nu\color{black}\bt} \tphi) - L_k^t T^{s''}_{\mathrm{res}}({\varpi (\bar g^{-1})}^{tt},\,\rd^3_{ttt} \tphi)] \|_{L^2(\Sigma_t)} \ls \ep,
	\end{split}
	\end{equation}
	\begin{equation}\label{eq:commute.with.L.6}
	\|\wo2 s'' \de^{jq} (\rd^2_{tj} {(\varpi (\bar g^{-1})^{tt})}) \Db^{s''-2} \rd_q (\f{\Gamma^\lambda \rd_\lambda \tphi}{(g^{-1})^{tt}})\|_{L^2(\Sigma_t)} \ls \ep,
	\end{equation}
	\begin{equation}\label{eq:commute.with.L.weight.comm.2}
	\|\wo2 (\rd_t(g^{-1})^{tt}) [\Db^{s''},\varpi]\rd^2_{tt}\tphi\|_{L^2(\Sigma_t)} \ls \ep.
	\end{equation}
	\begin{equation}\label{eq:commute.with.L.7}
	\|\wo2 L_k^t T^{s''}_{\mathrm{res}} (\rd_t (\varpi({\bar{g}^{-1}})^{tt}), \, \rd^2_{tt} \tphi)\|_{L^2(\Sigma_t)}\ls \ep.
	\end{equation}
	The estimates \eqref{eq:commute.with.L.1}--\eqref{eq:commute.with.L.7} will be proven in Proposition~\ref{prop:commute.with.L.1.to.7}.
	
	On the other hand, we will prove in Proposition~\ref{prop:est.C.in.L} below that the term $C$ in \eqref{eq:def.C} can be bounded as follows for all $r' \geq 1$:
		\begin{equation}\label{eq:C.main.est}
		\|\la x\ra^{-\f {r'} 2} C\|_{L^2(\Sigma_t)} + \|\la x\ra^{-\f {r'} 2} \rd_i C\|_{L^2(\Sigma_t)} \ls \ep.
		\end{equation}

	\pfstep{Step~3: Putting everything together} We rewrite $L^t \rd_t C = \f 1N e_0 C + \f {\bt^i}{N} \rd_i C$ using \eqref{def:e0} and \eqref{eq:silly.tangential}. We now apply Corollary~\ref{cor:main.weighted.energy} with $v = L_k \Db^{s''} \tphi$, $f_1 = F+ \f {\bt^i}{N} \rd_i C$, $f_2 = C$, $h = \f 1N$ and $r\geq 2$. Note that the bound \eqref{eq:h.assumption.for.IBP} holds for $h = \f 1N$ thanks to \eqref{eq:g.main}. Thus
		\begin{equation}\label{eq:rdLDbtph}
		\begin{split}
		&\: \sup_{t\in [0, T)} \|\la x\ra^{-(r+2\alp)} \rd L_k \Db^{s''} \tphi \|_{L^2(\Sigma_t)}^2  \\
		\ls &\: \|\la x\ra^{- \f r2} \rd L_k \Db^{s''} \tphi \|_{L^2(\Sigma_0)}^2 + \sup_{t\in [0, T)}\|\la x\ra^{- \f r2}C \|_{L^2(\Sigma_{t})}^2 \\ 
		&\: + \int_0^T ( \|\la x\ra^{-\f{r}{2}} f_1 \|_{L^2(\Sigma_t)}^2+  \|\la x\ra^{-\f{r}{2}} C \|_{L^2(\Sigma_t)}^2   +    \|\la x\ra^{-\f{r}{2}} \rd_x C \|_{L^2(\Sigma_t)}^2 )\, dt \\ 
		\ls &\: \ep^2 + \int_0^T \|\la x\ra^{-\f{r}{2}} f_1 \|_{L^2(\Sigma_t)}^2\, dt,
		\end{split}
		\end{equation}
		where in the last line, we have used the initial data bounds in \eqref{eq:assumption.rough.energy} and \eqref{eq:assumption.rough.energy.commuted}, as well as controlled $C$ using \eqref{eq:C.main.est}.
		
		Notice that by choosing $r' = r-1$, \eqref{eq:C.main.est} and \eqref{eq:g.main} imply that
		$\| \wo2 \f {\bt^i}{N} \rd_i C\|_{L^2(\Sigma_t)} \ls \ep$.
		Combining this with \eqref{eq:F.main.est}, we thus obtain
		$$\| \langle x \rangle^{-\f r 2} f_1 \|_{L^2(\Sigma_t)} \ls  \ep + \ep^{\f 3 2}\sum_{Z_k \in \{E_k,L_k\} }\| \partial Z_k \Db^{s''}\tphi\|_{L^2(\Sigma_t)}.$$
		
		Plugging this into \eqref{eq:rdLDbtph}, and then using \eqref{prop:weight.gain.EL} and Proposition~\ref{prop:Ds.est}, we thus obtain
	
	$$	\|\rd L_k \Db^{s''} \tphi \|_{L^2(\Sigma_t)} \ls \ep +\ep^{\f 3 2}\sum_{Z_k \in \{E_k,L_k\} }\| \partial Z_k \Db^{s''}\tphi\|_{L^2(\Sigma_t)} .$$
	
	Then we use \eqref{EtphiH3/2x.old}  and the  smallness of $\ep$ to absorb the $\ep^{\f 3 2}\sum_{Z_k \in \{E_k,L_k\} }\| \partial Z_k \Db^{s''}\tphi\|_{L^2(\Sigma_t)}$ terms into the left-hand side: both \eqref{LtphiH3/2x} and \eqref{EtphiH3/2x} follow immediately. \qedhere
\end{proof}

\begin{proposition}\label{prop:commute.with.L.1.to.7}
	The estimates \eqref{eq:commute.with.L.1}--\blue{\eqref{eq:commute.with.L.7}} hold for $r \geq 2$.
\end{proposition}
\begin{proof}
	\pfstep{Step~1: Proof of \eqref{eq:commute.with.L.1}} This follows from Proposition~\ref{prop:commute.with.L} applied with $v= \Db^{s''} \tphi$ and Proposition \ref{prop:Ds.est}.
	
	\pfstep{Step~2: Proof of \eqref{eq:commute.with.L.2}} This follows from Proposition \ref{prop:comm.ED.1} and \eqref{eq:frame.1}.

	\pfstep{Step~3: Proof of \eqref{eq:commute.with.L.weight.comm.0}} This follows from Proposition~\ref{prop:commute.with.E.weight.comm}, H\"older's inequality and \eqref{eq:frame.1}, \eqref{eq:g.main}.
	
	\pfstep{Step~4: Proof of \eqref{eq:commute.with.L.weight.comm}} By H\"older's inequality, \eqref{eq:frame.1}, \eqref{eq:g.main}, we get \begin{align*}
	&	\|\la x\ra^{-\f r2} \{ L_k^\mu [(g^{-1})^{\nu\color{black}\bt}[\Db^{s''},\varpi]\rd^3_{\mu\nu\color{black}\bt} \tphi] - L_k^t (g^{-1})^{tt}[\Db^{s''},\varpi]\rd^3_{ttt} \tphi \}\|_{L^2(\Sigma_t)}\\ \ls\ &  \sum_{ (\nu\color{black},\beta,\mu) \neq (t,t,t)} \|\la x\ra^{-\f r2} L_k^\mu  (g^{-1})^{\nu\color{black}\bt} \|_{L^\i(\Sigma_t)} \|[\Db^{s''},\varpi]\rd^3_{\nu\color{black}\bt \mu} \tphi \|\color{black} \|_{L^2(\Sigma_t)} \\ \ls\ &  \ep^{\f 3 2} \sum_{ (\nu\color{black},\beta,\mu) \neq (t,t,t)} \|[\Db^{s''},\varpi]\rd^3_{\nu\color{black}\bt \mu} \tphi \|\color{black}_{L^2(\Sigma_t)} \ls \ep^{\f 5 2}.
	\end{align*}   
	(Note that $- L_k^t (g^{-1})^{tt}[\Db^{s''},\varpi]\rd^3_{ttt} \tphi$ exactly removes the $(\nu\color{black},\beta,\mu)=(t,t,t)$ contribution from the first term.) Note also that in the final inequality, we have used  \eqref{eq:commute.with.E.weight.comm.reduced} from Proposition \ref{prop:commute.with.E.weight.comm} (since there is no $(\nu\color{black},\beta,\mu)=(t,t,t)$ term).

	\pfstep{Step~5: Proof of \eqref{eq:commute.with.L.3}} This follows from \eqref{eq:commute.with.E.3.prelim} in Proposition~\ref{prop:commute.with.E.3} and \eqref{eq:frame.1}.
	
	\pfstep{Step~6: Proof of \eqref{eq:commute.with.L.4}} First, using H\"older's inequality and \eqref{eq:g.main}, we obtain
	\begin{align*}
	&	\|\la x \ra^{-\f r2} s'' \de^{jq}  \{ (\rd_j (g^{-1})^{\nu\color{black}\bt}) L_k \rd^3_{q\nu\color{black}\bt} \Db^{s''-2}\tphi + L_k^t (\rd_j (g^{-1})^{tt}) \rd_t \Db^{s''-2} \rd_q (\f{\Gamma^\lambda \rd_\lambda \tphi}{(g^{-1})^{tt}}) \} \|_{L^2(\Sigma_t)} \\
	\ls\ & \sum_{ (\nu\color{black},\beta)\neq (t,t)}  \|  \rd_j (g^{-1})^{\nu\color{black}\bt} \|_{L^\i(\Sigma_t)}	\| \la x \ra^{-\f r2}L_k \rd^3_{q\nu\color{black}\bt} \Db^{s''-2}\tphi  \|_{L^2(\Sigma_t)}\\ 
	\ &+ \|    \rd_j (g^{-1})^{ t t} \|_{L^\i(\Sigma_t)}	\|\la x \ra^{-\f r2} \{L_k \rd^3_{q  t t} \Db^{s''-2}\tphi  + L^t_k \rd_t \Db^{s''-2} \rd_q (\f{\Gamma^\lambda \rd_\lambda \tphi}{(g^{-1})^{tt}})\}\|_{L^2(\Sigma_t)}  \\
	\ls
	\ &  \ep^{\f 3 2} \sum_{ (\nu\color{black},\beta)\neq (t,t)} \| \la x \ra^{-\f r2}L_k \rd^3_{q\nu\color{black}\bt} \Db^{s''-2}\tphi  \|_{L^2(\Sigma_t)}\\
	\ & + \ep^{\f 32} \|\la x \ra^{-\f r2} \{L_k \rd^3_{q  t t} \Db^{s''-2}\tphi  + L^t_k \rd_t \Db^{s''-2} \rd_q (\f{\Gamma^\lambda \rd_\lambda \tphi}{(g^{-1})^{tt}})\}\|_{L^2(\Sigma_t)}.
	\end{align*}
	To conclude, we apply Proposition~\ref{prop:commute.with.E.4} with $Y=L_k$ (where the bounds \eqref{eq:Y.est.assume} are given by \eqref{eq:frame.1}, \eqref{eq:frame.2}): more specifically we use \eqref{eq:Y.commute.gen} for the first term  and \eqref{eq:Y.commute.gen.2} for the second term.

	\pfstep{Step~7: Proof of \eqref{eq:commute.with.L.5}}  This follows from H\"older's inequality, \eqref{eq:frame.1} and Proposition~\ref{prop:commute.with.E.5}. (Note that $L_k^t T^{s''}_{\mathrm{res}}({\varpi (\bar g^{-1})}^{tt},\,\rd^3_{ttt} \tphi)$ exactly removes the contribution in $L_k^\mu T^{s''}_{\mathrm{res}}({\varpi \barg},\,\rd^3_{\mu\nu\color{black}\bt} \tphi)$ where $(\mu,\nu\color{black},\bt) = (t,t,t)$.)

	\pfstep{Step~8: Proof of \eqref{eq:commute.with.L.6}} By Sobolev embedding, $\Db^{-2} \rd_q: L^4(\Sigma_t) \to L^2(\Sigma_t)$ is bounded. Thus by H\"older's inequality, Lemma~\ref{lem:stupid.2} and Proposition~\ref{prop:Ds.est},
	\begin{equation*}
	\begin{split}
	&\: \|s'' \de^{jq} (\rd^2_{tj} {(\varpi (\bar g^{-1})^{tt})}) \Db^{s''-2} \rd_q (\f{\Gamma^\lambda \rd_\lambda \tphi}{(g^{-1})^{tt}})\|_{L^2(\Sigma_t)} \\
	\ls &\: \|\rd^2_{tj} {(\varpi (\bar g^{-1})^{tt})} \|_{L^4(\Sigma_t)} \|\Db^{s''} \f{\Gamma^\lambda \rd_\lambda \tphi}{(g^{-1})^{tt}}\|_{L^2({\Sigma_t})} \ls \ep^{\f 32} \|\rd \Db^{s'} \tphi\|_{L^2(\Sigma_t)} \ls \ep^{\f 52},
	\end{split}
	\end{equation*} 
	where we have used H\"older's inequality and the condition $0<s'-s''<\frac{1}{3}$ to deduce $\|\rd^2_{tj} (g^{-1})^{tt}\|_{L^4(\Sigma_t)} \ls \ep^{\f 3 2}$ from \eqref{eq:g.main}.

	\pfstep{Step~9: Proof of \eqref{eq:commute.with.L.weight.comm.2}} By H\"older's inequality, \eqref{eq:g.main}, Lemma~\ref{lem:cutoff.commute}, Lemma~\ref{lem:invert.tt} and Proposition~\ref{prop:Ds.est},
	\begin{equation*}
	\begin{split}
	&\: \|\wo2 (\rd_t(g^{-1})^{tt}) [\Db^{s''},\varpi]\rd^2_{tt}\tphi\|_{L^2(\Sigma_t)} \\
	\ls &\: \|\wo2 (\rd_t(g^{-1})^{tt})\|_{L^\i(\Sigma_t)} \|\Db^{s''-2} \rd^2_{tt}\tphi\|_{L^2(\Sigma_t)} \ls \ep^{\f 52}.
	\end{split}
	\end{equation*}
	
	\pfstep{Step~10: Proof of \eqref{eq:commute.with.L.7}} By H\"older's inequality and Corollary~\ref{cor:commute.3} (with $\th_1 = s'$, $\th_2 = s''$, $p = \f 2{s'-s''}$), \eqref{eq:g.main}, \eqref{eq:frame.1}, Lemma~\ref{lem:invert.tt} and Proposition~\ref{prop:Ds.est},
	\begin{equation*}
	\begin{split}
	&\: \| \wo2 L_k^t T^{s''}_{\mathrm{res}} (\rd_t (\varpi ({\bar g^{-1}})^{tt}), \, \rd^2_{tt} \tphi)\|_{L^2(\Sigma_t)} \\
	\ls &\: \|\wo2 L_k^t\|_{L^\i(\Sigma_t)} \| \rd_t (\varpi ({\bar g^{-1}})^{tt}) \|_{W^{1,\f 2{s'-s''}}(\Sigma_t)} \|\Db^{s'-1} \rd^2_{tt} \tphi\|_{L^2(\Sigma_t)} \\
	\ls &\: \ep^{\f 32}\|\Db^{s'-1} \rd^2_{tt} \tphi\|_{L^2(\Sigma_t)} \ls \ep^{\f 32} \|\rd \Db^{s'} \tphi\|_{L^2(\Sigma_t)} \ls \ep^{\f 52}.
	\end{split}
	\end{equation*}
	%Finally, by Lemma~\ref{lem:invert.tt},
	%$$\|\Db^{s'-1} \rd^2_{tt} \tphi\|_{L^2(\Sigma_t)} \ls \ep^{\f 34}.$$
	%This gives \eqref{eq:commute.with.L.7}.
	\qedhere
	
\end{proof}

\begin{proposition}\label{prop:est.C.in.L}
	Let $C$ be as in \eqref{eq:def.C}. Then  for all $r\geq 1$
	$$\|\wo2 C\|_{L^2(\Sigma_t)} + \|\wo2 \rd_i C\|_{L^2(\Sigma_t)} \ls \ep.$$
\end{proposition}
\begin{proof}
	
	We consider the four terms in \eqref{eq:def.C} respectively in Steps~1--4 below.
	
	%$$C:= [\Box^1, \Db^{s''}]\tphi + s'' \de^{jq} (\rd_j (g^{-1})^{tt}) \Db^{s''-2} \rd_q (\f{\Gamma^\lambda \rd_\lambda \tphi}{(g^{-1})^{tt}}) + T^{s''}_{\mathrm{res}} ((g^{-1})^{tt}, \, \rd^2_{tt} \tphi).$$
	
	\pfstep{Step~1: $[\Box^1, \Db^{s''}]\tphi$} That $\|\wo2 [\Box^1, \Db^{s''}]\tphi\|_{L^2(\Sigma_t)}\ls \ep$ follows from Proposition~\ref{prop:frac.1} and Proposition~\ref{prop:Ds.est}; that $\|\wo2 \rd_i[\Box^1, \Db^{s''}]\tphi\|_{L^2(\Sigma_t)}\ls \ep$ is a consequence of Proposition~\ref{prop:comm.ED.1}.
	
	\pfstep{Step~2: $s'' \de^{jq} (\rd_j {(\varpi (\bar g^{-1})^{tt})}) \Db^{s''-2} \rd_q (\f{\Gamma^\lambda \rd_\lambda \tphi}{(g^{-1})^{tt}})$} By \eqref{eq:g.main}, $\|\wo2 \rd_j {(\varpi (\bar g^{-1})^{tt})} \|_{L^\i(\Sigma_t)}\ls \ep^{\f 32}$. Hence, using H\"older's inequality, Lemma~\ref{lem:stupid.2}, the $L^2$-boundedness of $\Db^{-2}\rd_q$, and Proposition~\ref{prop:Ds.est}, we obtain
	$$\|\wo2 s'' \de^{jq} (\rd_j (\varpi (\bar g^{-1})^{tt})) \Db^{s''-2} \rd_q (\f{\Gamma^\lambda \rd_\lambda \tphi}{(g^{-1})^{tt}})\|_{L^2(\Sigma_t)}\ls \ep^{\f 32} \|\Db^{s''}(\f{\Gamma^\lambda \rd_\lambda \tphi}{(g^{-1})^{tt}})\|_{L^2(\Sigma_t)} \ls \ep^{\f 5 2}.$$
	
	To estimate the derivative, we use also the product rule, H\"older's inequality, \eqref{eq:g.main}, Lemma~\ref{lem:stupid.2} and Proposition~\ref{prop:Ds.est} to obtain
	\begin{equation*}
	\begin{split}
	&\: \left\|\wo2 \rd_i\left(s'' \de^{jq} (\rd_j (\varpi (\bar g^{-1})^{tt})) \Db^{s''-2} \rd_q (\f{\Gamma^\lambda \rd_\lambda \tphi}{(g^{-1})^{tt}})\right) \right\|_{L^2(\Sigma_t)} \\
	\ls &\: \|\rd_j (\varpi (\bar g^{-1})^{tt})\|_{L^\i(\Sigma_t)} \|\Db^{s''} (\f{\Gamma^\lambda \rd_\lambda \tphi}{(g^{-1})^{tt}})\|_{L^2(\Sigma_t)} + \|\rd^2_{ij} {(\varpi (\bar g^{-1})^{tt})} \|_{L^{{\infty}}(\Sigma_t)} \|\Db^{s''-1} (\f{\Gamma^\lambda \rd_\lambda \tphi}{(g^{-1})^{tt}})\|_{L^{{2}}(\Sigma_t)} \\
	\ls &\: \ep^{\f 32}  \|\Db^{s''} (\f{\Gamma^\lambda \rd_\lambda \tphi}{(g^{-1})^{tt}})\|_{L^2(\Sigma_t)}\ls \ep^{\f 52 }.
	\end{split}
	\end{equation*}
	
	\pfstep{Step~3: $(g^{-1})^{tt} [\Db^{s''},\varpi]\rd^2_{tt}\tphi$} To bound $(g^{-1})^{tt} [\Db^{s''},\varpi]\rd^2_{tt}\tphi$ itself, we use  H\"older's inequality, the estimate \eqref{eq:g.main}, Lemma~\ref{lem:cutoff.commute}, Lemma~\ref{lem:invert.tt} and Proposition~\ref{prop:Ds.est} to obtain
	$$\|\wo2 (g^{-1})^{tt} [\Db^{s''},\varpi]\rd^2_{tt}\tphi\|_{L^2(\Sigma_t)} \ls \|\wo2 (g^{-1})^{tt} \|_{L^\i(\Sigma_t)} \|\Db^{s''-2} \rd^2_{tt} \tphi\|_{L^2(\Sigma_t)} \ls \ep.$$
	
	For the derivative, we first use the product rule to distribute the $\rd_i$ derivative and then argue in a similar way as above, i.e.
	\begin{equation*}
	\begin{split}
	&\: \|\wo2 \rd_i \{ (g^{-1})^{tt} [\Db^{s''},\varpi]\rd^2_{tt}\tphi\} \|_{L^2(\Sigma_t)} \\
	\ls &\: \|\wo2 \rd_i(g^{-1})^{tt}\|_{L^\i(\Sigma_t)} \|\Db^{s''-2} \rd^2_{tt} \tphi\|_{L^2(\Sigma_t)} + \|\wo2 (g^{-1})^{tt}\|_{L^\i(\Sigma_t)} \|\Db^{s''-1} \rd^2_{tt} \tphi\|_{L^2(\Sigma_t)} \ls \ep.
	\end{split}
	\end{equation*}
	
	\pfstep{Step~4: $T^{s''}_{\mathrm{res}} (\varpi{(\bar g^{-1})^{tt}}, \, \rd^2_{tt} \tphi)$} For $T^{s''}_{\mathrm{res}} ({\varpi (\bar g^{-1})^{tt}}, \, \rd^2_{tt} \tphi)$ itself, we use Corollary~\ref{cor:commute.3}, the estimate \eqref{eq:g.main}, Lemma~\ref{lem:invert.tt} and Proposition~\ref{prop:Ds.est} to obtain
	$$\|T^{s''}_{\mathrm{res}} (\varpi {(\bar g^{-1})^{tt}}, \, \rd^2_{tt} \tphi)\|_{L^2(\Sigma_t)} \ls \|\varpi {(\bar g^{-1})^{tt}}\|_{W^{1,\infty}(\Sigma_t)}\|\Db^{s''-1}\rd^2_{tt} \tphi\|_{L^2(\Sigma_t)} \ls \ep.$$
	
	For the derivative, we first use product rule to obtain 
	\begin{equation}\label{eq:rdiC.T.term}
	\|\rd_i T^{s''}_{\mathrm{res}} (\varpi {(\bar g^{-1})^{tt}}, \, \rd^2_{tt} \tphi)\|_{L^2(\Sigma_t)} \ls \|T^{s''}_{\mathrm{res}} (\rd_i (\varpi {(\bar g^{-1})^{tt}}), \, \rd^2_{tt} \tphi)\|_{L^2(\Sigma_t)} + \|T^{s''}_{\mathrm{res}} (\varpi {(\bar g^{-1})^{tt}}, \, \rd^3_{itt} \tphi)\|_{L^2(\Sigma_t)}.
	\end{equation}
	The first term in \eqref{eq:rdiC.T.term} can be estimated using Corollary~\ref{cor:commute.3}, the estimate \eqref{eq:g.main}, Lemma~\ref{lem:invert.tt} and Proposition~\ref{prop:Ds.est} to obtain
	$$\|T^{s''}_{\mathrm{res}} (\rd_i (\varpi {(\bar g^{-1})^{tt})}, \, \rd^2_{tt} \tphi)\|_{L^2(\Sigma_t)} \ls \|\rd_i(\varpi {(\bar g^{-1})^{tt}})\|_{W^{1,\frac{2}{s'-s''}}(\Sigma_t)}\|\Db^{s'-1}\rd^2_{tt} \tphi\|_{L^2(\Sigma_t)} \ls \ep.$$
	The second term in \eqref{eq:rdiC.T.term} can also be estimated as follows using Corollary~\ref{cor:commute.3}, the estimate \eqref{eq:g.main}, Lemma~\ref{lem:invert.tt} and Proposition~\ref{prop:Ds.est}
	$$\|T^{s''}_{\mathrm{res}} (\varpi (\bar g^{-1})^{tt}, \, \rd^3_{itt} \tphi)\|_{L^2(\Sigma_t)}\ls \| \varpi (\bar g^{-1})^{tt}\|_{W^{2,\frac{2}{s'-s''}}(\Sigma_t)} \|\Db^{s'-1} \rd^2_{tt}\tphi\|_{L^2(\Sigma_t)} \ls \ep.$$
\end{proof}

\subsection{Control of $E_k \tphi$ and $L_k \tphi$ in $H^{1+s'}$}\label{sec:frac.L.E.commuted}
We turn to the estimates that are analogous to \eqref{LtphiH3/2x} and \eqref{EtphiH3/2x} but with vector fields and fractional derivatives taken in a slightly different order.

\begin{prop}\label{frac.inversion.prop} The following estimates are satisfied:
	\begin{equation} \label{Ztphi.H1+s'}
\sum_{Z_k \in \{E_k,L_k\}} \| Z_k \tphi \|_{H^{1+s'}(\Sigma_t)} \ls \ep,
	\end{equation}
		\begin{equation} \label{dtZtphi.Hs'}
\sum_{Z_k \in \{E_k,L_k\}} \| \partial_t Z_k \tphi \|_{H^{s'}(\Sigma_t)} \ls \ep.
	\end{equation}
\end{prop}

\begin{proof} Take $Z_k \in \{L_k, E_k\}$: the goal is to show $\|\Db^{s''}( \rd_{\nu\color{black}} Z_k \tphi) \|_{L^{2}(\Sigma_t)} \ls \ep$. We write the following identity \begin{align*}
& \Db^{s''}( \rd_{\nu\color{black}} Z_k \tphi)= \Db^{s''} \left( (\partial_{\nu\color{black}} Z_k^i) \partial_i \tphi \right)+   \Db^{s''} \left(Z_k^i \rd_{\nu\color{black}} \rd_i  \tphi\right)\\ = & \Db^{s''} \left( (\partial_{\nu\color{black}} Z_k^i) \partial_i \tphi \right)+ {\varpi} Z_k  \rd_{\nu\color{black}}\Db^{s''}  \tphi+[ \Db^{s''},  {\varpi} Z_k^i] \rd_{\nu\color{black}} \rd_i \tphi \\
= & \underbrace{\Db^{s''} \left( \varpi (\partial_{\nu\color{black}} Z_k^i) \partial_i \tphi \right)}_{I} + \underbrace{{\varpi}\rd_{\nu\color{black}} Z_k  \Db^{s''}  \tphi }_{II}- \underbrace{ {\varpi}(\rd_{\nu\color{black}} Z_k^i)  \rd_i  \Db^{s''}\tphi}_{III} + \underbrace{[ \Db^{s''},  \varpi Z_k^i] \rd_{\nu\color{black}} \rd_i \tphi}_{IV}.
	\end{align*}

We will treat each term separately.
	
	For $I$, we use Lemma~\ref{lem:frac.product}
	\begin{equation*}
	\begin{split}   \| I \|_{L^2(\Sigma_t)} \ls &\:  \| \varpi  (\partial Z_k^i)   \|_{L^{\i}(\Sigma_t)}   \| \partial_i \tphi \|_{H^{s'}(\Sigma_t)}+  \| \varpi  (\partial Z_k^i)   \|_{H^{s'}(\Sigma_t)}   \| \partial_i \tphi \|_{L^{\i}(\Sigma_t)} \\
	\ls &\: \ep^{\f 5 2}+  \| \varpi  (\partial Z_k^i)   \|_{H^{1}(\Sigma_t)}   \| \partial_i \tphi \|_{L^{\i}(\Sigma_t)} \ls \ep ,
	\end{split}
	\end{equation*} 
	where for the second inequality we used \eqref{eq:frame.2} and Proposition \ref{prop:Ds.est} and for the last one  we used  \eqref{eq:frame.2}, \eqref{eq:frame.3} and the bootstrap assumption \eqref{BA:Li}.
	
	For $II$, we use \eqref{LtphiH3/2x} and \eqref{EtphiH3/2x} and we get directly $$ \| II \|_{L^2(\Sigma_t)} \ls \ep.$$
	
	For $III$, we use \eqref{eq:frame.2} and Proposition \ref{prop:Ds.est} to obtain $$ \| IV \|_{L^2(\Sigma_t)} \ls  \| \partial Z_k^i \|_{L^{\i}(\Sigma_t \cap B(0,3R))} \| \partial_i \Db^{s''}\tphi \|_{L^2(\Sigma_t)} \ls \ep.$$
	 
	 For $IV$, we use Proposition \ref{prop:commute.2} with $f=Z_k^i \varpi $, $h=\rd_{\nu\color{black}} \rd_i \tphi$, $\theta=s''$ and $p=\infty$: 
	 $$ \| IV\|_{L^2(\Sigma_t)} \ls  \|  Z_k^i \|_{W^{1,\i}(\Sigma_t \cap B(0,3R))} \|\Db^{-1}\partial_i \rd_{\nu\color{black}} \Db^{s''}\tphi \|_{L^2(\Sigma_t)} \ls \| \rd \Db^{s''}\tphi \|_{L^2(\Sigma_t)}\ls \ep,$$ where for the second inequality we have used \eqref{eq:frame.1}, \eqref{eq:frame.2} and the $L^2$-boundedness of the operator $ \Db^{-1}\partial_i $, and for the last inequality we have used Proposition \ref{prop:Ds.est}. \qedhere
	 
%	 For $VI$, we make use of \eqref{eq:commute.weight.2.2} in Lemma \ref{lem:commute.weight.2}, \eqref{eq:frame.1} and Proposition \ref{prop:Ds.est} as such: $$ \| VI \|_{L^2(\Sigma_t)} \ls \| \langle x \rangle^{-1}  Z^i_k \|_{L^{\i}(\Sigma_t)}  \| \langle x \rangle \partial_i [\Db^{s''},\varpi ]   \tphi \|_{L^2(\Sigma_t)} \ls   \| \Db^{s''}\rd   \tphi \|_{L^2(\Sigma_t)} \ls  \ep.$$
%	 
%	 Finally we conclude the proof by $VII$: by Proposition \ref{prop:Ds.est} and \eqref{eq:frame.1} we obtain $$ \| VII \|_{L^2(\Sigma_t)} \ls \| Z^i_k \|_{L^{\i}(\Sigma_t \cap B(0,3R))}    \| \Db^{s''}\rd   \tphi \|_{L^2(\Sigma_t)} \ls  \ep.$$

	\end{proof}
\section{Energy estimates for $\phi_{reg}$}\label{sec:rphi}

In this section, we prove energy estimates for $\rphi$. We will prove that $\rphi$ is bounded in $H^{2+s'}(\Sigma)$, uniformly in $\de$. Since $\rphi$ is initially more regular, the proof of the energy estimates for $\rphi$ is also considerably easier than the higher order energy estimates for $\tphi$.

We begin with the energy estimates for up to the second derivative of $\rphi$. These bounds follow easily from the general energy estimates derived in Sections~\ref{sec:EE} and \ref{sec:Box.rdi.commutator}.
\begin{prop} \label{phiregH1.prop}
	The following estimates hold:
	\begin{align} \label{phiregH1}
	\sup_{t\in [0,T_B)} (\| \partial \phi_{reg} \|_{L^{2}(\Sigma_t)} + \| \partial^2  \phi_{reg} \|_{L^{2}(\Sigma_t)})  \lesssim \epsilon, \\
	\label{phiregH1.flux}
	\max_k \sup_{u_k \in \mathbb R} \sum_{Z_k \in \{L_k, E_k\}} (\| Z_k \rd_x \rphi\|_{L^2(C^k_{u_k}([0,T_B)))} + \| Z_k \rphi\|_{L^2(C^k_{u_k}([0,T_B)))})  \ls \ep.
	\end{align}
\end{prop}
\begin{proof}
	We recall that $\phi_{reg}$ is compactly supported on $B(0,R)$ and satisfies $\Box_g \phi_{reg} =0$. Thus we can apply Proposition~\ref{prop:EE} and Proposition~\ref{prop:commute.with.spatial.direct} simultaneously with $U_0=-\infty$ and $U_1=+\infty$ to get 
	\begin{align}
	\label{phiregH1.pf} \sup_{t\in [0,T_B)} (\| \partial \phi_{reg} \|_{L^{2}(\Sigma_t)} + \| \partial \rd_x  \phi_{reg} \|_{L^{2}(\Sigma_t)})  \lesssim \epsilon, \\
	\label{phiregH1.flux.pf} \max_k \sup_{u_k \in \mathbb R} \sum_{Z_k \in \{L_k, E_k\}} (\| Z_k \rd_x \rphi\|_{L^2(C^k_{u_k}([0,T_B)))} + \| Z_k \rphi\|_{L^2(C^k_{u_k}([0,T_B)))}) \ls \ep,
	\end{align}
	after using \eqref{data1} to bound the initial data term.
	
	Compared with the desired estimates \eqref{phiregH1} and \eqref{phiregH1.flux}, the only thing missing is a bound on $\|\rd^2_{tt} \rphi \|_{L^2(\Sigma_t)}$. Using the wave equation $\Box_g \rphi = 0$, we can rewrite $\rd^2_{tt}\rphi$ by \eqref{eq:wave.bg} as terms which can be bounded using \eqref{phiregH1.pf} above together with \eqref{eq:g.main} and \eqref{eq:Gamma}, yielding the desired estimate. \qedhere
\end{proof}

We then turn to the energy estimates for the $2+s'$ derivatives of $\rphi$. For this we will also use the commutator estimates with fractional derivatives proven in Section~\ref{unlochighestfrac}.
\begin{prop}  \label{d^2phiregH5/2.prop} The following estimate holds for all $t\in [0,T_B)$:
	\begin{equation} \label{d^2phiregH5/2}
	\|\Db^{s'} \partial^2  \phi_{reg} \|_{H^{s'}(\Sigma_t)}  \lesssim \epsilon.
	\end{equation}
	%\begin{equation} \label{dphiregC1/2}
	%\|  \phi_{reg} \|_{L^{\infty}(\Sigma_t)}+\| \partial  \phi_{reg} \|_{C^{s'-\alpha}(\Sigma_t)}  \lesssim \epsilon.
	%\end{equation}

\end{prop}

\begin{proof}
	
	\pfstep{Step~1: Using Corollary~\ref{cor:main.weighted.energy}} We recall that $\phi_{reg}$ is compactly supported on $B(0,R)$ and satisfies $\Box_g \phi_{reg} =0$. We apply the energy estimates in Corollary~\ref{cor:main.weighted.energy} with $v = \Db^{s'} \rd_i \rphi$, $f_2=0$ and 
	$$f_1 := \Box_g \Db^{s'} \rd_i \rphi = [\Box_g, \Db^{s'}] \rd_i \rphi + \Db^{s'} [\Box_g,\rd_i]\rphi$$
	to get that for every $T\in (0,T_B)$,
	\begin{equation}\label{eq:rphi.first.EE}
	\begin{split}
	&\: \sup_{t\in [0,T)} \| \la x \ra^{-r+2\alp} \rd \Db^{s'} \partial_i \rphi\|_{L^2(\Sigma_t)}^2\\
	\ls &\:  \| \wo2 \rd \Db^{s'}\partial_i \rphi \|_{L^2(\Sigma_0)}^2 + \int_0^T (\|\wo2 ([\Box_g, \Db^{s'}] \rd_i \rphi + \Db^{s'} [\Box_g,\rd_i]\rphi) \|_{L^2(\Sigma_\tau)}^2 \\
	\ls &\: \ep^2 + \int_0^T \|\wo2 [\Box_g, \Db^{s'}] \rd_i \rphi  \|_{L^2(\Sigma_\tau)}^2 \, \ud \tau + \int_0^T \|\wo2 \Db^{s'} [\Box_g,\rd_i]\rphi\|_{L^2(\Sigma_\tau)}^2 \, \ud \tau,
	\end{split}
	\end{equation}
	where in the last line we have controlled the data term using \eqref{data1}.
	
	%	We turn to the estimate of $\partial \partial_x  \Db^{s'} \rphi$. By Proposition~\ref{prop:EE.weighted} and \eqref{eq:weight.gain.generic_v}, for all $T\in [0,T_B)$
	%	\begin{equation*}
	%	\begin{split}
	%	&\: \sup_{t\in [0,T)} \|\rd \Db^{s'} \partial_i \rphi\|_{L^2(\Sigma_t)}^2\\
	%	\ls &\: \sup_{t\in [0,T)} \|  \partial \partial_i \rphi \|_{L^2(\Sigma_t)}^2 + \|\la x \ra^{-\f r2} \rd \Db^{s'}  \partial_i \rphi\|_{L^2(\Sigma_0)}^2 \\
	%	&\: + \int_0^{T} \|\rd \Db^{s'} \partial_i \rphi \|_{L^2(\Sigma_\tau)}^2 \, \ud \tau + \int_0^T \|\wo2 \Box_g \Db^{s'}  \partial_i\rphi\|_{L^2(\Sigma_\tau)}^2 \, \ud \tau \\
	%	\ls &\: \ep^2 + \int_0^{T} \|\rd \Db^{s'}  \partial_i\rphi \|_{L^2(\Sigma_\tau)}^2 \, \ud\tau + \int_0^T \|\wo2 \Box_g \Db^{s'}  \partial_i\rphi\|_{L^2(\Sigma_\tau)}^2 \,\ud \tau,
	%	\end{split}
	%	\end{equation*}
	%	where in the last line we have \eqref{phiregH1} and the regularity of the data see section \ref{data}.
	%	
	%	
	%	We now turn to the control of $ \|\wo2 \Box_g \Db^{s'}  \partial_i\rphi\|_{L^2(\Sigma_\tau)}$.
	\pfstep{Step~2: Bounding $ [\Box_g, \Db^{s'}] \rd_i \rphi$} Recall the decomposition $\Box_g = -\Box^1+\Box^2$ from \eqref{decompositionbox}. By \blue{\eqref{eq:frac.1_generic_v}} (applied to $v=\partial_i \rphi$), we obtain
	$$\|\wo2 [ \Box^1, \Db^{s'}]\partial_i\rphi\|_{L^2(\Sigma_t)} \ls \|\rd \Db^{s'}  \partial_i\rphi\|_{L^2(\Sigma_t)}.$$
	By \blue{\eqref{eq:frac.2_generic_v}} (applied to $v=\partial_i \rphi$), we get 
	$$\|\wo2 [ \Box^2, \Db^{s'}]\partial_{i}\rphi\|_{L^2(\Sigma_t)} \ls \|\rd \Db^{s'}  \partial\rphi\|_{L^2(\Sigma_t)}+ \| [\Box_g, \partial_{i}] \rphi\|_{L^2(\Sigma_t)}.$$
	Using Lemma~\ref{waveopcommspatiallemma} together with the metric estimates in \eqref{eq:g.main}, we have, using also \eqref{phiregH1}: 
	$$\| [\Box_g, \partial_{i}] \rphi\|_{L^2(\Sigma_t)} \ls \epsilon.$$
	Putting these bounds together, we obtain
	\begin{equation}\label{eq:rphi.main.1}
	\|\wo2 [ \Box_g, \Db^{s'}]\partial_i\rphi\|_{L^2(\Sigma_t)} \ls \ep + \|\rd \Db^{s'}  \partial_i\rphi\|_{L^2(\Sigma_t)}.
	\end{equation}
	
	\pfstep{Step~3: Bounding $\rd_i [\Box_g, \Db^{s'}] \rphi$} Using again the decomposition in \eqref{decompositionbox}, we have
	\begin{equation}\label{eq:Ds.comm.rphi}
	\Db^{s'} [\Box_g, \rd_i] \rphi = \underbrace{- \Db^{s'} [(\rd_i \gi^{\nu\color{black}\bt}) \rd^2_{\nu\color{black}\bt} \rphi]}_{=:I} + \underbrace{\Db^{s'} [ (\rd_i\Gamma^\lambda) (\rd_\lambda \rphi) ]}_{=:II}.
	\end{equation}
	By H\"older's inequality, Lemma~\ref{lem:frac.product}, \eqref{eq:g.main} and \eqref{eq:g.top.fractional},
	\begin{equation}\label{eq:rphi.step3.1}
	\begin{split}
	\|I \|_{L^2(\Sigma_t)} \ls &\: \| \Db^{s'} (\rd_i \gi^{\nu\color{black}\bt}) \|_{L^\i(\Sigma_t)} \| \rd^2_{\nu\color{black}\bt} \rphi \|_{L^2(\Sigma_t)} + \| \rd_i \gi^{\nu\color{black}\bt} \|_{L^\i(\Sigma_t)} \|\Db^{s'} \rd^2_{\nu\color{black}\bt} \rphi \|_{L^2(\Sigma_t)} \\
	\ls &\:  \ep^{\f 32}  (\|\rd^2 \rphi \|_{L^2(\Sigma_t)} + \| \Db^{s'} \rd^2 \rphi \|_{L^2(\Sigma_t)}) \ls  \ep^{\f 32}  \| \Db^{s'} \rd^2 \rphi \|_{L^2(\Sigma_t)} \\
	\ls &\: \ep^{\f 32} \| \Db^{s'} \rd \rd_x \rphi \|_{L^2(\Sigma_t)},
	\end{split}
	\end{equation}
	where in the last line we used $\|\Db^{s'} \rd^2 \rphi \|_{L^2(\Sigma_t)} \ls \| \Db^{s'} \rd \rd_x \rphi \|_{L^2(\Sigma_t)}$, which in turn follow from the wave equation. More precisely, since $\Box_g \rphi = 0$, we use \eqref{eq:wave.bg}, Lemma~\ref{lem:frac.product}, \eqref{eq:g.main} and \eqref{eq:Gamma} to obtain
	\begin{equation}\label{eq:Ds.dtt.rphi}
	\begin{split}
	&\: \|\Db^{s'} (\rd^2_{tt} \rphi)\|_{L^2(\Sigma_t)}  \\
	\ls &\: \|\Db^{s'} (\frac{\Box_g \rphi}{(g^{-1})^{t t}} -\bg^{i\lambda}\rd^2_{i\lambda}\rphi + \f{\Gamma^\lambda \rd_\lambda \rphi}{(g^{-1})^{tt}}) \|_{L^2(\Sigma_t)}  \\
	\ls &\: \|\rd_i \Db^{s'} ( \bg^{i\lambda} \rd_\lambda \rphi)\|_{L^2(\Sigma_t)} + \| \Db^{s'} [(\rd_i \bg^{i\lambda}) (\rd_\lambda \rphi)]\|_{L^2(\Sigma_t)}  + \|\Db^{s'} ( \f{\Gamma^\lambda \rd_\lambda \rphi}{(g^{-1})^{tt}} )\|_{L^2(\Sigma_t)} \\
	\ls &\: \|\Db^{s'} \rd \rd_x \rphi \|_{L^2(\Sigma_t)}.
	\end{split}
	\end{equation}

	The term $II$ in \eqref{eq:Ds.comm.rphi} can be treated similarly. Using Lemma~\ref{lem:frac.product}, \eqref{eq:Gamma}, we obtain
	\begin{equation}\label{eq:rphi.step3.2}
	\begin{split}
	\|II \|_{L^2(\Sigma_t)} \ls &\: \| \Db^{s'}(\varpi\rd_i \Gamma^\lambda) \|_{L^2(\Sigma_t)} \| \rd_{\lambda} \rphi \|_{L^\infty(\Sigma_t)} + \|\varpi\rd_i \Gamma^\lambda \|_{L^\i(\Sigma_t)} \|\Db^{s'} \rd_{\lambda} \rphi \|_{L^2(\Sigma_t)} \\
	\ls &\:  \ep^{\f 32} \| \rd_{\lambda} \rphi \|_{L^\i(\Sigma_t)} + \ep^{\f 32} \|\Db^{s'} \rd_{\lambda} \rphi \|_{L^2(\Sigma_t)} \\
	\ls &\: \ep^{\f 32} \| \rd \rphi \|_{L^\i(\Sigma_t)} + \ep^{\f 32} \| \rd \rphi \|_{L^2(\Sigma_t)}  + \ep^{\f 32} \|\rd \Db^{s'}  \partial_i\rphi\|_{L^2(\Sigma_t)} \\
	\ls &\: \ep^{\f 94} + \ep^{\f 32} \|\rd \Db^{s'}  \partial_i\rphi\|_{L^2(\Sigma_t)},
	\end{split}
	\end{equation}
	where we controlled $\|\Db^{s'} \rd_{\lambda} \rphi \|_{L^2(\Sigma_t)}$ by interpolating between $\| \rd_\lambda \rphi \|_{L^2(\Sigma_t)}$ and $\|\rd_\lambda \Db^{s'}  \partial_i\rphi\|_{L^2(\Sigma_t)}$ (for instance using Plancherel's theorem), and finally we used \eqref{BA:Li} to control both $\| \rd \rphi \|_{L^\i(\Sigma_t)}$ and $\| \rd \rphi \|_{L^2(\Sigma_t)}$
	
	Putting together \eqref{eq:Ds.comm.rphi}, \eqref{eq:rphi.step3.1} and \eqref{eq:rphi.step3.2}, we obtain
	\begin{equation}\label{eq:rphi.main.2}
	\|\wo2 \Db^{s'} [\Box_g, \rd_i] \rphi \|_{L^2(\Sigma_t)} \ls \ep + \|\rd \Db^{s'}  \partial_i\rphi\|_{L^2(\Sigma_t)}.
	\end{equation}
	
	\pfstep{Step~4: Putting everything together} Combining the estimates in \eqref{eq:rphi.first.EE}, \eqref{eq:rphi.main.1} and \eqref{eq:rphi.main.2}, we obtain 	\begin{equation*}
	\begin{split}
	\sup_{t\in [0,T)} \|\la x\ra^{-r-2\alp} \rd \Db^{s'} \partial_x \rphi\|_{L^2(\Sigma_t)}^2
	\ls  \ep^2 + \int_0^{T} \| \rd \Db^{s'}  \partial_x\rphi \|_{L^2(\Sigma_\tau)}^2 \, \ud\tau .
	\end{split}
	\end{equation*}
	By Proposition~\ref{prop:weight.gain.easy.generic_v} (applied to $v = \rd_x\rphi$), we can strengthen the weights on the left-hand side, i.e.
	\begin{equation*}
	\begin{split}
	\sup_{t\in [0,T)} \|\rd \Db^{s'} \partial_x \rphi\|_{L^2(\Sigma_t)}^2
	\ls  \ep^2 + \int_0^{T} \| \rd \Db^{s'}  \partial_x\rphi \|_{L^2(\Sigma_\tau)}^2 \, \ud\tau .
	\end{split}
	\end{equation*}
	
	By Gr\"onwall's inequality, we obtain $$ \sup_{t\in [0,T_B)} \|\rd \Db^{s'} \partial_x \rphi\|_{L^2(\Sigma_t)}^2 \lesssim \epsilon. $$
	Combining this with \eqref{eq:Ds.dtt.rphi} yields the desired conclusion \eqref{d^2phiregH5/2}. \qedhere \end{proof}  

%\begin{proposition}
%	\eqref{BA:flux.for.rphi} holds with $\ep^{\f 34}$ replaced by $C\ep$.
%\end{proposition}
%\begin{proof}
%	
%\end{proof}

\section{Conclusion of the proof of Theorem~\ref{thm:energyest}}\label{sec:wave.final}

In this section, we conclude the proof of Theorem~\ref{thm:energyest}\blue{.} Theorem~\ref{thm:energyest} consists of parts~\ref{wavethm.part1}, \ref{wavethm.part2}, and \ref{wavethm.part3}, which will be proven, respectively, in Proposition \ref{conclusion.prop1},  Proposition \ref{conclusion.prop2} and Proposition \ref{conclusion.prop3}. (For part~\ref{wavethm.part2}, we recall the definition of $\mathcal{E}$ in Definition \ref{def:Lipschitz.control.norm}.)

\begin{prop}[Statement \ref{wavethm.part1} of Theorem~\ref{thm:energyest}] \label{conclusion.prop1}
	There  exists $C = C(s',s'',R,\upkappa_0) >0$ such that \eqref{BA:rphi}--% \eqref{bootstrapsmallnessenergy}, \eqref{tphiH2bootstrap}, \eqref{tphiH3/2bootstrap}, \eqref{EtphiH2bootstrap}, \eqref{bootstrapbadunlocenergyhyp}, \eqref{BA:away.from.singular},  \eqref{BA:flux.for.rphi}, \eqref{BA:flux.for.tphi.improved}, \eqref{BA:flux.for.tphi.improved.2}, 
	\eqref{BA:flux.for.tphi} hold with  $C \ep$ in place of $\ep^{\frac{3}{4}}$.
\end{prop}

\begin{proof} We look at each of the bootstrap assumptions \eqref{BA:rphi}--% \eqref{bootstrapsmallnessenergy}, \eqref{tphiH2bootstrap}, \eqref{tphiH3/2bootstrap}, \eqref{EtphiH2bootstrap}, \eqref{bootstrapbadunlocenergyhyp}, \eqref{BA:away.from.singular},  \eqref{BA:flux.for.rphi}, \eqref{BA:flux.for.tphi.improved}, \eqref{BA:flux.for.tphi.improved.2}, 
	\eqref{BA:flux.for.tphi}. We point out the precise locations in the earlier sections which improve these bounds from $\ep^{\frac{3}{4}}$ to $C \ep$.
	\begin{itemize}
		\item 	Improvement of \eqref{BA:rphi}: it follows directly from \eqref{d^2phiregH5/2} in Proposition \ref{d^2phiregH5/2.prop} and \eqref{phiregH1} in Proposition \ref{phiregH1.prop}.
		
		\item 	Improvement of \eqref{bootstrapsmallnessenergy}: it follows from \eqref{energyglobal} and \eqref{energygoodcommutedglobal} in Proposition \ref{prop:easy.energy}, using also Proposition~\ref{prop:main.frame.est} to address the commutator term involving $[\partial,Z_k]$.
		\item 	Improvement of \eqref{tphiH2bootstrap}: it follows directly from \eqref{badlocenergyestimate} in Proposition \ref{prop:easy.energy}.
		\item 	Improvement of \eqref{tphiH3/2bootstrap}: it follows directly from \eqref{tphiH3/2x} in Proposition \ref{prop:Ds.est}.
		\item 	Improvement of \eqref{EtphiH2bootstrap}: it was already stated and proven in Proposition \ref{prop:highest.everything}.
		\item 	Improvement of \eqref{bootstrapbadunlocenergyhyp}: the $\| \partial \tphi \|_{L^2(\Sigma_t\cap S^k_{2\delta})}$ term follows directly from \eqref{locenergyestimate} in Proposition \ref{prop:local.small.energy}. The $\| Z_k \partial \tphi \|_{L^2(\Sigma_t\cap S^k_{2\delta})}$ term follows  \eqref{ELlocenergyestimate} and the commutation of $Z_k$ with $\partial$, using \eqref{locenergyestimate}, \eqref{eq:frame.1}, \eqref{eq:frame.2}.
		
		\item 	Improvement of \eqref{BA:away.from.singular}: this follows directly from \eqref{phiextestimate} in Proposition \ref{prop:phiext}.
		
		\item 	Improvement of \eqref{BA:flux.for.rphi}: it follows directly from	\eqref{phiregH1.flux} in Proposition \ref{phiregH1.prop}.
		\item 	Improvement of \eqref{BA:flux.for.tphi.improved}: the first term in \eqref{BA:flux.for.tphi.improved} is bounded by \eqref{phiextestimate} in Proposition~\ref{prop:phiext}, while the second term in \eqref{BA:flux.for.tphi.improved} is bounded by \eqref{energyglobal} in Proposition \ref{prop:easy.energy}.
		\item 	Improvement of \eqref{BA:flux.for.tphi.improved.2}: the second term in \eqref{BA:flux.for.tphi.improved.2} is bounded by \eqref{energygoodcommutedglobal} in Proposition~\ref{prop:easy.energy}. 
		
		To control the first term in \eqref{BA:flux.for.tphi.improved.2}, i.e.~to bound $\sup_{u_{k} \in\RR} \| L_k \rd_x \tphi\|_{L^2(C^{k}_{u_k}([0,T_B)))}$, we first notice that it suffices to bound $\sup_{u_{k} \in\RR} \| \rd_x L_k  \tphi\|_{L^2(C^{k}_{u_k}([0,T_B)))}$ since the commutator can be estimated with the help of Proposition \ref{prop:main.frame.est} and \eqref{BA:Li}. This latter term can in turn be bounded by $\| E_k L_k \tphi\|_{L^2(C^{k}_{u_k}([0,T_B)))}$ and $\| X_k L_k \tphi\|_{L^2(C^{k}_{u_k}([0,T_B)))}$ thanks to Lemma \ref{lem:rd.in.terms.of.XEL} (and the support properties of $\tphi$). The estimate for $\| E_k L_k \tphi\|_{L^2(C^{k}_{u_k}([0,T_B)))}$ follows directly from \eqref{energygoodcommutedglobal}, while the estimate for $\| X_k L_k \tphi\|_{L^2(C^{k}_{u_k}([0,T_B)))}$ follows from using the wave equation \eqref{waveop} and applying the estimates from Proposition \ref{prop:main.metric.est}, Proposition \ref{prop:main.Ricci.est}, \eqref{BA:Li} and \eqref{energygoodcommutedglobal}.
		\item 	Improvement of \eqref{BA:flux.for.tphi}: it follows directly from \eqref{badlocenergyestimate} in Proposition \ref{prop:easy.energy}.	
	\end{itemize}
\end{proof}

\begin{prop} [Statement \ref{wavethm.part2} of Theorem~\ref{thm:energyest}]\label{conclusion.prop2} The following estimate holds: $$ \mathcal{E} \ls \ep.$$
\end{prop}

\begin{proof}
	$\mathcal{E} $ is composed of a sum of terms given in  Definition \ref{def:Lipschitz.control.norm}. We treat each of these terms one by one. Most of these bounds have already been stated in the proof of Proposition \ref{conclusion.prop1}. 
	\begin{itemize}
		\item $\| \partial \Db^{s'} \tphi\|_{L^2(\Sigma_t)}\ls \ep$ :  already obtained with the improvement of \eqref{tphiH3/2bootstrap}.
		\item $ \| E_k \rd \tphi\|_{L^2(\Sigma_t)} \ls \ep$ :  already obtained with the improvement of \eqref{bootstrapsmallnessenergy}.
		\item $\| \rd E_k \Db^{s''} \tphi\|_{L^2(\Sigma_t)}\ls \ep$ : this follows directly from \eqref{EtphiH3/2x} in Proposition \ref{prop:LkEkDsphi.main}.
		\item $\de^{\f 12} \| \partial^2 \tphi\|_{L^2(\Sigma_t)}\ls \ep$ :  already obtained with the improvement of \eqref{tphiH2bootstrap}.
		\item $\de^{\f 12}   \|\rd E_k \rd \tphi\|_{L^2(\Sigma_t)}\ls \ep$ : the bound  $\de^{\f 12}   \|\rd E_k \rd_x \tphi\|_{L^2(\Sigma_t)}\ls \ep$ (i.e.\ the particular case where the first derivative is spatial) follows directly from \eqref{highest.everything} in Proposition \ref{prop:highest.everything}. To address the remaining term $\de^{\f 12}   \|\rd E_k \rd_t \tphi\|_{L^2(\Sigma_t)}$, we use \eqref{defnormal} as \begin{align*}
		&\rd E_k \rd_t \tphi= \rd E_k \left( N( \n\tphi+ \beta^i \rd_i \tphi) \right)= \rd \left( (E_k N) ( \n \tphi+\beta^i \rd_i \tphi)\right)+ \rd \left( N ( E_k\n \tphi+\beta^i E_k \rd_i \tphi)\right)\\ 
		= &  \underbrace{(\partial E_k N) ( \n \tphi+\beta^i \rd_i \tphi)}_{I}+ \underbrace{(E_k N) \partial ( \n \tphi+\beta^i \rd_i \tphi)}_{II}+ \underbrace{(\rd  N) ( E_k\n \tphi+\beta^i E_k \rd_i \tphi)}_{III}\\ &+   \underbrace{N ( \rd E_k\n \tphi+\beta^i  \rd E_k \rd_i \tphi)}_{IV}+   \underbrace{(N \partial \beta^i )  E_k \rd_i \tphi}_{V}.
		\end{align*} We treat each term individually: 
		\begin{itemize}
			\item Term $I$: by \eqref{eq:g.main}, \eqref{defnormal}, and the bootstrap assumption \eqref{BA:Li}, we have 
			\begin{equation*}
			\begin{split}
			\delta^{\f 1 2}\| I \|_{L^{2}(\Sigma_t)} \ls &\: \delta^{\f 1 2} \| \partial E_k N \|_{L^2(B(0,R))} \|\n \tphi+\beta^i \rd_i \tphi\|_{L^{\i}(\Sigma_t)} \\
			\ls &\: \delta^{\f 1 2} \| \partial E_k N \|_{L^2(B(0,R))} \|\rd \tphi\|_{L^{\i}(\Sigma_t)}  \ls \delta^{\f 1 2} \ep^{\f 9 4}\ls \ep.  
			\end{split}
			\end{equation*}
			
			\item Term $II$: by \eqref{eq:g.main}, \eqref{defnormal} and \eqref{energyglobal},   \eqref{badlocenergyestimate} we have $$\delta^{\f 1 2}\| II \|_{L^{2}(\Sigma_t)} \ls \delta^{\f 1 2} \|  E_k N \|_{L^{\infty}(B(0,R))} \|\partial(\n \tphi+\beta^i \rd_i \tphi)\|_{L^2(\Sigma_t)} \ls \delta^{\f 1 2} \ep^{\f 5 2} \delta^{-\f 1 2}\ls \ep.  $$
			\item Term $III$: by \eqref{eq:g.main}, \eqref{defnormal} and \eqref{energyglobal},   \eqref{badlocenergyestimate} we have $$\delta^{\f 1 2}\| III \|_{L^{2}(\Sigma_t)} \ls \delta^{\f 1 2} \|  \partial N \|_{L^{\infty}(B(0,R))} \| E_k\n \tphi+\beta^i E_k\rd_i \tphi \|_{L^2(\Sigma_t)} \ls \delta^{\f 1 2} \ep^{\f 5 2} \delta^{-\f 1 2}\ls \ep.  $$
			\item Term $IV$: by \eqref{eq:g.main} and \eqref{highest.everything} in Proposition \ref{prop:highest.everything}, we have $$\delta^{\f 1 2}\| IV \|_{L^{2}(\Sigma_t)} \ls \delta^{\f 1 2}(  \|\partial E_k \n \tphi\|_{L^2(\Sigma_t)} +\|\partial E_k \rd_x \tphi\|_{L^2(\Sigma_t)}) \ls \delta^{\f 1 2} \ep \delta^{-\f 1 2}\ls \ep.  $$	\item Term $V$: by \eqref{eq:g.main}, \eqref{eq:frame.1} and \eqref{badlocenergyestimate} we have $$\delta^{\f 1 2}\| V \|_{L^{2}(\Sigma_t)} \ls \delta^{\f 1 2} \| N \partial \beta^i \|_{L^{\i}(B(0,R)} \|E_k \partial_x \tphi \|_{L^2(\Sigma_t)} \ls \delta^{\f 1 2} \ep \delta^{-\f 1 2} \ls \ep.  $$
		\end{itemize}
		\item $\de^{-\f 12} \|\partial \tphi\|_{L^2(\Sigma_t \cap S_{2\de}^k)} \ls \ep$ and $\de^{-\f 12} \| E_k \partial \tphi\|_{L^2(\Sigma_t \cap S_{2\de}^k)} \ls \ep$ : already obtained with the  improvement of \eqref{bootstrapbadunlocenergyhyp}.
		
		\item $\| \partial^2 \tphi\|_{L^2(\Sigma_t \setminus \Sd^k)}\ls \ep$ : already obtained in the improvement of \eqref{BA:away.from.singular}.
		\item $\|\rd^2 \Db^{s'} \rphi \|_{L^2(\Sigma_t)} \ls \ep$ : already obtained in the improvement of \eqref{BA:rphi}.
		
	\end{itemize}
\end{proof}

\begin{prop} [Statement \ref{wavethm.part3} of Theorem~\ref{thm:energyest}] \label{conclusion.prop3} The estimates \eqref{eq:main.thm.rphi}--\eqref{eq:smooththeorem.2} are satisfied.	

\end{prop}

\begin{proof} We prove each of the estimates \eqref{eq:main.thm.rphi}--\eqref{eq:smooththeorem.2} individually. Some of these estimates are already obtained in the proof of Proposition \ref{conclusion.prop1} or Proposition \ref{conclusion.prop2}.
	\begin{itemize}
		\item Proof of	\eqref{eq:main.thm.rphi}:  it follows directly from \eqref{d^2phiregH5/2} in Proposition \ref{d^2phiregH5/2.prop} and \eqref{phiregH1} in Proposition \ref{phiregH1.prop}.
		\item Proof of	\eqref{eq:main.thm.tphi.1}: it was already proven with the improvement of \eqref{tphiH3/2bootstrap}.
			\item Proof of	\eqref{eq:main.thm.tphi.2} : It follows directly from \eqref{Ztphi.H1+s'} and \eqref{dtZtphi.Hs'} in Proposition \ref{frac.inversion.prop}.

			\item Proof of	\eqref{eq:smooththeorem.1}: the first inequality$ \| \partial^2 \tphi \|_{L^2(\Sigma_t)}\ls \ep \cdot \delta^{-\f 1 2}$ was already proven with the improvement of \eqref{tphiH2bootstrap}. The inequality $ \sum_{ \substack{ Y_k^{(1)}, Y_k^{(2)}, Y_k^{(3)} \in \{ X_k, E_k, L_k\} \\ \exists i, Y_k^{(i)} \neq X_k} } \|Y_k^{(1)} Y_k^{(2)} Y_k^{(3)} \tphi \|_{L^2(\Sigma_t)} \ls \ep\cdot \de^{-\f 12}$ follows directly from \eqref{highest.everything} in Proposition \ref{prop:highest.everything}.
			\item Proof of	\eqref{eq:smooththeorem.top}: this follows directly from \eqref{eq:three.derivatives} in Proposition \ref{prop:three.derivatives}.
			\item Proof of	\eqref{eq:smooththeorem.2}: it was already proven with the improvement of \eqref{BA:away.from.singular}.
	\end{itemize}
\end{proof}

\section{Lipschitz estimates and improved H\"older bounds for $\phi$} \label{dphiLinftysection}
In this section, we prove Lipschitz estimates for $\tphi$, as well as improved $C^{0,\f{s''}{2}}$ H\"older estimates for $\rd\rphi$ and for $\rd\tphi$ away from the singular zone. 

While proving the Lipschitz estimates, we will prove stronger Besov type estimates (\blue{recall}~Section~\ref{sec:intro.embedding} and \cite[Section~1.1.4]{LVdM1}). When combined with the energy estimates that we have already obtained, the result in this section improves the bootstrap assumptions \eqref{rphiBbootstrap}, \eqref{tphiBbootstrap} and \eqref{BA:Li}.

The following is the main result of this section (recall the definition of $\mathcal E$ in Definition~\ref{def:Lipschitz.control.norm}):
\begin{thm} [Main Lipschitz and improved H\"older estimates]\label{dphiboundedprop}
	Let $\rho_k(u_k,\theta_k,t)=\widetilde{\rho}(\frac{u_k}{\delta})$ be a cutoff function, where $\widetilde{\rho}:\mathbb R\to [0,1]$ is smooth function with $\widetilde{\rho} \equiv 0$ on $[2, \infty)$, and $\widetilde{\rho} \equiv 1$ on $(-\infty, 1]$.
	
	The following estimates hold for all $t \in [0,T_B)$ (recall the definition of the Besov space $ \Bes$ in Definition~\ref{def:Besov}):
	\begin{enumerate}
		\item For $k \in \{1,2,3\}$, $\rd\tphi$ obeys the following estimate near the singular zone for any $k'\neq k$ : 
		\begin{equation} \label{dtphiboundedint}	
		\|\rho_k  \cdot \partial \tphi \|_{ \Bes}   \lesssim \mathcal E.
		\end{equation}  		
		\item For $k \in \{1,2,3\}$, $\rd\tphi$ obeys the following estimate away from the singular zone for any $k'\neq k$:
		\begin{equation} \label{dtphiboundedext}
		%	\| (1-\rho) \cdot \partial \tphi \|_{ C^{\frac{s'''}{2}}(\Sigma_t) } \ls 
		\| \partial \tphi \|_{ C^{0,\frac{s''}{2}}(\Sigma_t \cap C^k_{\geq \delta}) }\lesssim \mathcal E,
		\end{equation}	\begin{equation} \label{dtphiboundedext2}
		\| (1-\rho_k ) \cdot \partial \tphi \|_{ \Bes } 
		\lesssim \mathcal E.
		\end{equation} 
		\item The regular part $\phi_{reg}$ of $\phi$ satisfies the following estimate for any $k$, $k'$ with $k\neq k'$:
		\begin{equation}\label{eq:phireg.Linfty}
		\| \partial \phi_{reg} \|_{C^{0,\f{s''}2}(\Sigma_t)} \ls \mathcal E,\quad \| \partial \phi_{reg} \|_{ \Bes} \ls \mathcal E.
		\end{equation}
		\item	As a consequence, the following estimate holds for $\phi$:
		%\begin{equation} \label{dtphibounded}
		%\| \partial \tphi \|_{ \Bes} \ls 	\| \partial \tphi \|_{ L^\infty(\Sigma_t)}\lesssim \epsilon.
		%\end{equation} 
		\begin{equation} \label{dphibounded}
	\| \partial \phi \|_{ L^\infty(\Sigma_t)}\lesssim\mathcal E.
		\end{equation}
	\end{enumerate}
\end{thm}

\subsection{Localized or anisotropic Sobolev embeddings} \label{LPsection}

In this subsection, we prove two general embedding results, namely Theorem \ref{embeddingThmInterior} and Theorem \ref{embeddingexterior}. These are the functional bounds \eqref{eq:intro.anisotropic.2} and \eqref{eq:intro.anisotropic.1} discussed in the introduction. They will be applied in the later subsections to $\rd\tphi$, $\rd\rphi$ (or appropriately localized versions) to prove Theorem~\ref{dphiboundedprop}.

Notice that all the general embedding results derived in this subsection will be applied in the $(u_k, u_{k'})$ coordinates. In order to keep the exposition general, and also distinguish the coordinates here from the $(x^1, x^2)$ coordinates in the elliptic gauge, \textbf{we will use $(y^1, y^2)$ to denote a general coordinate system on $\RR^2$}. In the following estimates, $\rd_{y^2}$ can be thought of as a good derivative, and in applications it corresponds to $\srd_{u_{k'}}$.

Before we turn to the actual embedding results, introduce the notations regarding Fourier transform, and an anisotropic Littlewood--Paley theory decomposition.

\begin{defn} 
\begin{enumerate}
\item Given $f=f(y_1,y_2) \in \mathcal S(\RR^2)$, denote by $\hat{f}=\hat{f}(\xi_1,\xi_2)$, or $(\mathcal Ff) = (\mathcal Ff)(\xi_1,\xi_2)$, the Fourier transform of $f$.
\item Let $s >0$. Fractional derivatives are defined as in Definition~\ref{def:fractional.Sobolev.norm}, except now in $(y^1, y^2)$ coordinates, i.e.~$\la D_y \ra^s f:= \mathcal{F}^{-1}(\la \xi\ra^s \mathcal F (f))$. Define also a homogeneous version by $\blue{|D_y|} ^s f:= \mathcal{F}^{-1}(| \xi|^s \mathcal F (f))$.
\item Let $\{P_k \}_{k \in \mathbb N \cup \{0\}}$ be the Littlewood--Paley projections as in Definition~\ref{def:Littlewood.Paley}, except with $(y^1, y^2)$ in place of $(u_k, u_{k'})$.
\item Define the anisotropic Littlewood--Paley projections $\{ P_{kl} \}_{k \in \mathbb N \cup \{ 0\},\, l \in \mathbb Z}$ as follows. Take $\varphi: \mathbb R\to [0,1]$ be even, smooth and such that $\underline{\varphi}(\eta) = \begin{cases}
1 & \mbox{if $|\eta|\leq 1$} \\
0 & \mbox{if $|\eta|\geq 2$} \end{cases}$. For each $l \in \mathbb Z$, define $\underline{P}_l$ by
$$\underline{P}_l f := \mathcal F^{-1} \Big[ (\underline{\varphi}(2^{-l}\xi_2) - \underline{\varphi}(2^{-l+1}\xi_2)) \mathcal F f \Big].$$
Then, for $k \in \mathbb N \cup \{0\}$, $l \in \mathbb Z$, define
$$P_{kl} := P_k \circ \underline{P}_l,$$
where $P_k$ is as in point 3 above.
\item Define also the notation that
$$f_k := P_k(f),\quad f_{kl}:= P_{kl}(f). $$
\end{enumerate}
\end{defn}

\subsubsection{Anisotropic localized Sobolev embedding}

\begin{thm} \label{embeddingThmInterior}
	Let $\sigma>\f 12$. The following holds for all Schwartz function $f$ with an implicit constant depending only on $\sigma$:
	$$\| f \|_{L^{\infty}(\RR^2)},\, \|  f \|_{ B_{\infty,1}(\RR^2)} \ls  \inf_{\de>0}(\de^{-\f 12}\| f\|_{L^2(\RR^2)} + \de^{\sigma-\f 12}\| \partial_{y^2} |D_{ y}|^{\sigma} f \|_{L^2(\RR^2)} + \de^{\f 12}\| \rd_{ y} f \|_{L^2(\RR^2)} + \de^{-\f 12} \| \partial_{ y^2} f\|_{L^2(\RR^2)}).$$ 
	
	Here, $B_{\infty,1}(\RR^2)$ is the Besov norm as in Definition~\ref{def:Besov}, except with $(y^1,y^2)$ in place of $(u_k, u_{k'})$.
\end{thm}

\begin{proof}
	By the triangle inequality, $\| f \|_{L^{\infty}(\RR^2)} \ls \underset{k\geq 0}{\sum} \|f_k \|_{L^\i(\RR^2)} = \|  f \|_{ B_{\infty,1}(\RR^2)}$. It suffices therefore to bound $\|  f \|_{ B_{\infty,1}(\RR^2)}$. 
	
	By scaling, it suffices to show $\|  f \|_{ B_{\infty,1}(\RR^2)}\ls 1$ (with an implicit constant \underline{independent} of $\de$), assuming there exists $\de>0$ such that
	\begin{align} \label{LP1}
	\| f\|_{L^2(\RR^2)} \leq &\: \delta^{ \frac{1}{2}}, \\
	\label{LP2}
	\| \partial_{ y^2} |D_{ y}|^{\sigma} f \|_{L^2(\RR^2)} \leq &\: \delta^{-\sigma +\frac{1}{2}}, \\
	\label{LP3}
	\| \rd_{ y} f \|_{L^2(\RR^2)} \leq &\: \delta^{ -\frac{1}{2}}, \\
	\label{LP4}
	\| \partial_{ y^2} f\|_{L^2(\RR^2)} \leq &\: \delta^{ \frac{1}{2}}.
	\end{align}
	
	Now we estimate, using the Cauchy--Schwarz inequality, the Plancherel identity, and the easy volume estimate $|\{  |\xi| \sim  2^k , |\xi_{ 2}| \sim 2^l \}| \sim 2^{k} \cdot 2^{l}\color{black}$, that
	$$  \| \hat{f}_{k l } \|_{L^{1}(\RR^2)} \lesssim  \| \hat{f}_{k l } \|_{L^{2}(\RR^2)} \cdot 2^{\frac{k}{2}} \cdot 2^{\frac{l}{2}} \color{black}\lesssim \| f_{k l } \|_{L^{2}(\RR^2)} \cdot 2^{\frac{k}{2}} \cdot 2^{\frac{l}{2}}.$$
	It then follows from \eqref{LP1}--\eqref{LP4} and the support properties of the Littlewood--Paley pieces that
	\begin{align}
	\| \hat{f}_{k l } \|_{L^{1}(\RR^2)} \lesssim &\: \| f_{k l } \|_{L^{2}(\RR^2)} \cdot 2^{\frac{k}{2}} \cdot 2^{\frac{l}{2}} \ls  \delta^{\frac{1}{2}} \cdot 2^{\frac{k}{2}} \cdot 2^{-\frac{l}{2}},\color{black} \label{NLP1}\\
	\| \hat{f}_{k l } \|_{L^{1}(\RR^2)}  \lesssim &\: \| \partial_{ y^2} |D_{ y} |^{\sigma}f_{k l} \|_{L^2(\RR^2)} \cdot 2^{k \cdot (\frac{1}{2}-\sigma)} \cdot 2^{-\frac{l}{2}} \ls \delta^{-\sigma +\frac{1}{2}} \cdot 2^{k \cdot (\frac{1}{2} - \sigma)} \cdot 2^{-\frac{l}{2}}, \label{NLP2}\\
	\| \hat{f}_{k l } \|_{L^{1}(\RR^2)}  \lesssim &\: \| \rd_{ y} f_{k l} \|_{L^2(\RR^2)} \cdot 2^{-\frac{k}{2}} \cdot 2^{\frac{l}{2}} \ls \delta^{ -\frac{1}{2}} \cdot 2^{-\frac{k}{2}} \cdot 2^{\frac{l}{2}}, \label{NLP3}\\
	\| \hat{f}_{k l } \|_{L^{1}(\RR^2)}  \lesssim &\: \| \partial_{ y^2} f_{k l} \|_{L^2(\RR^2)}  \cdot 2^{\frac{k}{2}} \cdot 2^{-\frac{l}{2}} \ls \delta^{ \frac{1}{2}} \cdot 2^{\frac{k}{2}} \cdot 2^{-\frac{l}{2}}. \label{NLP4}
	\end{align}
	
	We divide the sum into four cases, depending on the values of $k$ and $l$: \begin{enumerate}
		\item \label{LPsub1} When $\delta^{-1} \lesssim 2^k$ and $l \leq 0$, we use \eqref{NLP3}  to obtain
		
		%	$$ \| \sum_{ l\in \mathbb{Z} / 2^{l} \lesssim 2^{k \cdot (\frac{1}{2}-\epsilon_0)}} u_{k l } \|_{L^{\infty}(\RR^2)} \lesssim \delta^{ -\frac{1}{4}-\frac{\epsilon_0}{2}} \cdot 2^{ -k \cdot (\frac{1}{4}+\frac{\epsilon_0}{2})} . $$ 
		
		%	Now we sum in $k$ and still using geometric series we get 
		
		\begin{equation} \label{LPsub1eq}
		\sum_{ \delta^{-1} \lesssim 2^k, l \leq 0}  	\|\hat{f}_{k l }  \|_{L^{1}(\RR^2)} \lesssim 1.
		\end{equation}
		\item  \label{LPsub2} \blue{W}hen $\delta^{-1} \lesssim 2^k$ and $l \geq 0$\blue{,} we use \eqref{NLP2} to obtain
		
		%	$$ \| \sum_{ l\in \mathbb{Z} / 2^{l} \gtrsim 2^{k \cdot (\frac{1}{2}-\epsilon_0)}} u_{k l } \|_{L^{\infty}(\RR^2)} \lesssim \delta^{-\alpha + \frac{1}{4}+\frac{\epsilon_0}{2}} \cdot 2^{ -k \cdot (\alpha -\frac{1}{4}-\frac{\epsilon_0}{2})} . $$ 
		
		%	Now we sum in $k$ and we get 
		
		\begin{equation} \label{LPsub2eq}
		\sum_{ \delta^{-1} \lesssim 2^k, l \geq 0}  	\| \hat{f}_{k l }  \|_{L^{1}(\RR^2)} \lesssim 1.
		\end{equation}
		\item \label{LPsub3} \blue{W}hen $ 2^k \lesssim \delta^{-1}  $ and $l \leq 0$\blue{,} we use \eqref{NLP1} to obtain
		
		%	$$ \| \sum_{ l\in \mathbb{Z} / 2^{l} \lesssim \delta^{-\frac{1}{2}+\epsilon_0}} u_{k l } \|_{L^{\infty}(\RR^2)} \lesssim \delta^{ \frac{1}{2}} \cdot 2^{  \frac{k}{2}} . $$ 
		
		%	Now we sum in $k$ and we get 
		
		\begin{equation} \label{LPsub3eq}
		\sum_{ \delta^{-1} \gtrsim 2^k, l \leq 0} 	\| \hat{f}_{k l }  \|_{L^{1}(\RR^2)} \lesssim 1.
		\end{equation}
		
		\item \label{LPsub4} \blue{W}hen $2^k \lesssim \delta^{-1}$ and $l \geq 0$\blue{,} we use \eqref{NLP4} to obtain
		
		%$$ \| \sum_{ l \in \mathbb{Z} /  \delta^{-\frac{1}{2}+\epsilon_0} \lesssim 2^{l} } u_{k l } \|_{L^{\infty}(\RR^2)} \lesssim \delta^{ \frac{1}{2}} \cdot 2^{  \frac{k}{2}} . $$ 
		
		%Now we sum in $k$ and we get 
		
		\begin{equation} \label{LPsub4eq}
		\sum_{ \delta^{-1} \gtrsim 2^k, l \geq 0} 	\|\hat{f}_{k l }  \|_{L^{1}(\RR^2)}\lesssim 1.
		\end{equation}

		\color{black}
	\end{enumerate} 
	
	Now, combining \eqref{LPsub1eq}--\eqref{LPsub4eq}, it is clear by the triangle inequality  that 
		
		$$ \| f  \|_{B_{\infty,1}(\RR^2)} \leq  \sum_{k \in \mathbb{N}\cup \{0\}} \| \hat{f}_{k}  \|_{L^{1}(\RR^2)}  \leq \sum_{k \in \mathbb{N}\cup \{0\},\, l  \in \mathbb{Z}} \| \hat{f}_{k l }  \|_{L^{1}(\RR^2)} \lesssim 1.   $$

\end{proof}

\subsubsection{Anisotropic Sobolev embedding into H\"older spaces}

Our next goal is an anisotropic Sobolev embedding which maps into H\"older spaces on a half space. The main result is given in Theorem~\ref{embeddingexterior} below. We will start with the following lemma, which is a variant of the desired estimate, but on all of $\RR^2$.

\begin{lem} \label{lem:sobolev.holder}
	Let $s \in (0,\f 12)$. The following estimate holds for all sufficiently regular functions $f$ (with an implicit constant depending on $s$):
	\begin{equation*}
	\| f \|_{C^{0,\f s2}(\RR^2)} \lesssim \| f \|_{H^1(\RR^2)}+	\| \partial_{ y^2} |D_{ y}|^{s}f \|_{L^2(\RR^2)}.
	\end{equation*}
\end{lem}

\begin{proof} As in the proof of Theorem \ref{embeddingThmInterior}, we bound the $L^\infty$ norm of each Littlewood--Paley piece in different ways using the Hausdorff--Young, Cauchy--Schwarz inequalities, Plancherel's theorem and the volume estimate in Fourier space. Hence, denoting $ \|f \|:= \| f \|_{H^1(\RR^2)}+	\| \partial_{ y^2} |D_{ y}|^{s} f \|_{L^2(\RR^2)}$, we have
	\begin{align*}
	\| f_{k l } \|_{L^{\infty}(\RR^2)}  \lesssim &\: \| f_{k l } \|_{L^{2}(\RR^2)} \cdot 2^{\frac{k}{2}} \cdot 2^{\frac{l}{2}} \lesssim \|f \| \cdot 2^{\frac{k}{2}} \cdot 2^{\frac{l}{2}}, \\
	\| f_{k l } \|_{L^{\infty}(\RR^2)}  \lesssim &\: \| \partial_{ y^2} |D_{ y}|^{s}f_{k l} \|_{L^2(\RR^2)} \cdot 2^{k \cdot (\frac{1}{2}- s)} \cdot 2^{-\frac{l}{2}} \lesssim \|f \| \cdot 2^{k \cdot (\frac{1}{2}-s)} \cdot 2^{-\frac{l}{2}}, \\
	\| f_{k l } \|_{L^{\infty}(\RR^2)}   \lesssim &\: \| \rd_{ y} f_{k l} \|_{L^2(\RR^2)} \cdot 2^{-\frac{k}{2}} \cdot 2^{\frac{l}{2}} \lesssim \|f \| \cdot 2^{-\frac{k}{2}} \cdot 2^{\frac{l}{2}} ,\\
	\| f_{k l } \|_{L^{\infty}(\RR^2)}  \lesssim &\: \| \partial_{ y^2} f_{k l} \|_{L^2(\RR^2)}  \cdot 2^{\frac{k}{2}} \cdot 2^{-\frac{l}{2}} \lesssim \|f \| \cdot 2^{\frac{k}{2}} \cdot 2^{-\frac{l}{2}}.
	\end{align*}
	
	Without loss of generality, we take  $\|f \|=1$. Thus, 	
	\begin{align}
	\| f_{k l } \|_{L^{\infty}(\RR^2)}  \lesssim &\: 2^{\f k2} \cdot 2^{\frac{l}{2}} , \label{eq:fkl.Holder.1}\\
	2^{\f{ks}2} \cdot \| f_{k l } \|_{L^{\infty}(\RR^2)}  \lesssim  &\: 2^{k \cdot (\frac{1}{2} - \f s2)} \cdot 2^{-\frac{l}{2}}, \label{eq:fkl.Holder.2}\\
	2^{\f{ks}2} \cdot\| f_{k l } \|_{L^{\infty}(\RR^2)}   \lesssim  &\:2^{k \cdot (\f s2 -\frac{1}{2})} \cdot 2^{\frac{l}{2}}, \label{eq:fkl.Holder.3}\\
	\| f_{k l } \|_{L^{\infty}(\RR^2)}    \lesssim  &\: 2^{\f k2} \cdot 2^{-\frac{l}{2}}. \label{eq:fkl.Holder.4}
	\end{align}
	
	For all $k\geq 0$, we use \eqref{eq:fkl.Holder.2} and \eqref{eq:fkl.Holder.3} respectively to sum $l\geq (1-s)k$ and $l\leq (1-s)k$ to obtain
	$$\sum_{l\in \mathbb Z} 2^{\f{ks}2} \cdot \| f_{k l } \|_{L^{\infty}(\RR^2)} \ls \sum_{ l \geq (1-s)k}2^{\f{ks}2} \cdot \| f_{k l } \|_{L^{\infty}(\RR^2)} + \sum_{ l \leq (1-s)k}2^{\f{ks}2} \cdot \| f_{k l } \|_{L^{\infty}(\RR^2)}  \lesssim   1.$$
	%	When $k  \leq 0$, we use \eqref{eq:fkl.Holder.4} for $l\geq 0$ and \eqref{eq:fkl.Holder.1} for $l\leq 0$ to obtain		$$\sum_{k\leq 0} \sum_{l \in \mathbb Z} \| f_{k l } \|_{L^{\infty}(\RR^2)} \ls \sum_{k \leq 0} 2^{\f k2}(\sum_{l\geq 0} 2^{-\f l2} + \sum_{l\leq 0} 2^{\f l2}) \ls \sum_{k\leq 0} 2^{\f k2}  \lesssim 1.$$
	
	Recalling that $f_k = \sum_{l \in \mathbb{Z} } f_{k l}$, the above inequalities and the triangle inequality thus implies
	$$  \| f \|_{C^{0,\f s2}(\RR^2)} \sim %sum_{k\leq 0} \|f_k\|_{L^\infty(\RR^2)} +  
	\sup_{k \geq 0} 2^{\f{sk}2} \cdot \| f_{k} \|_{L^{\infty}(\RR^2)}   \lesssim 1$$
	where we have used the Littlewood--Paley characterization of the H\"older space.	\qedhere
\end{proof}

Using the above lemma, we obtain the main result of this subsubsection via a reflection argument.

\begin{thm} \label{embeddingexterior}
	Let $s\in (0,\f 12)$, $a\in \mathbb R$ and $\Omega_{L}:= (-\infty,a) \times \RR$ be the open left half plane. 
	
	Then the following holds for all $v\in \mathcal S(\RR^2)$ with an implicit constant depending only on $s$:
	$$\| v \|_{C^{0,\f s2}(\Omega_L)} \lesssim \| v_{|\Omega_L} \|_{H^1(\Omega_{L})} + \| \partial_{ y^2} |D_{ y}|^{s }v\|_{L^2(\RR^2)}.$$
	
	Moreover, $v_{|\Omega_L}$ can be extended to a $C^{0,\f s2}{(\RR^2)}$ function $Rv:\mathbb R^2\to \mathbb R$ such that 
	$$\|Rv\|_{C^{0,\f s2}(\RR^2)} \ls \| v_{|\Omega_L} \|_{H^1(\Omega_{L})} + \| \partial_{ y^2} |D_{ y}|^{s }v\|_{L^2(\RR^2)}.$$
\end{thm}
\begin{proof} 
	Our strategy is to extend $v_{|\Omega_L}$ into a function $Rv:\RR^2\to \RR$, which may differ from $v$, but for which we can prove boundedness using Lemma~\ref{lem:sobolev.holder}.
	%\footnote{Indeed, with our hypothesis, it is still possible that $v\notin L^{\infty}(\RR^2)$. When we will apply this result to $f=\partial \phi$ to prove that $\partial \phi \in L^{\infty}(\Omega_L)$, we can however infer that $\partial \phi \in L^{\infty}(\RR^2-\Omega_L)$ using Theorem \ref{embeddingThmInterior}.}.
	
	By a standard Sobolev extension result (see \cite[Theorem~5.19]{Adams}), there exists a bounded linear extension operator $E: H^1(\Omega_L) \to H^1(\mathbb R^2)$ satisfying $Ef_{|\Omega_L} = f$ (which is also bounded $E:L^2(\Omega_L) \to L^2(\mathbb R^2)$). 
	
	As a consequence, defining $Rf = E(f_{|\Omega_L})$, we have
	\begin{equation}\label{eq:reflection.boundedness.1}
	\|Rv \|_{H^1(\RR^2)}  \ls \| v_{|\Omega_L} \|_{H^1(\Omega_{L})};
	\end{equation}
	and using moreover that $\rd_{ y^2}$ is tangential to the boundary together with an interpolation argument, we obtain
	\begin{equation}\label{eq:reflection.boundedness.2}
	\| \partial_{ y^2} |D_{ y}|^{s }(Rv)\|_{L^2(\RR^2)} \ls  \| \partial_{ y^2} |D_{ y}|^{s } v\|_{L^2(\RR^2)}.
	\end{equation}
	
	Since $Rv_{|\Omega_L} = v$, by \eqref{eq:reflection.boundedness.1}, \eqref{eq:reflection.boundedness.2} and Lemma~\ref{lem:sobolev.holder}, 
	$$\| v\|_{C^{{0,\f s2}}(\Omega_L)} \ls \| Rv\|_{C^{0,\f s2}(\RR^2)} \ls  \|Rv \|_{H^1(\RR^2)} + \| \partial_{ y^2} |D_{ y}|^{s }(Rv)\|_{L^2(\RR^2)} \ls \| v_{|\Omega_L} \|_{H^1(\Omega_{L})} + \| \partial_{ y^2} |D_{ y}|^{s }v\|_{L^2(\RR^2)}. \qedhere$$
\end{proof}

\subsection{Converting the estimates into $(u_k,u_{k'})$ coordinates} 
In this subsection, we convert the $L^2$ bounds in $\mathcal E$ (see Definition~\ref{def:Lipschitz.control.norm}), which are defined with respect to the $(x^1,x^2)$ coordinate system in the elliptic gauge and with the geometric vector field $E_k$, into estimates in the $(u_k, u_{k'})$ coordinate system. This will later allow us to apply the embedding results obtained in Section~\ref{LPsection} in the $(u_k, u_{k'})$ coordinate system.

In the remainder of the section, recall the coordinate system $(u_k, u_{k'})$ and the notations introduced in Section~\ref{ukuk'coordinatesection}. In particular, recall that $(\srd_{u_k},\srd_{u_{k'}})$ denote the coordinate partial derivatives in the $(u_k, u_{k'})$ coordinates.

\subsubsection{{Equivalence of $L^p$ and $W^{1,p}$ norms}}

\begin{lemma}\label{lem:Lp.W1p.eq}
For any $p \in [1,\infty]$, 
$$ \| f\|_{L^p_{x^1,x^2}(\Sigma_t)} \ls \| f\|_{L^p_{u_k,u_{k'}}(\Sigma_t)} \ls \| f\|_{L^p_{x^1,x^2}(\Sigma_t)},\quad  \| f\|_{W^{1,p}_{x^1,x^2}(\Sigma_t)} \ls \| f\|_{W^{1,p}_{u_k,u_{k'}}(\Sigma_t)} \ls \| f\|_{W^{1,p}_{x^1,x^2}(\Sigma_t)}.$$
\magenta{Similar estimates hold when the $L^p$ and $W^{1,p}$ norms are taken over subsets of $\Sigma_t$.}
\end{lemma}
\begin{proof}
This is an immediate consequence of \blue{\eqref{eq:COV.1}--\eqref{eq:COV.2}}. \qedhere
\end{proof}

Because of the above lemma, \textbf{for the remainder of the section, we will write $L^p(\Sigma_t)$, etc.~without precisely indicating whether the coordinate system $(x^1,x^2)$ or $(u_k, u_{k'})$ is used.}

\subsubsection{$L^2$ estimates involving $E_k$}
Lemma~\ref{lem:Lp.W1p.eq} controls the change of variables for isotropic $L^p$ or $W^{1,p}$ spaces. In this subsection, we translate some estimates in $\mathcal E$ that involve the good derivative $E_k$, and write them in terms of $\srd_{u_{k'}}$; see Lemma \ref{lem:higher.order.coord.change}. 

We begin with a simple lemma.
\begin{lem}
For $k \neq k'$, the following holds for all $t\in [0,T_B)$:
	\begin{equation}\label{eq:mu.gEX.est}
	\| \mu_k \cdot g(E_k,X_{k'})^{-1}\|_{W^{1,\i}(\Sigma_t\cap B(0,R))} \ls 1.
	\end{equation}
\end{lem}
\begin{proof}
	This follows from \eqref{mu.main.estimate}, \eqref{eq:g.main}, \eqref{eq:frame.1}, \eqref{eq:frame.2} and \eqref{anglecontrol}. \qedhere
\end{proof}

\begin{lem}[\blue{$L^2$ estimates under coordinate change}] \label{lem:higher.order.coord.change}
	For $k \neq k'$, the following holds for all $t\in [0,T_B)$:
	\begin{equation}\label{eq:lower}
\|\partialukp \rd \tphi\|_{L^2(\Sigma_t\cap S^k_{2\de})} \ls \de^{\f 12} \cdot \mathcal E.
	\end{equation}	
		\begin{equation} \label{LP3core}
	\| \partialukukp  \partial \tphi \|_{L^2(\Sigma_t)} + \| \partialukpukp \partial \tphi \|_{L^2(\Sigma_t)} \lesssim \delta^{-\frac{1}{2}} \cdot \mathcal E,
	\end{equation}
\end{lem}
\begin{proof}
	The bound \eqref{eq:lower} \color{black} follows from \eqref{partialukpEX}, Lemma~\ref{lem:Lp.W1p.eq} and the definition of $\mathcal E$ (Definition~\ref{def:Lipschitz.control.norm}).

 For \eqref{LP3core}, \color{black} we first control the $\srd_{u_{k'}}$ or $\srd_{u_k}$ derivative by $\rd_x$ using Lemma~\ref{lem:Lp.W1p.eq}. We then write $\srd_{u_{k'}}$ in terms of $E_k$ using \eqref{partialukpEX}. Finally,  applying the product rule and \eqref{eq:mu.gEX.est}, and using Definition~\ref{def:Lipschitz.control.norm}, we obtain
	\begin{equation*}
	\begin{split}
	&\: \| \partialukukp  \partial \tphi \|_{L^2(\Sigma_t)} + \| \partialukpukp \partial \tphi \|_{L^2(\Sigma_t)} \ls  \| \rd_x \partialukp \rd\tphi \|_{L^2(\Sigma_t)}  = \| \rd_x [\mu_{k'} \cdot g(E_k,X_{k'}) ^{-1} E_k \rd\tphi ]\|_{L^2(\Sigma_t)} \\
	\ls &\: \| \mu_k \cdot g(E_k,X_{k'})^{-1} \|_{W^{1,\i}(\Sigma_t\cap B(0,R))} ( \|E_k \rd \tphi\|_{L^2(\Sigma_t)} + \|\rd_x E_k \rd\tphi\|_{L^2(\Sigma_t)})  \ls \de^{-\f 12} \cdot \mathcal E.
	\end{split}
	\end{equation*}
\qedhere	
	
\end{proof}

\subsubsection{$L^2$ estimates involving fractional derivatives}

\begin{lem} \label{fracchangethm}  Let $(y_1,y_2)$, $(z_1, z_2)$ be two systems of coordinates on $\RR^2$ such that 
	
	\begin{equation} \label{det=1}
	1\ls \Big| \det \left( \frac{\partial z_i}{\partial y_j}\right)_{i j} \Big| \ls 1,\quad \Big|  \frac{\partial z_i}{\partial y_j} \Big| \lesssim 1.
	\end{equation}

	Then, for every $0<\sigma <1$, the following holds for all $f \in \mathcal S(\RR^2)$:
	\begin{equation} \label{fracchangecoord2}
	\|  \la D_{{z_1,z_2}} \ra^{\sigma} {f} \|_{{L^2_z}(\RR^2)}\ls \| \la D_{y_1,y_2} \ra ^{\sigma} {f} \|_{{L^2_y}(\RR^2)} \lesssim  \|  \la D_{{z_1,z_2}} \ra^{\sigma} {f} \|_{{L^2_z}(\RR^2)}.
	\end{equation}

\end{lem}

\begin{proof} Define the change of variable {map} $y: z \in \RR^2 \rightarrow y(z) \in \RR^2$ and define the linear map $\varPhi_y: f \in L^2(\RR^2) \rightarrow \varPhi_y(f):= f \circ y \in L^2(\RR^2)$. 
 
	The bounds \eqref{det=1} obviously imply that $\Phi_y:L^2(\RR^2) \to L^2(\RR^2)$, $\Phi_y: H^1(\RR^2)\to H^1(\RR^2)$ are bounded maps with bounded inverses. The desired conclusion thus follows from interpolation. \qedhere	
	
\end{proof}  

Returning to our setting, this implies by Lemma \ref{lem:jacobian} that
\begin{lemma}\label{lem:Dx.Du.s.change}
	For $s\in \{s', s''\}$, the following estimate holds for all Schwartz function $f$:
	$$ \|\langle D_x\rangle^{s} f \|_{L^2(\Sigma_t)} \ls \|\langle D_{u_k,u_{k'}}\rangle^{s} f \|_{L^2_{u_k,u_{k'}}(\Sigma_t)} \ls \|\langle D_x\rangle^{s} f \|_{L^2(\Sigma_t)}.$$
\end{lemma}

After the above preliminaries, we are ready to translate the control for $\mathcal E$ into $L^2$ estimates on the derivatives of $\rd\tphi$ in the $(u_k, u_{k'})$ coordinate system. We begin with the

\begin{lem}\label{lem:du.Ds.dtphi}
For $k \neq k'$, the following holds for all $t\in [0,T_B)$:
	$$\| \partialukp  \fracDu^{s''} \partial \tphi \|_{L^2_{u_k, u_{k'}}(\Sigma_t)} \ls \mathcal E.$$
\end{lem}
\begin{proof}
	Consider the following chain of estimates:
	\begin{align}
	&\: \| \partialukp  \fracDu^{s''} \partial \tphi \|_{{L^2}(\Sigma_t)} = \|   \fracDu^{s''} \partialukp \partial \tphi \|_{{L^2}(\Sigma_t)} \label{eq:du.Ds.dtphi.1}\\
	\ls &\: \|   \Db^{s''} \partialukp \partial \tphi \|_{L^2(\Sigma_t)} = \|   \Db^{s''} (\mu_{k'} g(E_k,X_{k'})^{-1} E_k \partial \tphi )\|_{L^2(\Sigma_t)} \label{eq:du.Ds.dtphi.2}\\
	\ls &\: \|\Db^{s''} (\varpi \mu_{k'} g(E_k,X_{k'})^{-1}) \|_{L^\i(\Sigma_t)} \|   \Db^{s''} ( E_k \partial \tphi )\|_{L^2(\Sigma_t)}  \ls \| \Db^{s''} E_k \partial \tphi\|_{L^2(\Sigma_t)} \label{eq:du.Ds.dtphi.3}\\
	\ls &\:  \| E_k \Db^{s''} \rd \tphi\|_{L^2(\Sigma_t)} + \| \Db^{s''-1} \rd_x \partial \tphi\|_{L^2(\Sigma_t)} \label{eq:du.Ds.dtphi.4}\\
	\ls &\: \| \rd E_k \Db^{s''}  \tphi\|_{L^2(\Sigma_t)} + \|\Db^{s''} \rd \tphi\|_{L^2(\Sigma_t)}. \label{eq:du.Ds.dtphi.5}
	\end{align}
	
	For \eqref{eq:du.Ds.dtphi.2}, we first use  Lemmas~\ref{lem:Dx.Du.s.change} and \ref{lem:Lp.W1p.eq}, and  then \eqref{partialukpEX}. To obtain the first inequality in \eqref{eq:du.Ds.dtphi.3}, we use Lemma~\ref{lem:frac.product} (and $\mathrm{supp}(\tphi) \subseteq B(0,R)$). In the second inequality, $\| \mu_{k'} g(E_k,X_{k'})^{-1} \|_{W^{1,\infty}(B(0,3R))}$ {is bounded} using \eqref{eq:mu.gEX.est}. For \eqref{eq:du.Ds.dtphi.4}, we write $E_k = E_k^i \rd_i$ and use Proposition~\ref{prop:commute.2} and \eqref{eq:frame.2} to estimate the commutator $[\Db^{s''}, E_k^i]$.  For \eqref{eq:du.Ds.dtphi.5}, the first term is obtained after commuting $[\rd, E_k]$, and using \eqref{eq:frame.2} again; while the second term is obtained by the $L^2$-boundedness of the inhomogeneous Riesz transform $\Db^{-1} \rd_x$. Finally, note that both terms on \eqref{eq:du.Ds.dtphi.5} are controlled by $\mathcal E$, which concludes the proof. \qedhere
\end{proof}

\begin{lem}\label{lem:rphi.in.ukukp.energy}
For $k \neq k'$, the following holds for all $t\in [0,T_B)$:
	$$\| \partialukp  \fracDu^{s''} \rd \phi_{reg} \|_{L^2_{u_k, u_{k'}}(\Sigma_t)} \ls \mathcal E.$$
\end{lem}
\begin{proof}
	This is similar to Lemma~\ref{lem:du.Ds.dtphi} except it is much easier because we control $\| \rd \rd_x \Db^{s''}  \phi_{reg}\|_{L^2(\Sigma_t)}$ and $\|\Db^{s''} \rd \phi_{reg}\|_{L^2(\Sigma_t)}$ (instead of only $\| \rd E_k \Db^{s''}  \phi_{reg}\|_{L^2(\Sigma_t)}$ and $\|\Db^{s''} \rd \phi_{reg}\|_{L^2(\Sigma_t)}$); we omit the details. \qedhere
\end{proof}

\subsection{Boundedness of $ \partial \tphi$ in the singular region: proof of \eqref{dtphiboundedint}}\label{sec:Lipschitz.rd.tphi}

In the next few subsections, we will prove the estimates asserted in Theorem \ref{dphiboundedprop}; see the conclusion of the proof in Section~\ref{sec:Lipschitz.conclusion}. We begin with \eqref{dtphiboundedint}.

\begin{proof}[Proof of \eqref{dtphiboundedint}]
	We  apply Theorem~\ref{embeddingThmInterior} to $f = \rho_k \cdot \partial \tphi$ in the coordinate system $(y^1,y^2) = (u_k,u_{k'})$ and with $\sigma = 1$. Note that $\partial_{y^2}$ in the notations of Theorem~\ref{embeddingThmInterior} corresponds to $\partialukp$. In order to use Theorem~\ref{embeddingThmInterior}, it suffices to show that
	\begin{align*} 
&	\underbrace{\de^{-\f 12} \| \rho_k \rd\tphi \|_{L^2(\Sigma_t)}}_{=:I},\, \: \underbrace{\delta^{\f 12}(\| \partialukpukp  (\rho_k \rd\tphi ) \|_{L^2(\Sigma_t)} + \| \partialukukp  (\rho_k \rd\tphi ) \|_{L^2(\Sigma_t)}\blue{)}}_{=:II},\\
	&\: \qquad\qquad \underbrace{\de^{\f 12} \left(\| \partialuk (\rho_k \rd\tphi) \|_{ L^2(\Sigma_t)}+ \|  \partialukp (\rho_k \rd\tphi) \|_{L^2(\Sigma_t)}\right)}_{=:III},\, \underbrace{\de^{-\f 12} \| \partialukp (\rho_k \rd\tphi)\|_{L^2(\Sigma_t)}}_{=:IV} \ls \mathcal E.
	\end{align*}
	
	To control the terms $I$, $II$, $III$ and $IV$ above, first note that the cutoff function $\rho_k$ satisfies
	\begin{equation}\label{eq:rhok.prop}
	|\rho_k|\ls 1,\quad \partialukp \rho_k = 0, \quad |\partialuk \rho_k| \ls \de^{-1}, \quad\mathrm{supp}(\rho_k) \subseteq S^k_{2\de}.
	\end{equation}
	
	We first bound term $I$, using \eqref{eq:rhok.prop} and the $\de^{-\f 12} \|\partial \tphi\|_{L^2(\Sigma_t \cap S_{2\de}^k)}$ term in $\mathcal E$.
	
	To bound term $II$, we use \eqref{eq:rhok.prop} together with the estimates in Lemma~\ref{lem:higher.order.coord.change}\color{black}. 
	
	To estimate $III$, we use \eqref{eq:rhok.prop} and the bounds for $\de^{-\f 12} \|\partial \tphi\|_{L^2(\Sigma_t \cap S_{2\de}^k)} $ and $\de^{\f 12} \| \partial^2 \tphi\|_{L^2(\Sigma_t)}$ in $\mathcal E$.
	
	Finally, the bound for term $IV$ can be obtained using \eqref{eq:rhok.prop} and Lemma~\ref{lem:higher.order.coord.change}\color{black}. \qedhere
\end{proof}

\subsection{H\"older estimates for $ \partial \tphi$ away from the singular zone: proof of \eqref{dtphiboundedext} and \eqref{dtphiboundedext2}}
%The objective of this sub-section is to establish \eqref{dtphiboundedext}.

Even though we are interested in the H\"older estimates in the coordinates of the elliptic gauge (see \eqref{dtphiboundedext}, \eqref{dtphiboundedext2} and Definition~\ref{def:Holder}), in order to make use of the good derivative, we will apply Theorem \ref{embeddingexterior} in the $(u_k, u_{k'})$ coordinate system. Nevertheless, it is easy to check that the H\"older norms in these two coordinate systems are equivalent as we will state in the following lemma.
\begin{lemma}\label{lem:Holder.eq}
	For any $\sigma \in (0,1)$, and any open domain $\Omega\subseteq \Sigma_t$ with a Lipschitz boundary, the following holds for all Schwartz functions $f$:
	$$\|f\|_{C^{ \sigma}_{u_k, u_{k'}}(\Sigma_t\cap \Omega)} \ls \|f\|_{C^{ \sigma}(\Sigma_t \cap \Omega)} \ls \|f\|_{C^{ \sigma}_{u_k, u_{k'}}(\Sigma_t \cap \Omega)}.$$
\end{lemma}
\begin{proof}
	This is an immediate consequence of Lemma \ref{jacobian}. \qedhere
\end{proof}

We are now ready to prove \eqref{dtphiboundedext} and \eqref{dtphiboundedext2}.
\begin{proof}[Proof of \eqref{dtphiboundedext} and \eqref{dtphiboundedext2}]
	\pfstep{Step~1: Proof of \eqref{dtphiboundedext}} In view of Lemma~\ref{lem:Holder.eq}, it suffices to prove H\"older estimates in the $(u_k, u_{k'})$ coordinates. 
	We  apply Theorem \ref{embeddingexterior} to $v= \partial \tphi$, $\Omega_L:= \{ u_k {\geq} \de,\  u_{k'} \in \RR\} \subseteq \Sigma_t$, $s=s''$ in the coordinate system $(u_k,u_{k'})$. Note (as in Section~\ref{sec:Lipschitz.rd.tphi}) that $\partial_{y^2} = \partialukp$.
	
By Theorem \ref{embeddingexterior}, we know that \eqref{dtphiboundedext} holds as long as we can verify
		\begin{equation}\label{eq:embeddingexterior.to.verify}
		\| (\rd\tphi)_{|\Omega_L} \|_{H^1(\Omega_{L})} + \| \partialukp \la D_{u_k.u_{k'}}\ra^{s''} \rd\tphi\|_{L^2(\Sigma_t)} \ls \mathcal E.
		\end{equation}
		
		Now the first term in \eqref{eq:embeddingexterior.to.verify} can be controlled using the $\| \partial \Db^{s'} \tphi\|_{L^2(\Sigma_t)}$ and $\| \partial^2 \tphi\|_{L^2(\Sigma_t \setminus \Sd^k)}$ terms in the definition of $\mathcal E$ (Definition~\ref{def:Lipschitz.control.norm}); while the second term is controlled by Lemma~\ref{lem:du.Ds.dtphi}.

		\pfstep{Step~2: Proof of \eqref{dtphiboundedext2}} First, notice that the application of Theorem \ref{embeddingexterior} in Step~1 in fact gives a stronger result: namely, $\rd\tphi$ admits an extension $R\rd\tphi$ defined on the whole $\Sigma_t$ so that $\|R\rd\tphi\|_{C^{0,\f s2}_{u_k,u_{k'}}(\Sigma_t)} \ls \mathcal E$.
		
		Let $\varpi$ be the cutoff function as in the beginning of Section \ref{sec:frac.notation}. It can be checked explicitly that $\|\varpi\cdot (1 - \rho_k)\|_{B^{u_k,u_{k'}}_{\infty,1}} \ls 1$. (If the reader \blue{prefers} not to carry out the explicit estimate for the corresponding oscillatory integral, one can more easily check that $\varpi\cdot (1 - \rho_k)$ obeys the assumptions of Theorem~\ref{embeddingThmInterior}, and apply Theorem~\ref{embeddingThmInterior} to deduce the Besov bound.)
		
		Using the fact that $\|f_1 f_2\|_{B^{u_k,u_{k'}}_{\infty,1}} \ls  \|f_1 \|_{B^{u_k,u_{k'}}_{\infty,1}} \|f_2 \|_{B^{u_k,u_{k'}}_{\infty,1}}$, we have $\|\varpi\cdot (1 - \rho_k) \cdot R\rd \tphi\|_{B^{u_k,u_{k'}}_{\infty,1}} \ls \mathcal E$. Finally, one checks that the support properties for $\varpi$, $1-\rho_k$ and $\tphi$ imply that $(1-\rho_k) \cdot \rd\tphi = \varpi\cdot (1 - \rho_k) \cdot R\rd \tphi$. This concludes the proof of \eqref{dtphiboundedext2}. \qedhere

\end{proof}

\subsection{H\"older estimates for the regular part: proof of \eqref{eq:phireg.Linfty}}\label{sec:Holder.rphi}

Having completed the estimates for $\tphi$, we now turn to the estimate for $\rphi$.

\begin{proof}[Proof of \eqref{eq:phireg.Linfty}]
	We begin with the first estimate in \eqref{eq:phireg.Linfty}. Pick any $k'\neq k$. By the equivalence of the H\"older norms (Lemma~\ref{lem:Holder.eq}), it suffices to prove that $\rd\phi_{reg}$ is in $C^{0,\f{s''}2}_{u_k, u_{k'}}(\Sigma_t)$. For this, we apply Lemma~\ref{lem:sobolev.holder} in the $(u_k, u_{k'})$ coordinate system.
	
	It suffices to check
	$$\| \rd\rphi \|_{L^2(\Sigma_t)} +  \|\srd_{u_k} \rd \phi_{reg}\|_{L^2(\Sigma_t)} + \| \srd_{u_{k'}} \rd \phi_{reg}\|_{L^2(\Sigma_t)} + \| \srd_{u_k} \langle D_{u_k,u_{k'}} \rangle^{s'} \rd \phi_{reg}\|_{L^2(\Sigma_t)} \ls \mathcal{E}.$$
	The bounds for the first three terms follow directly from the definition of $\mathcal E$ (Definition~\ref{def:Lipschitz.control.norm}) and Lemma~\ref{lem:Lp.W1p.eq}, while the last term is controlled in Lemma~\ref{lem:rphi.in.ukukp.energy}.
	
	Finally, since $C^{0,\f{s''}2}_{u_k, u_{k'}}\subseteq B^{u_k,u_{k'}}_{\infty,1}(\Sigma_t)$ continuously, we obtain the second estimate in \eqref{eq:phireg.Linfty}. \qedhere
\end{proof}

\subsection{Conclusion of the proof of Theorem \ref{dphiboundedprop}: proof of \eqref{dphibounded}}\label{sec:Lipschitz.conclusion}

In view of the estimates derived in Sections~\ref{sec:Lipschitz.rd.tphi}--\ref{sec:Holder.rphi}, in order to conclude the proof of Theorem \ref{dphiboundedprop}, it suffices to prove \eqref{dphibounded}.
\begin{proof}[Proof of \eqref{dphibounded}]
	In the following, we use that $B^{u_k,u_{k'}}_{\infty,1}(\Sigma_t)$ embeds continuously into $L^\infty(\Sigma_t)$ whenever $k\neq k'$. 
	
	\begin{itemize}
		\item Combining \eqref{dtphiboundedint} and \eqref{dtphiboundedext2}, and using the triangle inequality, we have, for every $k$ and every $k'\neq k$,
		\begin{equation}\label{eq:dphibounded.pf.1}
		\|\rd \tphi\|_{L^\infty(\Sigma_t)} \ls \|\rd \tphi\|_{B^{u_k,u_{k'}}_{\infty,1}(\Sigma_t)} \ls \mathcal E.
		\end{equation}
		\item The second inequality in \eqref{eq:phireg.Linfty} implies that for any choice of $k\neq k'$,
		\begin{equation}\label{eq:dphibounded.pf.2}
		\|\rd \phi_{reg}\|_{L^\infty(\Sigma_t)} \ls \|\rd \phi_{reg}\|_{B^{u_k,u_{k'}}_{\infty,1}(\Sigma_t)} \ls \mathcal E.
		\end{equation}
	\end{itemize}
	
	Combining \eqref{eq:dphibounded.pf.1}, \eqref{eq:dphibounded.pf.2}, and using the triangle inequality, we have 
	$$\|\rd\phi\|_{L^\infty(\Sigma_t)} \leq \|\rd \phi_{reg}\|_{L^\infty(\Sigma_t)} + \sum_{k=1}^3  \|\rd \tphi\|_{L^\infty(\Sigma_t)} \ls \mathcal E.\qedhere$$
\end{proof}

\bibliographystyle{plain}
\bibliography{Threewaves}

\end{document}